\documentclass[11pt]{article}
\usepackage{times}
\usepackage[dvips, paper=letterpaper, top=.9in, bottom=.85in, left=.95in, right=.95in]{geometry}
\usepackage{hyperref}
\hypersetup{colorlinks,citecolor=blue}
\usepackage{color}
\usepackage{amsmath,amssymb,amsthm,amsfonts,latexsym,bbm,xspace,graphicx,float,mathtools,dsfont,bm}
\usepackage{mathrsfs}
\usepackage{mathtools}

\usepackage{comment}
\usepackage{tikz}
\usepackage{algorithm,algpseudocode}
\usepackage{verbatim}
\usepackage{nicefrac}
\usepackage{todonotes}

\usepackage{dsfont}
\usepackage{bm}
\usepackage{enumitem}
\newtheorem{remark}{Remark}

\usepackage{cleveref}

\usepackage{enumitem}

\newtheorem{theorem}{Theorem}
\newtheorem{corollary}{Corollary}
\newtheorem{lemma}{Lemma}

\newtheorem{fact}{Fact}
\newtheorem{definition}{Definition}

\newtheorem{observation}{Observation}
\newtheorem{example}{Example}

\newcommand{\ind}{\mathds{1}}

\newcommand{\core}{\textsc{Core}}
\newcommand{\maxcore}{\textsc{Max-Core}}
\newcommand{\icore}{\textsc{Core}^{CZ}}

\newcommand{\rev}{\textsc{Rev}}

\newcommand{\srev}{\textsc{SRev}}

\newcommand{\prev}{\textsc{PRev}}

\newcommand{\estprev}{\widetilde{\textsc{PRev}}}
\newcommand{\vBeta}{\boldsymbol{\beta}}
\newcommand{\vC}{\mathbf{c}}
\newcommand{\vr}{\mathbf{r}}
\newcommand{\Med}{\textsc{Median}}
\newcommand{\opt}{\textsc{OPT}}

\newcommand{\sold}{\textsc{SOLD}}

\newcommand{\supp}{\textsc{Supp}}

\newcommand{\lamij}{\lambda_{ij}(t_{ij},\beta_{ij},\delta_{ij})}
\definecolor{MyGray}{rgb}{0.8,0.8,0.8}

\newcommand{\poly}{\text{poly}}

\def \cT  {{\mathcal{T}}}
\def \cV  {{\mathcal{V}}}
\def \cC  {{\mathcal{C}}}
\def \cD  {{\mathcal{D}}}
\def \cB  {{\mathcal{B}}}
\def \cG  {{\mathcal{G}}}
\def \cM  {{\mathcal{M}}}
\def \cF  {{\mathcal{F}}}
\def \cP  {{\mathcal{P}}}
\def \cE  {{\mathcal{E}}}
\def \cQ  {{\mathcal{Q}}}
\def \cJ  {{\mathcal{J}}}

\def \eps  {\varepsilon}

\def \Mpp  {{\mathcal{M}_{\text{PP}}}}
\def \Mtpt  {{\mathcal{M}_{\text{TPT}}}}
\def \Mopt  {{\mathcal{M}_{\text{OPT}}}}

\newcommand{\optlp}{\textsc{OPT}_{\textsc{LP}}}

\def \M  {{\mathcal{M}}}

\def \DD  {{\mathcal{D}}}

\def \bd  {{\bm{d}}}
\def \bx  {{\bm{x}}}

\def \R   {{\mathbb{R}}}

\def \ptruncated {{\cP^{tr(\eps)}}}
\def \ptruncatedI {{\cP_{i}^{tr(\eps,\cP)}}}
\def \ptruncatedEmp {{\widehat{\cP}^{tr(\eps)}}}
\def \pbox {{\cP^{box(\eps)}}}

\newcommand{\val}{\textsc{Val}}

\newcommand{\rt}{\text{RT}}

\newcommand{\notshow}[1]{{}}
\newenvironment{prevproof}[2]{\noindent {\em {Proof of {#1}~\ref{#2}:}}}{$\Box$\vskip \belowdisplayskip}

\newcommand{\argmax}{\mathop{\arg\max}}
\newcommand{\argmin}{\mathop{\arg\min}}
\newcommand{\E}{\mathop{\mathbb{E}}}

\def \Wconstraint {(1)}
\def \WconstraintNew {(1$'$)}
\def \LambdaMarginalConstraint {(3)}
\def \PiConstraint {(2)}
\def \ReduceDemandConstaint {(6)}
\def \MarginalToGlobalConstraint {(7)}
\def \CompareMarginalConstraint {(4)}
\def \BoundMeanDeltaConstraint {(8)}
\def \BoundSumDeltaConstraint {(9)}
\def \HatLambdaDistributionConstraint {(5)}

\def \WConstraintFP {(3)}
\def \BetaConstraintFP {(1)}
\def \WBetaConstraintFP {(4)}
\def \CiConstraintFP {(2)}

\newcommand{\dem}{\textsc{Dem}}
\newcommand{\adem}{\textsc{ADem}}
\newcommand{\xos}{\textsc{Xos}}
\newcommand{\alg}{\textsc{ALG}}

\newcommand{\Cor}{{\textsc{Corner}}}
\usepackage{commath}

\begin{document}

\title{Computing Simple Mechanisms: Lift-and-Round over Marginal Reduced Forms}
\author{Yang Cai\footnote{Supported by a Sloan Foundation Research Fellowship and the NSF Award CCF-1942583 (CAREER).} \\Computer Science Department\\Yale University\\yang.cai@yale.edu
 \and Argyris Oikonomou\footnote{Supported by a Sloan Foundation Research Fellowship and the NSF Award CCF-1942583 (CAREER).}\\Computer Science Department\\Yale University\\ argyris.oikonomou@yale.edu 
 \and Mingfei Zhao\footnote{Work done in part while the author was studying at Yale University, supported by the NSF Award CCF-1942583 (CAREER).}\\ Google Research\\mingfei@google.com
}

\maketitle

\begin{abstract}
    We study revenue maximization in multi-item multi-bidder auctions under the natural \emph{item-independence} assumption -- a classical problem in Multi-Dimensional Bayesian Mechanism Design. One of the biggest challenges in this area is developing algorithms to compute (approximately) optimal mechanisms that are not brute-force in the size of the bidder type space, which is usually exponential in the number of items in multi-item auctions. Unfortunately, such algorithms were only known for basic settings of our problem when bidders have unit-demand ~\cite{ChawlaHMS10,ChawlaMS15} or additive valuations~\cite{Yao15}. 

In this paper, we significantly improve the previous results and design the first algorithm that runs in time \emph{polynomial in the number of items and the number of bidders} to compute mechanisms that are $O(1)$-approximations to the optimal revenue when bidders have XOS valuations, resolving the open problem raised in~\cite{ChawlaM16,CaiZ17}. Moreover, the computed mechanism has a simple structure: It is either a posted price mechanism or a two-part tariff mechanism. As a corollary of our result, we show how to compute an approximately optimal and simple mechanism efficiently using \emph{only sample access} to the bidders' value distributions. Our algorithm builds on two innovations that allow us to search over the space of mechanisms efficiently: (i) a new type of succinct representation of mechanisms -- the \emph{marginal reduced forms}, and (ii) a novel \emph{Lift-and-Round procedure} that concavifies the problem.

\notshow{

\argyrisnote{
Notes:
	Should we mention that the results hold with high probability?
	Should we mention that for XOS we need an adjusted demand oracle? We haven't named it yet.
}
\mingfeinote{I think we can mention that in the intro.}

\mingfeinote{

A note for the abstract:

Result: We can compute simple mechanisms.

Technique: 
1. Marginal Reduced Form (Define it) + Relaxation

Cai and Zhao proved that for any mechanism, there exists a set of dual parameters such that the revenue is bounded.

We search over the product space of the dual parameters.

2. Multiplicative approximation for the Marginal Reduced Form.

3. Robustify the analysis from Cai and Zhao, to prove that any feasible solution of the LP can be bounded using simple mechanism. 

Make the learnable result from C-D constructive.

}

}

\end{abstract}

\thispagestyle{empty}
\addtocounter{page}{-1}
\newpage


\section{Introduction}\label{sec:intro}
Revenue-maximization in multi-item auctions has been recognized as a central problem in Economics and more recently in Computer Science. 
While Myerson's celebrated work showed that a simple mechanism is optimal in single-item settings~\cite{Myerson81}, the optimal multi-item mechanism is known to be prohibitively complex and notoriously difficult to characterize even in basic settings. Facing the challenge, a major research effort has been dedicated to understanding the computational complexity for finding an approximately revenue-optimal  mechanism in multi-item settings. Despite significant progress
, there is still a substantial gap in our understanding of the problem, 
for example, in the natural and extensively studied \emph{item-independent} setting, first introduced in the influential paper by Chawla, Hartline, and Kleinberg~\cite{ChawlaHK07}.

Formally, the item-independent setting is defined as follows: A seller is selling $m$ heterogeneous items to $n$ bidders, where the $i$-th bidder's type is drawn independently from an $m$-dimensional product distribution $D_i=\bigtimes_{j\in[m]} D_{ij}$.\footnote{$[m]$ denotes $\{1,2,...,m\}$.  $D_{ij}$ is the distribution of bidder $i$'s value for item $j$. The definition is extended to XOS in \Cref{sec:prelim}.} We only understand the computational complexity of finding the revenue-optimal mechanism in the item-independent setting for the two most basic valuations: unit-demand and additive valuations. First, we know that finding an exactly optimal mechanism is computationally intractable even for a single bidder with either unit-demand~\cite{chen2015complexity} or additive valuation~\cite{DaskalakisDT14}. Second, there exists a polynomial time algorithm that computes a mechanism whose revenue is at least a constant fraction of the optimal revenue when bidders have unit-demand~\cite{ChawlaHMS10,ChawlaMS15} or additive valuations~\cite{Yao15}. However, unit-demand and additive valuations are only two extremes within a broader class of value functions known as the constrained additive valuations, where the bidder's value is additive subject to a downward-closed feasibility constraint.\footnote{A bidder has constrained-additive valuation if the bidder's value for a bundle $S$ is defined as $\max_{V\in 2^S\cap \mathcal{I}}\sum_{j\in V} t_j$, where $t_j$ is the bidder's value for item $j$, and $\mathcal{I}$ is a downward-closed set system over the items specifying the feasible bundles. Note that constrained-additive valuations contain familiar valuations such as additive,  
unit-demand, 
or matroid-rank valuations. 
} Furthermore, all constrained additive valuations are contained in an even more general class known as the XOS valuations. Beyond unit-demand and additive valuations, our understanding was limited, and we only knew how to compute an approximately optimal mechanism when bidders are symmetric, i.e., all $D_i$'s are identical~\cite{ChawlaM16,CaiZ17}. Finding a polynomial time algorithm for asymmetric bidders was thus raised as a major open problem in both~\cite{ChawlaM16,CaiZ17}. In this paper, we resolve this~open~problem.

\medskip \noindent \hspace{0.7cm}\begin{minipage}{0.92\columnwidth}
\begin{enumerate}
\item[{\bf Result I:}] For the item-independent setting with (asymmetric) XOS bidders, 
there exists an algorithm that computes a Dominant Strategy Incentive Compatible (DSIC) and Individually Rational (IR) mechanism that achieves at least $c\cdot\opt$ for some absolute constant $c>0$, where $\opt$ is the optimal revenue achievable by any Bayesian Incentive Compatible (BIC) and IR mechanism. Our algorithm has running time polynomial in 
$\sum_{i\in[n],j\in[m]}|\cT_{ij}|$, where $\cT_{ij}$ is the support~of~$D_{ij}$. See \Cref{thm:main XOS-main body} for the formal statement.
\end{enumerate}
\end{minipage}

\noindent 

\vspace{-.1in}
\paragraph{Computing Approximately Optimal Mechanisms under Structured Distributions.} When the bidders' types are drawn from arbitrary distributions, a line of works provide algorithms for finding almost revenue-optimal mechanisms in multi-item settings in time polynomial in the total number of types, i.e., $\sum_{i\in [n]}|\supp(D_i)|$ ($\supp(D_i)$ denotes the support of $D_i$)
~\cite{AlaeiFHHM12,CaiDW12a,CaiDW12b,CaiDW13a,CaiDW13b,cai2021efficient}. 
However, the total number of types could be exponential in the number of items, e.g., there are $\sum_{i\in [n]}\left(\prod_{j\in[m]} |\cT_{ij}|\right)$ types in the item-independent case, making these algorithms unsuitable. For unstructured type distributions, such dependence is unavoidable as even describing the distributions requires time $\Omega\left(\sum_{i\in [n]}|\supp(D_i)|\right)$. What if the type distributions are \emph{structured and permit a more succinct description}, e.g., product measures? Arguably, high-dimensional distributions  that arise in practice (such as bidders' type distributions in multi-item auctions) are rarely arbitrary, as arbitrary high-dimensional distributions cannot be represented or learned efficiently; see e.g.~\cite{Daskalakis18} for a discussion. Indeed, one of the biggest challenges in Bayesian Algorithmic Mechanism Design is designing algorithms to compute (approximately) optimal mechanisms that are not brute-force in the size of the bidder type space when the type distributions are structured. In this paper, we develop computational tools to exploit the item-independence to obtain an exponential speed-up in running time. 



\vspace{-.15in}
\paragraph{Simple vs. Optimal.} An additional feature of our algorithm is that the mechanisms computed have a simple structure. It is either a \emph{posted price mechanism} or a \emph{two-part tariff mechanism}. Given the description of the two mechanisms, it is clear that both of them are~DSIC~and~IR. 

\vspace{-.15in}
\paragraph{Rationed Posted Price Mechanism (RPP).} There is a price $p_{ij}$ for bidder $i$ to purchase item $j$. The bidders arrive in some arbitrary order, and each bidder can purchase \emph{at most one} item among the available ones at the given price.\footnote{Usually, posted price mechanisms do not restrict the maximum number of items a bidder can buy. We consider a rationed version of posted price mechanism to make the computational task easy.} 

\vspace{-.15in}
\paragraph{Two-part Tariff Mechanism (TPT).}  All bidders face the same set of prices $\{p_j\}_{j\in[m]}$. Bidders arrive in some arbitrary order. For each bidder, we show her the available items and the associated price for each item, then ask her to pay an entry fee depending on the bidder's identity and the available items. If the bidder accepts the entry fee, she proceeds to purchase any of the available items at the given prices; if she rejects the entry fee, then she pays nothing and receives nothing.

A recent line of works focus on designing simple and approximately optimal mechanisms
~\cite{ChawlaHK07,ChawlaHMS10,Alaei11,HartN12,KleinbergW12,CaiH13,BILW14,Yao15,RubinsteinW15,CaiDW16,ChawlaM16,CaiZ17}. The main takeaway of these results is that in the item-independent setting, there exists a simple mechanism that achieves a constant fraction of the optimal revenue. The most general setting where such a simple $O(1)$-approximation is known is exactly the setting in \textbf{Result I}, where bidders have XOS valuations~\cite{CaiZ17}. More specifically, \cite{CaiZ17} show that there is a RPP or TPT that achieves a constant fraction of the optimal revenue, however their result is purely existential and does not suggest how to compute these simple mechanisms. Our result makes their existential result constructive.

Finally, combining our result with the learnability result for multi-item auctions in~\cite{CaiD17}, we can extend our algorithm to the case when we only have sample access to the distributions.

\medskip \noindent \hspace{0.7cm}\begin{minipage}{0.92\columnwidth}
\begin{enumerate}
\item[{\bf Result II:}] For constrained-additive bidders, there exists an algorithm that computes a simple, DSIC, and IR mechanism whose revenue is at least $c\cdot\opt-O(\varepsilon \cdot \poly(n,m))$ for some absolute constant $c>0$ in time polynomial in $n$, $m$, and $1/\varepsilon$, given sample access to bidders' type distributions, and assuming each bidder's value for each item lies in $[0,1]$. See \Cref{thm:sample access} for the formal~statement.
\end{enumerate}
\end{minipage}
\vspace{-.1in}



\subsection{Our Approach and Techniques}\label{subsec:technique}
Our main technical contribution is \emph{a novel relaxation of the revenue optimization problem that can be solved approximately in polynomial time} and \emph{an accompanying rounding scheme that converts the solution to a simple and approximately optimal mechanism}.\footnote{An influential framework known as the ex-ante relaxation has been widely used in Mechanism Design, but is insufficient for our problem. See \Cref{sec:exante} for a detailed discussion.} Our first step is to replace the objective of revenue with a duality-based benchmark of the revenue proposed in~\cite{CaiZ17}. One can view the new objective as maximizing the virtual welfare, similar to Myerson's elegant solution for the single-item case. The main difference is that, while one can use \emph{a fixed set of virtual valuations for any allocation} in the single-item case, due to the multi-dimensionality of our problem, \emph{the virtual valuations must depend on the allocation, causing the virtual welfare to be a non-concave function in the allocation}. In this paper, we develop algorithmic tools to concavify and approximately optimize the virtual welfare maximization problem. We believe our techniques will be useful to address other similar challenges in Multi-Dimensional Mechanism Design.

More specifically, for every BIC and IR mechanism $\cM$ with allocation rule $\sigma$ and payment rule $p$, one can choose a set of \emph{dual parameters}  $\theta(\sigma)$ based on $\sigma$ to construct an upper bound $U(\sigma,\theta(\sigma))$ for the revenue of $\cM$. We refer to $\theta$ as the dual parameters because $\theta$ corresponds to a set of ``canonical'' dual variables, which can be used to derive the virtual valuations via the Cai-Devanur-Weinberg duality framework~\cite{CaiDW16}. The upper bound  $U(\sigma,\theta(\sigma))$ is then simply the corresponding virtual welfare. The computational problem is to find an allocation $\sigma$ that (approximately) maximizes $U(\sigma,\theta(\sigma))$. With such a $\sigma$, we could use the result in~\cite{CaiZ17} to convert it to a simple and approximately optimal mechanism. Unfortunately, the function $U(\sigma,\theta(\sigma))$ is highly non-concave in $\sigma$,\footnote{See \Cref{sec:core non-concavity} for an example of the non-concavity of the function. } and thus hard to maximize efficiently. See~\Cref{sec:tour to LP} for a detailed~discussion.

\vspace{-.15in}
\paragraph{LP Relaxation via Lifting.} We further relax our objective, i.e., $U(\sigma,\theta(\sigma))$, to obtain a computationally tractable problem. One specific difficulty in optimizing $U(\sigma,\theta(\sigma))$ comes from the fact that $\theta(\sigma)$ is highly non-linear in $\sigma$. We address this difficulty in two steps. In the first step of our relaxation, we flip the dependence of $\sigma$ and $\theta$ by relaxing the problem to the following two-stage optimization problem (\Cref{fig:LP-fixed-parameter}): 

\vspace{.05in}
\textbf{- Stage I:} Maximize $H(\theta)$ subject to some constraints. $H(\theta)$ is the optimal value of the Stage II~problem.

\vspace{.05in}

\textbf{- Stage II:} Maximize an LP over $\sigma$ with $\theta$-dependent constraints.
\vspace{.05in}

\noindent This makes the problem much more structured and significantly disentangles the complex dependence between $\sigma$ and $\theta$. Yet we still do not know how to solve it efficiently. In the second step of our relaxation, we merge the two-stage optimization into a single LP. In particular, we \emph{lift the problem to a higher dimensional space and optimize over joint distributions of the allocation $\sigma$ and the dual parameters $\theta$ via an LP (\Cref{fig:bigLP}).}  Since the number of dual parameters is already exponential in the number of bidders and the number of items, it is too expensive to represent such a joint distribution explicitly. We show it is unnecessary to search over all joint distributions. By leveraging the independence across bidders and items, it suffices for us to consider a set of succinctly representable distributions -- the ones whose marginals over the dual parameters are product measures. See \Cref{sec:tour to LP} for a more detailed discussion on the development of our relaxation.

\vspace{-.15in}
\paragraph{``Rounding'' any Feasible Solution to a Simple Mechanism.}Can we still approximate the optimal solution of the LP relaxation using a simple mechanism? Unfortunately, the result from~\cite{CaiZ17} no longer applies. 
We provide a \emph{generalization of~\cite{CaiZ17}}, that is, given any feasible solution of our LP relaxation, we can construct in polynomial time a simple mechanism whose revenue is at least a constant fraction of the objective value of the feasible solution (\Cref{thm:bounding-lp-simple-mech}). Our proof (in \Cref{sec:proof of bounding LP with simple mech})
provides several novel ideas to handle the new challenges due to the relaxation, which may be of independent~interest.
\vspace{-.15in}

\paragraph{Marginal Reduced Forms.} We deliberately postpone the discussion on  how we represent the allocation of a mechanism until now. A widely used succinct representation a mechanism $\cM$ is known as the reduced form or the interim allocation rule:
$\{r_{ij}(t_i)\}_{i\in[n],j\in[m], t_{i}\in\bigtimes_{j\in[m]}\cT_{ij}}$ where $r_{ij}(t_i)$ is the probability for bidder $i$ to receive item $j$ when \emph{her type is $t_i=(t_{i1},\ldots, t_{im})$}~\cite{CaiDW12a}. Despite being more succinct than the ex-post allocation rule, the reduced form is still too expensive to store in our setting, as its size is exponential in $m$. 
A key innovation in our relaxation is the introduction of an even more succinct representation -- the \emph{marginal reduced forms} and a \emph{multiplicative approximation} to the polytope of all feasible marginal reduced forms. Although this is a natural concept, to the best of our knowledge, we are the first to introduce and make use of it. 
The marginal reduced form is represented as : $\left\{w_{ij}(t_{ij})\right\}_{i\in[n],j\in[m], t_{ij}\in\cT_{ij}}$, where $w_{ij}(t_{ij})$ is the probability for bidder $i$ to receive item $j$ in $\cM$ and \emph{her value for item $j$ is $t_{ij}$}.\footnote{We refer to $\{w_{ij}(t_{ij})\}_{i\in[n],j\in[m], t_{ij}\in\cT_{ij}}$ as the marginal reduced forms as they are the marginals of the reduced forms multiplied by the probability that $t_{ij}$ is bidder $i$'s value for item $j$, i.e., $\frac{w_{ij}(t_{ij})}{\Pr_{D_{ij}}[t_{ij}]}=\E_{t_{i,-j}\sim \bigtimes_{\ell\neq j}D_{i\ell}}\left[r_{ij}\left(t_{ij},t_{i,-j}\right)\right]$.} Importantly, the size of a marginal reduced form is polynomial in the input size of our problem. As our LP relaxation uses marginal reduced forms as decision variables, it is crucial for us to be able to optimize over the polytope $P$ that contains all \emph{feasible marginal reduced forms}. To the best of our knowledge, $P$ does not have a succinct explicit description or an efficient separation oracle. To overcome the obstacle, we provide \emph{an efficient separation oracle for a different polytope $Q$ that is a multiplicative approximation to $P$, i.e., $c\cdot P\subseteq Q\subseteq P$ for some absolute constant $c\in(0,1)$} (\Cref{thm:multiplicative approx for constraint additive-main body}). Using the separation oracle for $Q$, we can find a $c$-approximation to the optimum of the LP relaxation efficiently. Note that a sampling technique
was developed in~\cite{CaiDW12b} to approximate the polytope of feasible reduced forms. However, their technique only provides an ``additive approximation to the polytope'', which is insufficient for our purpose. Indeed, our multiplicative approximation holds for a wide class of polytopes that frequently appear in Mechanism Design (\Cref{thm:special case of multiplicative approx}). 
We believe our technique has further applications, for example, to convert the additive FPRAS of Cai-Daskalakis-Weinberg~\cite{CaiDW12a,CaiDW12b,CaiDW13a,CaiDW13b} to a multiplicative FPRAS.
\vspace{-.1in}

\subsection{Related Work}\label{sec:related work}
\paragraph{Simple vs. Optimal.} 
We provide an algorithm for the most general setting where an $O(1)$-approximation to the optimal revenue is known using simple mechanisms. It is worth mentioning that a recent result by D\"{u}tting et al.~\cite{DuttingKL20} shows that simple mechanisms can be used to obtain a $O(\log\log m)$-approximation to the optimal revenue even when the bidders have subadditive valuations. We leave it as an interesting open problem to extend our algorithm to bidders with subadditive valuations.

\vspace{-.15in}
\paragraph{$(1-\varepsilon)$-Approximation in Item-Independent Settings.} We focus on constant factor approximations for general valuations. For more specialized valuations, e.g., unit-demand/additive, there are several interesting results for finding $(1-\varepsilon)$-approximation to the ``optimal mechanism''. For example, PTASes are known if we restrict our attention to finding the optimal simple mechanism for a single bidder, e.g., item-pricing~\cite{CaiD11b} or partition mechanisms~\cite{rubinstein2016computational}. For multiple bidders, PTASes are known for bidders with additive valuations under extra assumptions on  distributions (such as i.i.d., MHR,\footnote{That is, $f_{ij}(v)/1-F_{ij}(v)$ is monotone non-decreasing (MHR) for each $i,j$, where $f_{ij}$ is the pdf and $F_{ij}$ is the cdf.} etc.)~\cite{DaskalakisW12,CaiH13}.  The only result that does not require simplicity of the mechanism or extra assumptions on the distribution is \cite{kothari2019approximation}, but their algorithm is only a quasi-polynomial time approximation scheme (QPTAS) and computes a $(1-\varepsilon)$-approximation to the optimal revenue for a single unit-demand bidder.
\vspace{-.15in}

\paragraph{Structured Distributions beyond Item-Independence.} When the type distributions can be represented as other structured distributions such as Bayesian networks, Markov Random Fields, or Topic Models, recent results show how to utilize the structure to improve the learnability, approximability, and communication complexity of multi-item auctions~\cite{BrustleCD20,CaiO21,cai2021recommender}.  We believe that tools developed in this work would be useful to obtain similar improvement in terms of the computational complexity for computing approximately optimal mechanisms for structured distributions beyond item-independence. 

\vspace{-.15in}

\section{Preliminaries}\label{sec:prelim}


We focus on revenue maximization in the combinatorial auction with $n$ independent bidders and $m$ heterogeneous items. We denote bidder $i$'s type $t_i$ as $\{t_{ij}\}_{j\in [m]}$, where $t_{ij}$ is bidder $i$'s private information about item $j$. For each $i$, $j$, we assume $t_{ij}$ is drawn independently from the distribution $D_{ij}$. 
Let $D_i=\times_{j=1}^m D_{ij}$ be the distribution of bidder $i$'s type and $D=\times_{i=1}^n D_i$ be the distribution of the type profile.  
We only consider discrete distributions in this paper. We use $\cT_{ij}$ (or $\cT_i, \cT$) and $f_{ij}$ (or $f_i, f$) to denote the support and the probability mass function of $D_{ij}$ (or $D_i, D$). For notational convenience, we let  $t_{-i}$ to be the types of all bidders except $i$ and $t_{<i}$ (or $t_{\leq i})$ to be the types of the first $i-1$ (or $i$) bidders. Similarly, we define $D_{-i}$, $\cT_{-i}$
  and $f_{-i}$ for the corresponding distribution, support of the distribution, and probability mass function. 

\vspace{-.15in}
\paragraph{Valuation Functions.} For every bidder $i$, denote her valuation function as $v_i(\cdot,\cdot):\cT_i\times 2^{[m]}\to \mathbb{R}_+$. For every $t_i\in \cT_i,S\subseteq 2^{[m]}$, $v_i(t_i,S)$ is bidder $i$'s value for receiving a set $S$ of items, when her type is $t_i$. In the paper, we are interested in constrained-additive and XOS valuations. For every $i\in [n]$, bidder $i$'s valuation $v_i(\cdot,\cdot)$ is \emph{constrained-additive} if the bidder can receive a set of items subject to some downward-closed feasibility constraint $\cF_i$. Formally, $v_i(t_i,S)=\max_{R \in 2^S\cap \cF}\sum_{j\in R} t_{ij}$ for every type $t_i$ and set $S$. It contains classic valuations such as additive ($\cF_i=2^{[m]}$) and unit-demand ($\cF_i=\cup_{j\in[m]}\{j\}$). For constrained-additive valuations, we use $t_{ij}$ to denote bidder $i$'s value for item $j$. For every $i\in [n]$, bidder $i$'s valuation $v_i(\cdot,\cdot)$ is \emph{XOS (or fractionally-subadditive)} if each $t_{ij}$ represents a set of $K$ non-negative numbers $\{\alpha_{ij}^{(k)}(t_{ij})\}_{k\in[K]}$, for some integer $K$, and $v_i(t_i,S) = \max_{k\in[K]} \sum_{j\in S}\alpha_{ij}^{(k)}(t_{ij})$, for every type $t_i$ and set $S$. We denote by $V_{ij}(t_i)=v_i(t_i,\{j\})$ the value for a single item $j$. Since the value of the bidder for item $j$ only depends on $t_{ij}$, we denote $V_{ij}(t_{ij})$ as~the~singleton~value.




\notshow{
\begin{definition}
For every bidder $i$, we consider 
\begin{itemize}
\item \textbf{Constrained Additive:} Given a set of downward closed feasibility constraints $\cF$,
$v_i(t_i,S)=\max_{R\subseteq S,R \in \cF}\sum_{j\in R} t_{ij}$,
where we abuse the notation and denote by $v_i(t_{ij},\{j\})=t_{ij}$.
When $\cF=\cup_{j\in[m]}\{j\}$, that is a bidder is interested in buying at most one item, the valuation is called unit demand. When $\cF=2^{[m]}$, the valuation is called additive.
\item \textbf{XOS:} For a finite number $K$, each type represents a set of $K$ non-negative numbers $t_{ij}=\{a_{ij}^{(k)}(t_{ij})\}_{k\in[K]}$.
Moreover $v_i(t,S) = \max_{k\in[K]} \sum_{j\in S}a_{ij}^{(k)}(t_{ij})$ and we denote by $V_{ij}(t)=v(t,\{j\})$. Since the value of the agent for the $j$-th item only depends on $t_{ij}$, we overuse notation and also denote $V_{ij}(t_{ij})=v(t,\{j\})$.
\end{itemize}
\end{definition}
}

\vspace{-.15in}
\paragraph{Mechanisms.} A mechanism $\cM$ can be described as a tuple $(\sigma,p)$, where $\sigma$ is the \emph{interim allocation} rule of $\cM$ and $p$ stands for the payment rule. Formally, for every bidder $i$, type $t_i$ and set $S$, $\sigma_{iS}(t_i)$ is the interim probability that bidder $i$ with type $t_i$ receives exact bundle $S$.  
We use standard concepts of BIC, DSIC and IR for mechanisms. See~\Cref{appx:prelim} for the formal definitions. For any BIC and IR mechanism $\cM$, denote $\rev(\cM)$ the revenue of $\cM$. Denote $\opt$ the optimal revenue among all BIC and IR mechanisms.
Throughout this paper, the two classes of simple mechanisms we focus on are \emph{rationed posted price} (RPP) mechanisms and \emph{two-part tariff} (TPT) mechanisms, which are both described in \Cref{sec:intro}. We denote $\prev$ the optimum revenue achievable among all RPP mechanisms.  
\notshow{

Relevant to our work is to measure the bit complexity of the problems that we are facing.
In the following definitions we show how we measure the bit complexity of the problem when the bidder's type are Constraint Additive or XOS.
\begin{definition}[Bit Complexity of Instance]
Given a product distribution $D=\prod_{i\in[n]}$ over types for constraint additive or XOS valuation and let $T=\sum_{i\in[n]}\sum_{j\in[m]}\sum_{t_{ij}\in\cT_{ij}}|\cT_{ij}|$.
\begin{itemize}
\item Denote as $b_f$ the bit complexity of elements in $\{f_{ij}(t_{ij})\}_{i\in[n],j\in[m],t_{ij}\in\cT_{ij}}$.
\item If the valuations of the bidders are constraint additive, denote as $b_v$ the bit complexity of elements in $\{t_{ij}\}_{i\in[n],j\in[m],t_{ij}\in\cT_{ij}}$.
\item If the valuations of the bidders are XOS, denote as $b_v$ the bit complexity of elements in $\{\alpha^{(k)}_{ij}\}_{i\in[n],j\in[m],k\in[K],t_{ij}\in\cT_{ij}}$.
\end{itemize}
We refer to the value $\max(b_l,b_f)$ as the bit complexity of the instance.
\end{definition}}
\vspace{-.15in}
\paragraph{Access to the Bidders' Valuations.} We define several ways to access a bidder's valuation. 
\begin{definition}[Value and Demand Oracle]\label{def:value and demand oracle}
A \emph{value oracle} for a valuation $v(\cdot,\cdot)$ takes a type $t$ and a set of items $S\subseteq[m]$ as input, and returns the bidder's value $v(t,S)$ for the bundle $S$. A \emph{demand oracle} for a valuation $v(\cdot,\cdot)$ takes a type $t
$ and a  collection of non-negative prices $\{p_j\}_{j\in[m]}$ as input, and returns a utility-maximizing bundle, i.e. $S^*\in \arg\max_{S\subseteq [m]} \left(v(t,S) - \sum_{j\in S}p_j\right)$. In this paper, we use $\dem_i(\cdot,\cdot)$ to denote the demand oracle for bidder $i$'s valuation  $v_i(\cdot,\cdot)$.
\end{definition}
For constrained-additive valuations, our result only requires query access to a value oracle and a \emph{demand oracle} for every bidder $i$'s valuation $v_i(\cdot,\cdot)$. For XOS valuations, we need a stronger demand oracle that allows ``scaled types'' as input. We refer to the stronger oracle as the \emph{adjustable demand oracle}. 

\begin{definition}[Adjustable Demand Oracle]\label{def:adjustable demand oracle}
An \emph{adjustable demand oracle} for bidder $i$'s XOS valuation $v_i(\cdot,\cdot)$
takes a type $t$, a collection of non-negative coefficients  $\{b_j\}_{j\in[m]}$, and a collection of non-negative prices $\{p_j\}_{j\in[m]}$ as input. For every item $j$, $b_j$ is a scaling factor for $t_{ij}$, meaning that each of the $K$ numbers $\{\alpha_{ij}^{(k)}(t_{ij})\}_{k\in[K]}$, i.e. the contribution of item $j$ under each additive function, is multiplied by $b_j$. The oracle outputs a favorite bundle $S^*$ with respect to the adjusted contributions and the prices $\{p_j\}_{j\in[m]}$, as well as the additive function $\{\alpha_{ij}^{(k^*)}(t_{ij})\}_{j\in [m]}$ for some $k^*\in [K]$ that achieves the highest value on $S^*$. 
Formally, 
$
(S^*,k^*) \in \argmax_{S\subseteq [m],k\in[K]} \left\{\sum_{j\in S}b_j \alpha_{ij}^{(k)}(t_{ij}) - \sum_{j\in S}p_j\right\}.
$
We use  $\adem_i(\cdot,\cdot,\cdot)$ to denote the adjustable demand oracle for bidder $i$'s XOS valuation $v_i(\cdot,\cdot)$.
\end{definition}

\vspace{-.05in}
The adjustable demand oracle can be viewed as a generalization of the demand oracle for XOS valuations. In the above definition, if every coefficient $b_j$ is 1, then the adjustable demand oracle outputs the utility-maximizing bundle $S^*$ (as in the demand oracle) and the additive function that achieves the value for this set. For general $b_j$'s, the adjustable demand oracle scales item $j$'s contribution to bidder $i$'s value by a $b_j$ factor.
The output bundle $S^*$ maximizes the adjusted utility.\footnote{ {Note that for every collection of scaling factors, the query to the adjusted demand oracle is simply a demand query for a different XOS valuation. If all additive functions of $t_i$ are explicitly given, then the adjusted demand oracle can be simulated in~time~$O(mK)$.}}

\vspace{-.05in}
\begin{definition}[Bit Complexity of an Instance]\label{def:bit complexity}
Given any instance of our problem represented as the tuple $(\cT,D,v=\{v_i(\cdot,\cdot)\}_{i\in [n]})$,  
Denote as $b_f$ the bit complexity of elements in $\{f_{ij}(t_{ij})\}_{i\in[n],j\in[m],t_{ij}\in\cT_{ij}}$.
For constrained-additive valuations, denote as $b_v$ the bit complexity of elements in $\{t_{ij}\}_{i\in[n],j\in[m],t_{ij}\in\cT_{ij}}$. For XOS valuations, denote as $b_v$ the bit complexity of elements in $\{\alpha^{(k)}_{ij}(t_{ij})\}_{i\in[n],j\in[m],t_{ij}\in\cT_{ij}, k\in[K]}$.
We define the value $\max(b_v,b_f)$ to be the bit complexity of the instance.
\end{definition}


\notshow{
\begin{definition}[\todo{Is demand oracle folklore?} Demand Oracle for Constraint Additive Valuation $v_i(\cdot,\cdot)$]
Consider a set of non-negative prices $\{p_j\}_{j\in[m]}$ and a type $t_i\in \cT_i$ for bidder $i$,
a demand oracle $\dem(\cdot,\cdot)$ for a constraint additive valuation $v_i(\cdot,\cdot)$,
takes an input the bidder's type, the price vector and returns the utility maximizing bundle.
More formally, $\dem(t_i,\{p_j\}_{j\in[m]})$ outputs an a subset $S\subseteq [m] $ such that
$$
S\in \arg\max_{S\subseteq [m]} v_i(t_i,S) - \sum_{j\in S}p_j
$$
\end{definition}

However,
when we try to compute a simple mechanism for bidders with XOS valuations,
we will need to have access to a demand oracle that can alter the weight that the agent can have for different items.
We call this oracle the adjustable demand oracle and it is defined below.
}

\notshow{

\section{XOS Valuation}

\begin{definition}
Let $\DD_i$ be a product distribution over types of the $i$-th agent over $m$ heterogeneous items.
We denote by $T_i$ the support of $\DD_i$. The valuation of the bidder for the items is XOS.
We consider only anonymous posted price mechanisms with entry fee.

For each bidder,
we assume access to a demand oracle $\dem_i(\cdot,\cdot): T \times \R^m\rightarrow \mathbb{R}$,
that is, for a type $t \in T$ and a set of prices the item $\{p_j\}_{j \in[m]}$,
$\dem_i(t,\{p_j\}_{j\in[m]})$ is the maximum utility of the $i$-th agent with realized type $t$ that she can get under those prices.
Moreover we assume sample access to the agents distribution. \argyrisnote{I think it would be usefull to specify the model we are working on.
I don't think we lose anything by considering sample access to the distribution.}
\yangnote{Yang: We don't have multiplicative approximation if we assume only sample access. I will suggest work with the explicit model first, then in the end have a section discuss how our results can be used to derive results under the sample accesss model.}
\end{definition}

Let $\M(\{\sigma_{i,S}(\cdot)\}_{i\in[n], S \subseteq 2^{[m]}})$ be any BIC mechanism,
where by $\sigma_{i,S}(t_i)$ is the probability that the $i$-th agent with type $t_i \in T_i$ receives exactly bundle $S$.
Let $\epsilon = \sup_{\substack{i\in[n],j\in[m]\\t_{i,j} \in T_{i,j}}}f_i(t_{i,j})$.
Note that given supports $\{T_{i,j}\}_{i\in[n],j\in[m]}$ such that $\epsilon = \sup_{\substack{i\in[n],j\in[m]\\t_{i,j} \in T_{i,j}}}f_i(t_{i,j})$,
we can construct $\{T'_{i,j}\}_{i\in[n],j\in[m]}$ such that $\epsilon' = \sup_{\substack{i\in[n],j\in[m]\\t_{i,j} \in T'_{i,j}}}f_i(t_{i,j})$ and $|T'_{i,j}| \leq \left\lceil\frac{\epsilon}{\epsilon'}\right\rceil|T_{i,j}|$ by subdividing each $t_{i,j}\in T_{i,j}$ to $\{t_{i,j}^{(i)}\}_{i \in \left[\left\lceil\frac{\epsilon}{\epsilon'}\right\rceil\right]}\subseteq T_{i,j}'$. \yangnote{Yang: Could you make the notation consistent? Sometimes it's $T_{i,j}$ and sometimes it's $\cT_{ij}$. Also, what does $t_{i,j}^{(i)}$ mean?}





Let $W_{i}$ be the polytope of feasible welfares normalised to $[0,1]$ for the $i$-th agent.
That is,
there exists an allocation $\pi_i$, such that $w_{i,j}(t_{i,j})V_{i,j}(t_{i,j}) = \sum_{\substack{t_i'\in T_i \\t_{i,j}' = t_{i,j}}} f_i(t_{i,-j})\pi_{i}(t_{i},S)p_i(t_i, S,j)$,
where $p_i(t_i,S,j)$ is the supporting price for item $j$, when the $i$-th agent has type $t_i$ and receives set $S$. \yangnote{Yang: the definition of $W_i$ is different from later sections.}

}
\section{Linear Program Relaxation via Lifting}\label{sec:program}
In this section, we present the linear program relaxation for computing an approximately optimal simple mechanism. The main result of our paper is as follows:

\begin{theorem}\label{thm:main XOS-main body}
Let $T=\sum_{i,j}|\cT_{ij}|$ and $b$ be the bit complexity of the problem instance (\Cref{def:bit complexity}). 
For any $\delta>0$, there exists an algorithm that computes a RPP mechanism or a TPT mechanism, such that the revenue of the mechanism is at least $c\cdot \opt$ for some absolute constant $c>0$ with probability $1-\delta-\frac{2}{nm}$.
For constrained-additive valuations, our algorithm assumes query access to a value oracle and a demand oracle of bidders' valuations. For XOS valuations, our algorithm assumes query access to a value oracle and an adjustable demand oracle. The algorithm has running time $\poly(n,m,T,b,\log (1/\delta))$.
\end{theorem}

\vspace{-.05in}
{For any matroid-rank valuation, i.e., the downward-closed feasibility constraint is a matroid, the value and demand oracle can be simulated in polynomial time using greedy algorithms. For more general constraints, it is standard to assume access to the value and demand oracle. We also show that the adjustable demand oracle (rather than a demand oracle) is necessary to obtain our XOS result. In \Cref{lem:example_adjustable_oracle}, we prove that (even an approximation of) $\adem_i$ can not be implemented in polynomial time, given access to the value oracle, demand oracle, and XOS oracle.
}

As most of the technical barriers already exist in the constrained-additive case, for exposition purposes, we focus on constrained-additive valuations in the main body (unless explicitly stated).\footnote{The linear program for XOS valuations can be found in \Cref{fig:XOSLP} in \Cref{sec:program-XOS}.} Before stating our LP, we first provide a brief recap of the existential result by Cai and Zhao~\cite{CaiZ17} summarized in \Cref{lem:caiz17-constrained-additive}.\footnote{The statement is for constrained-additive bidders. See \Cref{appx_cai-zhao} for the statement for XOS bidders.}

\vspace{-.05in}

\begin{definition}\label{def:core-constrained-additive}
For any $i\in [n],j\in [m]$, 
and any feasible\footnote{
For constrained-additive bidders, an interim allocation $\sigma$ is feasible if it can be implemented by a mechanism whose allocation rule always respects all bidders’ feasibility constraints. It is without loss of generality to consider feasible interim allocations.} interim allocation $\sigma$, and non-negative numbers {$\tilde{\vBeta}=\{\tilde{\beta}_{ij}\in \cT_{ij}\}_{i\in[n],j\in[m]}$, $\vC=\{c_i\}_{i\in[n]}$} and $\vr=\{r_{ij}\}_{i\in [n],j\in [m]}\in [0,1]^{nm}$ (referred to as the dual parameters), 
let $\core(\sigma,\tilde{\vBeta},\vC,\vr)$ be the welfare under allocation $\sigma$ truncated at $\tilde{\beta}_{ij}+c_i$ for every $i,j$.
~Formally, $$
\core(\sigma,\tilde{\vBeta},\vC,\vr)=\sum_i\sum_{t_i}f_i(t_i)\cdot \sum_{S\subseteq [m]}\sigma_{iS}(t_i)\sum_{j\in S}t_{ij}\cdot \left(\ind[t_{ij}< \tilde{\beta}_{ij}+c_i]+r_{ij}\cdot \ind[t_{ij}= \tilde{\beta}_{ij}+c_i] \right).$$
\notshow{
if the bidders have constrained-additive valuations, and
$$\core(\sigma,\tilde{\vBeta},\vC,\vr)=\sum_i\sum_{t_i}f_i(t_i)\cdot \sum_{S\subseteq [m]}\sigma_{iS}(t_i)\sum_{j\in S}\gamma_{ij}^S(t_i)\cdot \left(\ind[V_{ij}(t_{ij})< \tilde{\beta}_{ij}+c_i]+r_{ij}\ind[V_{ij}(t_{ij})= \tilde{\beta}_{ij}+c_i] \right)$$ if the bidders have XOS valuations.
Here $\gamma_{ij}^S(t_i)=\alpha_{ij}^{k^*(t_i,S)}(t_{ij})$, where $\displaystyle k^*(t_i,S)=\arg\max_{k\in[K]}\big(\sum_{j\in S}\alpha^k_{ij}(t_{ij})\big)$.
}
\end{definition}

\begin{lemma}\cite{CaiZ17}\label{lem:caiz17-constrained-additive}
Given any BIC and IR mechanism $\cM$ with interim allocation $\sigma$,  
where $\sigma_{iS}(t_i)$ is the interim probability for bidder $i$ to receive exactly bundle $S$ when her type is $t_i$, there exist 
non-negative numbers {$\tilde{\vBeta}^{(\sigma)}=\{\tilde{\beta}_{ij}^{(\sigma)}\in \cT_{ij}\}_{i\in[n],j\in[m]}$, $\vC^{(\sigma)}=\{c_i^{(\sigma)}\}_{i\in[n]}$} 
 and $\vr^{(\sigma)}\in [0,1]^{nm}$ that satisfy\footnote{\cite{CaiZ17} provides an explicit way to calculate $\tilde{\vBeta}^{(\sigma)}, \vC^{(\sigma)}, \vr^{(\sigma)}$. We only include the crucial properties of these parameters here.} 
\begin{enumerate}
    \item $\sum_{i\in[n]} \left(\Pr_{t_{ij}}[t_{ij}>\tilde\beta_{ij}^{(\sigma)}]+r_{ij}^{(\sigma)}\cdot \Pr_{t_{ij}}[t_{ij}=\tilde\beta_{ij}^{(\sigma)}]\right)\leq \frac{1}{2},\forall j$,
    \item $\frac{1}{2}\cdot\sum_{t_i\in \cT_i}f_i(t_i)\cdot \sum_{S:j\in S}\sigma_{iS}(t_i)\leq \Pr_{t_{ij}}[t_{ij}>\tilde\beta_{ij}^{(\sigma)}]+r_{ij}^{(\sigma)}\cdot \Pr_{t_{ij}}[t_{ij}=\tilde\beta_{ij}^{(\sigma)}],\forall i,j$,
    \item $\sum_{i\in [n]} c_i^{(\sigma)}\leq 8\cdot \prev$,
\end{enumerate}
and the corresponding $\core(\sigma,\tilde{\vBeta}^{(\sigma)},\vC^{(\sigma)},\vr^{(\sigma)})$ satisfies the following inequalities:
\begin{enumerate}\addtocounter{enumi}{3}
    \item $\rev(\cM)\leq 28\cdot \prev+4\cdot\core(\sigma,\tilde{\vBeta}^{(\sigma)},\vC^{(\sigma)},\vr^{(\sigma)})$,
    \item $\core(\sigma,\tilde{\vBeta}^{(\sigma)},\vC^{(\sigma)},\vr^{(\sigma)})\leq 64\cdot \prev + 8\cdot \rev(\cM_1^{(\sigma)})$, where $\cM_1^{(\sigma)}$ is some TPT mechanism.
\end{enumerate}
\end{lemma}

\begin{remark}\label{remark:tie-breaking}
For continuous type distributions, there exists $\tilde{\vBeta}^{(\sigma)}$ that satisfy both Property 1 and 2 of \Cref{lem:caiz17-constrained-additive} with $r_{ij}^{(\sigma)}=1,\forall i,j$ for every $\sigma$. For discrete distributions, such a $\tilde{\vBeta}^{(\sigma)}$ may not exist. This is simply a tie-breaking issue, and the role of $\vr^{(\sigma)}$ is to fix it. Roughly speaking, $r_{ij}^{(\sigma)}$ is the probability that bidder $i$ wins item $j$, when she is indifferent between purchasing or not. Readers can treat $\vr^{(\sigma)}$ as the all-one vector to get the intuition behind our approach.
\end{remark}


By combining Property 4 and 5 of \Cref{lem:caiz17-constrained-additive}, Cai and Zhao~\cite{CaiZ17} proved that the revenue of any BIC, IR mechanism $\cM$ is bounded by a constant number of $\prev$ and the revenue of some TPT mechanism.
Recall that $\prev$ is the optimal revenue achieved by an RPP mechanism, which is exactly the Sequential Posted Price mechanism if we restrict the bidders' valuations to unit-demand. Thus we can compute a set of posted prices that approximates $\prev$ by Chawla et al.~\cite{ChawlaHMS09}.  

\subsection{Tour to Our Relaxation}\label{sec:tour to LP}
To facilitate our discussion about the key components and the intuition behind the relaxation, we present the development of our relaxation and along the way examine several failed attempts. In~\Cref{thm:bounding-lp-simple-mech}, we show that the optimal solution of the relaxed problem can indeed be approximated by simple mechanisms. Due to space limitations, we do not include details on the approximation analysis in this section, but focus on our intuition behind each step of our relaxation. Interested readers can find the proof of~\Cref{thm:bounding-lp-simple-mech} in~\Cref{sec:proof of bounding LP with simple mech}. We also assume $r_{ij}$ to be $1$ for every $i$ and~$j$ to keep the notation light. 

\vspace{-.18in}
\paragraph{Step 0: Replace Revenue with the Duality-Based Benchmark.}
Instead of optimizing the revenue, we optimize the upper bound of revenue. As guaranteed by \Cref{lem:caiz17-constrained-additive}, for any BIC and IR mechanism $\cM=(\sigma,p)$, its revenue is upper bounded by $O(\prev+\core(\sigma,\theta(\sigma)))$, where we use $\theta(\sigma)$ to denote the set of  dual parameters $(\tilde{\vBeta}^{(\sigma)},\vC^{(\sigma)})$ guaranteed to exist by~\Cref{lem:caiz17-constrained-additive}. 
Since we can approximate $\prev$, it suffices to first approximately maximize $\core(\sigma,\theta(\sigma))$ over all feasible interim allocations $\sigma$, then compute the TPT in \Cref{lem:caiz17-constrained-additive} based on the computed $\sigma$. $\core(\sigma,\theta(\sigma))$ is the truncated welfare, but the truncation depends on $\sigma$ in a complex way, causing the function to be highly non-concave in $\sigma$ (\Cref{ex:non-concave interim}).

\vspace{-.15in}
\paragraph{Step 1: Two-Stage Optimization.}
To overcome the barrier mentioned above, we consider a two-stage optimization problem (\Cref{fig:LP-fixed-parameter}) by \emph{switching the order of dependence} between the interim allocation $\sigma$ and dual parameters $\theta=(\vBeta,\vC)$. In Stage I, we optimize some function $H$ over the dual parameters $\theta=(\vBeta,\vC)$, where $H(\vBeta,\vC)$ is the optimum of the Stage II problem for every fixed set of parameters $(\vBeta,\vC)$. Constraint {\BetaConstraintFP} and {\CiConstraintFP} in the Stage I problem are due to Property 1 and 3 of \Cref{lem:caiz17-constrained-additive} respectively. In Stage II, for any fixed set of parameters $\theta=(\vBeta,\vC)$, we optimize $\core(\sigma,\theta)$ over all feasible $\sigma$ such that the tuple $(\sigma,\vBeta,\vC)$ satisfy Property 1, 2, and 3 of \Cref{lem:caiz17-constrained-additive}. We choose  the interim allocation $\sigma$ as the variables, $\core(\sigma,\vBeta,\vC)$ as the objective, and include Constraint {\WBetaConstraintFP}, which corresponds to Property 2 of \Cref{lem:caiz17-constrained-additive}. Why is the two-stage optimization a relaxation? For any interim allocation $\sigma$, (i) the corresponding set of dual parameters $\theta(\sigma)$ is a feasible solution of the first-stage optimization problem, and (ii) $\sigma$ is feasible in the second-stage optimization w.r.t. $\theta(\sigma)$, so $(\theta(\sigma),\sigma)$ is a feasible solution of the two-stage optimization problem.
\vspace{-.1in}

\begin{figure}[H]
\colorbox{MyGray}{
\noindent
\begin{minipage}{.37\textwidth}
\small
\textbf{Stage I:}
$$\textbf{max } H(\vBeta,\vC)$$
\begin{align*}
\textbf{s.t.} &\quad\BetaConstraintFP\quad\sum_{i\in [n]}\Pr_{t_{ij}}[t_{ij}\geq \beta_{ij}]\leq\frac{1}{2}\qquad \forall j\\
&\quad\CiConstraintFP\quad\sum_{i\in[n]} c_i \leq 8\cdot \prev
\end{align*}
\end{minipage}
\hfill\vline\hfill

\begin{minipage}{.60\textwidth}
\small
\textbf{~~Stage II:}
$${H(\vBeta,\vC)=}\textbf{max  } \sum_{i\in [n]}\sum_{t_i\in \cT_i}f_i(t_i)\cdot \sum_{S\subseteq [m]}\sigma_{iS}(t_i)\sum_{j\in S}t_{ij}\cdot \ind[t_{ij}\leq \beta_{ij} + c_i]$$
\vspace{-.3in}
  \begin{align*}
\textbf{s.t.} &\quad\WConstraintFP \quad \text{$\sigma$ is feasible} \\
&\quad\WBetaConstraintFP\quad\frac{1}{2}\sum_{t_i\in \cT_i}f_i(t_i)\cdot \sum_{S:j\in S}\sigma_{iS}(t_i)\leq\Pr_{t_{ij}}[t_{ij}\geq \beta_{ij}] & \forall i,j\\
\end{align*}
\end{minipage}
}
\vspace{-.15in}

\caption{Two-stage Optimization over $\theta=(\vBeta,\vC)$ and the allocation $\sigma$
}~\label{fig:LP-fixed-parameter}
\vspace{-.3in}
\end{figure}


We now focus on the Stage II problem and try to solve it efficiently for a fixed set of parameters $\theta$. The objective is a linear function of the variables $\sigma$, yet the set of variables $\sigma=\{\sigma_{iS}(t_i)\}_{i\in [n], S\subseteq [m], t_i\in \cT_i}$ has exponential size. Luckily, the problem can be expressed more succinctly. For any interim allocation $\sigma$ and dual parameters $\theta=(\vBeta,\vC)$, the objective ($\core(\sigma,\theta)$) can be simplified as follows: 
{
\begin{equation}\label{equ:core-marginal}
\begin{aligned}
\core(\sigma,\theta)&=\sum_{\substack{i\in [n]\\ t_i\in \cT_i}}f_i(t_i) \sum_{S\subseteq [m]}\sigma_{iS}(t_i)\sum_{j\in S}t_{ij}\cdot \ind[t_{ij}\leq \beta_{ij}+c_i]
&=\sum_{\substack{i\in [n], j\in [m]\\t_{ij}\in \cT_{ij}}} \widehat{w}_{ij}(t_{ij})t_{ij}\cdot \ind[t_{ij}\leq \beta_{ij}+c_i],
\end{aligned}
\end{equation}
}
\vspace{-.15in}

\noindent where $\widehat{w}_{ij}(t_{ij})=f_{ij}(t_{ij})\cdot\sum_{t_{i,-j}}f_{i,-j}(t_{i,-j})\cdot \sum_{S:j\in S}\sigma_{iS}(t_{ij},t_{i,-j})$ for every $i\in [n],j\in [m],t_{ij}\in \cT_{ij}$. We refer to $\{\widehat{w}_{ij}(t_{ij})\}_{i\in [n],j\in [m],t_{ij}\in \cT_{ij}}$ as the \textbf{marginal reduced form} of the interim allocation rule $\sigma$. $\widehat{w}_{ij}(t_{ij})$ represents the probability that bidder $i$'s value for item $j$ is $t_{ij}$ and she receives item $j$, and the probability is taken over the randomness of the allocation, other bidders' types, as well as her own values for all the other items. Now for every fixed dual parameters $\theta$, $\core$ is expressed as a linear function of the much more succinct representation $\widehat{w}=\{\widehat{w}_{ij}(t_{ij})\}_{i, j,t_{ij}}$ that has polynomial description size. We rewrite the Stage II problem as an LP using the variables $\widehat{w}$.
Denote $\core(\hat{w},\theta)$ the last term of \Cref{equ:core-marginal}, which is the objective of the problem. By the definition of $\widehat{w}$, Constraint {\WBetaConstraintFP} is equivalent to
\vspace{-.1in}

\begin{equation}\label{equ:constraints_on_w}
\frac{1}{2}\cdot\sum_{t_{ij}\in \cT_{ij}}\widehat{w}_{ij}(t_{ij})\leq
\Pr_{t_{ij}}[t_{ij}\geq \beta_{ij}],\qquad\forall i,j    
\end{equation}
\vspace{-.1in}

\noindent which is a linear constraint on $\widehat{w}$. Let $\cP_1$ be the convex polytope that contains all marginal reduced forms $\widehat{w}$ that can be implemented by some feasible allocation $\sigma$ (corresponds to Constraint {\WConstraintFP}) and $\cP_2$ be the set of all  $\widehat{w}$ that satisfy all constraints in \Cref{equ:constraints_on_w}. The Stage II problem is equivalent to the LP $\max_{\widehat{w}\in \cP_1\cap \cP_2}\core(\widehat{w},\theta)$. Unfortunately, since $\cP_1$ does not have an explicit succinct description or an efficient separation oracle, it is unclear if the problem can be solved efficiently.

\notshow{

\vspace{-.15in}
\paragraph{Step 1: Optimizing over a fixed set of dual parameters $\theta$.}
To optimize $\core(\sigma,\theta(\sigma))$, 
our first insight is to \emph{switch the order of dependence}, and choose a set of dual parameters $\theta$ first then optimize over the allocation. In particular, we aim to maximize $\core(\sigma,\theta)$ for every set of dual parameters $\theta$. We notice that for any interim allocation $\sigma$ and dual parameters $\theta=(\tilde{\vBeta},\vC)$,
$\core(\sigma,\theta)$ can be simplified as follows: 

\vspace{-.1in}
\begin{equation}\label{equ:core-marginal}
\begin{aligned}
\core(\sigma,\theta)&=\sum_i\sum_{t_i}f_i(t_i)\cdot \sum_{S\subseteq [m]}\sigma_{iS}(t_i)\sum_{j\in S}t_{ij}\cdot \ind[t_{ij}\leq \tilde{\beta}_{ij}+c_i] \\
&=\sum_i\sum_j\sum_{t_{ij}} \widehat{w}_{ij}(t_{ij})\cdot t_{ij}\cdot \ind[t_{ij}\leq \tilde{\beta}_{ij}+c_i],
\end{aligned}
\end{equation}
\vspace{-.15in}

\noindent where $\widehat{w}_{ij}(t_{ij})=f_{ij}(t_{ij})\cdot\sum_{t_{i,-j}}f_{i,-j}(t_{i,-j})\cdot \sum_{S:j\in S}\sigma_{iS}(t_{ij},t_{i,-j})$ for every $i\in [n],j\in [m],t_{ij}\in \cT_{ij}$. We refer to $\{\widehat{w}_{ij}(t_{ij})\}_{i\in [n],j\in [m],t_{ij}\in \cT_{ij}}$ as the \textbf{marginal reduced form} of the interim allocation rule $\sigma$. $\widehat{w}_{ij}(t_{ij})$ represents the probability that bidder $i$'s value for item $j$ is $t_{ij}$ and she receives item $j$. The probability is taken over the randomness of the allocation, other bidders' types, as well as her own values for all the other items.

Now for every fixed dual parameters $\theta$, $\core$ can still be written as a linear function of the much more succinct representation $\widehat{w}=\{\widehat{w}_{ij}(t_{ij})\}_{i, j,t_{ij}}$ (compared to the interim allocation $\sigma$). We denote $\core(\hat{w},\theta)$ the last term of \Cref{equ:core-marginal}. For a given $\theta$, the only constraint on the allocation, i.e., Property 2 of \Cref{lem:caiz17-constrained-additive}, can also be expressed as a linear constraint on $\widehat{w}$:
\vspace{-.1in}

\begin{equation}\label{equ:constraints_on_w}
\frac{1}{2}\cdot\sum_{t_{ij}\in \cT_{ij}}\widehat{w}_{ij}(t_{ij})\leq
\Pr_{t_{ij}}[t_{ij}\geq \tilde\beta_{ij}],\qquad\forall i,j    
\end{equation}
\vspace{-.15in}

\noindent Let $\cP_1$ be the convex polytope that contains all marginal reduced forms $\widehat{w}$ that can be implemented by some feasible allocation $\sigma$, and $\cP_2$ be the set of all  $\widehat{w}$ that satisfy all constraints in \Cref{equ:constraints_on_w}. To maximize $\core(\sigma,\theta)$, it suffices to solve $\max_{\widehat{w}\in \cP_1\cap \cP_2}\core(\widehat{w},\theta)$. Unfortunately, $\cP_1$ does not have an explicit succinct description or an efficient separation oracle. 




}

\vspace{-.15in}
\paragraph{Step 2: Marginal Reduced Form Relaxation.} To overcome this barrier, we consider a relaxation of $\cP_1$, where the feasibility constraint is only enforced on each bidder separately. 
We refer to this step as the \emph{marginal reduced~form~relaxation}. We use $\widehat{w}_i=\{\widehat{w}_{ij}(t_{ij})\}_{j\in [m],t_{ij}\in \cT_{ij}}$ to denote a feasible \emph{single-bidder marginal reduced form} for bidder $i$.  
Formally, we define the feasible region $W_i$ of $\widehat{w}_i$ in~Definition~\ref{def:W_i-constrained-add}.  




\begin{definition}[Constrained-additive valuations: single-bidder marginal reduced form polytope]\label{def:W_i-constrained-add}
For every $i\in [n]$, suppose bidder $i$ has a constrained-additive valuation with feasibility constraint $\cF_i$. Bidder $i$'s single-bidder marginal reduced form polytope $W_i\subseteq [0,1]^{\sum_{j\in[m]}|\cT_{ij}|}$ is defined as follows: $\widehat{w}_i\in W_i$ if and only if there exists an allocation rule $\{\sigma_S(t_i)\}_{t_i\in \cT_i, S\in \cF_i}$, i.e., $\sigma_S(t_i)$ is the probability that $i$ receives set $S$ when her type is $t_i$, such that
\textbf{(i)} $\sum_{S\in \cF_i}\sigma_S(t_i)\leq 1$, $\forall t_i\in \cT_i$, and \textbf{(ii)}
$\widehat{w}_{ij}(t_{ij})=f_{ij}(t_{ij})\cdot\sum_{t_{i,-j}}f_{i,-j}(t_{i,-j})\cdot \sum_{S:j\in S}\sigma_S(t_i)$, for all $j\in[m]$ and $t_{ij}\in \cT_{ij}$.
\end{definition}

\noindent Throughout this section, we assume access to a separation oracle of $W_i$ for every bidder $i$. In \Cref{thm:multiplicative approx for constraint additive-main body}, we present an efficient separation oracle for another polytope $\widehat{W}_i$ that is a multiplicative approximation to $W_i$, i.e., $\widehat{W}_i$ is sandwiched between $c\cdot W_i$ and $W_i$ for some absolute constant $c\in(0,1)$, using only queries to bidder $i$'s demand oracle. We will argue later that we can efficiently approximate our problem with the separation oracle for $\widehat{W}_i$. 

Here is our relaxation {to the (rewritten) Stage II problem}: Instead of forcing $\widehat{w}$ to be \emph{implementable jointly} ($\widehat{w}\in \cP_1$), we consider the relaxed region $\cP'\supseteq \cP_1$: $\widehat{w}\in \cP'$ if and only if: (i) $\widehat{w}_i
\in W_i$, for all bidder $i\in [n]$, and (ii) $\sum_{i}\sum_{t_{ij}}\widehat{w}_{ij}(t_{ij})\leq 1,\forall j\in [m]$. 
In other words, $\cP'$ guarantees that, for every bidder $i$, $\widehat{w}_i$ is a feasible single-bidder marginal reduced form for $i$, and the supply constraint is met in terms of marginal reduced forms (rather than ex-post allocations).

\begin{figure}[ht!]
\colorbox{MyGray}{
\noindent
\begin{minipage}{.97\textwidth}
\small
\textbf{~~Relaxed Stage II:}
$${H(\vBeta,\vC)=}\textbf{max  } \sum_{i\in[n]} \sum_{j\in[m]} \sum_{t_{ij}\in \cT_{ij}} \widehat{w}_{ij}(t_{ij})\cdot t_{ij}\cdot \ind[t_{ij}\leq \beta_{ij} + c_i]$$
\vspace{-.3in}
  \begin{align*}
\textbf{s.t.} &\quad\WConstraintFP \quad \widehat{w}_i \in W_i & \forall i \\
&\quad\WBetaConstraintFP\quad\frac{1}{2}\sum_{t_{ij}}\widehat{w}_{ij}(t_{ij})\leq\Pr_{t_{ij}}[t_{ij}\geq \beta_{ij}] & \forall i,j\\
& \qquad\quad~\widehat{w}_{ij}(t_{ij})\geq 0 & \forall i,j,t_{ij}
\end{align*}
\end{minipage}
}
\caption{The Relaxed Stage II Problem over the Marginal Reduced Forms
}~\label{fig:LP-relaxed-stageii}
\vspace{-.25in}
\end{figure}

\noindent The main benefit of this relaxation is computational. Without the relaxation, we need a multiplicative approximation of $\cP_1$. \Cref{thm:multiplicative approx for constraint additive-main body} provides such an approximation if we can exactly maximizes the social welfare -- a computational task that is substantially harder than answering demand queries. Indeed, we are not aware of any efficient algorithm that exactly maximizes the social welfare with only access to demand oracles of every bidder. 
The relaxed problem $\max_{\widehat{w}\in \cP'\cap \cP_2}\core(\widehat{w},\theta)$ is captured by the LP in \Cref{fig:LP-relaxed-stageii}.\footnote{We omit the supply constraint $\sum_{i}\sum_{t_{ij}}\widehat{w}_{ij}(t_{ij})\leq 1$ as it is implied by Constraint~\BetaConstraintFP~in the Stage I problem and Constraint~\WBetaConstraintFP.}

Consider the two-stage optimization with the relaxed Stage II problem. 
For every fixed parameters $\theta$, the relaxed Stage II problem can be solved efficiently (assuming a separation oracle of $W_i$ for every $i$). Unfortunately, we do not know how to solve the two-stage optimization problem efficiently, as the number of different dual parameters is exponential in $n$ and $m$, and enumerating through all possible choices of dual parameters is not an option. To overcome this obstacle, we need ideas explained in the following step.

\notshow{
\vspace{-.15in}
\paragraph{Step 3: Two-Stage Optimization.}
In this step, we obtain a two-stage optimization problem that is a relaxation of the original problem. For a fixed set of dual parameters $\theta=(\vBeta,\vC)$, we now know that the relaxed problem $\max_{\widehat{w}\in \cP'\cap \cP_2}\core(\widehat{w},\theta)$ can be captured by a linear program shown on the RHS of \Cref{fig:LP-fixed-parameter}. This is the second-stage optimization of our problem, and we denote the optimal solution of this LP by $H(\theta)$. In the first-stage optimization, we find the optimal dual parameter $\theta$ that maximizes $H(\theta)$. To make sure this optimum can still be approximated by simple mechanisms, we need to pose certain constraints on $\theta$, that is, they need to respect Constraints \BetaConstraintFP~ and~\CiConstraintFP~corresponding to property 1 and 4 of \Cref{lem:caiz17-constrained-additive}.~\footnote{We ignore the supply constraint $\sum_{i}\sum_{t_{ij}}\widehat{w}_{ij}(t_{ij})\leq 1$ in the LP as it is implied by Constraints~\WBetaConstraintFP~and~\BetaConstraintFP.} 



\begin{figure}[ht!]
\colorbox{MyGray}{
\noindent
\begin{minipage}{.37\textwidth}
\small
\textbf{Stage I:}
$$\textbf{max } H(\vBeta,\vC)$$
\begin{align*}
s.t. &\quad\BetaConstraintFP\quad\sum_{i\in [n]}\Pr_{t_{ij}}[t_{ij}\geq \beta_{ij}]\leq\frac{1}{2}\qquad \forall j\\
&\quad\CiConstraintFP\quad\sum_{i\in[n]} c_i \leq 8\cdot \prev
\end{align*}
\end{minipage}
\hfill\vline\hfill

\begin{minipage}{.59\textwidth}
\small
\textbf{~~Stage II:}
$${H(\vBeta,\vC)=}\textbf{max  } \sum_{i\in[n]} \sum_{j\in[m]} \sum_{t_{ij}\in \cT_{ij}} \widehat{w}_{ij}(t_{ij})\cdot t_{ij}\cdot \ind[t_{ij}\leq \beta_{ij} + c_i]$$
\vspace{-.3in}
  \begin{align*}
s.t. &\quad\WConstraintFP \quad \widehat{w}_i \in W_i & \forall i \\
&\quad\WBetaConstraintFP\quad\frac{1}{2}\sum_{t_{ij}}\widehat{w}_{ij}(t_{ij})\leq\Pr_{t_{ij}}[t_{ij}\geq \beta_{ij}] & \forall i,j\\
& \qquad\quad~\widehat{w}_{ij}(t_{ij})\geq 0 & \forall i,j,t_{ij}
\end{align*}
\end{minipage}
}
\caption{Two-stage Optimization over $\theta=(\vBeta,\vC)$ and the Marginal Reduced Forms
}~\label{fig:LP-fixed-parameter}
\vspace{-.3in}
\end{figure}

\noindent Why is the two-stage optimization a relaxation? First, for any interim allocation $\sigma$, the corresponding set of dual parameters $\theta(\sigma)$ is a feasible solution of the first-stage optimization problem. Second, for any interim allocation $\sigma$, its corresponding marginal reduced form is feasible in the second-stage optimization w.r.t. $\theta(\sigma)$. Unfortunately, we do not know how to solve the two-stage optimization problem efficiently, \argyrisnote{Argyris: Maybe writing it as "as the set of feasible dual parameter has size exponential in $n$ and $m$"} as the number of dual parameters is exponential in $n$ and $m$, enumerating through all possible dual parameters is not an option. To overcome this obstacle, we need ideas explained in the following step.

}



\vspace{-.15in}
\paragraph{Step 3: Lifting the problem to a higher dimensional space.}

Instead of enumerating all possible dual parameters $\theta$, we optimize over \emph{distributions of the parameters}. To guarantee that the number of decision variables in our program is polynomial, we focus on \emph{product distributions} over the parameters. Formally, for every $i,j$, let $\cC_{ij}$ be a distribution over $\cV_{ij}\times \Delta$, where $\cV_{ij}$ and $\Delta$ are the set of possible values of $\beta_{ij}$ and $c_i$ accordingly, after discretization 
(See Footnote~\ref{footnote:discretization in lp} in  Figure~\ref{fig:bigLP} for a formal definition). All $\cC_{ij}$'s are independent. In our program, we use decision variables $\{\hat{\lambda}_{ij}(\beta_{ij},\delta_{ij})\}_{i\in[n],j\in[m],\beta_{ij}\in \cV_{ij},\delta_{ij}\in \Delta}$ to represent the distribution $\cC_{ij}$, i.e., $\hat\lambda_{ij}(a,b)=\Pr_{(\beta_{ij},\delta_{ij}) \sim \cC_{ij}}[\beta_{ij} = a \land \delta_{ij} = b]$. {Notice that if the parameters are drawn from a product distribution, 
the parameter ``$c_i$'' may be different for each item $j$. To distinguish them, we use $\delta_{ij}$ to replace the original parameter $c_i$ in our program.  } 

Now we maximize the expected value of the $\core$ function over all product distributions $\times_{i,j}\cC_{ij}$ (represented by decision variables $\hat{\lambda}$) and the allocations (represented by the marginal reduced form $\widehat{w}$). {If the parameters $\theta$ and allocation $\widehat{w}$ are generated independently,} the expected $\core$ is not a linear objective, since the contributed truncated welfare in $\core$ is $\widehat{w}_{ij}(t_{ij})\cdot \hat{\lambda}_{ij}(\beta_{ij},\delta_{ij})\cdot t_{ij}\cdot \ind[t_{ij}\leq \beta_{ij}+\delta_{ij}]$ for every $t_{ij},\beta_{ij},\delta_{ij}$. \emph{To linearize the objective, we lift the problem to a higher dimensional space and consider joint distributions over the parameters and allocations.} 
We do not consider arbitrary joint distributions, and only focus on the ones that correspond to the following generative process: first draw $(\beta,\delta)$ from a product distribution (according to $\hat\lambda$), then choose a feasible allocation $\widehat{w}^{(\beta,\delta)}=\{\widehat{w}^{(\beta,\delta)}_{ij}(t_{ij})\}_{i,j,t_{ij}}$  conditioned on $(\beta,\delta)$. Since there are too many parameters $(\beta,\delta)$, we certainly cannot afford to store all $\widehat{w}^{(\beta,\delta)}$'s explicitly.
\notshow{Even though we do not directly use the $\widehat{w}^{(\beta,\delta)}$'s as decision variables, they are indeed crucial for understanding many of the constraints in our LP relaxation in \Cref{fig:bigLP}. See \Cref{subsec:result} for more discussion.}
Instead, for each bidder $i$ and item $j$ we introduce a new set of decision variables $\{\lambda_{ij}(t_{ij},\beta_{ij},\delta_{ij})\}_{t_{ij}\in\cT_{ij},\beta_{ij}\in \cV_{ij},\delta_{ij}\in \Delta}$, where $\lambda_{ij}(t_{ij},\beta_{ij},\delta_{ij})$ is the marginal probability for the following three events to happen simultaneously in our generative process: \textbf{(a)} $(\beta_{ij},\delta_{ij})$ are the parameters for $i$ and $j$. \textbf{(b)} Bidder $i$ receives item $j$. \textbf{(c)} Bidder $i$'s value for item $j$ is $t_{ij}$. Formally,
\vspace{-.12in}
\begin{equation}\label{equ:explain-lambda}
\lambda_{ij}(t_{ij},\beta_{ij},\delta_{ij})=\hat{\lambda}_{ij}(\beta_{ij},\delta_{ij})\cdot\sum_{\{(\beta_{i'j'},\delta_{i'j'})\}_{(i',j')\not=(i,j)}}\left({\widehat{w}_{ij}^{(\beta,\delta)}(t_{ij})}/{f_{ij}(t_{ij})}\right)\cdot \prod_{(i',j')\not=(i,j)}\hat{\lambda}_{i'j'}(\beta_{i'j'},\delta_{i'j'})
\end{equation}
\vspace{-.15in}

\noindent With the new variables $\lambda_{ij}(t_{ij},\beta_{ij},\delta_{ij})$'s, we can express the objective as an linear function: 
\vspace{-.12in}
$$ \sum_{i\in[n]} \sum_{j\in[m]} \sum_{t_{ij}\in \cT_{ij}} 
 f_{ij}(t_{ij})\cdot t_{ij}\cdot \sum_{\beta_{ij}\in \cV_{ij},\delta_{ij} \in \Delta} \lambda_{ij}(t_{ij},\beta_{ij}, \delta_{ij})\cdot \ind[t_{ij}\leq \beta_{ij} + \delta_{ij}].$$
\vspace{-.15in}


\notshow{
There is a set of decision variables $\lambda$ in our program representing the marginal expected value of $\frac{\widehat{w}_{ij}^{(\beta,\delta)}(t_{ij})}{f_{ij}(t_{ij})}$ on coordinate $(\beta_{ij},\delta_{ij})$. Formally, for every $t_{ij}\in \cT_{ij},\beta_{ij}\in \cV_{ij},\delta_{ij}\in \Delta$, there is a decision variable $\lambda_{ij}(t_{ij},\beta_{ij},\delta_{ij})$ in our program that satisfies
\begin{equation}\label{equ:explain-lambda}
\lambda_{ij}(t_{ij},\beta_{ij},\delta_{ij})=\hat{\lambda}_{ij}(\beta_{ij},\delta_{ij})\cdot\sum_{\{(\beta_{i'j'},\delta_{i'j'})\}_{(i',j')\not=(i,j)}}\frac{\widehat{w}_{ij}^{(\beta,\delta)}(t_{ij})}{f_{ij}(t_{ij})}\cdot \prod_{(i',j')\not=(i,j)}\hat{\lambda}_{i'j'}(\beta_{i'j'},\delta_{i'j'})
\end{equation}
}
\noindent Our program can be viewed as an ``expected version'' of the the two-stage optimization, when the parameters $\theta=(\beta,\delta)\sim \bigtimes_{i,j}\cC_{ij}$. In other words, we only require the constraints to be satisfied in expectation. We discuss our relaxation in more details in \Cref{subsec:result}.

\notshow{
\begin{equation}\label{equ:explain-w}
w_{ij}(t_{ij})=\sum_{\beta,\delta}\widehat{w}_{ij}^{(\beta,\delta)}(t_{ij})\cdot \prod_{i',j'}\hat{\lambda}_{i'j'}(\beta_{i'j'},\delta_{i'j'})\quad\forall i,j,t_{ij}\in \cT_{ij}
\end{equation}
}

\vspace{-.1in}
\subsection{Our LP and a Sketch of the Proof of~\Cref{thm:main XOS-main body}}\label{subsec:result}

We present a sketch of the proof of \Cref{thm:main XOS-main body} for  constrained-additive bidders and our main linear program (\Cref{fig:bigLP}). Although the LP has many constraints and may seem intimidating at first, all constraints follow quite naturally from our derivation in \Cref{sec:tour to LP}. See \Cref{sec:explain LP constraints} for more~details.

The first step of our proof is to estimate $\prev$ using \Cref{thm:chms10} from~\cite{ChawlaHMS09}. 

\vspace{-.1in}
\begin{lemma}[Theorem~14 and Appendix~F in \cite{ChawlaHMS09}]\label{thm:chms10}

There exists an algorithm that with probability at least 
$1-\frac{2}{nm}$, computes a Rationed Posted Price mechanism $\cM$ such that 
$\rev(\cM)\geq \frac{1}{6.75}(1-\frac{1}{nm})\cdot \prev$. The algorithm runs in time $poly(n,m,\sum_{i,j}|\cT_{ij}|)$.
\end{lemma}
\vspace{-.1in}

 Denote $\cE$ the event that an RPP in \Cref{thm:chms10} is computed successfully. For simplicity, we will condition on the event $\cE$ for the rest of this section. Let $\estprev$ be the revenue of the RPP mechanism found in \Cref{thm:chms10}.
 
 Next, we argue that the LP in~\Cref{fig:bigLP} (or \Cref{fig:XOSLP} when the valuations are XOS) can be solved efficiently. Note that there are $\poly(n,m,\sum_{i,j}|\cT_{ij}|)$ constraints except for Constraint $\Wconstraint$, where we need to enforce the feasibility of single-bidder marginal reduced forms. It suffices to construct an efficient separation oracle for $W_i$ for every $i$. However, to the best of our knowledge, $W_i$ does not have a succinct explicit description or an efficient separation oracle. For constrained-additive valuations, we construct another polytope $\widehat{W}_i$ such that: (i) $\widehat{W}_i$ is a multiplicative approximation of $W_i$, i.e., $c\cdot W_i\subseteq \widehat{W}_i\subseteq W_i$ for some absolute constant $c\in (0,1)$, and (ii) There exists an efficient separation oracle for $\widehat{W}_i$ given access to the demand oracle. 

\begin{theorem}\label{thm:multiplicative approx for constraint additive-main body}
Let $T=\sum_{i,j}|\cT_{ij}|$ and $b$ be the bit complexity of the problem instance (\Cref{def:bit complexity}). For any $i\in [n]$ and $\delta\in (0,1)$, there is an algorithm  that constructs a convex polytope $\widehat{W}_i\in [0,1]^{\sum_{j\in [m]}|\cT_{ij}|}$ using $\poly(n,m,T,\log(1/\delta))$ samples from $D_i$, such that with probability at least $1-\delta$,

\begin{enumerate}
\vspace{-.1in}
  \item $\frac{1}{12}\cdot W_i\subseteq \widehat{W}_i\subseteq W_i$, and the vertex-complexity (\Cref{def:vertex-complexity}) of $\widehat{W}_i$ is $\poly(n,m,T,b,\log(1/\delta))$.
\vspace{-.05in}
\item There exists a separation oracle $SO$ for $\widehat{W}_i$, given access to the demand oracle for bidder $i$'s valuation. The running time of $SO$ on any input with bit complexity $b'$ is $\poly(n,m,T,b,b',\log(1/\delta))$ and makes $\poly(n,m,T,b,b',\log(1/\delta))$ queries to the demand oracle.
  \end{enumerate}
  \vspace{-.1in}
The algorithm constructs the polytope and the separation oracle $SO$ in time 
$\poly(n,m,T,b,\log(1/\delta))$.
\end{theorem}
  \vspace{-.02in}

\noindent Indeed, we prove a more general result regarding polytopes that can be expressed as a ``mixture of polytopes'' (\Cref{thm:special case of multiplicative approx}), which can be viewed as a generalization of the technique developed in~\cite{CaiDW12bfull} for approximating the polytope of all feasible reduced forms. 
We postpone the proof of \Cref{thm:multiplicative approx for constraint additive-main body} to \Cref{sec:mrf for constraint additive}. 

To solve the LP relaxation, we replace $W_i$ by $\widehat{W}_i$ in the LP in~\Cref{fig:bigLP} for every $i\in [n]$, and solve the LP in polynomial time using the ellipsoid method. Clearly, this solution is also feasible for the original LP in~\Cref{fig:bigLP}. Moreover, since $\widehat{W}_i$ contains $c\cdot W_i$, we can show that \textbf{the objective value of the solution we computed is at least $c\cdot \optlp$}, where $\optlp$ the optimum of the LP in Figure~\ref{fig:bigLP}. Our proof of \Cref{thm:multiplicative approx for constraint additive-main body} heavily relies on the fact that $W_i$ is a down-monotone polytope,\footnote{A polytope $\cP\subseteq [0,1]^d$ is down-monotone if and only if for every $\bx\in \cP$ and $\textbf{0}\leq \bx'\leq \bx$, we have $\bx'\in \cP$.}  which does not hold in the XOS case.  For XOS valuations, we construct the polytope $\widehat{W}_i$ with a weaker guarantee:
For every vector $x$ in $W_i$, there exists another vector $x'$ in $\widehat{W}_i$ such that for every coordinate $j$, $x_j/x_j'\in [a,b]$ for some absolute constant $0<a<b$, and vice versa. See \Cref{sec:mrf for xos} for a complete proof of \Cref{thm:main XOS-main body} (including the XOS case).

 


Next, we argue that the LP optimum can be approximated by simple mechanisms. \cite{CaiZ17} shows that for any BIC and IR mechanism $\cM$, $\core(\sigma,\tilde{\vBeta}^{(\sigma)},\vC^{(\sigma)},\vr^{(\sigma)})$ (as stated in Lemma~\ref{lem:caiz17-constrained-additive}) can be bounded by 
a constant number of $\prev$ and the revenue of a TPT (see Property 5 of \Cref{lem:caiz17-constrained-additive}). 
We generalize their result by proving that for any feasible solution of the LP, its objective can be bounded by (a constant times) the revenue of a RPP or TPT mechanism, and the mechanism can be computed efficiently given the feasible~solution.


\begin{definition}\label{def:Q_j}
Let $(w,\lambda,\hat\lambda, \bd=(d_i)_{i\in [n]})$ be any feasible solution of the LP in \Cref{fig:bigLP}. 
For every $j\in [m]$, define 
$Q_j = \frac{1}{2}\cdot\sum_{i\in[n]}\sum_{t_{ij}\in \cT_{ij}}  f_{ij}(t_{ij})\cdot t_{ij}\cdot
    \sum_{{\beta_{ij}\in \cV_{ij}, \delta_{ij} \in \Delta}}\lambda_{ij}(t_{ij},\beta_{ij},\delta_{ij})\cdot\ind[t_{ij}\leq \beta_{ij}+\delta_{ij}].\footnote{\text{Recall that $\lambda_{ij}(t_{ij},\beta_{ij},\delta_{ij})$ is introduced in Step 3 of~\Cref{sec:tour to LP}. See~\Cref{fig:bigLP} for the formal definition.}}
$





\end{definition}



Clearly, for any feasible solution of the LP, the objective function is $2\cdot \sum_{j\in [m]}Q_j$. We prove in \Cref{thm:bounding-lp-simple-mech} that $2\cdot \sum_{j\in [m]}Q_j$ can be bounded by the revenue of $\Mtpt$ (Mechanism~\ref{def:constructed-SPEM}) and the RPP $\Mpp$ (\Cref{thm:chms10}).  As we can efficiently compute a feasible solution whose objective is $\Omega(\optlp)$, \Cref{thm:bounding-lp-simple-mech} implies that we can compute in polynomial time a simple mechanism whose revenue is at least $\Omega(\optlp+\prev)$.

\begin{theorem}\label{thm:bounding-lp-simple-mech}
Let $(w,\lambda,\hat\lambda, \bd)$ 
be any feasible solution of the LP in \Cref{fig:bigLP}.
Let $\Mpp$ be the rationed posted price mechanism computed in \Cref{thm:chms10}. Let $\Mtpt$ be the two-part tariff mechanism shown in Mechanism~\ref{def:constructed-SPEM} {with prices $\{Q_j\}_{j\in[m]}$}. Then the objective function of the solution $2\cdot \sum_{j\in [m]}Q_j$ is bounded by $c_1\cdot \rev(\Mpp)+c_2\cdot \rev(\Mtpt)$, for some absolute constant $c_1,c_2>0$. 
Moreover, both $\Mpp$ and $\Mtpt$ can be computed in time $\poly(n,m,\sum_{i,j}|\cT_{ij}|)$ with access to the demand oracle for the bidders'~valuations. 
\end{theorem}

The proof of \Cref{thm:bounding-lp-simple-mech} combines the ``shifted \core'' technique by Cai and Zhao~\cite{CaiZ17} with several novel ideas to handle the new challenges due to the relaxation. We postpone it to \Cref{sec:proof of bounding LP with simple mech}.




\begin{figure}[ht!]
\colorbox{MyGray}{
\begin{minipage}{.98\textwidth}
\small
$$\quad\textbf{max  } \sum_{i\in[n]} \sum_{j\in[m]} \sum_{t_{ij}\in \cT_{ij}} 
 f_{ij}(t_{ij})\cdot t_{ij}\cdot \sum_{\substack{\beta_{ij}\in \cV_{ij},\delta_{ij} \in \Delta}} \lambda_{ij}(t_{ij},\beta_{ij}, \delta_{ij})\cdot \ind[t_{ij}\leq \beta_{ij} + \delta_{ij}]$$
\vspace{-.3in}
  \begin{align*}
\textbf{s.t.}\\
 &\quad\textbf{Allocation Feasibility Constraints:}\\
 &\quad\Wconstraint \quad w_i \in W_i & \forall i \\
&\quad\PiConstraint \quad \sum_i\sum_{t_{ij}\in \cT_{ij}}w_{ij}(t_{ij})\leq 1 & \forall j\\  
 &\quad\textbf{Natural Feasibility Constraints:}\\
    &\quad\LambdaMarginalConstraint\quad f_{ij}(t_{ij})\cdot\sum_{\beta_{ij}\in \cV_{ij}}\sum_{\delta_{ij}\in \Delta} \lambda_{ij}(t_{ij},\beta_{ij},\delta_{ij}) = w_{ij}(t_{ij}) & \forall i,j,t_{ij}\in \cT_{ij}\\
    &\quad\CompareMarginalConstraint\quad\lambda_{ij}(t_{ij},\beta_{ij},\delta_{ij})\leq \hat\lambda_{ij}(\beta_{ij}, \delta_{ij}) & \forall i,j, t_{ij},\beta_{ij}\in \cV_{ij},\delta_{ij}\\
    &\quad\HatLambdaDistributionConstraint\quad \sum_{\substack{\beta_{ij}\in \cV_{ij},\delta_{ij} \in \Delta}} \hat\lambda_{ij}(\beta_{ij},\delta_{ij}) = 1 & \forall i,j\\
    &\quad\textbf{Problem Specific Constraints:}\\
    &\quad\ReduceDemandConstaint \quad \sum_{i\in[n]}\sum_{\beta_{ij} \in \cV_{ij}}\sum_{\delta_{ij}\in \Delta} \hat\lambda_{ij}(\beta_{ij},\delta_{ij}) \cdot \Pr_{t_{ij}\sim D_{ij}}[t_{ij}\geq \beta_{ij}] \leq \frac{1}{2}& \forall j\\
    & \quad\MarginalToGlobalConstraint\quad {\frac{1}{2}\sum_{t_{ij} \in \cT_{ij}} f_{ij}(t_{ij}) \left(\lambda_{ij}(t_{ij},\beta_{ij},\delta_{ij})+ \lambda_{ij}(t_{ij},\beta_{ij}^+,\delta_{ij})\right) \leq} \\
    &\qquad\qquad\qquad\hat\lambda_{ij}(\beta_{ij},\delta_{ij}) \cdot \Pr_{t_{ij}}[t_{ij}\geq \beta_{ij}]+\hat\lambda_{ij}(\beta_{ij}^+,\delta_{ij}) \cdot \Pr_{t_{ij}}[t_{ij}\geq\beta_{ij}^+] & \forall i,j,\beta_{ij}\in \cV_{ij}^0,\delta_{ij}\in \Delta\\
    &\quad\BoundMeanDeltaConstraint\quad \sum_{\substack{\beta_{ij}\in \cV_{ij}, \delta_{ij} \in \Delta}} \delta_{ij}\cdot \hat\lambda_{ij}(\beta_{ij},\delta_{ij}) \leq d_i  & \forall i,j\\
    &\quad\BoundSumDeltaConstraint\quad\sum_{i\in[n]} d_i \leq 
   111\cdot \estprev\\
     &\qquad\quad \lambda_{ij}(t_{ij},\beta_{ij},\delta_{ij})\geq 0,\hat\lambda_{ij}(\beta_{ij},\delta_{ij})\geq 0,w_{ij}(t_{ij})\geq 0,d_i\geq 0 & \forall i,j,t_{ij},\beta_{ij}\in \cV_{ij},\delta_{ij}
\end{align*}

\textbf{Variables:}~\footnote{For every $i,j$, let $\cV_{ij}^0 = \cT_{ij}$ be the set of all possible values of $t_{ij}$. To address the tie-breaking issue in \Cref{remark:tie-breaking}, let $\eps_r>0$ be an arbitrarily small number,
 and define $\cV_{ij}^+ = \{t_{ij} + \eps_r: t_{ij} \in \cT_{ij}\}$ and $\cV_{ij} = \cV^0_{ij} \cup \cV^+_{ij}$. Let $\Delta$ be a geometric discretization of range $[\estprev/n,55\cdot\estprev]$. Formally, $\delta\in \Delta$ if and only if $\delta=\frac{2^x}{n}\cdot \estprev$ for some integer $x$ such that $0\leq x\leq\lceil\log(55n)\rceil$. Finally, 
for each $\beta \in \cV_{ij}^0$, let $\beta^+=\beta+\eps_r\in \cV_{ij}^+$. Note that the LP  (or the LP in~\Cref{fig:XOSLP}) do not depend on the choice of $\varepsilon_r$, so we can choose $\varepsilon_r$ to be sufficiently small. In fact, let $b$ be an upper bound of the bit complexity of the problem instance, and the bit complexity of any feasible solution of our LP. Our proof works as long as $\eps_r<\min\{\frac{1}{2^{\poly(b)}},\frac{\prev}{\sum_{i,j}|\cT_{ij}|}\}$.\label{footnote:discretization in lp}}

\vspace{.1in}
- $\lambda_{ij}(t_{ij},\beta_{ij},\delta_{ij})$, for all $i,j$ and $t_{ij}\in \cT_{ij}$, $\beta_{ij}\in \cV_{ij}, \delta_{ij}\in \Delta$. See Step 3 of \Cref{sec:tour to LP} for an explanation of this variable.

\vspace{.1in}
- $\hat\lambda_{ij}(\beta_{ij},\delta_{ij})$, for all $i$, $j$, $\beta_{ij}\in \cV_{ij}, \delta_{ij}\in \Delta$, denoting the distribution $\cC_{ij}$ over $\cV_{ij}\times \Delta$.

\vspace{.1in}
- $w_{ij}(t_{ij})$, for all $i\in [n],j\in [m],t_{ij}\in \cT_{ij}$, denoting the expected marginal reduced form. We denote $w_i=\{w_{ij}(t_{ij})\}_{j\in[m],t_{ij}\in\cT_{ij}}$ the vector of all variables associated with bidder $i$.

\vspace{.1in}
- $d_i$, for all $i\in [n]$, denoting an upper bound of the expectation of $\delta_{ij}$ over distribution $\cC_{ij}$ for all $j$.
\end{minipage}
}
\caption{LP Relaxation for Constrained-Additive Bidders}~\label{fig:bigLP}

\end{figure}

\begin{algorithm}[h!]
\floatname{algorithm}{Mechanism}
\begin{algorithmic}[1]
\setcounter{ALG@line}{-1}
\State Before the mechanism starts, the seller computes the \emph{price} $Q_j$ (\Cref{def:Q_j}) for every item $j$.  
\State 
Bidders arrive sequentially in the lexicographical order.
\State When bidder $i$ arrive, the seller shows her the set of available items $S_i(t_{<i})\subseteq [m]$, as well as their prices. Note that $S_i(t_{<i})$ is the set of items that are not purchased by the first $i-1$ bidders, which depends on $t_{<i}$. We use $S_1(t_{<1})$ to denote $[m]$.
\State Bidder $i$ is asked to pay an \emph{entry fee}. The seller samples a type $t_i'\sim D_i$, and sets the entry fee as: $\xi_i(S_i(t_{<i}),t_i') = \max_{S'\subseteq S_i(t_{<i})}\left(v_i(t_i',S')-\sum_{j\in S'}Q_j\right)$. The entry fee is bidder $i$'s utility for her favorite set under prices $Q_j$'s if her type is $t'_i$.
\State If bidder $i$ (with type $t_i$) agrees to pay the entry fee 
$\xi_i(S_i(t_{<i}),t_i') $, then she can enter the mechanism and take her favorite set $S^*\in\argmax_{S'\subseteq S_i(t_{<i})}\left(v_i(t_i,S')-\sum_{j\in S'}Q_j\right)$, by paying $\sum_{j\in S^*}Q_j$. Otherwise, the bidder gets nothing and pays 0. 
\end{algorithmic}
\caption{{\sf \quad Two-part Tariff Mechanism $\Mtpt$}}\label{def:constructed-SPEM}
\end{algorithm}

\begin{remark} \label{remark:lexi-order}
\Cref{thm:bounding-lp-simple-mech} indeed holds even if the bidders arrive in an arbitrary order in $\Mtpt$. We choose the lexicographical order only to keep the notation light. 
\end{remark}

 We complete the last step of our proof by showing $\opt= O(\optlp+\prev)$ in~\Cref{lem:bound rev by opt}. More specifically, we show that for any mechanism $\cM=(\sigma,p)$, the tuple $(\sigma,\tilde{\vBeta}^{(\sigma)},\vC^{(\sigma)},r^{(\sigma)})$ stated in \Cref{lem:caiz17-constrained-additive} corresponds to a feasible solution of the LP in~\Cref{fig:bigLP} whose objective is at least $\core(\sigma,\tilde{\vBeta}^{(\sigma)},\vC^{(\sigma)},r^{(\sigma)})$. Hence, the revenue of $\cM$ is upper bounded by $\prev$ and $\optlp$. The proof is postponed to \Cref{subsec:proof_lem:bound rev by opt}. Indeed, for both \Cref{thm:bounding-lp-simple-mech} and \Cref{lem:bound rev by opt}, we prove a general statement that applies to XOS valuations, which requires a generalized LP and definitions. See \Cref{thm:bounding-lp-simple-mech-XOS} and \Cref{lem:bound rev by opt-XOS} in \Cref{sec:appx_program} for details. 

\begin{lemma}\label{lem:bound rev by opt}
For any BIC and IR mechanism $\cM$, $\rev(\cM)\leq 28\cdot \prev+4\cdot\optlp$.
\end{lemma}

\subsection{Interpretation of Our LP in~\Cref{fig:bigLP}.}~\label{sec:explain LP constraints} We explain our LP in this section. The objective is the expected \core, as explained in Step 3 in \Cref{sec:tour to LP}. 
According to our definition of $\lamij$ and Constraint~\LambdaMarginalConstraint, $\{w_{ij}(t_{ij})\}_{i,j,t_{ij}}$ corresponds to the expected marginal reduced form, that is, $w_{ij}(t_{ij})$ is the expected probability for bidder $i$ to receive item $j$ and her value for item $j$ is $t_{ij}$.
Constraints~\Wconstraint~and~\PiConstraint~simply sets the feasible region of the expected marginal reduced form $w$. They follow directly from the fact that every realized $\widehat{w}^{(\beta,\delta)}$ is feasible (see Step 3 in \Cref{sec:tour to LP}).  
Constraint~\CompareMarginalConstraint~follows from \Cref{equ:explain-lambda} and the fact that every ${\widehat{w}_{ij}^{(\beta,\delta)}(t_{ij})}/{f_{ij}(t_{ij})}$ is in $[0,1]$. Constraint~\HatLambdaDistributionConstraint~implies that $\{\hat\lambda_{ij}(\beta_{ij},\delta_{ij})\}_{\beta_{ij},\delta_{ij}}$ correspond to a distribution $\cC_{ij}$.

Constraints~\ReduceDemandConstaint~-~\BoundSumDeltaConstraint~are specialized for our problem, which guarantees that the LP optimum can be bounded by simple mechanisms. Constraint~\ReduceDemandConstaint~follows from taking expectation on both sides of Constraint~\BetaConstraintFP~in \Cref{fig:LP-fixed-parameter}, over the randomness of $\beta_{ij}$. Constraints~\BoundMeanDeltaConstraint~and~\BoundSumDeltaConstraint~correspond to Constraint~\CiConstraintFP~in \Cref{fig:LP-fixed-parameter}. Here we bound the expectation of $\delta_{ij}$ by a unified upper bound $d_i$ for every $j$.\footnote{This corresponds to the fact that in \Cref{lem:caiz17-constrained-additive} (and \Cref{fig:LP-fixed-parameter}), there is a single $c_i$ that represents $\delta_{ij}$.} It is worth emphasizing Constraint~\MarginalToGlobalConstraint, which corresponds to Constraint \WBetaConstraintFP~in \Cref{fig:LP-fixed-parameter} (and \Cref{fig:LP-relaxed-stageii}). Instead of taking expectations over the dual parameters, we force the constraint to hold for every $\beta_{ij}$ and $\delta_{ij}$. 
This is an important property that is crucial in our analysis (see \Cref{fn:constraint 7} in \Cref{lem:Q and Q-hat}). Readers may notice that Constraint~\PiConstraint~is implied by Constraints {\LambdaMarginalConstraint}, {\ReduceDemandConstaint} and {\MarginalToGlobalConstraint}. We keep~Constraint~\PiConstraint~so that it is clear that the supply constraint is enforced over the (expected) marginal reduced~form. 
The Problem Specific Constraints~\ReduceDemandConstaint-\BoundSumDeltaConstraint~ in the LP in Figure~\ref{fig:bigLP} are expected versions of the constraints in the two-stage optimization problem in Figure~\ref{fig:LP-fixed-parameter}, which are directly inspired by the Properties~1, 2, and 3 in Lemma~\ref{lem:caiz17-constrained-additive}.
They are crucial to guarantee that optimal value of the LP in Figure~\ref{fig:bigLP} is still approximable by simple mechanisms.


\notshow{
\section{Backup section (only for checking)}

\begin{align*}
    (P) = 
    \max& \frac{1}{2} \sum_{i\in[n]} \sum_{j\in[m]} \sum_{t_{i,j}\in T_{i,j}}  f_{i,j}(t_{i,j})\cdot V_{i,j}(t_{i,j})\cdot
    \\
    &\quad \quad\sum_{\substack{\beta\in V_{i,j}(T_{i,j}),\delta \in \Delta V_{i}}} \lambda_{i,j}(t_{i,j},\beta, \delta)\cdot \mathds{1}[V_{i,j}(t_{i,j}) \leq \beta + \delta]\\
    s.t.\quad & \pi_i \in \Pi_i  & \forall i \in[n] \\
    & \sum_{\beta}\lambda_{i,j}(t_{i,j},\beta) = \pi_{i,j}(t_{i,j}) & \forall i\in[n],j\in[m],t_{i,j}\in T_{i,j}\\
    & \sum_{i\in[n]}\sum_{t_{i,j}\in T_{i,j}} f_{i,j}(t_{i,j})\cdot \pi_{i,j}(t_{i,j}) \leq 1, & \forall j \in[m] \\
    & \lambda_{i,j}(t_{i,j},\beta, \delta) \leq \lambda_{i,j}'(\beta, \delta) & \forall i\in[n], j \in[m], t_{i,j} \in T_{i,j}\\ 
    &\sum_{t_{i,j}} f_{i,j}(t_{i,j}) \lambda_{i,j}(t_{i,j},\beta) \leq \frac{1}{2}\cdot \lambda'_{i,j}(\beta) \Pr_{t_{i,j}\sim D_{i,j}}[t_{i,j} \geq \beta] & \forall \beta \in V_{i,j}(T_{i,j})\\
    & \lambda_{i}'(\delta) = \sum_{\beta \in V_{i,j}(T_{i,j})} \lambda_{i,j}'(\beta, \delta) \text{(Maybe not needed)} &\forall i\in[n], j \in[m] \\
    & \sum_{\delta \in \Delta V_i} \delta \lambda_{i}'(\delta) \leq \delta_i  & \forall i \in [n]\\
    & \sum_{i\in[n]} \delta_i \leq O(1) \cdot \srev \\
     & \forall i \in[n],j\in [m],t\in T_{i,j},\beta,\gamma \in V_{i,j}(t_{i,j})\\
\end{align*}
}

\notshow{

\section{Old Main Program (for reference only)}

The variables in the LP are as follows:
\begin{enumerate}
    \item $\pi_{ij}(t_{ij})$, for all $i\in [n],j\in [m],t_{ij}\in \cT_{ij}$. We denote $\pi_i=(\pi_{ij}(t_{ij}))_{j,t_{ij}}$ the vector of all variables associated with bidder $i$.
    \item $w_{ij}(t_{ij})$, for all $i\in [n],j\in [m],t_{ij}\in \cT_{ij}$. We denote $w_i=(w_{ij}(t_{ij}))_{j,t_{ij}}$ the vector of all variables associated with bidder $i$.
    \item $\lambda_{ij}^{(r)}(t_{ij},\beta_{ij},\delta_{ij}),\lambda_{ij}^{(n)}(t_{ij},\beta_{ij},\delta_{ij})$, for all $i,j$ and $t_{ij}\in \cT_{ij}$, $\beta_{ij}\in \cV_{ij}, \delta_{ij}\in \Delta_i$.
    \item $\hat\lambda_{ij}^{(r)}(\beta_{ij},\delta_{ij}), \hat\lambda_{ij}^{(n)}(\beta_{ij},\delta_{ij})$, for all $\beta_{ij}\in \cV_{ij}, \delta_{ij}\in \Delta_i$, 
    \item $d_i$, for all $i\in [n]$.
\end{enumerate} \todo{To add some explanation of the variables.} In the LP, we will abuse the notations and let $\lambda_{ij}(t_{ij},\beta_{ij},\delta_{ij})=\lambda_{ij}^{(r)}(t_{ij},\beta_{ij},\delta_{ij})+\lambda_{ij}^{(n)}(t_{ij},\beta_{ij},\delta_{ij})$, $\lambda_{ij}(t_{ij},\beta_{ij})=\sum_{\delta_{ij}\in\Delta_i}\lambda_{ij}(t_{ij},\beta_{ij},\delta_{ij})$. Similarly, let $\hat\lambda_{ij}(\beta_{ij},\delta_{ij})=\hat\lambda_{ij}^{(r)}(\beta_{ij},\delta_{ij})+\hat\lambda_{ij}^{(n)}(\beta_{ij},\delta_{ij})$, $\hat\lambda_{ij}^{(r)}(\beta_{ij})=\sum_{\delta_{ij}\in\Delta_i}\hat\lambda_{ij}^{(r)}(\beta_{ij},\delta_{ij})$, $\hat\lambda_{ij}(\beta_{ij})=\sum_{\delta_{ij}\in\Delta_i}\hat\lambda_{ij}(\beta_{ij},\delta_{ij})$. 


\begin{align*}
    (P)
    &\quad\text{maximize  } \sum_{i\in[n]} \sum_{j\in[m]} \sum_{t_{ij}\in \cT_{ij}}  f_{ij}(t_{ij})\cdot V_{ij}(t_{ij})\cdot
    \\
    &\sum_{\substack{\beta_{ij}\in \cV_{ij},\delta_{ij} \in \Delta_i}} \left(\lambda_{ij}(t_{ij},\beta_{ij}, \delta_{ij})\cdot \ind[V_{ij}(t_{ij})<\beta_{ij} + \delta_{ij}]+\lambda_{ij}^{(r)}(t_{ij},\beta_{ij}, \delta_{ij})\cdot \ind[V_{ij}(t_{ij})=\beta_{ij} + \delta_{ij}]\right)\\
    s.t.\quad & (1)\quad (\pi_i,w_i) \in W_i & \forall i \\
    & (2)\quad f_{ij}(t_{ij})\cdot\sum_{\beta_{ij}\in \cV_{ij}} \lambda_{ij}(t_{ij},\beta_{ij}) = w_{ij}(t_{ij}) & \forall i,j,t_{ij}\in \cT_{ij}\\
    &(3) \quad \sum_i\sum_{t_{ij}\in \cT_{ij}}\pi_{ij}(t_{ij})\leq 1 & \forall j\\
    &(4) \quad \sum_{i\in[n]}\sum_{\beta_{ij} \in \cV_{ij}} \hat\lambda_{ij}(\beta_{ij}) \cdot \Pr_{t_{ij}\sim D_{ij}}[V_{ij}(t_{ij})> \beta_{ij}]+\hat\lambda_{ij}^{(r)}(\beta_{ij}) \cdot \Pr_{t_{ij}\sim D_{ij}}[V_{ij}(t_{ij})=\beta_{ij}] \leq \frac{1}{2} & \forall j\\
    & (5)\quad \frac{1}{2}\sum_{t_{ij} \in \cT_{ij}} f_{ij}(t_{ij}) \lambda_{ij}(t_{ij},\beta_{ij}) \leq  \hat\lambda_{ij}(\beta_{ij}) \cdot \Pr_{t_{ij}}[V_{ij}(t_{ij})> \beta_{ij}]+\hat\lambda_{ij}^{(r)}(\beta_{ij}) \cdot \Pr_{t_{ij}}[V_{ij}(t_{ij})=\beta_{ij}] & \forall i,j,\beta_{ij}\in \cV_{ij}\\
    & (6)\quad\lambda_{ij}^{(r)}(t_{ij},\beta_{ij},\delta_{ij})\leq \hat\lambda_{ij}^{(r)}(\beta_{ij}, \delta_{ij}),\quad \lambda_{ij}^{(n)}(t_{ij},\beta_{ij},\delta_{ij})\leq \hat\lambda_{ij}^{(n)}(\beta_{ij}, \delta_{ij}) & \forall i,j, t_{ij},\beta_{ij},\delta_{ij}\\
    & (7)\quad \sum_{\substack{\beta_{ij}\in \cV_{ij}, \delta_{ij} \in \Delta_i}} \delta_{ij}\cdot \hat\lambda_{ij}(\beta_{ij},\delta_{ij}) \leq d_i  & \forall i,j\\
    & (8)\quad\sum_{i\in[n]} d_i \leq 8\cdot \prev \text{\todo{Change the constant if needed.}}\\
    &(9)\quad \sum_{\substack{\beta_{ij}\in \cV_{ij},\delta_{ij} \in \Delta_i}} \hat\lambda_{ij}(\beta_{ij},\delta_{ij}) = 1 & \forall i,j\\
     & w_{ij}(t_{ij})\geq 0,\lambda_{ij}^{(r)}(t_{ij},\beta_{ij},\delta_{ij})\geq 0,\lambda_{ij}^{(n)}(t_{ij},\beta_{ij},\delta_{ij})\geq 0,\hat\lambda_{ij}^{(r)}(\beta_{ij},\delta_{ij})\geq 0,\hat\lambda_{ij}^{(n)}(\beta_{ij},\delta_{ij})\geq 0 & \forall i,j,t_{ij},\beta_{ij},\delta_{ij}
\end{align*}

}
\section{Sample Access to the Distribution}\label{sec:sample access}


In this section, we focus on the case where we have only access to the bidders' distribution. 
Our goal is again to compute an approximately optimal mechanism. Our plan is as follows: (i) for each $i\in[n]$ and $j\in[m]$, take $O(\log(1/\delta)/\varepsilon^2)$ samples from $D_{ij}$, and let $\widehat{D}$ be the uniform distribution over the samples. By the DKW inequality~\cite{DvoretzkyKW56}, $\widehat{D}_{ij}$ and $D_{ij}$ have Kolmogorov distance (\Cref{def:kolmogorov}) no more than $\varepsilon$ with probability at least $1-\delta$. (ii) We then apply our algorithm in \Cref{thm:main XOS-main body}
to compute an RPP or TPT that is approximately optimal w.r.t. distributions $\{\widehat{D}_{ij}\}_{i\in[n],j\in[m]}$. We show that the computed simple mechanism is approximately optimal for the true distributions as well. The proof of Theorem~\ref{thm:sample access} is postponed to \Cref{appx:sample access}.

\notshow{
Unfortunately we cannot apply Theorem~\ref{thm:bounding-lp-simple-mech} to $\cD$ to compute an approximately optimal TPT mechanism,
since we only have sample access to $\cD$. 
In order to deal with this issue,
for each agent and each item,
we construct an empirical distribution over values that the agent has.
We measure the distance the distance between our empirical and real distribution in Kolmogorov distance as described in Definition~\ref{def:kolmogorov}.
}

\begin{theorem}\label{thm:sample access}
Suppose all bidders' valuations are constrained additive. If for each $i\in[n]$ and $j\in[m]$,  $D_{ij}$ is supported on {numbers in $[0,1]$ with bit complexity no more than $b$}, 
then for any $\eps>0$ and $\delta>0$, with probability $1-\delta$, we can compute in time {$\poly(n,m,1/\varepsilon,\log (1/\delta),b)$} a rationed posted price mechanism or a two-part tariff mechanism, whose revenue is  at least $c\cdot \opt-O(nm^2\eps)$ for some absolute constant $c$. The algorithm takes $O\left(\frac{\log(nm/\delta)}{\varepsilon^2}\right)$ samples from each $D_{ij}$ and assumes query access to a demand oracle for each bidder.
\end{theorem}

\notshow{
We need another theorem from Brustle et al.~\cite{BrustleCD20} that states that the optimal revenue is Lipschitz with respect to the Kolmogorov distance between the type distributions.

\begin{theorem}[Lipschitz Continuity of the Optimal Revenue, adapted from Theorem 6 in~\cite{BrustleCD20}]\label{thm:continuous of OPT under Prokhorov}
For any distributions $\cD=\bigtimes_{i\in[n], j\in[m]} D_{ij}$ and $\widehat{\cD}=\bigtimes_{i\in[n], j\in[m]} \widehat{D}_{ij}$, where $D_{ij}$ and $D_{ij}$ are supported on $[0,1]$ for every $i\in[n]$ and $j\in [m]$

if $d_K(D_{ij},\widehat{D}_{ij})\leq \varepsilon$ for all $i\in[n]$ and $j\in [m]$, then $$\left \lvert\opt(\DD)-\opt(\hDD)\right\rvert\leq O\left(n\kappa+n\sqrt{m\LL H\kappa}\right),$$
		where $\kappa=O\left(n m \LL H \varepsilon+  m \LL \sqrt{nH\varepsilon}\right)$.
	\end{theorem}
}

\notshow{For this section,
we are going to consider general TPT mechanisms,
whose format is described in Mechanism~\ref{def:general SPEM}.
Notice that Mechanism~\ref{def:general SPEM} is a generalization of Mechanism~\ref{def:constructed-SPEM},
since now the seller can decide the entry fee to by any arbitrary (but deterministic) function that depends on the set of items left.
We prove formally that claim in Observation~\ref{obs:restricted to general tpt}.
We prove in Theorem~\ref{thm:sample tpt} that we can calculate a TPT mechanism $\widehat{\Mtpt}$ and a RPSM $\widehat{\Mpp}$,
whose revenue is within a constant multiplicative and additive factor of the optimal revenue induced by any TPT mechanism.

\argyrisnote{Maybe we should emphasize in Theorem~\ref{thm:bounding-lp-simple-mech} that the types samples to set the entry fee are collected before the mechanism starts and do not change.}
\begin{observation}\label{obs:restricted to general tpt}
For any TPT mechanism (as described in Mechanism~\ref{def:constructed-SPEM}),
When the $i$-th agent arrives at the mechanism,
then the entry fee that the buyer has to pay depends only on the set of items that remain and is chosen deterministic.
That is,
there exists a function $\xi_i:2^{[m]}\rightarrow\mathbb{R}_+$ such that the entry fee imposed to the $i$-th agent when the the set of items that remains is $S$, is $\xi_i(S)$.
\end{observation}

\begin{proof}
Consider an instance of Mechanism~\ref{def:constructed-SPEM} $\Mtpt$,
where $\{p_j\}_{j\in[m]}$ are the posted prices and $t_i$ is the sampled type used to determine the entry fee for the $i$-th agent.
Consider the function $\xi_i:2^{[m]}\rightarrow\mathbb{R}_+$ and define $\xi_i(S)=\max_{S'\subseteq S}\left(v_i(t_i',S')-\sum_{j\in S'}Q_j\right)$.
Then observe that the entry fee imposed at the $i$-th agent when the set of items that remain is exactly $S$ is equal to $\xi_i(S)$.
\end{proof}
}


\notshow{

\begin{theorem}\label{thm:sample tpt}
Assume that the agent's valuation is constraint additive and her type is sampled from the distribution $\cD$,
where we only have sample access to it.
Let $\Mtpt$ be the optimal general two part tariff mechanism (as described in Mechanism~\ref{def:general SPEM}).
We can compute in time $\poly\left(n,m,\frac{1}{\eps}\right)$ a RPSM mechanism $\widehat{\Mpp}$ and a general two part tariff mechanism $\widehat{\Mtpt}$  such that there exists constants $c_1,c_2,c_3>0$:
$$
\rev(\Mtpt,\cD) \leq c_1 \rev(\widehat{\Mpp},\cD) + c_2  \rev(\widehat{\Mtpt},\cD) + c_3 n m^2 H \eps
$$

\end{theorem}
\begin{proof}
For each agent $i\in[n]$ and item $j\in[m]$,
we denote by $\widehat{\cD}_{ij}$ the empirical distribution induced by \argyrisnote{$k(\eps)$ samples so that Kolmogorov distance is $\eps$} from $\cD_{ij}$. 
With probability at least $\argyrisnote{smth}$,
for each $i$ and $j$ we have that $||D_{ij}-\widehat{D}_{ij}||_K \leq \eps$.
For the rest of the proof we condition on the event that for each $i\in[n]$ and $j\in[m]$,
$||D_{ij}-\widehat{D}_{ij}||_K \leq \eps$.
For each agent $i\in[n]$ we denote by $\widehat{\cD}_i=\times_{j\in[m]}\widehat{\cD}_{ij}$ and $\widehat{\cD}=\prod_{i\in[n]}\widehat{\cD}_i$.

First we want to relate the expected revenue of $\Mopt$ when we sample the agents type from $\widehat{\cD}$.
We make use of a theorem proved in \cite{CaiD17}.

A corollary of Theorem~\ref{thm:robust mechanism} is the following

$${\Mtpt}({\cD}) \leq \Mtpt(\widehat{\cD}) + 4nm^2H\eps$$

\argyrisnote{I think we should slightly modify the statement of theorem~\ref{thm:bounding-lp-simple-mech}.}
Applying Theorem~\ref{thm:bounding-lp-simple-mech} on the constructed empirical distribution $\widehat{\cD}$ and we get we can compute in time $\poly(n,m,k(\eps))$ a RSPM mechanism $\Mpp$ and a TPT mechanism $\Mtpt$ such that there exists constants $c_1,c_2 \geq 0$ and \argyrisnote{with probability smth} 
\begin{align}
\Mtpt(\widehat{\cD}) \leq c_1 \widehat{\Mpp}(\widehat{\cD}) + c_2 \widehat{\Mtpt}(\widehat{\cD})\label{eq:real tpt}
\end{align}

Moreover we have that
\begin{align}
\widehat{\Mtpt}(\widehat{\cD}) \leq& \widehat{\Mtpt}({\cD}) +  4nm^2H\eps\label{eq:tpt}\\
\widehat{\Mpp}(\widehat{\cD}) \leq& \widehat{\Mpp}({\cD})+ 4nm^2H\eps \label{eq:pp}
\end{align}
Observe that the first inequality is a direct application of Theorem~\ref{thm:robust mechanism} on $\Mtpt$.

In order to prove the second inequality,
it is enough to observe that if we assume that the agents are unit demand, then the set of RPSM is the same as the set of TPT mechanisms without entry fee.
Thus we can directly apply Theorem~\ref{thm:robust mechanism} and assume that the agents valuation are Unit-demand.

By combining Equation~\ref{eq:real tpt}, Equation~\ref{eq:tpt} and Equation~\ref{eq:pp} we prove the statement.

\end{proof}

We conclude with the main theorem of this section.

\begin{theorem}\label{thm:sample opt}
Assume that the agent's valuation is constraint additive and her type is sampled from the distribution $\cD$,
where we only have sample access to it.
Let $\Mopt$ be the optimal BIC mechanism.
We can compute in time $\poly\left(n,m,\frac{1}{\eps}\right)$ a RPSM mechanism $\widehat{\Mpp}$ and a general two part tariff mechanism $\widehat{\Mtpt}$ such that $\argyrisnote{with probability smth}$ there exists constants $c_1,c_2,c_3>0$:
$$
\rev(\Mopt,\cD) \leq c_1 \rev(\widehat{\Mpp},\cD) + c_2  \rev(\widehat{\Mtpt},\cD) + c_3 n m^2 H \eps
$$
\end{theorem}

\begin{proof}
The proof follows by combining Lemma~\ref{lem:caiz17}, Theorem~\ref{thm:chms10} and Theorem~\ref{thm:sample tpt}.
\end{proof}

}
\newpage

\appendix
\section{Additional Preliminaries}\label{appx:prelim}
\begin{definition}~\cite{RubinsteinW15}\label{def:subadditive independent}
Let $\cD_i$ be the type distribution of bidder $i$ and denote by $\cV_i$ her distribution over valuations $v_i(t_i,\cdot)$ where $t_i\sim \cD_i$. We say that $\cV_i$ is subadditiver over independent items if
\begin{itemize}
\item $v_i(\cdot,\cdot)$ has no externalities, that is for any $S\subseteq[m]$, $t_i,t_i' \in \cT_i$ such that $t_{ij}=t_{ij}'$ for $j\in S$,
then $v_i(t_i,S)=v_i(t_i',S)$.
\item $v_i(\cdot,\cdot)$ is monotone, that is for each $t_i\in \cT_i$ and $S\subseteq T \subseteq[m]$, $v_i(t_i,S)\leq v_i(t_i,T)$
\item $v_i(\cdot,\cdot)$ is subadditive function, that is for all $t_i\in \cT_i$ and $S_1,S_2\subseteq [m]$, $v_i(t_i,S_1 \cup S_2)\leq v_i(t_i,S_1)+ v_i(t_i,S_2)$ 
\end{itemize}
\end{definition}

\paragraph{Mechanisms:} A mechanism $M$ in multi-item auctions can be described as a tuple $(x,p)$. For every type profile $t$, buyer $i$ and bundle $S\subseteq [m]$, $x_{iS}(t)$ is the probability of buyer $i$ receiving the exact bundle $S$ at profile $t$, $p_i(t)$ is the payment for buyer $i$ at the same type profile. To ease notations, for every buyer $i$ and types $t_i$, we use $p_i(t_i)=\E_{t_{-i}}[p_i(t_i,t_{-i})]$ as the interim price paid by buyer $i$ and $\sigma_{iS}(t_i)=\E_{t_{-i}}[x_{iS}(t_i,t_{-i})]$ as the interim probability of receiving the exact bundle $S$.

\paragraph{IC and IR constraints:} A mechanism $M=(x,p)$ is BIC if:
$$\sum_{S\subseteq [m]}\sigma_{iS}(t_i)\cdot v_i(t_i,S)-p_i(t_i)\geq \sum_{S\subseteq [m]}\sigma_{iS}(t_i')\cdot v_i(t_i,S)-p_i(t_i'), \forall i,t_i,t_i'\in \cT_i.$$

The mechanism is DSIC if:
$$\sum_{S\subseteq [m]}x_{iS}(t_i,t_{-i})\cdot v_i(t_i,S)-p_i(t_i,t_{-i})\geq \sum_{S\subseteq [m]}x_{iS}(t_i',t_{-i})\cdot v_i(t_i,S)-p_i(t_i',t_{-i}), \forall i,t_i,t_i'\in \cT_i,t_{-i}\in \cT_{-i}.$$

The mechanism is (interim) IR if:
$$\sum_{S\subseteq [m]}\sigma_{iS}(t_i)\cdot v_i(t_i,S)-p_i(t_i)\geq 0, \forall i,t_i\in \cT_i.$$

The mechanism is ex-post IR if:
$$\sum_{S\subseteq [m]}x_{iS}(t_i,t_{-i})\cdot v_i(t_i,S)-p_i(t_i,t_{-i})\geq 0, \forall i,t_i\in \cT_i,t_{-i}\in \cT_{-i}.$$

\begin{definition}[Separation Oracle for Convex Polytope $\cP$]
A Separation Oracle $SO$ for a convex polytope $\cP\subseteq \mathbb{R}^d$,
takes as input a point $\bx \in \mathbb{R}^d$ and if $\bx \in \cP$,
then the oracle says that the point is in the polytope.
If $\bx \notin \cP$,
then the oracle output a separating hyperplane,
that is it outputs a vector $\textbf{y}\in \mathbb{R}^d$ and $c\in \mathbb{R}$ such that $\textbf{y}^T x \leq c$,
but for $\textbf{z} \in \cP$, $\textbf{y}^T \textbf{z} > c$.
\end{definition}

\notshow{
\begin{definition}[Scalable Demand Oracle for XOS Valuations]\label{def:scalable demand oracle}
For every buyer $i$, a scalable demand oracle with input $\dem_i(t_i,\{b_j\}_{j\in[m]},\{p_j\}_{j\in[m]} )$ for an XOS valuation $v_i(\cdot,\cdot)$,
takes an input the bidder's type $t_i$, the non-negative coefficient vector $\{b_j\}_{j\in[m]}$ and the non-negative price vector $\{p_j\}_{j\in[m]}$, and outputs a subset of items $S^*\subseteq[m] $ and $k^*\in[K]$ such that
$$
(S^*,k^*) \in \arg\max_{S\subseteq [m],k\in[K]} \sum_{j\in S}b_j \alpha_{ij}^{(k)}(t_{ij}) - \sum_{j\in S}p_j
$$
\end{definition}


\begin{definition}[Value Oracle for XOS Valuation $v_i(\cdot,\cdot)$]
Consider a type $t_i\in \cT_i$ for bidder $i$ and a $k\in[K]$.
A value oracle $\val(\cdot,\cdot)$ for XOS valuation $v_i(\cdot,\cdot)$,
takes an input the bidder's type, and a $k\in[K]$ and outputs $\{a_{ij}^{(k)}(t_{ij})\}_{j\in[m]}$.
\end{definition}
}

\begin{definition}[Polytopes and Facet-Complexiy]\label{def:facet-complexity} We say $P$ has \emph{facet-complexity} at most $b$ if it can be written as $P:= \{\vec{x}\mid \vec{x} \cdot \vec{w}^{(i)} \leq c_i,\ \forall i \in \mathcal{I}\}$, where each $\vec{w}^{(i)}$ and $c_i$ has bit complexity at most $b$ for all $i \in \mathcal{I}$. We use the term \emph{convex polytope} to refer to a set of points that is closed, convex, bounded,\footnote{$P\subseteq \mathbb{R}^d$ is bounded if it is contained in $[-x,x]^d$ for some $x \in \mathbb{R}$.} and has finite facet-complexity.
\end{definition}

\begin{definition}[Vertex-Complexity]\label{def:vertex-complexity}
We use the term \emph{corner} to refer to non-degenerate extreme points of a convex polytope. In other words, $\vec{y}$ is a corner of the $d$-dimensional convex polytope $P$ if $\vec{y} \in P$ and there exist $d$ linearly independent directions $\vec{w}^{(1)},\ldots,\vec{w}^{(d)}$ such that $\vec{x} \cdot \vec{w}^{(i)} \leq \vec{y} \cdot \vec{w}^{(i)}$ for all $\vec{x} \in P, 1 \leq i \leq d$. We use \textbf{$\Cor(P)$} to denote the set of corners of a convex polytope $P$. We say $P$ has vertex-complexity at most $b$ if all vectors in $\Cor(P)$ have bit complexity no more than $b$.
\end{definition}

The following fact states that the vertex-complexity and facet-complexity of a polytope in $\mathbb{R}^d$ are off by at most a $d^2$ multiplicative factor.

\begin{fact}[Lemma 6.2.4 of~\cite{grotschel2012geometric}]\label{fact:corners2}
Let $P$ be a convex polytope in $\mathbb{R}^d$. If $P$ has facet-complexity at most $b$, its vertex-complexity is at most $O(b\cdot d^2)$. Similarly, if $P$ has vertex-complexity at most $\ell$, its facet-complexity is at most $O(\ell\cdot d^2)$.
\end{fact}

\begin{theorem}\label{thm:ellipsoid}[Ellipsoid Algorithm for Linear Programming~\cite{grotschel2012geometric}\footnote{Properties 1 and 2 follow from Theorem 6.4.9 of~\cite{grotschel2012geometric} and Property 3 follows from Remark 6.5.2 of~\cite{grotschel2012geometric}.}]
Let $P$ be a convex polytope in $\mathbb{R}^d$ specified via a separation oracle $SO$, and $\vec{c}$ is any fixed vector in $\mathbb{R}^d$. 
Assume that $P$'s facet-complexity and the bit complexity of $\vec{c}$ are no more than $b$. Then we can run the ellipsoid algorithm to optimize $\vec{c}\cdot\vec{x}$ over $P$, maintaining the following properties:
\begin{enumerate}
\item The algorithm will only query $SO$ on rational points with bit complexity $\poly(d,b)$.
\item The algorithm will solve the linear program in time $\poly(d, b,\rt_{SO}(\poly(d,b)))$, where $\rt_{SO}(x)$ is the running time of the SO on any input of bit complexity $x$.
\item The output optimal solution is a corner of $P$.
\end{enumerate}
\end{theorem}

\section{Some Examples}\label{sec:example}
\subsection{Non-Concavity of~\core}\label{sec:core non-concavity}

In this section, we show that the~$\core(\sigma,\theta(\sigma))$ function is non-concave in the interim allocation rule $\sigma$. We first provide the formal definition of $\theta(\sigma)$ for a single-bidder two-item instance, and we use $\icore(\sigma)$ to denote $\core(\sigma,\theta(\sigma))$.
\begin{definition}[Core for a single additive bidder over two items with continuous distributions - \cite{CaiZ17}]\label{def:core-single interim}
Consider a single bidder interested in two items, whose value is sampled from continuous distribution $D$ with support $T=\supp(D)$ and density function $f(t)$ for $t\in T$.
Consider a feasible interim allocation $\sigma=\{\sigma_1(t),\sigma_2(t)\}_{t\in \supp(D)}$,
that is $\sigma_1(t)$ ($\sigma_2(t)$ resp.) is the probability that the allocation rule awards item $1$ (item $2$ resp.) to a bidder with type $t$.
Define 
$$\beta_1(\sigma) = \argmin_{\beta\geq 0}\left[ \Pr_{t_1\sim D_1} \left[t_1\geq \beta\right]=\E_{t\sim D}\left[\sigma_1(t) \right]\right] \quad \quad \quad \beta_2(\sigma) = \argmin_{\beta\geq 0}\left[ \Pr_{t_2\sim D_2} \left[t_2\geq \beta\right]=\E_{t\sim D}\left[\sigma_2(t) \right] \right] $$
and
$$c(\sigma)= \argmin_{a\geq 0}\left\{\Pr_{t\sim D}\left[t_1\leq \beta_1(\sigma)+a\right]+\Pr_{t\sim D}\left[t_2\leq \beta_2(\sigma)+a\right]\geq\frac{1}{2}\right\}$$
The term $\icore$ for interim allocation $\sigma$ is defined as follows:
%
$$\icore(\sigma)=
\E_{t\sim D}\left[\sigma_{1}(t)t_{1}\cdot \ind[t_{1}\leq \beta_{1}(\sigma)+c(\sigma)]\right]
+
\E_{t\sim D}\left[\sigma_{2}(t)t_{2}\cdot \ind[t_{2}\leq \beta_{2}(\sigma)+c(\sigma)]\right]$$

\end{definition}

In Example~\ref{ex:non-concave interim} we show that 
$\icore(\sigma)$ is a non-concave function even in the setting with a single bidder and two items. The reason for the $\icore$ being non-concave lies in the fact that the interval which we truncate depends on the interim allocation $\sigma$.
Computing the concave hull of $\icore(\sigma)$ in the worst case requires exponential time in the dimension of the space, which is $m$ is our case. 

\begin{example}\label{ex:non-concave interim}
Consider a single additive bidder interested in two items whose values are both drawn from the uniform distribution $U[0,1]$. 
Consider two interim allocation rules $\sigma$ and $\sigma'$:
\begin{itemize}
    \item $\sigma$: Award the first item to the buyer if her value for it lies in the interval $[0,1/2]$ and never award the second item to the buyer.
    \item $\sigma'$: Always award the first item to the buyer and never award the second item to the buyer.
\end{itemize}

\noindent According to Definition~\ref{def:core-single},
for allocation rule $\sigma$, the dual parameters are $\beta_1(\sigma)=1/2$, $\beta_2(\sigma)=1$ and $c(\sigma)=0$,
which implies $\icore(\sigma)=1/8$.
Similarly for allocation rule $\sigma'$ we have $\beta_1(\sigma')=0$, $\beta_2(\sigma')=1$ and $c(\sigma')=0$,
which implies that $\icore(\sigma')=0$.

Consider the interim allocation $\sigma''$ that uses  allocation rule $\sigma$ with probability $50\%$ and $\sigma'$ with $50\%$.
Note that $\sigma''$ is in the convex combination of $\sigma$ and $\sigma'$ 
and more specifically $\sigma''=\frac{\sigma+\sigma'}{2}$.
For interim allocation $\sigma''$ we have that $\beta_1(\sigma'')=1/4$, $\beta_2(\sigma'')=1$ and $c(\sigma'')=0$,
which implies that $\icore(\sigma'')=\frac{1}{32}$.
We notice that the second item contributes nothing to the $\icore$,
but it ensures that $c=0$ regardless of the allocation of the first item.
Thus $\icore\left(\sigma''\right) < \frac{1}{2}(\icore(\sigma)+\icore(\sigma'))$, which implies that $\icore(\cdot)$ is not a concave function.
\end{example}

\subsection{Why can't we use the Ex-Ante Relaxation?}\label{sec:exante}

An influential framework known as the ex-ante relaxation has been widely used in Mechanism Design, but is insufficient for our problem. Informally speaking,
ex-ante relaxation reduces a multi-bidder objective to the sum of  single-bidder objectives subject to some global supply constraints over ex-ante allocation probabilities. 
To solve the ex-ante relaxation program efficiently, the single-bidder objective has to be concave and efficiently computable given the ex-ante probabilities~\cite{Alaei11}. 

In revenue maximization, the single-bidder objective -- the optimal revenue subject to ex-ante probabilities -- is indeed a concave function. However, we do not have a polynomial time algorithm to even compute the single-bidder objective given a set of fixed ex-ante probabilities.\footnote{The closest thing we know is a QPTAS for a unit-demand bidder. See \Cref{sec:related work}.} To fix this issue, one can try to find a concave function that is always a good approximation to the single-bidder objective for any ex-ante probabilities. To the best of our knowledge, such a concave function only exists for unit-demand bidders via the copies setting technique~\cite{ChawlaHMS10}. Alternatively, one can replace the global objective -- optimal revenue by the upper bound of revenue proposed in~\cite{CaiZ17}. Yet the corresponding single-bidder objective for one term $\core$ in the upper bound is highly non-concave, which makes the ex-ante relaxation not applicable.

Although the term $\core$ was originally defined for interim allocation rules (as in \Cref{def:core-single interim}), it can also be defined for ex-ante probabilities. We only define it for the single-bidder two-item case. Let $q=\{q_1,q_2\}\in[0,1]^2$,
and $\maxcore=\max_{\sigma\in \Sigma(q)} \icore(\sigma)$,
where $\Sigma(q)$ is the set of feasible interim allocations that awards the first item with probability at most $q_1$ and the second item with probability at most $q_2$.
Example~\ref{ex:non convex core} also shows that $\maxcore(\cdot)$ is a non-concave function by observing that $\sigma \in \Sigma(1/2,0)$,$\sigma' \in \Sigma(1,0)$ and $\sigma'' \in \Sigma(3/4,0)$.

\begin{definition}[Core for a single additive bidder over two items - \cite{CaiZ17}]\label{def:core-single}
Consider a single bidder interested in two items, whose value is sampled from $D_1\times D_2$.
Consider a supply constraints $q_1,q_2 \in [0,1]$. Note that $q_1$ (or $q_2$) is the probability that a mechanism awards the first item (or the second item) to the bidder.
Define 
$$\beta_1 = \argmin_{\beta\geq 0}\left[ \Pr_{t_1\sim D_1} \left[t_1\geq \beta\right]=q_1 \right] \quad \quad \quad \beta_2 = \argmin_{\beta\geq 0}\left[ \Pr_{t_2\sim D_2} \left[t_2\geq \beta\right]=q_2 \right] $$
and

$$c= \argmin_{a\geq 0}\left\{\Pr_{t_1\sim D_1}\left[t_1\leq \beta_1+a\right]+\Pr_{t_2\sim D_2}\left[t_2\leq \beta_2+a\right]\geq\frac{1}{2}\right\}$$
The term $\maxcore$ is defined as follows:
$$
\maxcore\left(q\right) = \max_{\substack{x_1:\cT_1 \rightarrow [0,1]\\ \sum_{t_1\in \cT_1}f_1(t_1) x_1(t_1)= q_1}}\sum_{\substack{t_1\in \cT_1 \\ t_1 \leq \beta_1 + c}} f_1(t_1)\cdot t_1 \cdot x_1(t_1) + 
\max_{\substack{x_2:\cT_2 \rightarrow [0,1]\\ \sum_{t_2\in \cT_2}f_2(t_2) x_2(t_2)=q_2}}
\sum_{\substack{t_2\in \cT_2 \\ t_2 \leq \beta_2 + c}} f_2(t_2)\cdot t_2 \cdot x_2(t_2)
$$
\end{definition}
\notshow{
\begin{definition}[Core for an additive bidder - \cite{ChawlaM16}]
Consider a single additive bidder sampled from $D$ with supply constraints $\{q_{j}\}_{j\in[m]}\in [0,1]^{m}$ (Note that $q_{j}$ is the probability that a mechanism awards item $j$ to bidder $i$).
For item $j\in[m]$, define 
$$\beta_{j} = \argmin_{\beta\geq 0}\left[ \Pr_{t_j\sim D_j} \left[t_j\geq \beta_j\right]=q_j \right] $$
$$c= \argmin_{c\geq 0}\left[ \Pr_{t\sim D}\left[\forall_{j\in[m]}t_j\leq \beta_j+c \right]\geq\frac{1}{2}\right]$$
The term $\core$ is defined as follows:
$$
\core\left(\{q_j\}_{j\in[m]}\right) =
$$
\end{definition}

}

In Example~\ref{ex:non convex core} we show that 
$\maxcore(q)$ is a non-concave function even in the setting with a single bidder and two items. The reason for the $\maxcore$ being non-concave lies in the fact that the interval which we truncate depends on the supply constraints $q$.
Computing the concave hull of $\maxcore(q)$ in the worst case requires exponential time in the dimension of the space, which is $m$ is our case. These facts make the ex-ante relaxation approach not applicable to solve our problem.

\begin{example}\label{ex:non convex core}
Consider a single additive bidder interested in two items whose values are both drawn from the uniform distribution $U[0,1]$. 
Consider the values $q=(1/2,0)$ and $q'=(1,0)$.
According to Definition~\ref{def:core-single},
for $q$ we have 
that $\beta^{(q)}_1 = 1/2,\beta^{(q)}_2 = 1$ and $c^{(q)}=0$ 
and for $q'$ we have $\beta^{(q')}_1= 0,\beta^{(q')}_2=1$ and $c^{(q')}=0$.
We notice that the second item contributes nothing to the $\maxcore$,
but it ensures that $c=0$ regardless of the supply demand for the first item.
Observe that $\maxcore(q) = 1/8 $ and
$\maxcore(q') =  0$. Let $q'' = (q+q')/2=(3/4,0)$.
For $q''$, observe that $\beta_1^{(q'')} = 1/4,\beta_2^{(q'')}=1$ and $c^{(q'')}=0$. We have $\maxcore(q'') = 1/32$. Thus $\maxcore\left(q''\right) < \frac{1}{2}(\maxcore(q)+\maxcore(q'))$, which implies that $\maxcore(\cdot)$ is not a concave function.
\end{example}


\section{Missing Details from Section~\ref{sec:program}}\label{sec:appx_program}


In this section, we provide a proof of \Cref{thm:bounding-lp-simple-mech}. Indeed, we prove a generalization that works for XOS buyers (\Cref{thm:bounding-lp-simple-mech-XOS}), with the generalized of the single-bidder marginal reduced form polytope~\Cref{def:W_i} and a generalized LP (\Cref{fig:XOSLP}). 

\begin{theorem}\label{thm:bounding-lp-simple-mech-XOS}
Let $(w,\lambda,\hat\lambda, \bd)$ (or  $(\pi,w,\lambda,\hat\lambda, \bd)$) be any feasible solution of the LP in \Cref{fig:bigLP} (or \Cref{fig:XOSLP}).
Let $\Mpp$ be the rationed posted price mechanism computed in \Cref{thm:chms10}. Let $\Mtpt$ be the two-part tariff mechanism shown in Mechanism~\ref{def:constructed-SPEM} with prices $\{Q_j\}_{j\in[m]}$ (\Cref{def:Q_j-XOS}). Then the objective function of the solution $2\cdot \sum_{j\in [m]}Q_j$ is bounded by $c_1\cdot \rev(\Mpp)+c_2\cdot \rev(\Mtpt)$, for some constant $c_1,c_2>0$. 
Moreover, both $\Mpp$ and $\Mtpt$ can be computed in time $\poly(n,m,\sum_{i,j}|\cT_{ij}|)$, 
with access to the demand oracle for the buyers' valuations. 
\end{theorem}

\subsection{Result by Cai and Zhao~\cite{CaiZ17} for XOS Valuations}\label{appx_cai-zhao}

The result by Cai and Zhao~\cite{CaiZ17} applies also to XOS valuations. Here we state their result for this general case. Note that this is a generalized definition and lemma for \Cref{def:core-constrained-additive} and \Cref{lem:caiz17-constrained-additive}.

\begin{definition}\label{def:core}
For any $i\in [n],j\in [m]$, let $\cV_{ij}^0 = \{ V_{ij}(t_{ij}): t_{ij} \in \cT_{ij}\}$. For any feasible interim allocation $\sigma$, and non-negative numbers {$\tilde{\vBeta}=\{\tilde{\beta}_{ij}\in \cV_{ij}^0\}_{i\in[n],j\in[m]}$, $\vC=\{c_i\}_{i\in[n]}$} and $\vr=\{r_{ij}\}_{i\in [n],j\in [m]}\in [0,1]^{nm}$ (which we refer to as the dual parameters), 
define $\core(\sigma,\tilde{\vBeta},\vC,\vr)$ as the welfare under allocation $\sigma$ truncated at $\tilde{\beta}_{ij}+c_i$ for every $i,j$.
Formally, $$\core(\sigma,\tilde{\vBeta},\vC,\vr)=\sum_i\sum_{t_i}f_i(t_i)\cdot \sum_{S\subseteq [m]}\sigma_{iS}(t_i)\sum_{j\in S}t_{ij}\cdot \left(\ind[t_{ij}< \tilde{\beta}_{ij}+c_i]+r_{ij}\cdot \ind[t_{ij}= \tilde{\beta}_{ij}+c_i] \right)$$
if the buyers have constrained-additive valuations, and
$$\core(\sigma,\tilde{\vBeta},\vC,\vr)=\sum_i\sum_{t_i}f_i(t_i)\cdot \sum_{S\subseteq [m]}\sigma_{iS}(t_i)\sum_{j\in S}\gamma_{ij}^S(t_i)\cdot \left(\ind[V_{ij}(t_{ij})< \tilde{\beta}_{ij}+c_i]+r_{ij}\ind[V_{ij}(t_{ij})= \tilde{\beta}_{ij}+c_i] \right)$$ if the buyers have XOS valuations.
Here $\gamma_{ij}^S(t_i)=\alpha_{ij}^{k^*(t_i,S)}(t_{ij})$, where $\displaystyle k^*(t_i,S)=\arg\max_{k\in[K]}\big(\sum_{j\in S}\alpha^k_{ij}(t_{ij})\big)$.
\end{definition}

\begin{lemma}\cite{CaiZ17}\label{lem:caiz17}
Given any BIC and IR mechanism $\cM$, there exists \textbf{(i)} a feasible interim allocation $\sigma$,\footnote{Note that when buyers have constrained-additive valuations, it suffice to take $\sigma$ to be the interim allocation rule of $\cM$. For XOS valuations, $\sigma$ will be the interim allocation of a modified version of $\cM$. See Section~5 in~\cite{CaiZ16a} for details.} where $\sigma_{iS}(t_i)$ is the interim probability for buyer $i$ to receive exactly bundle $S$ when her type is $t_i$, \textbf{(ii)} non-negative numbers {$\tilde{\vBeta}^{(\sigma)}=\{\tilde{\beta}_{ij}^{(\sigma)}\in \cV_{ij}^0\}_{i\in[n],j\in[m]}$, $\vC^{(\sigma)}=\{c_i^{(\sigma)}\}_{i\in[n]}$} 
 and $\vr^{(\sigma)}\in [0,1]^{nm}$ that depend on $\sigma$, and \textbf{(iii)} a two-part tariff mechanism $\cM_1^{(\sigma)}$ such that 
\begin{enumerate}
    \item $\sum_{i\in[n]} \left(\Pr_{t_{ij}}[V_{ij}(t_{ij})>\tilde\beta_{ij}^{(\sigma)}]+r_{ij}^{(\sigma)}\cdot \Pr_{t_{ij}}[V_{ij}(t_{ij})=\tilde\beta_{ij}^{(\sigma)}]\right)\leq \frac{1}{2}$.
    \item $\frac{1}{2}\cdot\sum_{t_i\in \cT_i}f_i(t_i)\cdot \sum_{S:j\in S}\sigma_{iS}(t_i)\leq \Pr_{t_{ij}}[V_{ij}(t_{ij})>\tilde\beta_{ij}^{(\sigma)}]+r_{ij}^{(\sigma)}\cdot \Pr_{t_{ij}}[V_{ij}(t_{ij})=\tilde\beta_{ij}^{(\sigma)}],\forall i,j$.
    \item $\rev(\cM)\leq 28\cdot \prev+4\cdot\core(\sigma,\tilde{\vBeta}^{(\sigma)},\vC^{(\sigma)},\vr^{(\sigma)})$.
    \item $\sum_{i\in [n]} c_i^{(\sigma)}\leq 8\cdot \prev$.
    \item $\core(\sigma,\tilde{\vBeta}^{(\sigma)},\vC^{(\sigma)},\vr^{(\sigma)})\leq 64\cdot \prev + 8\cdot \rev(\cM_1^{(\sigma)})$.
\end{enumerate}
\end{lemma}

\subsection{Single-Bidder Marginal Reduced Form Polytope for XOS Valuations}\label{sec:mrf-XOS}
In Definition~\ref{def:W_i} we define the single-bidder marginal reduced form polytope $W_i$ for XOS buyers, which differs from the single-bidder marginal reduced form polytope for constrained-additive valuation is several ways.
In Definition~\ref{def:W_i},
we define a distribution $\sigma_S^k$ over all possible subset of items $S\subseteq [m]$ and over the finite number $k\in[K]$ over additive functions that can be chosen when we evaluate the value that the buyer has for a set of items. In Definition~\ref{def:W_i-constrained-add}, the distribution $\sigma_S$ was only over sets in the set of feasible allocations. 

Similar to Definition~\ref{def:W_i-constrained-add},
$\pi_{ij}(t_{ij})$ is equal to $f_{ij}(t_{ij})$ times the probability that the $i$-th buyer receives the $j$-th item.
In contrast to Definition~\ref{def:W_i-constrained-add},
in Definition~\ref{def:W_i},
the value of $w_{ij}(t_{ij})$ is $\frac{f_{ij}(t_{ij})}{V_{ij}(t_{ij})}$ times the expected value that the buyer has for the item when we are allowed to choose which additive functions in $k\in[K]$ we count the value of the buyer, or we are even allowed to allocate an item to the buyer but count zero value for it (that is equivalent to just throwing away the item).

\begin{definition}[XOS valuations: single-bidder marginal reduced form polytope]\label{def:W_i}
For every $i\in [n]$, the single-bidder marginal reduced form polytope $W_i\subseteq [0,1]^{2\cdot \sum_j|\cT_{ij}|}$ is defined as follows. Let $\pi_i=(\pi_{ij}(t_{ij}))_{j,t_{ij}\in \cT_{ij}}$ and $w_i=(w_{ij}(t_{ij}))_{j,t_{ij}\in \cT_{ij}}$. Then $(\pi_i,w_i)\in W_i$ if and only if there exist a number $\sigma_S^{(k)}(t_i)\in [0,1]$ for every $t_i\in \cT_i, S\subseteq [m],k\in [K]$, such that
\begin{enumerate}
    \item $\sum_{S,k}\sigma_S^{(k)}(t_i)\leq 1 $, $\forall t_i\in \cT_i$.
    \item $
    {\pi_{ij}(t_{ij})=}f_{ij}(t_{ij})\cdot\sum_{t_{i,-j}}f_{i,-j}(t_{i,-j})\cdot \sum_{S:j\in S}\sum_{k\in [K]}\sigma_S^{(k)}(t_{ij},t_{i,-j})$, for all $i,j,t_{ij}\in \cT_{ij}$.
    \item ${w_{ij}(t_{ij})\leq} f_{ij}(t_{ij})\cdot\sum_{t_{i,-j}}f_{i,-j}(t_{i,-j})\cdot \sum_{S:j\in S}\sum_{k\in [K]}\sigma_S^{(k)}(t_{ij},t_{i,-j})\cdot \frac{\alpha_{ij}^{(k)}(t_{ij})}{V_{ij}(t_{ij})}$, for all $i,j,t_{ij}\in \cT_{ij}$.
\end{enumerate}
\end{definition}

\subsection{The Linear Program for XOS valuations}\label{sec:program-XOS}


\begin{figure}[H]
\colorbox{MyGray}{
\begin{minipage}{.98\textwidth}
\small
$$\quad\textbf{max  } \sum_{i\in[n]} \sum_{j\in[m]} \sum_{t_{ij}\in \cT_{ij}} 
 f_{ij}(t_{ij})\cdot V_{ij}(t_{ij})\cdot \sum_{\substack{\beta_{ij}\in \cV_{ij}\\\delta_{ij} \in \Delta}} \lambda_{ij}(t_{ij},\beta_{ij}, \delta_{ij})\cdot \ind[V_{ij}(t_{ij})\leq \beta_{ij} + \delta_{ij}]$$
\vspace{-.3in}
  \begin{align*}
 & \textbf{s.t.}\\
 &\quad\textbf{Allocation Feasibility Constraints:}\\
 &\quad\Wconstraint \quad (\pi_i,w_i) \in W_i & \forall i \\
&\quad\PiConstraint \quad \sum_i\sum_{t_{ij}\in \cT_{ij}}\pi_{ij}(t_{ij})\leq 1 & \forall j\\  
 &\quad\textbf{Natural Feasibility Constraints:}\\
    &\quad\LambdaMarginalConstraint\quad f_{ij}(t_{ij})\cdot\sum_{\beta_{ij}\in \cV_{ij}}\sum_{\delta_{ij}\in \Delta} \lambda_{ij}(t_{ij},\beta_{ij},\delta_{ij}) = w_{ij}(t_{ij}) & \forall i,j,t_{ij}\in \cT_{ij}\\
    &\quad\CompareMarginalConstraint\quad\lambda_{ij}(t_{ij},\beta_{ij},\delta_{ij})\leq \hat\lambda_{ij}(\beta_{ij}, \delta_{ij}) & \forall i,j, t_{ij},\beta_{ij}\in \cV_{ij},\delta_{ij}\\
    &\quad\HatLambdaDistributionConstraint\quad \sum_{\substack{\beta_{ij}\in \cV_{ij}\\\delta_{ij} \in \Delta}} \hat\lambda_{ij}(\beta_{ij},\delta_{ij}) = 1 & \forall i,j\\
    &\quad\textbf{Problem Specific Constraints:}\\
    &\quad\ReduceDemandConstaint \quad \sum_{i\in[n]}\sum_{\beta_{ij} \in \cV_{ij}}\sum_{\delta_{ij}\in \Delta} \hat\lambda_{ij}(\beta_{ij},\delta_{ij}) \cdot \Pr_{t_{ij}\sim D_{ij}}[V_{ij}(t_{ij})\geq \beta_{ij}] \leq \frac{1}{2}& \forall j\\
    & \quad\MarginalToGlobalConstraint\quad \frac{1}{2}\sum_{t_{ij} \in \cT_{ij}} f_{ij}(t_{ij}) \left(\lambda_{ij}(t_{ij},\beta_{ij},\delta_{ij})+ \lambda_{ij}(t_{ij},\beta_{ij}^+,\delta_{ij})\right) \leq \\
    &\qquad\qquad\qquad\hat\lambda_{ij}(\beta_{ij},\delta_{ij}) \cdot \Pr_{t_{ij}}[V_{ij}(t_{ij})\geq \beta_{ij}]+\hat\lambda_{ij}(\beta_{ij}^+,\delta_{ij}) \cdot \Pr_{t_{ij}}[V_{ij}(t_{ij})\geq\beta_{ij}^+] & \forall i,j,\beta_{ij}\in \cV_{ij}^0,\delta_{ij}\in \Delta\\
    &\quad\BoundMeanDeltaConstraint\quad \sum_{\substack{\beta_{ij}\in \cV_{ij}\\ \delta_{ij} \in \Delta}} \delta_{ij}\cdot \hat\lambda_{ij}(\beta_{ij},\delta_{ij}) \leq d_i  & \forall i,j\\
    &\quad\BoundSumDeltaConstraint\quad\sum_{i\in[n]} d_i \leq 
    {111\cdot \estprev}\\
     & \lambda_{ij}(t_{ij},\beta_{ij},\delta_{ij})\geq 0,\hat\lambda_{ij}(\beta_{ij},\delta_{ij})\geq 0,\pi_{ij}(t_{ij})\geq 0, w_{ij}(t_{ij})\geq 0, d_i\geq 0 & \forall i,j,t_{ij},\beta_{ij}\in \cV_{ij},\delta_{ij}
\end{align*}
\end{minipage}}
\caption{LP for XOS Valuations }~\label{fig:XOSLP}
\end{figure}

The LP for XOS valuations can be found in \Cref{fig:XOSLP}. Here $V_{ij}^0=\{V_{ij}(t_{ij}):t_{ij}\in \cT_{ij}\}$, $\cV_{ij}^+ = \{V_{ij}(t_{ij}) + \eps_r: t_{ij} \in \cT_{ij}\}$ and $\cV_{ij} = \cV^0_{ij} \cup \cV^+_{ij}$. We notice that this is consistent with our LP for constrained-additive buyers (\Cref{fig:bigLP}), as $V_{ij}(t_{ij})=t_{ij}$ for constrained-additive buyers.


Denote $\optlp$ the optimum objective of the LP in \Cref{fig:XOSLP}. Similar to the constrained-additive case, we have the following lemma.

\begin{lemma}\label{lem:bound rev by opt-XOS}
When buyers have XOS valuations, for any BIC and IR mechanism $\cM$, $\rev(\cM)\leq 28\cdot \prev+4\cdot\optlp$.
\end{lemma}

\subsection{Proof of \Cref{lem:bound rev by opt} and \Cref{lem:bound rev by opt-XOS}}\label{subsec:proof_lem:bound rev by opt}

\begin{proof}
{The proof is stated for XOS buyers, whose LP contains a new set of variables $\pi$ compared to the LP for constrained-additive buyers. When the buyers have constrained-additive valuations, we can simply treat $\pi$ to be the same as $w$. Also, note that $V_{ij}(t_{ij})=t_{ij}$ for constrained-additive valuations.} 

Let tuple $(\hat{\sigma},\tilde{\vBeta},\vC,\vr)$ be the one stated in Lemma~\ref{lem:caiz17} for $\cM$. Consider the following choice of variables of the LP in \Cref{fig:bigLP} (or \Cref{fig:XOSLP}). For every $i,j,t_{ij}$, let 

$$w_{ij}(t_{ij})=f_{ij}(t_{ij})\cdot\sum_{t_{i,-j}\in \cT_{i,-j}} f_{i,-j}(t_{i,-j})\sum_{S:j\in S}\hat\sigma_{iS}(t_{ij},t_{i,-j})$$
if the buyers' have constrained-additive valuations. Let
$$\pi_{ij}(t_{ij})=f_{ij}(t_{ij})\cdot\sum_{t_{i,-j}\in \cT_{i,-j}} f_{i,-j}(t_{i,-j})\sum_{S:j\in S}\hat\sigma_{iS}(t_{ij},t_{i,-j})$$
$$w_{ij}(t_{ij})=f_{ij}(t_{ij})\cdot\sum_{t_{i,-j}\in \cT_{i,-j}} f_{i,-j}(t_{i,-j})\sum_{S:j\in S}\hat\sigma_{iS}(t_{ij},t_{i,-j})\cdot\frac{\gamma_{ij}^{S}(t_i)}{V_{ij}(t_{ij})}.$$ 
if the buyers' have XOS valuation.
For each $c_i$, we notice that by \Cref{lem:caiz17} and \Cref{thm:chms10}, $0\leq c_i\leq 8\cdot \prev\leq 55\cdot\estprev$ when $n\cdot m\geq 110$. We round it up to the closest number in $\Delta$, and we denote it using $\hat{c}_i$. Clearly, $$\core(\sigma,\tilde{\vBeta},\vC,\vr)\leq \core(\sigma,\tilde{\vBeta},\hat{\vC},\vr).$$

 $\lambda_{ij}(t_{ij},\beta_{ij},\delta_{ij})$ and $\hat\lambda_{ij}(\beta_{ij},\delta_{ij})$ can be set to non-zero only if $\beta_{ij}\in\{\tilde{\beta}_{ij}, \tilde{\beta}_{ij}^+\}$ and $\delta_{ij}=\hat{c}_i$.
More specifically, we choose the variables as follows.  \begin{itemize}
    \item $\lambda_{ij}(t_{ij},\tilde{\beta}_{ij},c_i):= r_{ij}\cdot w_{ij}(t_{ij})/f_{ij}(t_{ij})$, 
    \item $\lambda_{ij}(t_{ij},\tilde{\beta}_{ij}^+,c_i):= (1-r_{ij})\cdot w_{ij}(t_{ij})/f_{ij}(t_{ij})$,
    \item $\hat\lambda_{ij}(\tilde{\beta}_{ij},c_i)=r_{ij}$,
    \item $\hat\lambda_{ij}(\tilde{\beta}_{ij}^+,c_i)=1-r_{ij}$;
    \item $d_i=\hat{c}_i$.
\end{itemize}


We show that this is indeed a feasible solution of the LP in \Cref{fig:bigLP} by verifying each constraint. 
We first prove that $(\pi_i,w_i)\in W_i$ for every $i$. This is clear for constrained-additive valuations. For XOS valuations,  consider the mapping $\sigma_{iS}^{(k)}(t_i)=\hat\sigma_{iS}(t_i)\cdot \ind[k=\argmax_{k'\in [K]}\sum_{j\in S}\alpha_{ij}^{(k')}(t_{ij})]$ for every $t_i$ (we break ties arbitrarily). Thus by the definition of $\gamma_{ij}^S$, we have $\hat\sigma_{iS}(t_i)\cdot \gamma_{ij}^S(t_i)=\sum_{k}\sigma_{iS}^{(k)}(t_i)\cdot \alpha_{ij}^{(k)}(t_i)$. Then clearly $(\pi_i,w_i)$ satisfies all of the conditions in~\Cref{def:W_i}. 

For Constraint~\LambdaMarginalConstraint, LHS equals to $f_{ij}(t_{ij})\cdot \left(\lambda_{ij}(t_{ij},\tilde{\beta}_{ij},c_i)+\lambda_{ij}(t_{ij},\tilde{\beta}^+_{ij},c_i)\right)=w_{ij}(t_{ij})$. Constraint~\PiConstraint~follows from the fact that $\hat{\sigma}$ is a feasible interim allocation rule, and each item $j$ can be allocated to at most one buyer for every type profile. Thus
$$\sum_i\sum_{t_{ij}}\pi_{ij}(t_{ij})=\sum_i\sum_{t_i}f_i(t_i)\sum_{S:j\in S}\hat\sigma_{iS}(t_i)\leq 1$$
By property 1 of \Cref{lem:caiz17} and the choice of $\eps_r$, we have that for every $j$,
\begin{align*}
&\sum_{i\in[n]}\sum_{\beta_{ij} \in \cV_{ij}} \hat\lambda_{ij}(\beta_{ij}) \cdot \Pr_{t_{ij}\sim D_{ij}}[V_{ij}(t_{ij})\geq \beta_{ij}] \\
=&\sum_{i\in[n]}\left(r_{ij}\cdot\Pr_{t_{ij}}[V_{ij}(t_{ij})\geq\tilde\beta_{ij}] + (1-r_{ij})\cdot \Pr_{t_{ij}}[V_{ij}(t_{ij})\geq\tilde\beta_{ij}^+] \right)\\
=&\sum_{i\in[n]} \left(\Pr_{t_{ij}}[V_{ij}(t_{ij})>\tilde\beta_{ij}]+r_{ij}\cdot \Pr_{t_{ij}}[V_{ij}(t_{ij})=\tilde\beta_{ij}]\right)\leq \frac{1}{2}
\end{align*}
Thus, Constraint~\ReduceDemandConstaint~is satisfied.
For Constraint~\MarginalToGlobalConstraint, we only need to verify the constraint for $\beta_{ij}=\tilde{\beta}_{ij}\in \cV_{ij}^0$.
LHS equals to $\frac{1}{2}\sum_{t_{ij}}w_{ij}(t_{ij})$. 
We notice that $\gamma_{ij}^S(t_i)\leq V_{ij}(t_{ij})$ for all $i,j,S$ (for XOS valuations). Thus  
$$\sum_{t_{ij}\in \cT_{ij}}w_{ij}(t_{ij})\leq \sum_{t_{ij}\in \cT_{ij}}f_{ij}(t_{ij})\cdot\left( \sum_{t_{i,-j}\in \cT_{i,-j}} f_{i,-j}(t_{i,-j})\sum_{S:j\in S}\hat\sigma_{iS}(t_{ij},t_{i,-j})\right)=\sum_{t_i\in \cT_i}f_i(t_i)\cdot \sum_{S:j\in S}\hat\sigma_{iS}(t_i)$$

Since $r_{ij}\in [0,1]$ for all $i,j$, by property 2 of \Cref{lem:caiz17} and the choice of $\varepsilon_r$, we have

\begin{align*}
\text{LHS of Constraint~\MarginalToGlobalConstraint} \leq& \frac{1}{2}\sum_{t_i\in \cT_i}f_i(t_i)\cdot \sum_{S:j\in S}\hat\sigma_{iS}(t_i)=\Pr_{t_{ij}}[V_{ij}(t_{ij})>\tilde{\beta}_{ij}]+r_{ij}\cdot \Pr_{t_{ij}}[V_{ij}(t_{ij})=\tilde{\beta}_{ij}] \\
=&\hat\lambda_{ij}(\beta_{ij}^+) \cdot \Pr_{t_{ij}}[V_{ij}(t_{ij})\geq \beta_{ij}^+] +\hat\lambda_{ij}(\beta_{ij}) \cdot \Pr_{t_{ij}}[V_{ij}(t_{ij})\geq\beta_{ij}]
\end{align*}

For Constraint~\CompareMarginalConstraint, since $\gamma_{ij}^S(t_i)\leq V_{ij}(t_{ij})$ for all $i,j,S$ and $\sum_{S\in [m]}\hat\sigma_{iS}(t_i)\leq 1$, we have $w_{ij}(t_{ij})\leq f_{ij}(t_{ij})$. Thus 
$$\lambda_{ij}(t_{ij},\tilde{\beta}_{ij},c_i)=r_{ij}\cdot \frac{w_{ij}(t_{ij})}{f_{ij}(t_{ij})}\leq r_{ij}= \hat\lambda_{ij}(\tilde{\beta}_{ij},c_i).$$ 
and
$$\lambda_{ij}(t_{ij},\tilde{\beta}_{ij}^+,c_i)=(1-r_{ij})\cdot \frac{w_{ij}(t_{ij})}{f_{ij}(t_{ij})}\leq 1-r_{ij}= \hat\lambda_{ij}(\tilde{\beta}_{ij}^+,c_i).$$ 
Constraint~\BoundMeanDeltaConstraint~and ~\HatLambdaDistributionConstraint~are straightforward since $\hat\lambda_{ij}(\beta_{ij},\delta_i)=r_{ij}\cdot \ind[\beta_{ij}=\tilde{\beta}_{ij}\wedge \delta_{ij}=\hat{c}_i]$,  $\hat\lambda_{ij}(\beta_{ij},\delta_i)=(1-r_{ij})\cdot \ind[\beta_{ij}=\tilde{\beta}^+_{ij}\wedge \delta_{ij}=\hat{c}_i]$, and $d_i=\hat{c}_i$. 

Lastly, for Constraint~\BoundSumDeltaConstraint, we notice that for every $i\in [n]$, $\hat{c_i}\leq \max\{\frac{\estprev}{n},2c_i\}$. Thus by property 4 of \Cref{lem:caiz17}, when $n\cdot m\geq 110$,
$$\sum_{i\in [n]}c_i\leq 2\cdot \sum_{i\in [n]}c_i+\estprev\leq 16\cdot\prev+\estprev\leq 111\cdot\estprev$$

Thus the solution is feasible. We are left to show that the objective of the above solution is at least $\core(\hat{\sigma},\tilde{\vBeta},\hat{\vC},\vr)$. In fact, by the definition of $w_{ij}(t_{ij})$,   

{
\begin{align*}
\core(\hat{\sigma},\tilde{\vBeta},\hat{\vC},\vr)=&\sum_{i\in[n]}\sum_{j\in[m]}\sum_{t_{ij}\in\cT_{ij}}
w_{ij}(t_{ij})\cdot V_{ij}(t_{ij})\cdot \left(\ind[V_{ij}(t_{ij})< \tilde{\beta}_{ij}+\hat{c}_i]+r_{ij}\cdot \ind[V_{ij}(t_{ij})= \tilde{\beta}_{ij}+\hat{c}_i] \right)\\
\leq &\sum_i\sum_j\sum_{t_{ij}}
w_{ij}(t_{ij})\cdot V_{ij}(t_{ij})\cdot \ind[V_{ij}(t_{ij})\leq \tilde{\beta}_{ij}+\hat{c}_i].
\end{align*}

This is exactly the objective of the LP in \Cref{fig:bigLP} 
according to the choice of $\eps_r$, since \begin{align*}
    \lambda_{ij}(t_{ij},\tilde{\beta}_{ij},\delta_{ij})\cdot\ind[V_{ij}(t_{ij})\leq \tilde{\beta}_{ij}+\hat{c}_i]+\lambda_{ij}(t_{ij},\tilde{\beta}^+_{ij},\delta_{ij})\cdot&\ind[V_{ij}(t_{ij})\leq \tilde{\beta}^+_{ij}+\hat{c}_i]\\=&\frac{w_{ij}(t_{ij})}{f_{ij}(t_{ij})}\cdot \ind[V_{ij}(t_{ij})\leq \tilde{\beta}_{ij}+\hat{c}_i]
\end{align*}
}
The proof is complete by invoking property 3 of \Cref{lem:caiz17}. 
\end{proof}

\subsection{Bounding the Difference between the Shifted~\core~and the Original~\core}\label{subsec:proof_lem:Q and Q-hat}

We first give the following definition as a generalization of \Cref{def:Q_j}.

\begin{definition}\label{def:Q_j-XOS}
Let $(\pi, w,\lambda,\hat\lambda, \bd=(d_i)_{i\in [n]})$ be any feasible solution of the LP in \Cref{fig:XOSLP}. 
For every $j\in [m]$, define 





$$Q_j = \frac{1}{2}\cdot\sum_{i\in[n]}\sum_{t_{ij}\in \cT_{ij}}  f_{ij}(t_{ij})\cdot V_{ij}(t_{ij})\cdot
    \sum_{\substack{\beta_{ij}\in \cV_{ij}\\\delta_{ij} \in \Delta}}\lambda_{ij}(t_{ij},\beta_{ij},\delta_{ij})\cdot\ind[V_{ij}(t_{ij})\leq \beta_{ij}+\delta_{ij}].
$$
\end{definition}

Recall that by Constraint \HatLambdaDistributionConstraint, for every $i,j$, $\hat\lambda_{ij}(\cdot,\cdot)$ can be viewed as a joint distribution $\cC_{ij}$ over $\cV_{ij}\times \Delta$, i.e. $\Pr_{(\beta_{ij},\delta_{ij}) \sim \cC_{ij}}[\beta_{ij} = a \land \delta_{ij} = b]= \hat\lambda_{ij}(\beta_{ij}=a, \delta_{ij} = b)$. Denote $\cB_{ij}$ the marginal distribution of $\beta_{ij}$ with respect to $\cC_{ij}$. Inspired by the ``shifted core'' technique by Cai and Zhao~\cite{CaiZ17}, we need further definitions to describe the welfare contribution by each item under a smaller threshold.  

\begin{definition}\label{def:hatQ_j}
For every $i\in [n]$, define\footnote{\label{fn:tau_i}If all $D_{ij}$s are continuous, then for every $i$ there exists $\tau_i$ that satisfies the following property:\\ $\sum_{j\in[m]} \Pr_{t_{ij}\sim D_{ij},\beta_{ij} \sim \cB_{ij}}[V_{ij}(t_{ij}) \geq \max(\beta_{ij}, Q_j + x)]\leq \frac{1}{2}$ and the inequality is achieved as equality for all $\tau_i>0$. However, this property might not be satisfied for discrete distributions. This is again a tie-breaking issue addressed in \Cref{remark:tie-breaking}. We refer the readers to Lemma 5 of \cite{CaiZ17} for a fix. For simplicity, in our proof we will assume that all $\tau_i$s satisfy the property mentioned above.  
} 
$$\tau_i = \inf_{x\geq 0}\left\{ \sum_{j\in[m]} \Pr_{t_{ij}\sim D_{ij},\beta_{ij} \sim \cB_{ij}}[V_{ij}(t_{ij}) \geq \max(\beta_{ij}, Q_j + x)]\leq \frac{1}{2}\right\}$$


Then for every $j\in [m]$, define
$$\hat Q_j = \frac{1}{2}\cdot\sum_{i\in[n]}\sum_{t_{ij}\in \cT_{ij}}  f_{ij}(t_{ij})\cdot V_{ij}(t_{ij})\cdot
    \sum_{\substack{\beta_{ij}\in \cV_{ij}\\\delta_{ij} \in \Delta}}\lambda_{ij}(t_{ij},\beta_{ij},\delta_{ij})\cdot \ind[V_{ij}(t_{ij})\leq \min\{\beta_{ij}+\delta_{ij},Q_j+\tau_i\}]$$
\end{definition}

We prove in \Cref{lem:Q and Q-hat_main body} that the difference between $\sum_{j\in[m]} Q_j$ and $\sum_{j\in[m]} \hat{Q}_j$ can be bounded using $\prev$. 

\begin{lemma}\label{lem:Q and Q-hat_main body}
For every $j\in [m]$, $Q_j\geq \hat{Q}_j$. Moreover, there exists some absolute constant $c>0$ such that
$$\sum_{j\in [m]}Q_j\leq \sum_{j\in [m]}\hat{Q}_j+c\cdot \prev$$
\end{lemma}

To prove \Cref{lem:Q and Q-hat_main body}, we consider the following variant of the RPP mechanism, where the posted prices are allowed to be randomized. 

\paragraph{Rationed Randomized Posted Price Mechanism (RRPP).} Before the auction starts, the seller first draws a posted price $p_{ij}$ from some distribution $\cG_{ij}$, for every buyer $i$ and item $j$. All $\cG_{ij}$s are independent.
The buyers then arrive in some arbitrary order, and each buyer $i$ can purchase \emph{at most one} item among the available ones at the realized posted price $p_{ij}$ for every item $j$.

Clearly any RRPP mechanism is also DSIC and IR. We notice that any RRPP mechanism can be viewed as a distribution of RPP mechanisms, as the seller can draw all the posted prices before the auction starts, and use the realized (and deterministic) set of posted prices to sell the items. Thus the highest revenue achievable among all RRPP mechanisms is the same as the optimum revenue among all RPP mechanisms, which is $\prev$. 

Before giving the proof of \Cref{lem:Q and Q-hat_main body}, we first prove a useful lemma that analyzes the revenue of RRPP. It's a generalization of Lemma 17 of \cite{CaiZ17}, which allows randomized posted prices.  

\begin{lemma}\label{lem:prev-useful}\cite{CaiZ17}
For every $i,j$, let $\cG_{ij}$ be any distribution over $\mathbb{R}_+$. All $\cG_{ij}$'s are independent from each other. Suppose both of the following inequalities hold, for some constant $a,b\in (0,1)$:
\begin{enumerate}
\item $\sum_{i\in [n]} \Pr_{t_{ij}\sim \cD_{ij},x_{ij}\sim \cG_{ij}}\left[V_i(t_{ij})\geq x_{ij}\right]\leq a,\forall j\in [m]$.
\item
$\sum_{j\in [m]} \Pr_{t_{ij}\sim \cD_{ij},x_{ij}\sim \cG_{ij}}\left[V_i(t_{ij})\geq x_{ij}\right]\leq b,\forall i\in [n]$.
\end{enumerate}
Then
\begin{equation*}
\sum_{i\in [n]}\sum_{j\in[m]} \E_{x_{ij}\sim \cG_{ij}}\left[x_{ij}\cdot \Pr_{t_{ij}\sim \cD_{ij}}\left[V_i(t_{ij})\geq x_{ij}\right]\right]\leq \frac{1}{(1-a)(1-b)}\cdot \prev.
\end{equation*}
\end{lemma}

\begin{proof}
Consider an RRPP that sells item $j$ to buyer $i$ at price $x_{ij}\sim \cG_{ij}$. The mechanism
visits the buyers in some arbitrary order. For every $i,j$ and every realized $x_{ij}$, we will bound the probability of buyer $i$ purchasing item $j$, over the randomness of $\{\cD_{i'j'}\}_{i'\in[n],j'\in[m]}$ and $\{\cG_{i'j'}\}_{(i',j')\neq (i,j)}$. 
Notice that when it is buyer $i$'s turn, she purchases exactly item $j$ and pays $x_{ij}$ if all of the following three conditions hold: (i) $j$ is still available, (ii) $V_i(t_{ij})\geq x_{ij}$ and (iii) $\forall k\neq j, V_i(t_{ik})< x_{ik}$. The second condition means buyer $i$ can afford item $j$. The third condition means she cannot afford any other item $k\neq j$. Therefore, buyer $i$ purchases exactly item $j$.

Now let us compute the probability that all three conditions hold, when $t_{ij}\sim \cD_{ij}$ and $x_{ij}\sim \cG_{ij}$ for all $i,j$. Since every buyer's valuation is subadditive over the items, item $j$ is purchased by someone else only if there exists a buyer $k\neq i$ who has $V_k(t_{kj})\geq x_{kj}$. 
By the union bound, the event described above happens with probability at most $\sum_{k\neq i} \Pr_{t_{kj},x_{kj}}\left[V_k(t_{kj})\geq x_{kj}\right]$, which is less than $a$ by Inequality 1 of the statement. Therefore, condition (i) holds with probability at least $1-a$. Clearly, condition (ii) holds with probability $\Pr_{t_{ij}}\left[V_i(t_{ij})\geq x_{ij}\right]$. Finally, condition (iii) holds with at least probability $1-b$, because the probability that there exists any item $k\neq j$ such that $V_i(t_{ik})\geq x_{ik}$ is no more than $\sum_{k\neq j}\Pr_{t_{ik},x_{ik}}[V_i(t_{ik})\geq x_{ik}]\leq b$ (Inequality 2 of the statement). Since the three conditions are independent, buyer $i$ purchases exactly item $j$ with probability at least $(1-a)(1-b)\cdot \Pr_{t_{ij}}\left[V_i(t_{ij})\geq x_{ij}\right]$. So the expected revenue of this mechanism is at least $(1-a)(1-b)\cdot \E_{x_{ij}\sim \cG_{ij}}\left[x_{ij}\cdot \Pr_{t_{ij}\sim \cD_{ij}}\left[V_i(t_{ij})\geq x_{ij}\right]\right]$.
\end{proof}

A direct corollary of \Cref{lem:prev-useful} is that $\sum_i \tau_i$ can be bounded using $\prev$.

\begin{lemma}\label{lem:bounding tau_i}
$\sum_{i\in[n]} \tau_i\leq 8\cdot \prev$.
\end{lemma}
\begin{proof}
By Constraint \ReduceDemandConstaint~of the LP in \Cref{fig:bigLP} (or \Cref{fig:XOSLP}), for every item $j$,
\begin{align*}
    \sum_{i\in [n]}\Pr_{\substack{t_{ij}\sim D_{ij}\\\beta_{ij}\sim \cB_{ij}}}[V_{ij}(t_{ij}) \geq \max(\beta_{ij}, Q_j+\tau_i)] &\leq \sum_{i\in [n]}\Pr_{\substack{t_{ij}\sim D_{ij}\\\beta_{ij}\sim \cB_{ij}}}[V_{ij}(t_{ij})\geq \beta_{ij}]\\
    &=\sum_{i\in [n]}\sum_{\beta_{ij}\in \cV_{ij}}\hat\lambda_{ij}(\beta_{ij})\Pr_{t_{ij}}[V_{ij}(t_{ij})\geq \beta_{ij}]\leq \frac{1}{2}
\end{align*}

By the definition of $\tau_i$, for every $i$ we have
$$\sum_{j\in [m]}\Pr_{t_{ij}\sim D_{ij},\beta_{ij}\sim \cB_{ij}}[V_{ij}(t_{ij}) \geq \max(\beta_{ij}, Q_j+\tau_i)]]\leq\frac{1}{2}.$$

Thus by \Cref{lem:prev-useful}, 
\begin{align*}
4\cdot \prev\geq& \sum_{i\in[n],j\in[m]} 
\E_{\beta_{ij}\sim\cB_{ij}}\left[\max(\beta_{ij}, Q_j+\tau_i)\cdot
\Pr_{t_{ij}\sim D_{ij}}[V_{ij}(t_{ij}) \geq \max(\beta_{ij}, Q_j+\tau_i)]\right]\\
\geq & \sum_{i\in[n]}\tau_i\sum_{j\in[m]}\E_{\beta_{ij}\sim\cB_{ij}}\left[\Pr_{t_{ij}\sim D_{ij}}[V_{ij}(t_{ij}) \geq \max(\beta_{ij}, Q_j+\tau_i)]\right]\\
=&\frac{1}{2}\sum_{i\in[n]} \tau_i,
\end{align*}
{The last equality comes from the fact that by the definition of $\tau_i$, $\sum_{j}\Pr_{t_{ij}\sim D_{ij},\beta_{ij}\sim \cB_{ij}}[V_{ij}(t_{ij}) \geq \max(\beta_{ij}, Q_j+\tau_i)]]=\frac{1}{2}$ for all $i$ such that $\tau_i>0$ (see \Cref{fn:tau_i}). We finish the proof.} 
\end{proof}

\begin{lemma}\label{lem:Q and Q-hat}[(Restatement of \Cref{lem:Q and Q-hat_main body})]
For every $j\in [m]$, $Q_j\geq \hat{Q}_j$. 
Moreover,
$$\sum_{j\in [m]}Q_j\leq \sum_{j\in [m]}\hat{Q}_j+236.5\cdot \prev$$ 
\end{lemma}

\begin{proof}

For every $j$, it's clear that $\hat{Q}_j\leq Q_j$ by the definition of $\hat{Q}_j$. It remains to bound $\sum_j(Q_j-\hat{Q}_j)$. In the proof, we abuse the notation and let $\hat\lambda_{ij}(\beta_{ij})=\sum_{\delta_{ij}\in\Delta}\hat\lambda_{ij}(\beta_{ij},\delta_{ij})$. Also for every $i,j,t_{ij},\beta_{ij}\in \cV_{ij}$, let $\lambda_{ij}(t_{ij},\beta_{ij})=\sum_{\delta_{ij}\in \Delta}\lambda_{ij}(t_{ij},\beta_{ij},\delta_{ij})$.

\begin{align*}
    &2\sum_{j\in [m]} \left(Q_j - \hat{Q}_j \right) \\
    \leq & \sum_{j} \sum_{i} \sum_{t_{ij}:V_{ij}(t_{ij}) \geq Q_j + \tau_i} f_{ij}(t_{ij}) V_{ij}(t_{ij}) \sum_{\substack{\beta_{ij}\in \cV_{ij}\\
    \delta_{ij} \in \Delta}} \lambda_{ij}(t_{ij},\beta_{ij}, \delta_{ij})\cdot \mathds{1}[V_{ij}(t_{ij}) \leq \beta_{ij}+ \delta_{ij}] \\
    \leq & \sum_{j} \sum_{i} \sum_{t_{ij}:V_{ij}(t_{ij}) \geq Q_j + \tau_i} f_{ij}(t_{ij}) \sum_{\substack{\beta_{ij}\in \cV_{ij}\\
    \delta_{ij} \in \Delta}} \left(\beta_{ij}+(V_{ij}(t_{ij})-\beta_{ij})^+\right)\cdot\lambda_{ij}(t_{ij},\beta_{ij}, \delta_{ij})\cdot \mathds{1}[V_{ij}(t_{ij}) \leq \beta_{ij}+ \delta_{ij}]\\
    \leq & \sum_{j} \sum_{i} \sum_{t_{ij}:V_{ij}(t_{ij}) \geq Q_j + \tau_i} f_{ij}(t_{ij}) \sum_{\beta_{ij}\in \cV_{ij}} \beta_{ij}\cdot\lambda_{ij}(t_{ij},\beta_{ij})\\
    &+\sum_{j} \sum_{i} \sum_{t_{ij}:V_{ij}(t_{ij}) \geq Q_j + \tau_i} f_{ij}(t_{ij}) \sum_{\substack{\beta_{ij}\in \cV_{ij}\\
    \delta_{ij} \in \Delta}} (V_{ij}(t_{ij})-\beta_{ij})^+\cdot\lambda_{ij}(t_{ij},\beta_{ij}, \delta_{ij})\cdot \mathds{1}[V_{ij}(t_{ij}) \leq \beta_{ij}+ \delta_{ij}]\\
\end{align*}

Here the first inequality uses the fact that $\lambda_{ij}(t_{ij},\beta_{ij},\delta_{ij})\cdot \ind[V_{ij}(t_{ij})\leq \beta_{ij}+\delta_{ij}]$ and\\ $\lambda_{ij}(t_{ij},\beta_{ij},\delta_{ij})\ind[V_{ij}(t_{ij})\leq \min\{\beta_{ij}+\delta_{ij},Q_j+\tau_i\}]$ can differ only when $V_{ij}(t_{ij})>Q_j+\tau_i \wedge V_{ij}(t_{ij})\leq \beta_{ij}+\delta_{ij}$. In the last inequality, we drop the indicator $\ind[V_{ij}(t_{ij})\leq \beta_{ij}+\delta_{ij}]$ for the first term. 

We bound the first term:
\notshow{
\begin{equation}\label{inequ:Q and Q_hat-term1}
\begin{aligned}
&\sum_{j} \sum_{i}  \sum_{t_{ij}:V_{ij}(t_{ij}) \geq Q_j+\tau_i} f_{ij}(t_{ij})\cdot \sum_{\beta_{ij}\in \cV_{ij}} \beta_{ij} \cdot \lambda_{ij}(t_{ij},\beta_{ij})\\
= & \sum_{i,j} \sum_{\beta_{ij}\leq Q_j+\tau_i} \beta_{ij} \sum_{t_{ij}:V_{ij}(t_{ij}) \geq Q_j+\tau_i} f_{ij}(t_{ij})\cdot  \lambda_{ij}(t_{ij},\beta_{ij})
+
\sum_{i,j}\sum_{\beta_{ij}> Q_j+\tau_i} \beta_{ij} \sum_{t_{ij}:V_{ij}(t_{ij}) \geq Q_j+\tau_i} f_{ij}(t_{ij})\cdot  \lambda_{ij}(t_{ij},\beta_{ij})\\
\leq &\sum_{i,j}
\sum_{\beta_{ij}\leq Q_j + \tau_i} \hat\lambda_{ij}(\beta_{ij}) \cdot \beta_{ij}
\sum_{t_{ij}:V_{ij}(t_{ij}) \geq Q_j+\tau_i} f_{ij}(t_{ij})
+ 
\sum_{i,j}
\sum_{\beta_{ij} > Q_j + \tau_i} 
2\hat\lambda_{ij}(\beta_{ij})\cdot\beta_{ij}\cdot\Pr_{t_{ij}\sim D_{ij}}[V_{ij}(t_{ij}) \geq \beta_{ij}]\\
\leq & 2\sum_{i,j} 
\sum_{\beta_{ij}\in \cV_{ij}} \hat\lambda_{ij}(\beta_{ij}) \cdot \beta_{ij}\cdot
\Pr_{t_{ij}\sim D_{ij}}[V_{ij}(t_{ij}) \geq \max(\beta_{ij}, Q_j+\tau_i)]\\
= & 2\sum_{i,j} 
\E_{\beta_{ij}\sim \cB_{ij}}\left[  \beta_{ij}\cdot
\Pr_{t_{ij}\sim D_{ij}}[V_{ij}(t_{ij}) \geq \max(\beta_{ij}, Q_j+\tau_i)]\right]\\
\leq & 8\cdot\prev
\end{aligned}
\end{equation}
}

\begin{equation}\label{inequ:Q and Q_hat-term1}
\begin{aligned}
&\sum_{j} \sum_{i}  \sum_{t_{ij}:V_{ij}(t_{ij}) \geq Q_j+\tau_i} f_{ij}(t_{ij})\cdot \sum_{\beta_{ij}\in \cV_{ij}} \beta_{ij} \cdot \lambda_{ij}(t_{ij},\beta_{ij})\\
= & \sum_{i,j} \sum_{\substack{\beta_{ij}\in \cV_{ij}\\\beta_{ij} < Q_j+\tau_i}} \beta_{ij} \sum_{t_{ij}:V_{ij}(t_{ij}) \geq Q_j+\tau_i} f_{ij}(t_{ij})\cdot  \lambda_{ij}(t_{ij},\beta_{ij})
+
\sum_{i,j}\sum_{\substack{\beta_{ij}\in \cV_{ij}\\\beta_{ij}\geq Q_j+\tau_i}} \beta_{ij} \sum_{t_{ij}:V_{ij}(t_{ij}) \geq Q_j+\tau_i} f_{ij}(t_{ij})\cdot  \lambda_{ij}(t_{ij},\beta_{ij})\\
\leq & \sum_{i,j} \sum_{\substack{\beta_{ij}\in \cV_{ij}\\\beta_{ij} < Q_j+\tau_i}} \beta_{ij} \sum_{t_{ij}:V_{ij}(t_{ij}) \geq Q_j+\tau_i} f_{ij}(t_{ij})\cdot  \lambda_{ij}(t_{ij},\beta_{ij})
\\
& \qquad\qquad\qquad\qquad + 
\sum_{i,j}\sum_{\substack{\beta_{ij}\in \cV_{ij}^0\\\beta_{ij}\geq Q_j+\tau_i}} \beta_{ij}^+ \sum_{t_{ij}:V_{ij}(t_{ij}) \geq Q_j+\tau_i} f_{ij}(t_{ij})\cdot  \left(\lambda_{ij}(t_{ij},\beta_{ij})+\lambda_{ij}(t_{ij},\beta_{ij}^+)\right)\\
\leq &\sum_{i,j}
\sum_{\substack{\beta_{ij}\in \cV_{ij}\\\beta_{ij} < Q_j + \tau_i} } \hat\lambda_{ij}(\beta_{ij}) \cdot \beta_{ij}
\sum_{t_{ij}:V_{ij}(t_{ij}) \geq Q_j+\tau_i} f_{ij}(t_{ij})
+ 
\sum_{i,j}
\sum_{\substack{\beta_{ij}\in \cV_{ij}^0\\\beta_{ij} \geq Q_j + \tau_i} }
2\hat\lambda_{ij}(\beta_{ij})\cdot\beta_{ij}^+\cdot\Pr_{t_{ij}\sim D_{ij}}[V_{ij}(t_{ij}) \geq \beta_{ij}]\\
& \qquad\qquad\qquad\qquad\qquad\qquad\qquad\qquad\qquad\qquad + 
\sum_{i,j}
\sum_{\substack{\beta_{ij}\in \cV_{ij}^+\\\beta_{ij} \geq Q_j + \tau_i} }
2\hat\lambda_{ij}(\beta_{ij})\cdot\beta_{ij}\cdot\Pr_{t_{ij}\sim D_{ij}}[V_{ij}(t_{ij}) \geq \beta_{ij}] \\
\leq &\sum_{i,j}
\sum_{\substack{\beta_{ij}\in \cV_{ij}\\\beta_{ij} < Q_j + \tau_i} } \hat\lambda_{ij}(\beta_{ij}) \cdot \beta_{ij}
\sum_{t_{ij}:V_{ij}(t_{ij}) \geq Q_j+\tau_i} f_{ij}(t_{ij})
+ 
\sum_{i,j}\left(
\sum_{\substack{\beta_{ij}\in \cV_{ij}\\\beta_{ij} \geq Q_j + \tau_i} }
2\hat\lambda_{ij}(\beta_{ij})\cdot\beta_{ij}\cdot\Pr_{t_{ij}\sim D_{ij}}[V_{ij}(t_{ij}) \geq \beta_{ij}] + \varepsilon_r |\cV_{ij}^0|\right)\\
\leq & 2\sum_{i,j} 
\sum_{\beta_{ij}\in \cV_{ij}} \hat\lambda_{ij}(\beta_{ij}) \cdot \beta_{ij}\cdot
\Pr_{t_{ij}\sim D_{ij}}[V_{ij}(t_{ij}) \geq \max(\beta_{ij}, Q_j+\tau_i)] + \varepsilon_r\cdot \sum_{i,j}|\cV_{ij}^0|\\
= & 2\sum_{i,j} 
\E_{\beta_{ij}\sim \cB_{ij}}\left[  \beta_{ij}\cdot
\Pr_{t_{ij}\sim D_{ij}}[V_{ij}(t_{ij}) \geq \max(\beta_{ij}, Q_j+\tau_i)]\right]+ \varepsilon_r\cdot \sum_{i,j}|\cV_{ij}^0|\\
\leq & 8\cdot\prev+ \varepsilon_r\cdot \sum_{i,j}|\cV_{ij}^0|\\
\leq & 9\cdot \prev
\end{aligned}
\end{equation}

The first inequality comes from the fact that for $\beta_{ij}\in \cV_{ij}^0$,
then $\beta_{ij} < \beta_{ij}^+$ and from the fact that for sufficiently small $\varepsilon_r>0$,
then $Q_j + \tau_i\leq\beta_{ij}$ iff $Q_j + \tau_i\leq\beta_{ij}^+$. The second inequality comes from Constraint~\CompareMarginalConstraint~and  \MarginalToGlobalConstraint~of the LP in \Cref{fig:bigLP} (or \Cref{fig:XOSLP}).\footnote{\label{fn:constraint 7}Note that this is the only place that Constraint~\MarginalToGlobalConstraint~is used in our proof.}
For the second last inequality, notice that by Constraint \ReduceDemandConstaint~of the LP in \Cref{fig:bigLP} (or \Cref{fig:XOSLP}), for every item $j$, 
\begin{align*}
    \sum_{i}\Pr_{\substack{t_{ij}\sim D_{ij}\\\beta_{ij}\sim \cB_{ij}}}[V_{ij}(t_{ij}) \geq \max(\beta_{ij}, Q_j+\tau_i)]]&\leq \sum_{i}\Pr_{\substack{t_{ij}\sim D_{ij}\\\beta_{ij}\sim \cB_{ij}}}[V_{ij}(t_{ij})\geq \beta_{ij}]\\
    &=\sum_{i}\sum_{\beta_{ij}\in\cV_{ij}}\hat\lambda_{ij}(\beta_{ij})\Pr_{t_{ij}}[V_{ij}(t_{ij})\geq \beta_{ij}]\leq \frac{1}{2}
\end{align*}

By the definition of $\tau_i$, for every buyer $i$ we have
$$\sum_{j}\Pr_{t_{ij}\sim D_{ij},\beta_{ij}\sim \cB_{ij}}[V_{ij}(t_{ij}) \geq \max(\beta_{ij}, Q_j+\tau_i)]]\leq\frac{1}{2}$$

The second last inequality then follows from \Cref{lem:prev-useful}. The last inequality follows from the fact that $\eps_r\leq \frac{\prev}{\sum_{i,j}|\cV_{ij}^0|}$.
For the second term, we have

\begin{align*}
&\sum_{j} \sum_{i} \sum_{t_{ij}:V_{ij}(t_{ij}) \geq Q_j + \tau_i} f_{ij}(t_{ij}) \sum_{\substack{\beta_{ij}\in \cV_{ij}\\\delta_{ij} \in \Delta_i}} (V_{ij}(t_{ij})-\beta_{ij})^+\cdot\lambda_{ij}(t_{ij},\beta_{ij}, \delta_{ij})\cdot \mathds{1}[V_{ij}(t_{ij}) \leq \beta_{ij}+ \delta_{ij}]\\
\leq & \sum_{j} \sum_{i} \sum_{\substack{\beta_{ij}\in \cV_{ij}\\\delta_{ij} \in \Delta_i}} \hat\lambda_{ij}(\beta_{ij}, \delta_{ij})\cdot\sum_{t_{ij}} f_{ij}(t_{ij})  (V_{ij}(t_{ij})-\beta_{ij})\cdot\mathds{1}[V_{ij}(t_{ij}) \leq \beta_{ij}+ \delta_{ij} \wedge V_{ij}(t_{ij})\geq \max(\beta_{ij}, Q_j + \tau_i)]\\
= & \sum_{i,j}\E_{(\beta_{ij},\delta_{ij})\sim \cC_{ij}}\left[\sum_{t_{ij}}f_{ij}(t_{ij})u_{ij}(t_{ij},\beta_{ij},\delta_{ij})\right],
\end{align*}

where
$$u_{ij}(t_{ij},\beta_{ij},\delta_{ij})=(V_{ij}(t_{ij})-\beta_{ij})\cdot\mathds{1}[V_{ij}(t_{ij}) \leq \beta_{ij}+ \delta_{ij} \wedge V_{ij}(t_{ij})\geq \max(\beta_{ij}, Q_j + \tau_i)]$$

We notice that by the definition of $\tau_i$, the following inequality holds for every buyer $i$.

\begin{equation}\label{inequ:Q and Q_hat-1} 
\sum_j\Pr_{t_{ij}\sim D_{ij},(\beta_{ij},\delta_{ij})\sim \cC_{ij}}[u_{ij}(t_{ij},\beta_{ij},\delta_{ij})>0]\leq \sum_{j}\Pr_{t_{ij}\sim D_{ij},\beta_{ij}\sim \cB_{ij}}[V_{ij}(t_{ij}) \geq \max(\beta_{ij}, Q_j+\tau_i)]]\leq\frac{1}{2}    
\end{equation}

Denote $\cC_i=\times_{j=1}^m \cC_{ij}$ and $\beta_{i}=(\beta_{ij})_{j\in [m]}, \delta_{i}=(\delta_{ij})_{j\in [m]}$, we have  
\begin{align*}
    &\sum_{i}\E_{(\beta_{i},\delta_{i})\sim \cC_{i}}\left[\sum_{t_i}f_i(t_i)\cdot\max_j u_{ij}(t_{ij},\beta_{ij},\delta_{ij})\right]\\
\geq & \sum_{i}\E_{(\beta_{i},\delta_{i})\sim \cC_{i}}\left[\sum_j\sum_{t_{ij}}f_{ij}(t_{ij})\cdot u_{ij}(t_{ij},\beta_{ij},\delta_{ij})\cdot \prod_{k\not=j}\Pr_{t_{ik}\sim D_{ik}}[u_{ik}(t_{ik},\beta_{ik},\delta_{ik})=0]\right]\\
=&\sum_{i}\sum_j\E_{(\beta_{ij},\delta_{ij})\sim \cC_{ij}}\left[\sum_{t_{ij}}f_{ij}(t_{ij})\cdot u_{ij}(t_{ij},\beta_{ij},\delta_{ij})\cdot \prod_{k\not=j}\Pr_{t_{ik}\sim D_{ik},(\beta_{ik},\delta_{ik})\sim \cC_{ik}}[u_{ik}(t_{ik},\beta_{ik},\delta_{ik})=0]\right]\\
\geq & \frac{1}{2}\cdot \sum_{i,j}\E_{(\beta_{ij},\delta_{ij})\sim \cC_{ij}}\left[\sum_{t_{ij}}f_{ij}(t_{ij})u_{ij}(t_{ij},\beta_{ij},\delta_{ij})\right]
\end{align*}

Here the equality uses the fact that all $\cC_{ij}$'s are independent. The last inequality comes from Inequality~\eqref{inequ:Q and Q_hat-1} and the union bound. Now the second term is bounded by \begin{align*}
& 2\cdot \sum_{i}\E_{(\beta_{i},\delta_{i})\sim \cC_{i}}\left[\sum_{t_i}f_i(t_i)\cdot\max_j u_{ij}(t_{ij},\beta_{ij},\delta_{ij})\right]\\
\leq & 2\cdot \sum_{i}\E_{(\beta_{i},\delta_{i})\sim \cC_{i}}\left[\sum_{t_i}f_i(t_i)\cdot\max_{j\in [m]}\left\{ (V_{ij}(t_{ij})-\beta_{ij})^+\cdot\mathds{1}[V_{ij}(t_{ij}) \leq \beta_{ij}+ \delta_{ij}]\right\}\right]
\end{align*}

\begin{definition}
For every $i,j,t_i\in \cT_i, S\subseteq [m]$, let 
$$\eta_i(t_i,S)=\E_{(\beta_{i},\delta_{i})\sim \cC_{i}}\left[\max_{j\in S}\left\{ (V_{ij}(t_{ij})-\beta_{ij})^+\cdot\mathds{1}[V_{ij}(t_{ij}) \leq \beta_{ij}+ \delta_{ij}]\right\}\right]$$

\notshow{For every $i,j,t_i\in \cT_i,\beta_i\in \times_j\cV_{ij}, S\subseteq [m]$, let 
$$\eta_i(t_i,\beta_i,S)=\E_{\delta_i}\left[\max_{j\in S}\left\{ (V_{ij}(t_{ij})-\beta_{ij})^+\cdot\mathds{1}[V_{ij}(t_{ij}) \leq \beta_{ij}+ \delta_{ij}]\right\}\right]$$
}

\end{definition}

\noindent Therefore, we can rewrite $2\cdot \sum_{i}\E_{(\beta_{i},\delta_{i})\sim \cC_{i}}\left[\sum_{t_i}f_i(t_i)\cdot\max_{j\in [m]}\left\{ (V_{ij}(t_{ij})-\beta_{ij})^+\cdot\mathds{1}[V_{ij}(t_{ij}) \leq \beta_{ij}+ \delta_{ij}]\right\}\right]$ as $\sum_i\sum_{t_i} f_i(t_i)\cdot \eta(t_i,[m])$.

\begin{definition}\label{def:Lipschitz}
A function $v(\cdot,\cdot)$ is \textbf{$a$-Lipschitz} if for any type $t,t'\in T$, and set $X,Y\subseteq [m]$,
$$\left|v(t,X)-v(t',Y)\right|\leq a\cdot \left(\left|X\Delta Y\right|+\left|\{j\in X\cap Y:t_j\not=t_j'\}\right|\right),$$ where $X\Delta Y=\left(X\setminus Y\right)\cup \left(Y\setminus X\right)$ is the symmetric difference between $X$ and $Y$.
\end{definition}

\begin{lemma}\label{lem:eta-lipschitz}
For every $i$, $\eta_i(\cdot, \cdot)$ is subadditive over independent items and $d_i$-Lipschitz. Note that $d_i$ is the variable in the LP in \Cref{fig:bigLP}.
\end{lemma}

\begin{proof}

For every $i,j,t_{ij},\beta_{ij},\delta_{ij}$, denote $h_{ij}(t_{ij},\beta_{ij},\delta_{ij})=(V_{ij}(t_{ij})-\beta_{ij})^+\cdot\mathds{1}[V_{ij}(t_{ij}) \leq \beta_{ij}+ \delta_{ij}]$.

We will first verify that for each $t_i\in T_i$, $\eta_i(t_i,\cdot)$ is monotone, subadditive and has no externalities.

Monotonicity: Let $S_1\subseteq S_2 \subseteq [m]$.
Then:
\begin{align*}
\eta_i(t_i,S_1)=\E_{(\beta_{i},\delta_{i})\sim \cC_{i}}\left[\max_{j\in S_1}\left\{ h_{ij}(t_{ij},\beta_{ij},\delta_{ij})\right\}\right] 
\leq  \E_{(\beta_{i},\delta_{i})\sim \cC_{i}}\left[\max_{j\in S_2}\left\{ h_{ij}(t_{ij},\beta_{ij},\delta_{ij})\right\}\right]
=  \eta_i(t_i,S_2)
\end{align*}

Subadditivity: For any set $S_1,S_2,S_3\subseteq [m]$ such that $S_1 \cup S_2 = S_3$,
it holds that:
\begin{align*}
\eta_i(t_i,S_3)&=\E_{(\beta_{i},\delta_{i})\sim \cC_{i}}\left[\max_{j\in S_3}\left\{ h_{ij}(t_{ij},\beta_{ij},\delta_{ij})\right\}\right] \\
&\leq\E_{(\beta_{i},\delta_{i})\sim \cC_{i}}\bigg[\max_{j\in S_1}\left\{ h_{ij}(t_{ij},\beta_{ij},\delta_{ij})\right\} + \max_{j\in S_2}\left\{ h_{ij}(t_{ij},\beta_{ij},\delta_{ij})\right\}\bigg] \\
&=\eta_i(t_i,S_1) + \eta_i(t_i,S_2)
\end{align*}
The first inequality follows because $S_1\cup S_2=S_3$.

No externalities: for every $S\subseteq [m]$ and $t_i,t_i'\in \cT_i$ such that $t_{ij}=t_{ij}'$ for every $j\in S$, we have
$$\eta_i(t_i,S)=\E_{(\beta_{i},\delta_{i})\sim \cC_{i}}\left[\max_{j\in S}\left\{ h_{ij}(t_{ij},\beta_{ij},\delta_{ij})\right\}\right]=\E_{(\beta_{i},\delta_{i})\sim \cC_{i}}\left[\max_{j\in S}\left\{ h_{ij}(t_{ij}',\beta_{ij},\delta_{ij})\right\}\right]=\eta_i(t_i',S)$$

Now we are going to prove that $\eta_i(\cdot,\cdot)$ is $d_i$-Lipschitz.
For $t_i,t_i'\in \cT_i$ and $X,Y \subseteq[m]$,
let $Z = \{j \in X\cap Y \land t_{ij}=t_{ij}'\}$.
It is enough to show that:
\begin{align*}
\eta_i(t_i,X) - \eta_i(t_i',Y) \leq {\left(|X\setminus Y|+(|X\cap Y|-|Z|)\right)\cdot d_i}= {(|X|-|Z|)\cdot d_i}\\ 
\eta_i(t'_i,Y) - \eta_i(t_i,X) \leq {\left(|Y\setminus X|+(|X\cap Y|-|Z|)\right)\cdot d_i}={(|Y|-|Z|)\cdot d_i}
\end{align*}

We are only going to show $\eta_i(t_i,X) - \eta_i(t_i',Y) \leq  {(|X|-|Z|)\cdot d_i} 
$,
since the other case is similar.
Because $\eta_i(t_i,\cdot)$ is monotone, it suffices to prove that:
$$
\eta_i(t_i,X) - \eta_i(t_i',Y) \leq \eta_i(t_i,X) - \eta_i(t_i',Z) \leq  {(|X|-|Z|)\cdot d_i } 
$$
For each $j \in Z$, $t_{ij} =t_{ij}'$, which implies that $\eta_i(t_i',Z) = \eta_i(t_i,Z)$.
Note that:
\begin{align*}
&\eta_i(t_i,X) - \eta_i(t_i',Z)\\ 
= & \eta_i(t_i,X) - \eta_i(t_i,Z)\\
= & \E_{(\beta_{i},\delta_{i})\sim \cC_{i}}\big[\max_{j\in X}\left\{ h_{ij}(t_{ij},\beta_{ij},\delta_{ij})\right\} - \max_{j\in Z}\left\{ h_{ij}(t_{ij},\beta_{ij},\delta_{ij})\right\} \big] \\
\leq & \E_{(\beta_{i},\delta_{i})\sim \cC_{i}}\big[\max_{j\in Z}\left\{ h_{ij}(t_{ij},\beta_{ij},\delta_{ij})\right\} +\sum_{j\in X \backslash Z} h_{ij}(t_{ij},\beta_{ij},\delta_{ij}) - \max_{j\in Z}\left\{ h_{ij}(t_{ij},\beta_{ij},\delta_{ij})\right\} \big] \\
= &  \E_{(\beta_{i},\delta_{i})\sim \cC_{i}} \left[\sum_{j\in X \backslash Z} h_{ij}(t_{ij},\beta_{ij},\delta_{ij}) \right]\\
= &  \sum_{j\in X \backslash Z} \E_{(\beta_{ij},\delta_{ij})\sim \cC_{ij}} \left[ h_{ij}(t_{ij},\beta_{ij},\delta_{ij})\right]\\
\leq &  \sum_{j\in X \backslash Z} \E_{(\beta_{ij},\delta_{ij})\sim \cC_{ij}} \left[ \delta_{ij} \right]\\
\leq & (|X|-|Z|)d_i 
\end{align*}
Here the second last inequality follows by Constraint \BoundMeanDeltaConstraint~of the LP in \Cref{fig:bigLP} (or \Cref{fig:XOSLP}).
\end{proof}

\begin{lemma}\label{lemma:concentrate}\cite{CaiZ17}
Let $g(t,\cdot)$ with $t\sim D=\prod_j D_j$ be a function drawn from a distribution that is  subadditive over independent items of ground set $I$. If $g(\cdot,\cdot)$ is $c$-Lipschitz, then if we let $a$ be the median of the value of the grand bundle $g(t,I)$, i.e. $a=\inf\left\{x\geq 0: \Pr_t[g(t,I)\leq x]\geq \frac{1}{2}\right\}$,
$$\mathds{E}_t[g(t,I)]\leq 2a+\frac{5c}{2}.$$
\end{lemma}


To finish the proof of \Cref{lem:Q and Q-hat}, we will bound $\sum_i\sum_{t_i}f_i(t_i)\cdot \eta_i(t_i,[m])$ using a modified two-part tariff mechanism. 
Consider the following variant of two-part tariff $\cM$, with a randomized posted price $\beta_{ij}\sim \cB_{ij}$ for buyer $i$ and item $j$, and restricting the buyer to purchase at most one item. The procedure of the mechanism is shown in Mechanism~\ref{mech:variant-tpt}~\footnote{The result holds for any buyers' order, we choose the lexicographical order to keep the notation light.}.

\begin{algorithm}[ht]
\floatname{algorithm}{Mechanism}
\begin{algorithmic}[1]
\setcounter{ALG@line}{-1}
\State Before the mechanism starts, the seller determines a \textbf{distribution} of posted price $\cB_{ij}$ for every buyer $i$ and item $j$. Recall that $\cB_{ij}$ is the marginal distribution of $\beta_{ij}$. 
\State Bidders arrive sequentially in the lexicographical order.
\State When every buyer $i$ arrives, the seller shows her the set of available items $S_i(t_{<i},\beta_{<i})\subseteq [m]$ (see the remark below), as well as the \emph{distribution} of the posted prices $\{\cB_{ij}\}_{j\in [m]}$. 
\State Buyer $i$ is asked to pay an \emph{entry fee} defined as follows:

\centerline{$\xi_i(S_i(t_{<i},\beta_{<i}))=\Med_{t_i\sim D_i}\{\eta_i(t_i,S_i(t_{<i},\beta_{<i}))\}$.} \noindent Here $\Med_x[h(x)]$ denotes the median of a non-negative function $h(x)$ on random variable $x$, i.e.
$\Med_x[h(x)]=\inf\{a\geq 0:\Pr_{x}[h(x)\leq a]\geq\frac{1}{2}\}$.
\State If buyer $i$ (with type $t_i$) agrees to pay the entry fee, then the seller will sample a realized posted price $\beta_{ij}\sim \cB_{ij}$ for every available item $j\in S_i(t_{<i},\beta_{<i})$. The buyer is restricted to purchase \textbf{at most one} item. The buyer then either chooses her favorite item $j^*=\argmax_{j\in S_i(t_{<i},\beta_{<i})}\left(V_{ij}(t_{ij})-\beta_{ij}\right)$, and pays $\beta_{ij^*}$, or leaves with nothing if $V_{ij}(t_{ij})<\beta_{ij},\forall j\in S_i(t_{<i},\beta_{<i})$. If the buyer refuses to pay the entry fee, she gets nothing and pays 0. 
\end{algorithmic}
\caption{{\sf \quad The Rationed Two-part Tariff with Randomized Posted Price $\cM$}}\label{mech:variant-tpt}
\end{algorithm}

\begin{remark}
We notice that in Mechanism~\ref{mech:variant-tpt}, the set of available items when each buyer $i$ comes to the auction, depends on both $t_{<i}$ and the realized prices for the first $i-1$ buyers (denoted by $\beta_{<i}$). Let $S_i(t_{<i},\beta_{<i})$ be the random set of available items when buyer $i$ comes to the auction. Let $S_1(t_{<1},\beta_{<1})=[m]$.
\end{remark} 

It's not hard to see that $\cM$ stated in Mechanism~\ref{mech:variant-tpt} is BIC and IR: When every buyer sees the set of available items and the distribution of posted prices, she can calculate her expected surplus of this set, over the randomness of the posted prices. Then she will accept the entry fee if and only if the expected surplus is at least the entry fee.

In Mechanism~\ref{mech:variant-tpt}, by union bound, for every item $j$,

$$\Pr[j\in S_i(t_{<i},\beta_{<i})]\geq 1-\sum_{k<i}\Pr_{t_{ij}\sim D_{ij},\beta_{ij}\sim \cB_{ij}}[V_{ij}(t_{ij})\geq \beta_{ij}]\geq \frac{1}{2}\quad{\text{ (Due to Constraint~\ReduceDemandConstaint)}}$$

We notice that for every realization of $t_{<i},\beta_{<i}$, after seeing the remaining item set $S_i(t_{<i},\beta_{<i})$, buyer $i$'s expected surplus if she enters the mechanism is:
$$\E_{(\beta_{i},\delta_{i})\sim \cC_{i}}\left[\max_{j\in S_i(t_{<i},\beta_{<i})}\left\{ (V_{ij}(t_{ij})-\beta_{ij})^+\right\}\right]\geq \eta_i(t_i,S_i(t_{<i},\beta_{<i})).$$

Thus the buyer will accepts the entry fee with probability at least $1/2$. Hence

\begin{equation}\label{inequ:Q and Q_hat-term2}
\begin{aligned}
\rev(\cM) & \geq  \frac{1}{2}\sum_i\E_{t_{<i},\beta_{<i}}[\xi_i(S_i(t_{<i},\beta_{<i}))]\geq \frac{1}{2}\cdot \sum_i(\frac{1}{2}\E_{t_i,t_{<i},\beta_{<i}}[\eta_i(t_i,S_i(t_{<i},\beta_{<i}))]-\frac{5}{4}d_i)\\
 & \geq \frac{1}{2}\cdot \sum_i(\frac{1}{4}\E_{t_i}[\eta_i(t_i,[m])]-\frac{5}{4}d_i) \geq\frac{1}{8}\cdot \sum_i \E_{t_i}[\eta_i(t_i,[m])]-5\cdot \prev 
\end{aligned}
\end{equation}



Here the second inequality is obtained by applying \Cref{lemma:concentrate} to function $\eta_i(t_i,\cdot)$ and ground set $I=S_i(t_{<i},\beta_{<i})$ for every $t_{<i},\beta_{<i}$. The last inequality comes from constraint \BoundSumDeltaConstraint~of the LP in \Cref{fig:bigLP}. The third inequality comes from the following:  


Fix any $t_i,\beta_i,\delta_i$. Let $j^*=\argmax_{j\in [m]}(V_{ij}(t_{ij})-\beta_{ij})^+\cdot\mathds{1}[V_{ij}(t_{ij}) \leq \beta_{ij}+ \delta_{ij}]$. Then 

\begin{align*}
&\E_{t_{<i},\beta_{<i}}\left[\max_{j\in S_i(t_{<i},\beta_{<i})}\left\{ (V_{ij}(t_{ij})-\beta_{ij})^+\cdot\mathds{1}[V_{ij}(t_{ij}) \leq \beta_{ij}+ \delta_{ij}]\right\}\right]\\
\geq & \E_{t_{<i},\beta_{<i}}\left[\max_{j\in [m]}\left\{ (V_{ij}(t_{ij})-\beta_{ij})^+\cdot\mathds{1}[V_{ij}(t_{ij}) \leq \beta_{ij}+ \delta_{ij}]\right\}\cdot \ind[j^*\in S_i(t_{<i},\beta_{<i})]\right]\\
= & \max_{j\in [m]}\left\{ (V_{ij}(t_{ij})-\beta_{ij})^+\cdot\mathds{1}[V_{ij}(t_{ij}) \leq \beta_{ij}+ \delta_{ij}]\right\}\cdot \Pr_{t_{<i},\beta_{<i}}[j^*\in S_i(t_{<i},\beta_{<i})]\\
\geq & \frac{1}{2}\cdot \max_{j\in [m]}\left\{ (V_{ij}(t_{ij})-\beta_{ij})^+\cdot\mathds{1}[V_{ij}(t_{ij}) \leq \beta_{ij}+ \delta_{ij}]\right\}
\end{align*}

Taking expectation over $t_i$ and $(\beta_i,\delta_i)\sim \cC_i$ on both sides, we have 
$$\E_{t_i,t_{<i},\beta_{<i}}[\eta_i(t_i,S_i(t_{<i},\beta_{<i}))]\geq \frac{1}{2}\E_{t_i}[\eta_i(t_i,[m])],$$


 
which is exactly the third inequality of \eqref{inequ:Q and Q_hat-term2}. By combining Inequalities \eqref{inequ:Q and Q_hat-term1} and \eqref{inequ:Q and Q_hat-term2}, we have
$$2\sum_j(Q_j - \hat{Q}_j) \leq 89\cdot\prev + 16\cdot\rev(\cM)$$

Finally, since in $\cM$, each buyer is restricted to purchase at most one item, it can be treated as a BIC and IR mechanism in the unit-demand setting. By \cite{ChawlaHMS10,KleinbergW12,CaiDW16}, $\rev(\cM)\leq 24\cdot \prev$. We complete our proof for Lemma~\ref{lem:Q and Q-hat}.


\notshow{
\begin{align*}
& \sum_{j\in[m]}\sum_{i\in[n]} \sum_{\substack{t_{i,j}\in T_{i,j}\\V_{i,j}(t_{i,j}) \geq Q_j^* + \tau_i}} f_{i,j}(t_{i,j})  \sum_{\beta\in V_{i,j}(T_{i,j})}
\sum_{\delta \in \Delta V_{i}} (V_{i,j}(t_{i,j})-\beta)\cdot
\lambda_{i,j}(t_{i,j},\beta, \delta)
\mathds{1}[\delta + \beta \geq V_{i,j}(t_{i,j}) \geq \beta ] \\
= & \sum_{j\in[m]}\sum_{i\in[n]} \sum_{t_{i,j}\in T_{i,j}} f_{i,j}(t_{i,j})  \sum_{\beta\in V_{i,j}(T_{i,j})}
\sum_{\delta \in \Delta V_{i}} (V_{i,j}(t_{i,j})-\beta)\cdot
\lambda_{i,j}(t_{i,j},\beta, \delta)
\mathds{1}[\delta + \beta \geq V_{i,j}(t_{i,j}) \geq \max( \beta, Q_j^* + \tau_i) ] \\
\leq & \sum_{j\in[m]} \sum_{i\in[n]} \sum_{t_{i,j}\in T_{i,j}} f_{i,j}(t_{i,j})  \sum_{\beta\in V_{i,j}(T_{i,j})}
 (V_{i,j}(t_{i,j})-\beta) \sum_{\delta \in \Delta V_{i}}
\lambda_{i,j}'(\beta, \delta)
\mathds{1}[\delta + \beta \geq V_{i,j}(t_{i,j}) \geq \max( \beta, Q_j^* + \tau_i) ] 
\end{align*}

Observe that the quantity above is the sum of the truncated utilities if we pretend that the buyers are additive,
when the price is set to 
$\beta_{i,j}$
of the $i$-th agent for the $j$-th item,
where $\beta_{i,j}$ and the truncation $\delta_{i,j}$ is sampled (independently from the other items and agents, but the distribution between $\beta_{i,j}$ and $\delta_{i,j}$ is coupled) from $B_{i,j} \times \Delta_{i,j}$.
By the constraint of the program,
the expected truncated utility of the $i$-th agent is at most $\delta_i$.
Let

$$u_{ADD} = \sum_{i\in[n]} \sum_{t_{i,j}\in T_{i,j}} f_{i,j}(t_{i,j})  \sum_{\beta\in V_{i,j}(T_{i,j})}
 (V_{i,j}(t_{i,j})-\beta) \sum_{\delta \in \Delta V_{i}}
\lambda_{i,j}'(\beta, \delta)
\mathds{1}[\delta + \beta \geq V_{i,j}(t_{i,j}) \geq \max( \beta, Q_j^* + \tau_i) ] 
$$

Let $u_{UD}$ be the expected truncated utility if the buyer was unit-demand,
notice that a lower bound of $u_{UD}$ is:
\begin{align*}
u_{UD} \geq & \sum_{i\in[n]} \sum_{t_{i,j}\in T_{i,j}} f_{i,j}(t_{i,j})  \sum_{\beta\in V_{i,j}(T_{i,j})}
 (V_{i,j}(t_{i,j})-\beta) \sum_{\delta \in \Delta V_{i}}
\lambda_{i,j}'(\beta, \delta)
\mathds{1}[\delta + \beta \geq V_{i,j}(t_{i,j}) > \max( \beta, Q_j^* + \tau_i) ] \\
& \quad \cdot \prod_{k\neq i}\Pr_{\substack{\beta \sim B_{i,k}\\t_{i,k}\sim D_{i,k}}}[V_{i,k}(t_{i,k}) \leq \max( \beta, Q_k^* + \tau_i) ] \\
\geq & \sum_{i\in[n]} \sum_{t_{i,j}\in T_{i,j}} f_{i,j}(t_{i,j})  \sum_{\beta\in V_{i,j}(T_{i,j})}
 (V_{i,j}(t_{i,j})-\beta) \sum_{\delta \in \Delta V_{i}}
\lambda_{i,j}'(\beta, \delta)
\mathds{1}[\delta + \beta \geq V_{i,j}(t_{i,j}) > \max( \beta, Q_j^* + \tau_i) ] \frac{1}{2} \\
= & \frac{u_{ADD}}{2}
\end{align*}

Now we are going to show that there exists a simple mechanism that generates revenue close to $u_{UD}$.
Let $B\times \Delta$ be the coupled distribution of $B_i \times \Delta_i$.
Let 
$$u^T_i(t, X) = \E_{(\beta, \delta) \sim B \times \Delta} [ \max_{j \in X} (V_{i,j}(t_{i,j})- \beta_j)\mathds{1}[\delta_{i,j} + \beta_j \geq V_{i,j}(t_{i,j}) \geq \beta_j ]$$

Note that $\sum_{i\in[n]}\E_{t_i \sim \DD_i}[u^T_i(t, [m])] \geq u_{UD}$.

\begin{lemma}
The function $u^T_i(\cdot, \cdot)$ is $\delta_i$-Lipschitz and subadditive over the items.
\end{lemma}

}


\notshow{

\todo{Let $t_i \in T_i$.
We denote by $W_{t_i}$ the the polytope of feasible welfares normalised to $[0,1]$  when the $i$-th agent has type $t_i$ (note that this can be achieved by dividing the supporting price that the agent has for the allocation by $V_{i,j}(t_{i,j})$).
Let $W_i$ be a mixture of polytopes $\{W_{t}\}_{ t \in T}$ with distribution $\{f(t)\}_{t \in T}$.
The bit complexity of the polytope depends on the supporting prices, which (I believe) its reasonable to assume that are bounded.

Note that for $t_i \in T_i$,
we can optimize any objective over $W_{t_i}$ by using a demand oracle. 
By sampling we can get a sampling an empirical polytope close to the polytope of $W_i$. Thus we can calculate the value of the LP in \Cref{fig:bigLP} with respect to the sampled polytope.

Observe that the way $W_i$ is defined here,
is the scaled by $f_i(t_i)$ version of the $W_i$ that was defined in the program

}

}

\end{proof}

\subsection{Analyzing the Revenue of $\Mtpt$ and the Proof of Theorem~\ref{thm:bounding-lp-simple-mech-XOS}}\label{sec:proof of bounding LP with simple mech}

In this section, we will show that $\sum_{j\in[m]} Q_j$ can be bounded using the revenue of the two-part tariff $\Mtpt$ defined in Mechanism~\ref{def:constructed-SPEM}. To analyze the revenue of $\Mtpt$, we require the following definition.





\begin{definition}\label{def:mu}


For any buyer $i$ and type $t_i$, let $C_i(t_i)=\{j\in [m]\mid V_{ij}(t_{ij})\leq Q_j+\tau_i\}$. For any $t_i\in \cT_i$ and set $S\subseteq [m]$, let

$$\mu_i(t_i,S)=\max_{S'\subseteq S}\left(v_i(t_i,S'\cap C_i(t_i))-\sum_{j\in S'}Q_j\right)$$ 
\end{definition}

\begin{lemma}\label{lem:mu_i lipschitz}
For every $i$, {if $v_i(\cdot,\cdot)$ is subadditive over  independent itmes}, then $\mu_i(\cdot, \cdot)$ is subadditive over independent items and $\tau_i$-Lipschitz.
\end{lemma}

\begin{proof}
First we are going to show that $\mu_i(t,\cdot)$ is monotone.
Note that for sets $S_1 \subseteq S_2$ the following holds:
\begin{align*}
    \mu_i(t_i,S_1)=& \max_{S'\subseteq S_1}\left(v_i(t_i,S'\cap C_i(t_i))-\sum_{j\in S'}Q_j\right)\\
    \leq &\max_{S'\subseteq S_2}\left(  v_i(t_i, S'\cap C_i(t_i)) -\sum_{j\in S'}Q_j\right)\\
    = & \mu_i(t_i,S_2)
\end{align*}

Now we are going to show that $\mu_i(t,\cdot)$ is subadditive.
Let $S_1,S_2,S_3\subseteq[m]$ such that $S_1 \cup S_2 = S_3$ and $S_1 \cap S_2 = S_c$.
Let $S_a = S_1 \backslash S_c$ and $S_b=S_2$,
then for any $t_i\in T_i$ we have the following:
\begin{align*}
    \mu_i(t_i,S_3)=& \max_{S'\subseteq S_3}\left(v_i(t_i,S'\cap C_i(t_i))-\sum_{j\in S'}Q_j\right)\\
    \leq &\max_{S'\subseteq S_3}\left(  v_i(t_i,\left(S_a\cap S'\right)\cap C_i(t_i)) + v_i(t_i,\left(S_b\cap S'\right)\cap C_i(t_i))-\sum_{j\in S'}Q_j\right)\\
    = & \max_{S'\subseteq S_3}\left(  \left(v_i(t_i,\left(S_a\cap S'\right)\cap C_i(t_i))- \sum_{j\in S_a\cap S'}Q_j \right)+ \left(v_i(t_i,\left(S_b\cap S'\right)\cap C_i(t_i))-\sum_{j\in S_b\cap S'}Q_j\right)\right)\\
    \leq & \max_{S'\subseteq S_a}\left(v_i(t_i,S'\cap C_i(t_i))-\sum_{j\in S'}Q_j\right) + \max_{S'\subseteq S_b}\left(v_i(t_i,S'\cap C_i(t_i))-\sum_{j\in S'}Q_j\right)\\
    = & \mu_i(t_i,S_a) + \mu_i(t_i,S_b) \\
    \leq & \mu_i(t_i,S_1) + \mu_i(t_i,S_2)
\end{align*}

The first inequality follows because $v_i(t_i,\cdot)$ is a subadditive function and $S_a\cup S_b=S_3$. The second inequality follows because {$\max$ is subadditive.}
The final inequality follows from the fact that $S_b =S_2$, $S_1 \supseteq S_a$ and that $\mu_i(t_i,\cdot)$ is monotone.

We now prove that $\mu_i(\cdot,\cdot)$ has no externalities. Fix any $S\subseteq [m]$ and $t_i,t_i'\in \cT_i$ such that $t_{ij}=t_{ij}'$ {for all $j\in S$}. We notice that by the definition of $C_i$, $S'\cap C_i(t_i)=S'\cap C_i(t_i')$ for all $S'\subseteq S$. Since $v_i(\cdot,\cdot)$ has no externalities, $v_i(t_i,S'\cap C_i(t_i))=v_i(t_i',S'\cap C_i(t_i'))$ for every $S'\subseteq S$. Thus $\mu_i(t_i,S)=\mu_i(t_i',S)$.

Now we are going to show that $\mu_i(\cdot,\cdot)$ is $\tau_i$-Lipschitz.
Let $t_i,t_i' \in \cT_i$ and $X,Y \subseteq[m]$ and $c= |\{j \in [m]: j \in X \Delta Y \text{ or }t_{ij} \neq t_{ij}'\}|$,~\footnote{$\Delta$ stands for the symmetric difference between two sets.} we need to show that:
$$
|\mu_i(t_i,X) - \mu_i(t_i',Y)| \leq c\cdot \tau_i
$$
Let $Z = \{j :j \in X\cap Y \text{ and } t_{ij}=t_{ij}'\}$.
Since $\mu_i(t_i,\cdot)$ is monotone, in order to show that $\mu_i(\cdot,\cdot)$ is $\tau_i$-Lipschitz,
it is enough to show that
\begin{align*}
\mu_i(t_i,X) - \mu_i(t_i',Y) \leq \mu_i(t_i,X) - \mu_i(t_i',Z) \leq c\cdot \tau_i \\
\mu_i(t'_i,Y) - \mu_i(t_i,X) \leq \mu_i(t'_i,Y) - \mu_i(t_i,Z) \leq c\cdot \tau_i
\end{align*}
We are only going to prove that $\mu_i(t_i,X) - \mu_i(t_i',Z) \leq c\cdot \tau_i$ since the other case is similar.
Because for each $j\in Z$, $t_{i,j}'=t_{i,j}$, then $\mu_i(t_i',Z)= \mu_i(t_i,Z)$.
We have
\begin{align*}
    \mu_i(t_i,X)=& \max_{S'\subseteq X}\left(v_i(t_i,S'\cap C_i(t_i))-\sum_{j\in S'}Q_j\right)\\
    \leq &\max_{S'\subseteq X}\left(\sum_{j\in S'\backslash Z}\left(v_i(t_i,\{j\}\cap C_i(t_i)) - Q_j \right) +  \left(v_i(t_i,(Z\cap S')\cap C_i(t_i))-\sum_{j\in Z\cap S'}Q_j\right)\right)\\
    \leq &\max_{S'\subseteq X}\left(\sum_{j\in S'\backslash Z}\left(V_{ij}(t_{ij}) - Q_j \right)^+\ind[V_{ij}(t_{ij}) \leq Q_j + \tau_i] +  \left(v_i(t_i,(Z\cap S')\cap C_i(t_i))-\sum_{j\in Z\cap S'}Q_j\right)\right)\\
    =& \max_{S'\subseteq Z}  \left(v_i(t_i,S'\cap C_i(t_i))-\sum_{j\in  S'}Q_j\right) + \sum_{j\in X\backslash Z}\left(V_{ij}(t_{ij}) - Q_j \right)^+\ind[V_{ij}(t_{ij}) \leq Q_j + \tau_i]\\
    \leq & \max_{S'\subseteq Z}  \left(v_i(t_i,S'\cap C_i(t_i))-\sum_{j\in S'}Q_j\right) + (|X| - |Z|)\tau_i\\
    \leq &\mu_i(t_i,Z) +  c\cdot \tau_i
\end{align*}
The first inequality follows because $v_i(t_i,\cdot)$ is subadditive.
The second inequality follows because $v_i(t_i,\{j\}\cap C_i(t_i)) - Q_j \leq (V_{ij}(t_{ij}) - Q_j)^+\ind[V_{ij}(t_{ij})\leq Q_j + \tau_i]$.
\end{proof}

\begin{lemma}\label{lem:lower bounding mu}
	For every type profile $t\in \cT$, let $\textsc{SOLD}(t)$ be the set of items sold in mechanism $\Mtpt$. Then
\begin{align*}
\E_{t}\left[\sum_{i\in[n]} \mu_i\left(t_i,S_i(t_{< i})\right) \right]&\geq \sum_j \Pr_{t}[j\notin \sold(t)]\cdot(2\hat{Q}_j-Q_j)\\
&\geq \sum_{j} \Pr_t\left[j\notin \sold(t)\right]\cdot Q_j - 473\cdot \prev 
\end{align*}
\end{lemma}

\begin{proof}
By the definition of polytope $W_i$, for every buyer $i$ and $t_i\in \cT_i$, there exists an vector of non-negative numbers $\{\sigma_{iS}^{(k)}(t_i)\}_{S\subseteq [m],k\in [K]}$, such that $\sum_{S,k} \sigma_{iS}^{(k)}(t_i)\leq 1$ and
\begin{equation}
\pi_{ij}(t_{ij})=f_{ij}(t_{ij})\cdot\sum_{t_{i,-j}}f_{i,-j}(t_{i,-j})\cdot \sum_{S:j\in S}\sum_{k\in [K]}\sigma_{iS}^{(k)}(t_{ij},t_{i,-j})
\end{equation}

\begin{equation}\label{equ:w and gamma}
w_{ij}(t_{ij})\cdot V_{ij}(t_{ij}) \leq f_{ij}(t_{ij})\cdot \sum_{t_{i,-j}\in \cT_{i,-j}} f_{i,-j}(t_{i,-j})\sum_{S:j\in S}\sum_k\sigma_{iS}^{(k)}(t_{ij},t_{i,-j})\cdot\alpha_{ij}^{(k)}(t_{ij})
\end{equation}


We have
\begin{align*}
&\E_{t}\left[\sum_i \mu_i\left(t_i,S_i(t_{< i})\right) \right]\\
\geq & \sum_i \E_{t_i, t_{-i}}\left[\sum_{S\subseteq [m]}\sum_k\sigma_{iS}^{(k)}(t_i)\cdot\mu_i\left(t_i,S_i(t_{<i})\cap S\right)\right]\\
\geq & \sum_i \E_{t_i, t_{-i}}\left[\sum_{S,k}\sigma_{iS}^{(k)}(t_i)\cdot\sum_{j\in S} \ind\left[j\in S_i(t_{<i})\right]\cdot\left(\alpha_{ij}^{(k)}(t_{ij})\cdot\ind[V_{ij}(t_{ij})\leq Q_j+\tau_i]-Q_j\right)^+ \right]\\
=& \sum_i \E_{t_i}\left[\sum_{j\in [m]}\Pr_{t_{-i}}[j\in S_i(t_{<i})]\cdot\sum_{S:j\in S}\sum_k\sigma_{iS}^{(k)}(t_i)\cdot\left(\alpha_{ij}^{(k)}(t_{ij})\cdot\ind[V_{ij}(t_{ij})\leq Q_j+\tau_i]-Q_j\right)^+ \right]\\
\geq & \sum_i\sum_j \Pr_{t}[j\notin \sold(t)]\cdot \E_{t_i}\left[\sum_{S: j\in S}\sum_k \sigma_{iS}^{(k)}(t_i)\cdot\left(\alpha_{ij}^{(k)}(t_{ij})\cdot\ind[V_{ij}(t_{ij})\leq Q_j+\tau_i]-Q_j\right)^+\right]\\
\geq & \sum_j \Pr_{t}[j\notin \sold(t)]\cdot \sum_i \sum_{t_i} f_i(t_i)\cdot \sum_{S: j\in S}\sum_k \sigma_{iS}^{(k)}(t_i)\cdot\left(\alpha_{ij}^{(k)}(t_{ij})\cdot\ind[V_{ij}(t_{ij})\leq Q_j+\tau_i]-Q_j\right)\\
= & \sum_j \Pr_{t}[j\notin \sold(t)]\cdot \sum_i \sum_{t_{ij}} f_{ij}(t_{ij})\sum_{t_{i,-j}}f_{i,-j}(t_{i,-j})\cdot\sum_{S:j\in S}\sum_k\sigma_{iS}^{(k)}(t_i)\cdot(\alpha_{ij}^{(k)}(t_{ij})\cdot\ind[V_{ij}(t_{ij})\leq Q_j+\tau_i]-Q_j)\\
\geq & \sum_j \Pr_{t}[j\notin \sold(t)]\cdot \sum_i \sum_{t_{ij}} w_{ij}(t_{ij}) V_{ij}(t_{ij})\cdot \ind[V_{ij}(t_{ij})\leq Q_j+\tau_i]-\sum_j\Pr_t[j\not\in \sold(t)]\cdot Q_j
\end{align*}


The first inequality is because $\mu_i(t_i,S)$ is monotone in set $S$ for any $i,t_i$ and $\sum_{S,k}\sigma_{iS}^{(k)}(t_i)\leq 1$. For any fixed $i,t_i$ and set $S$, if we let $S'$ be the set of items that are in $S\cap S_i(t_{<i})$ and satisfy that $\alpha_{ij}^{(k)}(t_{ij})\cdot\ind[V_{ij}(t_{ij})\leq Q_j+\tau_i]-Q_j\geq 0$. Clearly $S'\subseteq C_i(t_i)$. Then 
\begin{align*}
\mu_i\left(t_i,S_i(t_{<i})\cap S\right)&\geq v_i(t_i,S')-\sum_{j\in S'}Q_j\\
&= \max_{k'\in [K]}\sum_{j\in S'}\alpha_{ij}^{(k')}(t_{ij})-\sum_{j\in S'}Q_j\geq \sum_{j\in S'} \left(\alpha_{ij}^{(k)}(t_{ij})-Q_j\right)\\
&=\sum_{j\in S'} \left(\alpha_{ij}^{(k)}(t_{ij})\cdot\ind[V_{ij}(t_{ij})\leq Q_j+\tau_i]-Q_j\right)~~~~~(S'\subseteq C_i(t_i))
\end{align*} 

This inequality is exactly the second inequality above. The third inequality is because $\Pr_{t_{<i}}[j\in S_i(t_{<i})]\geq \Pr_{t}[j\notin \sold(t)]$ for all $j$ and $i$, as the LHS is the probability that the item is not sold after the seller has visited the first $i-1$ buyers and the RHS is the probability that the item remains unsold till the end of the mechanism $\Mtpt$. 
The last inequality follows from Inequality~\eqref{equ:w and gamma} and Constraint \PiConstraint~of the LP in \Cref{fig:bigLP} (or in \Cref{fig:XOSLP}):

$$\sum_i\sum_{t_{ij}}f_{ij}(t_{ij})\sum_{t_{i,-j}}f_{i,-j}(t_{i,-j})\cdot \sum_{S:j\in S}\sum_k\sigma_{iS}^{(k)}(t_i)=\sum_i\sum_{t_{ij}}\pi_{ij}(t_{ij})\leq 1$$

Notice that by \Cref{def:Q_j} and Constraint \LambdaMarginalConstraint~of the LP in \Cref{fig:bigLP} (or in \Cref{fig:XOSLP}), for every $i,j,t_{ij}$, 

$$
\sum_{\beta_{ij},\delta_{ij}}\lambda_{ij}(t_{ij},\beta_{ij},\delta_{ij})=w_{ij}(t_{ij})/f_{ij}(t_{ij})$$
Thus we have

\begin{align*}
&\sum_j \Pr_{t}[j\notin \sold(t)]\cdot \sum_i \sum_{t_{ij}} w_{ij}(t_{ij}) V_{ij}(t_{ij})\cdot \ind[V_{ij}(t_{ij})\leq Q_j+\tau_i]-\sum_j\Pr_t[j\not\in \sold(t)]\cdot Q_j\\
= & \sum_j \Pr_{t}[j\notin \sold(t)]\cdot \sum_i \sum_{t_{ij}} f_{ij}(t_{ij})V_{ij}(t_{ij})\cdot \ind[V_{ij}(t_{ij})\leq Q_j+\tau_i] \sum_{\beta_{ij},\delta_{ij}}\lambda_{ij}(t_{ij},\beta_{ij},\delta_{ij})\\
&\quad\quad-\sum_j\Pr_t[j\not\in \sold(t)]\cdot Q_j\\
\geq &\sum_j \Pr_{t}[j\notin \sold(t)]\cdot (2\hat{Q}_j-Q_j)~~~~~\text{(\Cref{def:hatQ_j})}\\
= &  \sum_j \Pr_{t}[j\notin \sold(t)]\cdot Q_j-  \sum_j\Pr_{t}[j\notin \sold(t)]\cdot2(Q_j-\hat{Q}_j)\\
\geq & \sum_j \Pr_{t}[j\notin \sold(t)]\cdot Q_j - \sum_j 2(Q_j-\hat{Q}_j )~~~~~{\text{(\Cref{lem:Q and Q-hat}, $Q_j-\hat{Q}_j\geq 0$ for all $j$)}}\\
\geq & \sum_j \Pr_{t}[j\notin \sold(t)]\cdot Q_j- 473\cdot\prev
~~~~~\text{(\Cref{lem:Q and Q-hat})}
\end{align*}
\end{proof}

Now we give the proof of \Cref{thm:bounding-lp-simple-mech-XOS}. Note that this is also the proof of \Cref{thm:bounding-lp-simple-mech}.

\begin{prevproof}{Theorem}{thm:bounding-lp-simple-mech-XOS}

For every $i,t_{<i}$, we apply \Cref{lemma:concentrate} to function $\mu_i(t_i,\cdot)$ and ground set $S_i(t_{<i})$. By \Cref{def:mu}, we have
\begin{equation}\label{equ:bounding-mu_i}
\E_{t_i}[\mu_i(t_i,S_i(t_{<i}))]\leq 2\cdot \Med_{t_i}(\mu_i(t_i,S_i(t_{<i})))+\frac{5}{2}\cdot \tau_i
\end{equation}

We now bound the revenue of $\Mtpt$. For every $i\in [n]$, $t_i\in \cT_i$ and $S\subseteq [m]$, let $\mu'_i(t_i,S)=\max_{S'\subseteq S}(v_i(t_i,S')-\sum_{j\in S'}Q_j)$ which is at least as large as $\mu_i(t_i,S)$. Then the 
surplus of buyer $i$ with true type $\hat{t}_i$, for the set $S_i(t_{<i})$ is $\mu'_i(\hat{t}_i,S_i(t_{<i}))$. By Mechanism~\ref{def:constructed-SPEM}, the entry fee $\xi_i(S_i(t_{<i}),t_i')=\mu'_i(t_i',S_i(t_{<i}))$ for every sampled type $t_i'$. Thus for every $t_{<i}$, we have
$$\Pr_{\hat{t}_i,t_i'\sim D_i}\left[\mu'_i(\hat{t}_i,S_i(t_{<i}))\geq \xi_i(S_i(t_{<i}),t_i')\geq \Med_{t_i}(\mu_i(t_i,S_i(t_{<i})))\right]\geq \frac{1}{8}$$

In other words, for every $t_{<i}$, with probability at least $1/8$, the buyer accepts the entry fee, and the entry fee is at least $\Med_{t_i}(\mu_i(t_i,S_i(t_{<i})))$. Thus the revenue of $\Mtpt$ from the entry fee is at least
\begin{align*}
&\frac{1}{8}\sum_i\E_{t_{<i}}\left[\Med_{t_i}(\mu_i(t_i,S_i(t_{<i})))\right]\\
\geq & \frac{1}{16}\sum_i\E_{t_i,t_{<i}}[\mu_i(t_i,S_i(t_{<i}))]-\frac{5}{32}\cdot \sum_i\tau_i~~~~~~\text{(Inequality~\eqref{equ:bounding-mu_i})}\\
\geq & \frac{1}{16}\sum_{j} \Pr_t\left[j\notin \sold(t)\right]\cdot Q_j - \frac{493}{16}\cdot\prev ~~~~~~\text{(\Cref{lem:lower bounding mu} and \Cref{lem:bounding tau_i})}
\end{align*}
We notice that for $\Mtpt$, the revenue from the posted prices are $\sum_j \Pr[j\in \sold(t)]\cdot Q_j$. Thus
\begin{align*}
\rev(\Mtpt)\geq & \frac{1}{16}\sum_{j} \Pr_t\left[j\notin \sold(t)\right]\cdot Q_j - \frac{493}{16}\cdot \prev +\sum_j \Pr[j\in \sold(t)]\cdot Q_j\\
\geq & \frac{1}{16}\cdot \sum_{j\in [m]}Q_j-\frac{493}{16}\cdot \prev
\end{align*}

Thus

$$2\cdot \sum_jQ_j\leq 986\cdot \prev+32\cdot\rev(\Mtpt)$$

We then show that $\Mtpt$ can be computed in polynomial time: By \Cref{def:Q_j}, the posted price $Q_j$ can be computed in time $\poly(n,m,\sum_{i,j}|\cT_{ij}|)$, given the feasible solution of the LP in \Cref{fig:bigLP} (or in \Cref{fig:XOSLP}). Given the set of available items $S_i(t_{<i})$, for every sampled type $t_i'$, calculating the entry fee requires a single query from the demand oracle. For every buyer $i$ with reported type $t_i$, the mechanism requires a single query from the demand oracle to obtain her favorite bundle among the set of available items, under prices $\{Q_j\}_{j\in [m]}$, and to determine whether the buyer will accept the entry fee.  

Lastly, by \Cref{thm:chms10}, we can compute an RPP $\Mpp$ with the desired running time and query complexity, such that $\Mpp\geq \frac{1}{6.75}(1-\frac{1}{nm})\cdot\prev$. We finish our proof.
\end{prevproof}


\section{Multiplicative Approximation of Down-Monotone and Boxable Polytopes}\label{sec:multi-approx-polytope}

In this section, we provide a proof of \Cref{thm:multiplicative approx for constraint additive-main body} and prove \Cref{thm:main XOS-main body} for constrained-additive valuations using \Cref{thm:multiplicative approx for constraint additive-main body}. We restate the theorem here.

\begin{theorem}\label{thm:main}
(Restatement of \Cref{thm:main XOS-main body} for constrained-additive valuations) Let $T=\sum_{i,j}|\cT_{ij}|$ and $b$ be the bit complexity of the problem instance (\Cref{def:bit complexity}).
For constrained-additive buyers, for any $\delta>0$, there exists an algorithm that computes a rationed posted price mechanism or a two-part tariff mechanism, such that the revenue of the mechanism is at least $c\cdot \opt$ for some absolute constant $c>0$ with probability $1-\delta-\frac{2}{nm}$. Our algorithm assumes query access to a value oracle and a demand oracle of buyers' valuations, and has running time $\poly(n,m,T,b,\log (1/\delta))$.
\end{theorem}

For \Cref{thm:multiplicative approx for constraint additive-main body}, we indeed prove a result for a natural family of polytopes. Throughout this section, we assume that the polytope we consider is  \emph{down-monotone}. Formally, a polytope $\cP\subseteq [0,1]^d$ is down-monotone if and only if for every $\bx\in \cP$ and $\textbf{0}\leq \bx'\leq \bx$, we have $\bx'\in \cP$.
To state our result, we need the following definitions.

\begin{definition}
For any two sets $A,B \subseteq \mathbb{R}^d$,
we denote by $A+B$ the \emph{Minkowski addition} of set $A$ and set $B$ where:
$$A+B = \{a + b : a \in A\text{ and } b \in B\}
$$

Note that if both $A$ and $B$ are convex, then $A+B$ is also convex.
\end{definition}

\begin{definition}
Let $\cP$ be a convex polytope, we define $a\cdot \cP:=\{a\bx: \bx\in\cP\}$ for any $a\geq 0$.
\end{definition}

\begin{definition}\label{def:mixture}
Let $\ell$ be a finite integer. For any set of convex sets $\{\cP_i\}_{i\in [\ell]}$ and a distribution $\cD=\{q_i\}_{i\in[\ell]}$, 
 the set $\cP = \sum_{i\in[\ell]}q_i \cP_i$
is called the mixture of $\{\cP_i\}_{i\in [\ell]}$ over distribution $\cD$.
\end{definition}

\begin{definition}\label{def:truncated polytope}
Let $\cP,\cQ\subseteq [0,1]^d$ be down-monotone polytopes. For each coordinate $j\in[d]$,
we define the width of $\cP$ at coordinate $j$ as $l_j(\cP)=\max_{\bx\in \cP}x_j$.
For any $\eps>0$, we define the {($\eps,\cQ$)-truncated polytope of $\cP$ (denoted as $\cP^{tr(\eps,\cQ)}$)} in the following way: $\bx\in \cP^{tr(\eps,\cQ)}$ if and only if there exists $\bx'\in \cP$ such that $x_j=x_j'\cdot \ind[l_j(\cQ)\geq \eps],\forall j\in [d]$.
We notice that since $\cP$ is down-monotone, $\cP^{tr(\eps,\cQ)}\subseteq \cP$. Moreover, $\cP^{tr(\eps,\cQ)}$ is convex if $\cP$ is convex.
We also use $\cP^{tr(\eps)}$ to denote $\cP^{tr(\eps,\cP)}$.
\end{definition}

\begin{definition}\label{def:box polytope}
Let $\cP\subseteq [0,1]^d$.
For any $\eps>0$, define the $\eps$-box polytope $\pbox$ of $\cP$ as follows: 
{$\pbox=\{\bx\subseteq [0,1]^d: x_j\leq \min\left( \eps,l_j(\cP)\right),\forall j\in [d]\}$.} Clearly, $\pbox$ is convex.
\end{definition}

\Cref{thm:special case of multiplicative approx} is the main theorem of this section. We prove that if $\cP$ is a mixture of a set of down-monotone, convex polytopes $\{\cP_i\}_{i\in[\ell]}$, and $\cP$ contains the polytope $c\cdot\pbox$ for some {$c\leq1$}, then there exists another down-monotone, convex polytope $\widehat{\cP}$ sandwiched between $c/6\cdot \cP$ and $\cP$. {And more importantly, we have an efficient separation oracle for $\widehat{\cP}$, whose running time is \emph{independent of $\ell$}, as long as we can efficiently optimize any linear objective for every $\cP_i$. The key feature of our separation oracle for $\widehat{\cP}$ is that its running time does not depend on $\ell$, as in our applications, $\ell$ is usually exponential in the input size.}

\begin{theorem}\label{thm:special case of multiplicative approx}
Let $\ell$ be a positive integer, and $\cP \subseteq [0,1]^d$ be a mixture of $\{\cP_i\}_{i\in[\ell]}$ over distribution $\cD=\{q_i\}_{i\in [\ell]}$,
where for each $i\in[\ell]$, $\cP_i\subseteq [0,1]^d$ is a convex and down-monotone polytope. 
Suppose for every $i\in [\ell]$, 
{there exists an oracle $\cQ_i(\cdot)$, whose output $\cQ_i(\bm{a})\in \argmax\{\bm{a}\cdot \bx:\bx\in \cP_i\}$ for any input $\bm{a}\in \mathbb{R}^d$. 
Given $\{l_j(\cP)\}_{j\in[d]}$, suppose $c\cdot\pbox\subseteq \cP$ for some $\eps>0$ and $c\in (0,1]$. Let $b$ be an upper bound on the bit complexity of $\cQ_i(\bm{a})$ for all $i\in[\ell]$ and $\bm{a}\in \mathbb{R}^d$, as well as on the bit complexity of $l_j(\cP)$ for all $j\in [d]$.} 
Let the parameter {$k\geq \Omega\left(d^4\left(b + \log\left( \frac{1}{\eps}\right)\right)\right)$.} 
 We can construct a convex and down-monotone  polytope $\widehat{\cP}$ using $N = \left\lceil\frac{8kd}{\eps^2}\right\rceil$ samples from $\cD$ such that with probability at least {$1-2de^{-2dk}$}, the following guarantees hold:
\begin{enumerate}
\item $\frac{c}{6}\cdot\cP\subseteq \widehat{\cP}\subseteq \cP$. 
\item There exists a separation oracle $SO$ for $\widehat{\cP}$,
whose running time on input with bit complexity $b'$, is $\poly\left(b,b',k,d,\frac{1}{\eps}\right)$ and requires $\poly\left(b,b', k,d,\frac{1}{\eps}\right)$ queries to oracles in  $\{\cQ_i\}_{i\in[\ell]}$ with inputs of bit complexity at most $\poly\left(b, b',k,d,\frac{1}{\eps}\right)$.
\end{enumerate}


\end{theorem}

\notshow{
\begin{theorem}\label{thm:special case of multiplicative approx-main body}
Let $\ell$ be a positive integer. Let $\cP \subseteq [0,1]^d$ be a mixture of $\{\cP_i\}_{i\in[\ell]}$ over distribution $\cD=\{q_i\}_{i\in [\ell]}$,
where for each $i\in[\ell]$, $\cP_i\subseteq [0,1]^d$ is a convex and down-monotone polytope. 
Assume for all $i$,
the bit complexity of each corner of $\cP_i$ is at most $b$. Suppose for every $i\in [\ell]$, there is an algorithm $\mathcal{A}$ that can maximize any linear objective over $\cP_i$, i.e. solve $\max\{\mathbf{a}\cdot \bx:\bx\in \cP_i\}$ for any $\mathbf{a}\in \mathbb{R}^d$ with bit complexity $y$, in time $RT_{\mathcal{A}}(y)$.
Suppose $c\cdot\pbox\subseteq \cP$ for some $\eps>0$ and absolute constant $c\in (0,1]$, then for any $\delta\in (0,1)$ there is an algorithm (with running time $\poly(d,\frac{1}{\eps},\log(1/\delta))$) that with probability at least $1-\delta$, constructs a convex and down-monotone polytope $\widehat{\cP}\in [0,1]^{d}$ using $\poly(d,\frac{1}{\eps},\log(1/\delta))$ samples from $\cD$, such that
\begin{enumerate}
\item $\frac{c}{6}\cdot\cP\subseteq \widehat{\cP}\subseteq \cP$. 
\item There exists a separation oracle $SO$ for $\widehat{\cP}$. The running time of $SO$ on any input with bit complexity $b'$ is $\poly(d,b,b',\frac{1}{\eps},\log(1/\delta),RT_{\mathcal{A}}(\poly(d,b,b',\frac{1}{\eps},\log(1/\delta))))$.
\end{enumerate}
\end{theorem}
}


The complete proof of \Cref{thm:special case of multiplicative approx} is postponed to \Cref{subsec:proof_general_polytope}. Here we give a sketch of the proof. We first prove that if the polytope $\cP$ contains $c$ times the $\eps$-box polytope, then the convex set $\frac{c}{2}(\ptruncated + \pbox)$ is sandwiched between $\frac{c}{2}\cP$ and $\cP$ (\Cref{lem:boxable polytope contained} in \Cref{subsec:proof_general_polytope}). Next, we construct the polytope $\widehat{\cP}$ that is close to $\frac{c}{2}(\ptruncated + \pbox)$. For $\eps>0$ and every $i\in [\ell]$, let $\ptruncatedI$ be the $(\eps,\cP)$-truncated polytope of $\cP_i$. It is clear that $\ptruncated$ is a mixture of $\{\ptruncatedI\}_{i\in [\ell]}$ over distribution $\cD$. We  construct our polytope $\widehat{\cP}$ using $\ptruncatedEmp$, the mixture of $\{\ptruncatedI\}_{i\in [\ell]}$ over an empirical distribution $\widehat{\cD}$ of $\cD$. Cai et al.~\cite{CaiDW12b,CaiDW13a} proved that with polynomially many samples from $\cD$, $\ptruncatedEmp$ and $\ptruncated$ are close within an additive error $\eps$ in the $\ell_\infty$-norm (\Cref{thm:approx mixture} in \Cref{subsec:proof_general_polytope}). By choosing $\widehat{\cP}=\frac{c}{3}(\ptruncatedEmp + \pbox)$, we show that $\widehat{\cP}$ is a multiplicative approximation to $\cP$.

To apply \Cref{thm:special case of multiplicative approx} to the single-bidder marginal reduced form polytope $W_i$, we first show that $W_i$ is a mixture of a set of polytopes $\{W_{i,t_i}\}_{t_i\in \cT_i}$ over $D_i$, where each $W_{i,t_i}$ contains {``all feasible single-bidder marginal reduced forms'' for a specific type $t_i$ (\Cref{def:W_i-t_i} in \Cref{sec:mrf for constraint additive}).}
For every $t_i$, we can maximize any linear objective over $W_{i,t_i}$ via a query to the demand oracle. 
Finally, we prove that $W_i$ contains ($c$ times) the $\eps$-box polytope of itself, for some $c\in (0,1)$ and $\eps>0$.

\subsection{Proof of \Cref{thm:special case of multiplicative approx}}\label{subsec:proof_general_polytope}

In this section we give a proof of \Cref{thm:special case of multiplicative approx}. We first prove the following observation about the Minkowski addition of down-monotone polytopes.  

\begin{observation}
Let $\cP\subseteq [0,1]^d$ be any down-monotone polytope. Then for every $0\leq a\leq b$, $a\cdot \cP\subseteq b\cdot \cP$. Let $\cP_1,\cP_2\subseteq [0,1]^d$ both be down-monotone polytopes. Then for every $0\leq a_1'\leq a_1$ and $0\leq b_1'\leq b_1$, $a_1'\cP_1+b_1'\cP_2\subseteq a_1\cP_1+b_1\cP_2$.  
\end{observation}
\begin{proof}
For the first half of the statement, for every $\bx\in a\cdot \cP$, $\frac{\bx}{a}\in \cP$. Since $\cP$ is down-monotone, $\frac{\bx}{b}\in \cP$. Thus $\bx\in b\cdot\cP$. As $a'_1\cdot \cP\subseteq a_1\cdot \cP$ and $b'_1\cdot \cP\subseteq b_1\cdot \cP$, the second half of the statement follows from the definition of the Minkowski addition.
\end{proof}


\notshow{
\begin{theorem}\label{thm:special case of multiplicative approx}[(Concrete Version of \Cref{thm:special case of multiplicative approx-main body})]
Let $\ell$ be a positive integer, and $\cP \subseteq [0,1]^d$ be a mixture of $\{\cP_i\}_{i\in[\ell]}$ over distribution $\cD=\{q_i\}_{i\in [\ell]}$,
where for each $i\in[\ell]$, $\cP_i\subseteq [0,1]^d$ is a convex and down-monotone polytope. 
Suppose for every $i\in [\ell]$, 
\yangnote{there exists an oracle $\cQ_i(\cdot)$, whose output $\cQ_i(\mathbf{a})\in \argmax\{\mathbf{a}\cdot \bx:\bx\in \cP_i\}$ for any input $\mathbf{a}\in \mathbb{R}^d$. 
Given $\{l_j(\cP)\}_{j\in[d]}$, suppose $c\cdot\pbox\subseteq \cP$ for some $\eps>0$ and $c\in (0,1]$. Let $b$ be an upper bound on the bit complexity of $\cQ_i(a)$ for all $i\in[\ell]$ and $\mathbf{a}\in \mathbb{R}^d$, as well as on the bit complexity of $l_j(\cP)$ for all $j\in [d]$.} \footnote{\yangnote{Yang: I removed the assumption that the bit complexity of a corner is at most $b$. I think the bit complexity of the output of $\cQ_i$ implies that the corners all have bit complexity at most $b$. We also need the bit complexity of $l_j(\cP)$.}}
Let the parameter \yangnote{$k\geq \Omega\left(d^4\left(b + \log\left( \frac{1}{\eps}\right)\right)\right)$.} 
 We can construct a convex and down-monotone  polytope $\widehat{\cP}$ using $N = \left\lceil\frac{8kd}{\eps^2}\right\rceil$ samples from $\cD$ such that with probability at least {$1-2de^{-2dk}$}, the following guarantees hold:
\begin{enumerate}
\item $\frac{c}{6}\cdot\cP\subseteq \widehat{\cP}\subseteq \cP$. 
\item There exists a separation oracle $SO$ for $\widehat{\cP}$,
whose running time on input with bit complexity $b'$, is $\poly\left(b,b',k,d,\frac{1}{\eps}\right)$ and requires $\poly\left(b,b', k,d,\frac{1}{\eps}\right)$ queries to oracles in  $\{\cQ_i\}_{i\in[\ell]}$ with inputs of bit complexity at most $\poly\left(b, b',\yangnote{k},d,\frac{1}{\eps}\right)$.
\end{enumerate}


\end{theorem}

}



We use the following result from an unpublished manuscript by Cai et al.~\cite{CaiDW21}. A special case of the result appeared as Theorem 4 in~\cite{CaiDW12bfull} (conference version by the same authors). In particular, the result we use here is stated for a mixture of polytopes, while Theorem 4 in~\cite{CaiDW12bfull} is only for the polytope of all feasible reduced forms, but the proof is essentially the same. Interested readers are welcome to email the first author for a proof of~\Cref{thm:approx mixture}.

\begin{theorem}[\cite{CaiDW21}]\label{thm:approx mixture}
Let $\ell$ be a positive integer. Let $\cP$ be a mixture of $\{\cP_i\}_{i\in[\ell]}$ over distribution $\cD=\{q_i\}_{i\in [\ell]}$,
where $\cP_i\subseteq \mathbb{R}^d$ is a convex polytope for every $i\in[\ell]$. Assume for all $i$,
the bit complexity of each corner of $\cP_i$ is at most $b$. For any $\eps > 0$ and integer 
{$k \geq \ \Omega\left(d^4\left(b + \log\left( \frac{1}{\eps}\right)\right)\right)$},
let $\cD'$ be the empirical distribution induced by $\lceil\frac{8kd}{\varepsilon^2}\rceil$ samples from $\cD$.
Let $\cP$ be the mixture of $\{\cP_i\}_{i\in[\ell]}$ over distribution $\cD'$. With probability at least $1-2de^{-2dk}$ we have that
\begin{enumerate}
    \item For all $\bx\in \cP$,
    there exists an $\bx'\in \cP'$ such that $||\bx-\bx'||_\infty \leq \eps$.
    \item For all $\bx'\in \cP'$,
    there exists an $\bx\in \cP$ such that $||\bx-\bx'||_\infty \leq \eps$.
\end{enumerate}
\end{theorem}

To prove \Cref{thm:special case of multiplicative approx}, we need the following lemmas.

\begin{lemma}\label{lem:boxable polytope contained}
Let $\cP\subseteq [0,1]^d$ be a convex and down-monotone polytope. If $c\cdot \pbox\subseteq \cP$ for some $\eps>0$ and $c\in(0,1]$, 
then $\frac{c}{2}\cP \subseteq \frac{c}{2}\ptruncated + \frac{c}{2}\pbox \subseteq \cP$.
\end{lemma}

\begin{proof}
First we prove that $\frac{c}{2}\cP \subseteq \frac{c}{2}\ptruncated + \frac{c}{2}\pbox$.
Note that it is enough to prove that $\cP \subseteq \ptruncated + \pbox$.
For any $\bx\in \cP$,
we consider the vectors $\bx',\bx''\in [0,1]^d$ such that
$$\bx_j'=\bx_j\cdot \ind[l_j(\cP)\geq \eps],\forall j\in [d]$$
$$\bx_j''=\bx_j\cdot \ind[l_j(\cP)<\eps],\forall j\in [d]$$
Note that $\bx = \bx' + \bx''$.
By the definition of $\ptruncated$, $\bx' \in \ptruncated$. For $\bx''$, we notice that for every $j\in [d]$, $\bx_j''=\bx_j\cdot \ind[l_j(\cP)<\eps]\leq l_j(\cP)\cdot \ind[l_j(\cP)<\eps]$. By the definition of $\pbox$,
$\bx''\in \pbox$. Thus $\bx = \bx' + \bx''\in \ptruncated +\pbox$.

For the other direction, note that $\ptruncated\subseteq \cP$ by \Cref{def:truncated polytope} and $c \cdot \pbox\subseteq \cP$ by assumption, so $\frac{c}{2}\pbox + \frac{c}{2}\ptruncated \subseteq \frac{1}{2}\cP + \frac{1}{2}\cP= \cP$.
\end{proof}

\begin{lemma}\label{lem:mixture of truncated distribution}
Let $\ell$ be a positive integer and $\cP \subseteq [0,1]^d$ be a mixture of $\{\cP_i\}_{i\in[\ell]}$ over distribution $\cD=\{q_i\}_{i\in [\ell]}$,
where for each $i\in[\ell]$, $\cP_i\subseteq [0,1]^d$ is a convex and down-monotone polytope.
Then $\ptruncated$ is a mixture of $\{\ptruncatedI\}$ over $\cD$,
where for each $i\in [\ell]$, $\ptruncatedI\subseteq \cP_i$ is the $(\eps,\cP)$-truncated polytope of $\cP_i$ (\Cref{def:truncated polytope}).
\end{lemma}

\begin{proof}
To prove our statement, we first show that for each $\widehat{\bx}\in \ptruncated$,
there exist $\{\widehat{\bx}^{(i)}\in \ptruncatedI\}_{i\in[\ell]}$ such that $\widehat{\bx}= \sum_{i\in[\ell]}q_i \widehat{\bx}^{(i)}$.
By definition of $\ptruncated$,
there exists $\bx \in \cP$ such that for each $j\in[d]$, $\widehat{\bx}_j = \bx_j\cdot \ind[l_j(\cP)\geq \eps]$.

Since $\bx \in \cP$ and $\cP$ is a mixture of $\{\cP_i\}_{i\in [\ell]}$ over $\cD$, there exist $\{\bx^{(i)} \in \cP_i\}_{i\in[\ell]}$ such that $\bx = \sum_{i\in[\ell]}q_i \bx^{(i)}$.
For each $i\in [\ell]$,
consider $\widehat{\bx}^{(i)}$ be defined such that for all $j\in[d]$:
$$\widehat{\bx}_j^{(i)} = \bx_j^{(i)}\ind[l_j(\cP)\geq \eps]
$$
Clearly, $\widehat{\bx}^{(i)}\in \ptruncatedI$ and $\widehat{\bx}=\sum_{i\in[\ell]}q_i \widehat{\bx}^{(i)}$. Similarly, we can argue that  any point $\widehat{\bx}$ that lies in the mixture of $\{\ptruncatedI\}$ over $\cD$ must also lie in $\ptruncated$, which concludes the proof.
\end{proof}

\begin{lemma}\label{lem:mixture of lower-closed}
Let $\ell$ be a positive integer, and $\cP\subseteq [0,1]^d$ be a mixture of $\{\cP_i\}_{i\in[\ell]}$ over distribution $\cD=\{q_i\}_{i\in[\ell]}$,
where $\cP_i$ is a convex and down-monotone polytope in $[0,1]^d$ for every $i$. Then $\cP$ is a convex and down-monotone polytope.
\end{lemma}

\begin{proof}
For every $\bx\in \cP$, there exists a set of vectors $\{\bx^{(i)}\}_{i\in [\ell]}$ such that $\bx^{(i)}\in \cP_i,\forall i$, and $\bx = \sum_{i\in[\ell]}q_i \cdot \bx^{(i)}$. Now consider any $\hat{\bx}$ such that $\mathbf{0}\leq \hat{\bx}\leq \bx$.
For each $i \in [\ell]$, let $\hat{\bx}^{(i)}\in [0,1]^d$ be the vector such that $\hat{\bx}^{(i)}_j=\bx^{(i)}_j\cdot \hat{\bx}_j/\bx_j,\forall j\in [d]$.
Clearly, $0\leq \hat{x}^{(i)}_j\leq x^{(i)}_j$ for all $j\in[d]$. 
Since $\cP_i$ is down-monotone, we have $\hat{\bx}^{(i)}\in \cP_i$. Note that for every $j\in [d]$, $\sum_{i}q_i\cdot\hat{\bx}^{(i)}_j=(\sum_{i\in[\ell]} q_i\bx^{(i)}_j)\cdot \hat{\bx}_j/\bx_j=\hat{\bx}_j$. Thus $\hat{\bx}=\sum_{i\in[\ell]} q_i\cdot\hat{\bx}^{(i)}\in \cP$.
\end{proof}

\notshow{
\begin{prevproof}{Theorem}{thm:special case of multiplicative approx}
We condition on the probability \todo{say probability}.
We are only going to prove that for each $p\in \frac{1}{2}\cP^t_{\eps} +\frac{1}{2}\cP^b_{\eps}$,
then there exists a $\widetilde{p}\in \frac{1}{2}\widehat{\cP}^t_{\eps} +\frac{1}{2}\cP^b_{\eps}$ such that:
$$\frac{p(j)}{\widetilde{p}(j)}\in \left[1-\delta,1+\delta\right]$$
since the other case is similar.
Any $p\in \frac{1}{2}\cP^b_{\eps} +\frac{1}{2}\cP^t_{\eps}$,
can be decomposed into $p = \frac{1}{2}p^b+ \frac{1}{2}p^t$ such that $p^b \in \cP^b_{\eps}$ and $p^t \in \cP^t_{\eps}$.
By Theorem~\ref{thm:approx mixture} we know that there exists a $\widehat{p}\in \widehat{\cP}_{\eps}^t$ such that for each $j\in[d]$:
\begin{align}
|\widehat{p}^t(j) - p^t(j)| \leq \delta \eps\label{eq:additive error}
\end{align}
We partition the coordinates of $p$ into four disjoint sets.
\begin{align*}
    S^1 =& \{j\in[d]: l_j(\cP) \geq \eps \land p(j) \geq \eps\} \\
    S^2 =& \{j\in[d]:l_j(\cP) \geq \eps \land  p(j) < \eps \land p^t(j) \leq \widehat{p}^t(j)\} \\
    S^3 =& \{j\in[d]:l_j(\cP) \geq \eps \land  p(j) < \eps \land p^t(j) > \widehat{p}^t(j)\}\\
    S^4 =& \{j\in[d]:l_j(\cP) < \eps\}\\
\end{align*}
We consider the point $\widetilde{p}^t$ such that:
$$
\widetilde{p}^t(j) = 
\begin{cases}
\widehat{p}^t(j) \quad & \text{if $j\in S^1$}\\
{p}^t(j) \quad & \text{if $j\in S^2$}\\
\widehat{p}^t(j) & \text{if $j\in S^3$} \\
0 & \text{if $j\in S^4$} 
\end{cases}
$$
Since for each $i\in[L]$, $\cP_i$ is a lower closed polytpe,
then also $\cP_{\eps,i}^t$ is a lowered closed polytope.
Thus by Lemma~\ref{lem:mixture of lower-closed} 
$\widehat{\cP}_{\eps}^t$ is a lowered-closed polytope.
Notice that for each $j\in[d]$, $\widetilde{p}^t(j) \leq \widehat{p}^t(j)$,
which implies that $\widetilde{p}^t\in \widehat{\cP}_{\eps}^t$.
We also consider the point $\widetilde{p}^b\in [0,1]^d$:
$$
\widetilde{p}^b(j) = 
\begin{cases}
p^b(j) \quad & \text{if $j\in S^1$}\\
p^b(j) \quad & \text{if $j\in S^2$}\\
\min(\eps, p^b(j) + p^t(j) - \widehat{p}^t(j)) & \text{if $j\in S^3$} \\
p^b(j) & \text{if $j\in S^4$}
\end{cases}
$$
Observe that for $j\in S^1 \cup S^2 \cup S^4$,
the elements of $\widetilde{p}^b$ are the same as $p^b\in \cP_{\eps}^b$.
Notice that for $j\in S^2$,
$h_j(\cP)\geq \eps$ and $\widetilde{p}^b(j) \leq \eps$,
which implies that $\widetilde{p}^b \in \cP_{\eps}^b$.

We consider the element $\widetilde{p}=\frac{1}{2}\widetilde{p}^b + \frac{1}{2}\widetilde{p}^t $Now we are going to show that $\widetilde{p} \in \frac{1}{2}\widehat{\cP}_{\eps}^b  +\frac{1}{2}\cP_{\eps}^t$,
we have that for each $j\in [d]$:
$$\frac{p(j)}{\widetilde{p}(j)}\in \left[1-\delta,1+\delta\right]$$
For $j\in S^1$,
note that by Equation~\eqref{eq:additive error} and that for $j\in S^1$,then $p(j) \geq \eps$
\begin{align*}
&-\frac{\delta}{2}p(j)\leq -\frac{\delta \eps}{2} \leq \widetilde{p}(j) - p(j) \leq \frac{\delta \eps}{2} \leq \frac{\delta}{2}p(j) \\
\Rightarrow & \frac{p(j)}{\widetilde{p}(j)}\in \left[1-\frac{\delta}{2},1+\frac{\delta}{2} \right]
\end{align*}

For $j\in S^2$,
observe that 
$$\widetilde{p}(j) = \frac{1}{2}\widetilde{p}^t(j)+ \frac{1}{2}\widetilde{p}^b(j) =
\frac{1}{2}{p}^t(j)+ \frac{1}{2}{p}^b(j) = p(j)$$

For $j\in S^3$,
if $\eps\geq p^b(j) + p^t(j) - \widehat{p}^t(j)$,
then 
\begin{align*}
\widetilde{p}(j) = & \frac{1}{2}\widetilde{p}^t(j)+ \frac{1}{2}\widetilde{p}^b(j)  \\
=& \frac{\widehat{p}^t(j)}{2}+ \frac{p^b(j) + p^t(j) - \widehat{p}^t(j)}{2}\\
=&\frac{1}{2}{p}^t(j)+ \frac{1}{2}{p}^b(j) = p(j)
\end{align*}

if $\eps< p^b(j) + p^t(j) - \widehat{p}^t(j)$,
then
\begin{align*}
\widetilde{p}(j) = & \frac{1}{2}\widetilde{p}^t(j)+ \frac{1}{2}\widetilde{p}^b(j)  \\
=& \frac{\widehat{p}^t(j)}{2}+ \frac{\eps}{2}\\
<& \frac{\widehat{p}^t(j)}{2}+ \frac{p^b(j) + p^t(j) - \widehat{p}^t(j)}{2}\\
=&\frac{1}{2}{p}^t(j)+ \frac{1}{2}{p}^b(j) = p(j)
\end{align*}
Observe that by combining the fact that for each $j\in [d]$ since $p^b\in \cP_{\eps}^b$ then $p^b(j)\leq \eps$ and Equation~\eqref{eq:additive error}, we get that
$$
p^b(j) + p^t(j) - \widehat{p}^t(j)\leq \eps+\delta \eps=(1+\delta)\eps
$$
thus
\begin{align*}
\widetilde{p}(j) = & \frac{1}{2}\widetilde{p}^t(j)+ \frac{1}{2}\widetilde{p}^b(j)  \\
=& \frac{\widehat{p}^t(j)}{2}+ \frac{\eps}{2}\\
\geq& \frac{\widehat{p}^t(j)}{2(1+\delta)}+ \frac{p^b(j) + p^t(j) - \widehat{p}^t(j)}{2(1+\delta)}\\
=&\frac{1}{1+\delta}\left(\frac{1}{2}{p}^t(j)+ \frac{1}{2}{p}^b(j)\right)\\
=&  \frac{1}{1+\delta}p(j)
\end{align*}
which implies that 
$$\frac{p(j)}{\widetilde{p}(j)}\in \left[1,1+\delta\right]$$
For $j\in S^4$,
since $h_j(\cP) < \eps$,
note that $p^t(j) = \widetilde{p}^t(j) = 0$,
which implies that 
$$\widetilde{p}(j) = \frac{1}{2}\widetilde{p}^b(j) = \frac{1}{2}{p}^b(j)= p(j)$$

Now we are going to prove that $\frac{(1-\delta)^2}{2}\cP \subseteq \widehat{\cP}\subseteq \cP$,
where $\widehat{\cP} = (1-\delta)\left( \frac{1}{2}\widehat{\cP}_{\eps}^t  +\frac{1}{2}\cP_{\eps}^b\right)$.
For each $\widetilde{p}\in\widehat{\cP}=(1-\delta)\left(\frac{1}{2}\widehat{\cP}_{\eps}^t  +\frac{1}{2}\cP_{\eps}^b\right)$,
we can decompose $\widetilde{p}=\frac{1-\delta}{2}\widehat{p}^t + \frac{1-\delta}{2}\widehat{p}^b$
such that $\widehat{p}^t\in \widehat{\cP}_{\eps}^t$ and $\widehat{p}^b\in \cP^b$.
By what we proved before,
we know that there exists a $p^t\in{\cP}_{\eps}^t$ such that  for each $j\in[d]$
$$
\frac{p^t(j)}{\widehat{p}^t(j)} \in \left[1-\delta,1+\delta \right]
$$
Consider the point $p = \frac{1}{2}p^t + \frac{1}{2}p^b \in \left(\frac{1}{2}{\cP}_{\eps}^t  +\frac{1}{2}\cP_{\eps}^b\right)$.
For for each $j\in[d]$
$$
\widetilde{p}(j)= \frac{1-\delta}{2}\widehat{p}^t(j) + \frac{1-\delta}{2}\widehat{p}^b(j) \leq \frac{1}{2}p(j)+ \frac{1}{2}\widehat{p}^b(j) = p(j)
$$
Since $\cP_\eps^b\subseteq \cP$ and $\cP$ is a lowered-closed polytope,
by Lemma~\ref{lem:boxable polytope contained}, we know that $p \in \frac{1}{2}{\cP}_{\eps}^t  +\frac{1}{2}\cP_{\eps}^b \subseteq \cP$
Since $\cP$ is a lowered closed polytope, $p \in \cP$ and for each $j\in[d]$ $\widetilde{p}(j) \leq p(j)$,
then $\widetilde{p}\in \cP$,which implies that $\widehat{\cP} \subseteq \cP$.

For any ${p} \in (1-\delta)^2\left( \frac{1}{2}{\cP}_{\eps}^t  +\frac{1}{2}\cP_{\eps}^b\right)$,
we can decompose ${p}=\frac{(1-\delta)^2}{2}{p}^t + \frac{(1-\delta)^2}{2}{p}^b$
such that ${p}^t\in {\cP}_{\eps}^t$ and ${p}^b\in \cP^b$.
By what we have proved before,
we know that there exists a $\widehat{p}^t$ such that for each $j\in[d]$:
$$
\frac{p^t(j)}{\widehat{p}^t(j)} \in \left[1-\delta,1+\delta \right]
$$
Consider the point $\widehat{p} = \frac{1-\delta}{2}\widehat{p}^t + \frac{1-\delta}{2}p^b \in \left(\frac{1-\delta}{2}{\cP}_{\eps}^t  +\frac{1-\delta}{2}\cP_{\eps}^b\right)$.
For each $j\in[d]$ note that
$$p(j) = \frac{(1-\delta)^2}{2}{p}^t(j) + \frac{(1-\delta)^2}{2}{p}^b(j) \leq \frac{1-\delta}{2}\widehat{p}^t(j) + \frac{1-\delta}{2}{p}^b(j)=\widehat{p}(j)$$

We know that $\cP_\eps^b$ is lowered-closed and by Lemma~\ref{lem:mixture of lower-closed} we know that $\widehat{\cP}_\eps^t$ is lowered-closed.
Thus $\widehat{\cP}$ is a lowered-closed polytope,
as the Minkowski sum of lowered-closed polytopes $\frac{(1-\delta)^2}{2}\widehat{\cP}_\eps^t$ and $\frac{(1-\delta)^2}{2}\cP_\eps^b$.
Since for each $p \in (1-\delta)^2\left(\frac{1}{2}\cP_\eps^t +\frac{1}{2}\cP_\eps^b\right)$ and for each $j\in[d]$,
there exists a $\widehat{p}\in \widehat{\cP}$ such that for each $j\in[d]$,
it holds that $p(j) \leq \widehat{p}(j)$, implies that since $\cP$ is lowered-closed  $p \in \widehat{\cP}$.
Thus $(1-\delta)^2\left(\frac{1}{2}\cP_\eps^t +\frac{1}{2}\cP_\eps^b\right)\subseteq \widehat{\cP}$.

By Lemma~\ref{lem:boxable polytope contained},
we know that $\frac{(1-\delta)^2}{2}\cP\subseteq (1-\delta)^2\left( \frac{1}{2}{\cP}_{\eps}^t  +\frac{1}{2}\cP_{\eps}^b\right) \subseteq \widehat{\cP}$.

\end{prevproof}

}


To prove \Cref{thm:special case of multiplicative approx},
we will also need the celebrated result of the equivalence between optimization and separation.

\begin{theorem}[\cite{KarpP80,GrotschelLS81}]\label{thm:equivalence of opt and sep}
Let $\cP\subseteq \mathbb{R}^d$ be a convex polytope and suppose we have access to an algorithm $\mathcal{A}(\bm{a}):\mathbb{R}^d\rightarrow\cP$, that takes input vector $\bm{a}\in \mathbb{R}^d$,
outputs a vector $\bx^* \in \cP$ with bit complexity at most $b$,
such that $\bx^* \in \text{argmax}\{\bm{a}\cdot \bx:\bx\in \cP\}$.
Then we can construct a separation oracle $SO$ for $\cP$, where on {any input $\bm{a}\in \mathbb{R}^d$ with bit complexity at most $b'$, $SO$ makes at most $\poly(d,b,b')$ queries to $\mathcal{A}$, and the input of each query has bit complexity no more than $\poly(d,b,b')$.} Moreover, the running time of $SO$ on $\bm{a}$ is at most  $\poly(d,b,b',RT_{\mathcal{A}}(\poly(d,b,b')))$. Here $RT_{\mathcal{A}}(c)$ is the running time of $\mathcal{A}$ with input whose bit complexity is at most $c$. 
\end{theorem}

\begin{prevproof}{Theorem}{thm:special case of multiplicative approx}

Consider the polytopes $\ptruncated$ and $\pbox$.
By Lemma~\ref{lem:mixture of truncated distribution}, $\ptruncated$ is a mixture of $\{\ptruncatedI\}_{i\in[\ell]}$ over distribution $\cD$. 
Let $\widehat{\cD}$ be the empirical distribution induced by $N=\lceil\frac{8kd}{\eps^2}\rceil$ samples from $\cD$.  
Let $\ptruncatedEmp$ be the mixture of 
$\{\ptruncatedI\}_{i\in[\ell]}$ over $\widehat{\cD}$.
By Theorem~\ref{thm:approx mixture},
we have that with probability at least {$1-2de^{-2dk}$}, both of the two following conditions hold:
\begin{enumerate}
    \item For each $\widehat{p}\in \ptruncatedEmp$,
    there exists a $p\in \ptruncated$ such that $||p-\widehat{p}||_\infty \leq \eps$.
    \item For each ${p}\in {\ptruncated}$,
    there exists a $\widehat{p}\in \ptruncatedEmp$ such that $||p-\widehat{p}||_\infty \leq \eps$.
\end{enumerate}
For the rest of the proof, we condition on the event that both conditions hold.
We consider the polytope $\widehat{\cP}=\frac{c}{3}\left(\ptruncatedEmp+\pbox\right)$.
First we are going to prove that $\widehat{\cP}\subseteq \cP$.
By condition 1, we have that for any $\widehat{p}^{tr}\in \ptruncatedEmp$,
there exists a ${p}^{tr}\in {\ptruncated}$ such that $||p^{tr}-\widehat{p}^{tr}||_{\infty}\leq \eps$.
Consider the vector $\tilde{p}^{tr}$ defined such that for each $j\in[d]$,
$$
\tilde{p}^{tr}_j = \min\left( \widehat{p}^{tr}_j, p^{tr}_j\right).
$$
Since for each $j\in[d]$, $\tilde{p}^{tr}_j \leq {p}^{tr}_j$ and $\ptruncated$ is down-monotone (\Cref{lem:mixture of lower-closed}),
we have $\tilde{p}^{tr}\in \ptruncated$.
Let vector $\tilde{p}^{box}$ be such that for every $j\in[d]$,
$$
\tilde{p}^{box}_j = \widehat{p}^{tr}_j - \tilde{p}^{tr}_j=\widehat{p}^{tr}_j -\min(\widehat{p}^{tr}_j, {p}^{tr}_j).
$$
Notice that for every $\bx\in \ptruncatedI$, $x_j=0$ for all $j$ such that $l_j(\cP)<\eps$. Since $\ptruncated$ and $\ptruncatedEmp$ are both mixtures of $\{\ptruncatedI\}_{i\in[\ell]}$, we have for every $\bx\in \ptruncated$ and $\widehat{\bx}\in \ptruncatedEmp$, $x_j=\widehat{x}_j=0$ for all $j$ such that $l_j(\cP)<\eps$. Therefore, we have $\tilde{p}^{box}_j\leq \eps\cdot \ind[l_j(\cP)\geq \eps]\leq \min(\eps,l_j(\cP))$, for every $j\in [d]$. The first inequality follows from the fact that $||p^{tr}-\widehat{p}^{tr}||_{\infty}\leq \eps$, 
and that if $l_j(\cP) < \eps$,
then $\widehat{p}^{tr}_j=p^{tr}_j=0$.
Thus $\tilde{p}^{box}\in \pbox$.


For every $\widehat{p}^{tr}\in \ptruncatedEmp$, we have found $\tilde{p}^{tr}\in \ptruncated$ and $\tilde{p}^{box}\in \pbox$ such that $\widehat{p}^{tr}=\tilde{p}^{tr}+\tilde{p}^{box}$. Thus

\begin{equation}\label{equ:proof-hatp}
\begin{aligned}
&\ptruncatedEmp\subseteq \ptruncated + \pbox \\
\Rightarrow & \ptruncatedEmp+\pbox \subseteq \ptruncated + 2\pbox\subseteq \frac{3}{c}\cdot\cP \\
\Rightarrow & \widehat{\cP}=\frac{c}{3}\left(\ptruncatedEmp+\pbox\right)\subseteq \cP
\end{aligned}
\end{equation}

The second line follows from the assumption that $c\pbox \subseteq \cP$ and $\ptruncated\subseteq \cP$ (by \Cref{def:truncated polytope} and the fact that $\cP$ is down-monotone), and $c\leq 1$. 
Similarly, by switching the role of $\ptruncated$ and $\ptruncatedEmp$, with condition 2, we also have ${\ptruncated}\subseteq \ptruncatedEmp + \pbox$. Thus
\begin{align*}
&{\ptruncated} + \pbox \subseteq \ptruncatedEmp + 2\pbox \\ 
\Rightarrow& \cP \subseteq {\ptruncated} + \pbox \subseteq 2\left(\ptruncatedEmp + \pbox \right)=\frac{6}{c}\widehat{\cP}\\
\Rightarrow & \frac{c}{6}\cP \subseteq \widehat{\cP}
\end{align*}
The second line follows from $\cP \subseteq {\ptruncated} + \pbox$ (Lemma~\ref{lem:boxable polytope contained}), {the origin $\mathbf{0}\in \ptruncatedEmp$}, and the definition of $\widehat{\cP}$. Thus $\frac{c}{6}\cP \subseteq \widehat{\cP}\subseteq \cP$.

To construct a separation oracle for $\widehat{\cP}$, it is sufficient to optimize any linear objective over $\widehat{\cP}$. For every $\bm{a}\in \mathbb{R}^d$, we are going to solve the maximization problem $\max\{\bm{a}\cdot \bx:~\bx\in \widehat{\cP}\}$.

Let $\{i_1,...,i_N\}$
be the $N$ samples from $\cD$, where $i_k\in [\ell]$ for $k\in [N]$.
We notice that
$$
\widehat{\cP} = \sum_{k\in[N]}\frac{c}{3N}\cdot\cP^{tr(\eps,\cP)}_{i_k} + \frac{c}{3}\pbox 
$$
is the Minkowski addition of a set of polytopes. Thus in order to maximize over $\widehat{\cP}$, it's sufficient to maximize over each polytope. In other words, it is sufficient to solve
$\max\{\bm{a}\cdot \bx:\bx\in \pbox\}$ and 
$\max\{\bm{a}\cdot \bx:\bx\in \cP^{tr(\eps,\cP)}_{i_k}\}$ for each $k\in [N]$.
First consider $\pbox$. Since the polytope is a ``box'' where the constraint for each coordinate $j$ is separate: $x_j\leq \min(\eps,l_j(\cP))$.  Thus the optimum $\bx^{box}\in \pbox$ satisfies that $x^{box}_j =\min(l_j(\cP),\eps)\cdot \ind[a_j>0]$. 
Thus the optimum $\bx^{box}$ can be computed in time {$O(d\cdot(b+\log 1/\eps+b'))$} 
and its bit complexity is at most {$O(d\cdot(b+\log 1/\eps))$, where $b'$ is the bit complexity of $\bm{a}$.} 

Now we show how to 
solve $\max\{\bm{a}\cdot \bx:\bx\in \cP_{i_k}^{tr(\eps,\cP)}\}$
using a single query to $\cQ_{i_k}(\cdot)$, for every $k\in [N]$.
Consider the vector $\bm{a}'\in \mathbb{R}^d$ such that $a_j'=a_j\cdot \ind[l_j(\cP)\geq \eps], \forall j\in [d]$.
Then clearly
$$\max\{\bm{a}\cdot \bx:\bx\in \cP_{i_k}^{tr(\eps,\cP)}\} = \max\{\bm{a}'\cdot \bx:\bx\in \cP_{i_k}\}$$
Let $\widehat{\bx}^{(i_k)}$ be the output from oracle $\cQ_{i_k}(\bm{a}')$,
then $\widehat{\bx}^{(i_k)}\in \text{argmax}\{\bm{a}'\cdot \bx:\bx\in \cP_{i_k}\}$.
Consider the element $\bx^{(i_k)}\in [0,1]^d$ such that for each $j\in [d]$, $x^{(i_k)}_j=\widehat{x}^{(i_k)}_j\cdot\ind[l_j(\cP)\geq \eps]$.
Then $\bx^{(i_k)}\in \text{argmax}\{\bm{a}'\cdot \bx:\bx\in \cP^{tr(\eps,\cP)}_{i_k}\}$ and its bit complexity is at most the bit complexity of $\widehat{\bx}^{(i_k)}$, which is at most $b$, {by our assumption on $\cQ_i(\cdot)$}.



Thus we have
$$\sum_{k\in[N]}\frac{c}{3N}\cdot \bx^{(i_k)} + \frac{c}{3}\cdot \bx^{box}\in\text{argmax}\{\bm{a}\cdot \bx:\bx\in \widehat{\cP}\}$$

{To sum up, we provide an algorithm to optimize any linear objective over $\widehat{\cP}$. Moreover, the output of our optimization algorithm always has bit complexity $\poly(b, k,d, 1/\eps)$. For any $\bm{a}\in \mathbb{R}^d$ with bit complexity $b'$, our optimization algorithm runs in time $\poly(b, b', k,d, 1/\eps)$ and make $N=\left\lceil \frac{8kd}{\eps^2}\right\rceil$ queries to the oracles in $\{\cQ_i\}_{i\in[\ell]}$. Using Theorem~\ref{thm:equivalence of opt and sep}, we can construct a separation oracle for $\widehat{\cP}$ that satisfies the properties in the statement of \Cref{thm:special case of multiplicative approx}.}

\end{prevproof}



\subsection{Approximating Single-Bidder Marginal Reduced Form Polytope for Constrained-Additive Buyers}\label{sec:mrf for constraint additive}

In this section,
we prove \Cref{thm:multiplicative approx for constraint additive-main body} using \Cref{thm:special case of multiplicative approx}.
The goal is to show that the single-bidder marginal reduced form polytope $W_i$ satisfies the requirements of Theorem~\ref{thm:special case of multiplicative approx}.
Recall that for every buyer $i$, the support of her type is $\cT_i$ and the support of her value for each item $j$ is $\cT_{ij}$. Additionally, $W_i$ is a subset of $[0,1]^{\sum_{j\in[m]}|\cT_{ij}|}$, and there is a coordinate for every $j\in [m]$ and every $t_{ij}\in \cT_{ij}$. To ease the notation, we will use $t_{ij}$'s to index the coordinates throughout this section. Since each buyer is constrained-additive, we denote $\cF_i$ the feasibility constraint of buyer $i$, and drop the subscript if the buyer is fixed or clear from context.

\begin{lemma}\label{obs:constraint additive length}
For each $i\in[n],j\in[m]$ and $t_{ij}\in\cT_{ij}$, $l_{t_{ij}}(W_i)=f_{ij}(t_{ij})$. Recall that $l_{t_{ij}}(W_i)$ is the width of $W_i$ at coordinate $t_{ij}$ (\Cref{def:truncated polytope}).
\end{lemma}

\begin{proof}
For every $i\in [n]$ and every $\widehat{w}_i\in W_i$, by \Cref{def:W_i-constrained-add}, there exists a number $\sigma_S(t_i)\in [0,1]$ for every $t_i\in \cT_i, S\in \cF_i$ such that
\begin{enumerate}
    \item $\sum_{S\in \cF_i}\sigma_S(t_i)\leq 1$, $\forall t_i\in \cT_i$.
    \item $\widehat{w}_{ij}(t_{ij})=f_{ij}(t_{ij})\cdot\sum_{t_{i,-j}}f_{i,-j}(t_{i,-j})\cdot \sum_{S\in \cF_i:j\in S}\sigma_S(t_i)$, for all $j\in[m]$ and $t_{ij}\in \cT_{ij}$.
\end{enumerate}

Thus for every $j,t_{ij}$, by combining both properties above, we have $$\widehat{w}_{ij}(t_{ij})\leq f_{ij}(t_{ij})\cdot\sum_{t_{i,-j}}f_{i,-j}(t_{i,-j})\cdot 1=f_{ij}(t_{ij}).$$

Moreover, for every $j,t_{ij}$, choosing $\widehat{\sigma}$ such that:~\footnote{$\widehat{\sigma}$ is simply the allocation that gives buyer $i$ item $j$ when her value for item $j$ is $t_{ij}$ and does not give her anything otherwise.}
$$
\widehat{\sigma}_S(t_i') =
\begin{cases}
1 \quad & \text{if $t_{ij}'=t_{ij} \land S=\{j\}$} \\
0 & \text{o.w.}
\end{cases}
$$
induces an element $\widetilde{w}_i\in W_i$ such that 
$\widetilde{w}_{ij}(t_{ij})=f_{ij}(t_{ij})$.
Thus $l_{t_{ij}}(W_i)=f_{ij}(t_{ij})$.


\notshow{
$$l_{t_{ij}}(W_i) \leq \max_{\widehat{\sigma}_i\in \widehat{\Sigma}_i} f_{ij}(t_{ij})\cdot\sum_{t_{i,-j}\in \cT_{i,-j}} f_{i,-j}(t_{i,-j})\sum_{S:j\in S}\hat\sigma_{i}((t_{ij},t_{i,-j}),S)\leq f_{ij}(t_{ij})$$
where $\Sigma_i$ is the set such that $\sigma \in \Sigma_i$ if
\begin{align*}
    \sum_{S\subseteq 2^{[m]}} \sigma_i(t_i,S) \leq 1 \quad \forall t_i \in \cT_i \\
    \sigma_i(t_i,S) \geq 0 \quad \forall t_i\in \cT_i, S \subseteq 2^{[m]}
\end{align*}

Moreover observe that when we use the allocations $\widehat{\sigma}_{i}$ such that:
$$
\widehat{\sigma}_{i}(t_i',S) =
\begin{cases}
1 \quad & \text{if $t_{ij}'=t_{ij} \land S=\{j\}$} \\
0 & \text{o.w.}
\end{cases}
$$
the element $\widetilde{w}\in W_i$ that corresponds to this allocation has $\widetilde{w}_{ij}(t_{ij})=f_{ij}(t_{ij})$.
Thus $l_{t_{ij}}(W_i)=f_{ij}(t_{ij})$.
}

\end{proof}

\begin{definition}\label{def:W_i-t_i}
For any buyer $i$ and type $t_i\in \cT_i$, consider $W_{t_i}\subseteq [0,1]^{\sum_{j\in[m]}|\cT_{ij}|}$ defined as follows: $w_i \in W_{t_i}$ if and only if there exists a collection of non-negative numbers $\{\sigma_{S}\}_{S\in \cF_i}$ such that
\begin{enumerate}
    \item $\sum_{S\in \cF_i}\sigma_S\leq 1$.
    \item $w_{ij}(t_{ij}')=\sum_{S\in \cF_i, j\in S}\sigma_S\cdot \ind[t_{ij}'=t_{ij}],\forall j\in [m],t_{ij}'\in \cT_{ij}$.
\end{enumerate}
\end{definition}

The following observation directly follows from \Cref{def:W_i-constrained-add} and \Cref{def:W_i-t_i}. 

\begin{observation}\label{obs:mix constraint additive}
$W_{i}$ is mixture of $\{W_{t_i}\}_{t_i\in \cT_i}$ over distribution $D_i$. Recall that $D_i$ is the distribution for buyer $i$'s type $t_i$.
\end{observation}

\begin{lemma} \label{lem:monotone-closed polytope W}
For every $i$ and $t_i \in \cT_i$, $W_{t_i}$ is 
a convex and down-monotone polytope. Moreover, given access to a demand oracle $\dem_i(t_i,\cdot)$ for buyer $i$, we can calculate an element $$w_i^*\in \text{argmax}_{w_i\in W_{t_i}} \bm{a}\cdot w_i,$$
for any $\bm{a}\in \mathbb{R}^{\sum_{j\in[m]}|\cT_{ij}|}$ with a single query to the demand oracle. Moreover, the bit complexity of $w_i^*$ is at most $\sum_{j\in[m]}|\cT_{ij}|$.
\end{lemma}
\begin{proof}
{For every feasible set $S\in \cF_i$, consider allocation $\lambda^{(S)}=\{\ind[S'=S]\}_{S'\in\cF_i}$. The set of $\sigma=\{\sigma_S\}_{S\in\cF_i}$ that satisfy property 1 of \Cref{def:W_i-t_i} is equivalent to the set of all convex combinations of the $\lambda^{(S)}$'s and the origin $\mathbf{0}$.
More specifically, $\sigma = \sum_{S\in \cF_i} \sigma_S\cdot \lambda^{(S)}+(1-\sum_{S\in\cF_i}\sigma_S)\cdot \mathbf{0}$. Hence, the set of $\sigma$'s is a convex polytope $P$ in $[0,1]^{|\cF_i|}$. Since $W_{t_i}$ is a projection of $P$ to $[0,1]^{\sum_{j\in[m]}|\cT_{ij}|}$, $W_{t_i}$ is also a convex polytope.}


Next, we prove $W_{t_i}$ is down-monotone. Consider any $w_i\in W_{t_i}$. By \Cref{def:W_i-t_i}, $w_{ij}(t_{ij}')=0$ for all $j,t_{ij}'$ such that $t_{ij}'\not=t_{ij}$. To prove that $W_{t_i}$ is down-monotone, it suffices to prove that for every $j_0\in [m]$, any vector $\widetilde{w}_i$, achieved by only decreasing the coordinate $t_{ij_0}$ from $w_i$, is still in $W_{t_i}$. Formally, $\widetilde{w}_i$ satisfies
\begin{itemize}
    \item $\widetilde{w}_{ij_0}(t_{ij_0})<w_{ij_0}(t_{ij_0})$.
    \item $\widetilde{w}_{ij}(t_{ij})=w_{ij}(t_{ij})$, if $j\not=j_0$.
    \item 
    $\widetilde{w}_{ij}(t_{ij}')=0$, for any $j\in[m]$, if $t_{ij}'\not=t_{ij}$.
\end{itemize}

Let $\{\sigma_S\}_{S\in \cF_i}$ be the vector of numbers associated with $w_i$ in \Cref{def:W_i-t_i}. Let $\alpha=\frac{\widetilde{w}_{ij_0}(t_{ij_0})}{w_{ij_0}(t_{ij_0})}<1$. Consider another vector $\{\widetilde{\sigma}_S\}_{S\in \cF_i}$ where for every $S\in \cF_i$, $$\widetilde{\sigma}_S=
\begin{cases}
\alpha\cdot\sigma_S+(1-\alpha)\cdot (\sigma_{S\cup\{j_0\}}+\sigma_S), & j_0\not\in S\wedge S\cup\{j_0\}\in \cF_i\\
\alpha\cdot \sigma_S, & j_0\in S\\
\sigma_S, & \text{o.w.}
\end{cases}$$

We notice that the above definition is well-defined because $S\in \cF_i$ as long as $S\cup\{j_0\}\in \cF_i$. {Also $\sum_{S\in \cF_i}\widetilde{\sigma}_S\leq 1$. Intuitively, $\sigma=\{\sigma_S\}_{S\in \cF_i}$ represents a randomized allocation of sets $S\in \cF_i$ to the bidder. Then $\{\widetilde{\sigma}_S\}_{S\in \cF_i}$ represents another randomized allocation: choose a set $S$ according to the the randomized allocation  $\sigma$, if $S$ contains $j_0$, then throw away $j_0$ with probability $1-\alpha$.} 

Now we have $\sum_{S\in \cF_i, j_0\in S}\widetilde{\sigma}_S=\alpha\sum_{S\in \cF_i, j_0\in S}\sigma_S=\widetilde{w}_{ij_0}(t_{ij_0})$ and $\sum_{S\in \cF_i, j\in S}\widetilde{\sigma}_S=w_{ij}(t_{ij})$, for all $j\not=j_0$. Thus $\widetilde{w}_i\in W_{t_i}$, and $W_{t_i}$ is down-monotone.

It remains to show that we can find an element in $\text{argmax}_{w_i\in W_{t_i}} \bm{a}\cdot w_i$ for any $\bm{a}\in \mathbb{R}^{\sum_{j\in[m]}|\cT_{ij}|}$, given access to a demand oracle.
W.l.o.g. we assume that for each $j\in[m]$ and $t_{ij}\in \cT_{ij}$,
$a_{ij}(t_{ij})\leq \frac{1}{2}\cdot t_{ij}$ (by scaling and the fact that $t_{ij}>0$).
Observe that any corner $w_i$ of the polytope $W_{t_i}$ corresponds to the choice of $\{\sigma_S\}_{S\in \cF_i}$ such that $\sigma_S=\ind[S=T]$ for some particular $T\in \cF_i$, i.e. $w_{ij}(t_{ij}')=\ind[j\in T\wedge t_{ij}'=t_{ij}]$.

Since 
$w_i^*\in \text{argmax}_{w_i\in W_{t_i}} \bm{a}\cdot w_i$ is a corner of $W_{t_i}$. We have
\begin{align*}
    & \text{max}_{w_i\in W_{t_i}} \bm{a}\cdot w_i\\
    =& \max_{S\in\cF_i} {\sum_{j\in S}} a_{ij}(t_{ij})\\ 
    =& \max_{S\in\cF_i} {\sum_{j\in S}} a_{ij}(t_{ij})^+\\
    =& \max_{S\in\cF_i} {\sum_{j\in S}} (t_{ij}-(t_{ij}-a_{ij}(t_{ij})^+))\\
\end{align*}

Here $x^+=\max\{x,0\}$. The second equality holds because $\cF_i$ is downward-closed.
Notice that $a_{ij}(t_{ij})^+\leq \frac{1}{2}t_{ij}$. Thus with a single query to the demand oracle, with type $t_i$ and prices $p_{ij}=t_{ij}-a_{ij}(t_{ij})^+\geq 0,\forall j$, we can find $\text{argmax}_{w_i\in W_{t_i}} \bm{a}\cdot w_i$. The bit complexity of $w_i^*$ is at most $\sum_{j\in[m]}|\cT_{ij}|$. 
\end{proof}

\begin{lemma}\label{lem:box polytope constraint additive}
Let $T = \sum_{j\in[m]}|\cT_{ij}|$. For any $\eps < \frac{1}{T}$, $(1-\eps T)W_{i}^{box(\eps)}\subseteq W_i$.
\end{lemma}

\begin{proof}

For any $w_i\in (1-\eps T )W_i^{box(\eps)}$, we prove that $w_i\in W_i$. 

Consider the following set of numbers $\{\sigma_{S}(t_i)\}_{t_i,S}$ (see \Cref{{def:W_i-constrained-add}}): For each $j\in [m], t_{ij}\in \cT_{ij}$, let $c_j(t_{ij})=\min\left(\frac{\eps}{f_{ij}(t_{ij})},1 \right)$ and
$$p_j(t_{ij})=\frac{w_{ij}(t_{ij})}{f_{ij}(t_{ij})c_j(t_{ij})\cdot \sum_{t_{i,-j}}f_{i,-j}(t_{i,-j})\cdot \prod_{j'\not=j}(1-c_{j'}(t_{ij'}))}.$$ 

{Note that for every $j'\in[m]$, there exists a value $t_{ij'}\in\cT_{ij'}$ such that $f_{ij'}(t_{ij'})\geq 1/|\cT_{ij'}|$. Due to our choice of $\eps$, the corresponding $c_{j'}(t_{ij'})<1$. Hence, $\sum_{t_{i,-j}}f_{i,-j}(t_{i,-j})\cdot \prod_{j'\not=j}(1-c_{j'}(t_{ij'}))>0$, and $p_j(t_{ij})$  is well-defined.}

For every $t_i$, define
\begin{align*}
\sigma_{S}(t_i)=
\begin{cases}
p_j(t_{ij})\cdot c_j(t_{ij})\cdot \prod_{j'\not=j}(1-c_{j'}(t_{ij'})), &\text{if  $S=\{j\}$ for some $j\in[m]$}\\
0, &\text{o.w.}
\end{cases}
\end{align*}

For every $j$ and $t_{ij}$, let $C_j(t_{ij})$ be the independent Bernoulli random variable that is $1$ with probability $c_j(t_{ij})$. Then for every $j$, 
$$
\Pr_{C_j(t_{ij}), t_{ij}\sim D_{ij}}\left[C_j(t_{ij})=1\right]=\sum_{t_{ij}\in \cT_{ij}} f_{ij}(t_{ij})\cdot\min\left(\frac{\eps}{f_{ij}(t_{ij})},1 \right) \leq \eps \cdot |\cT_{ij}|
$$

By the union bound,
\begin{align*}\sum_{t_{i,-j}}f_{i,-j}(t_{i,-j})\cdot \prod_{j'\not=j}(1-c_{j'}(t_{ij'}))=&\Pr_{\substack{t_{i,-j}\sim D_{i,-j}\\\forall j' \neq j, C_k(t_{ij'})}}[C_{j'}(t_{ij'}) =0,\forall j'\not=j]\\
=&1- \Pr_{\substack{t_{i,-j}\sim D_{i,-j}\\\forall j' \neq j, C_k(t_{ij'})}}[C_{j'}(t_{ij'}) =1,\exists j'\not=j]\\
\geq & 1-\sum_{j\neq j'}\Pr_{\substack{t_{ij}\sim D_{ij}\\ C_j(t_{ij}) }}[C_j(t_{ij})=1]\\
\geq & 1-\eps\cdot T
\end{align*}

Now we show that $w_i\in W_i$ by verifying both properties in \Cref{def:W_i-constrained-add}. {For the first property, since $w_i\in (1-\eps T)W_{i}^{box(\eps)}$, $0 \leq w_{ij}(t_{ij}) \leq \left(1-\eps \cdot T \right)\cdot\min\{\eps,l_{t_{ij}}(W_i)\}= \left(1-\eps \cdot T \right)\cdot f_{ij}(t_{ij})\cdot c_j(t_{ij})$. The equality is due to the definition of $c_j(t_{ij})$ and \Cref{obs:constraint additive length}.} Thus $p_j(t_{ij})\leq 1$ for every $j$ and $t_{ij}$.
{We have that $\sum_{S}\sigma_S(t_i)=\sum_{j\in[m]} p_j(t_{ij})\cdot c_j(t_{ij})\cdot \prod_{j'\not=j}(1-c_{j'}(t_{ij'}))\leq \sum_{j\in[m]} c_j(t_{ij})\cdot \prod_{j'\not=j}(1-c_{j'}(t_{ij'}))\leq \prod_{j\in[m]}\left(c_j(t_{ij})+(1-c_j(t_{ij})\right)=1$.}

The second property: 
\begin{align*}
&f_{ij}(t_{ij})\cdot\sum_{t_{i,-j}}f_{i,-j}(t_{i,-j})\cdot \sum_{S:j\in S}\sigma_S(t_{ij},t_{i,-j})\\
=&f_{ij}(t_{ij})\sum_{t_{i,-j}}f_{i,-j}(t_{i,-j})\cdot p_j(t_{ij})\cdot c_j(t_{ij})\cdot \prod_{j'\not=j}(1-c_{j'}(t_{ij'}))=w_{ij}(t_{ij})
\end{align*}
\notshow{
The third inequality:
\begin{align*}
&f_{ij}(t_{ij})\cdot\sum_{t_{i,-j}}f_{i,-j}(t_{i,-j})\cdot \sum_{S:j\in S}\sum_{k\in [K]}\sigma_S(t_{ij},t_{i,-j})\cdot \frac{\alpha_{ij}(t_{ij})}{V_{ij}(t_{ij})}\\
=&f_{ij}(t_{ij})\sum_{t_{i,-j}}f_{i,-j}(t_{i,-j})\cdot p_j(t_{ij})\cdot c_j(t_{ij})\cdot \prod_{j'\not=j}(1-c_{j'}(t_{ij'}))\cdot \frac{\max_k\alpha_{ij}(t_{ij})}{V_{ij}(t_{ij})}=\pi_{ij}(t_{ij})=w_{ij}'(t_{ij})
\end{align*}
}
Thus by \Cref{def:W_i-constrained-add}, $w_i\in W_i$. 

\notshow{
The old proof:

We consider the following allocation rule $\sigma_{S}^k(t_i)$ parametrised by variables $\{p_{j,t_{ij}}\in[0,1]\}_{j\in[m],t_{ij}\in T_{ij}}$.
We are going to prove that $W_{i,\epsilon}^r\subseteq W_i$.
When an agent with type $t_i$ appears, the mechanism works in the following way.
For each $j\in[m]$, let 
$$C(j) \sim Ber\left( \min\left(\frac{\epsilon}{f_{ij}(t_{ij})},1 \right)\right)$$
If there exists only one $j^*\in[m]$ such that $C(j^*)=1$,
then allocate only the $j^*$-th item to the agent activating the additive function $k^*\in[K]$ such that $\alpha^{k*}_{j^*}=V_{ij}(t_{ij})$.
Note that:
$$
\Pr_{\substack{t_{ij}\sim D_{ij}\\C(j) \sim Ber\left( \min\left(\frac{\epsilon}{f_{ij}(t_{ij})},1) \right)\right) }}\left[C_j=1  \right] \leq \sum_{t_{ij}\in \cT_{ij}} f_{ij}(t_{ij})\cdot\min\left(\frac{\eps}{f_{ij}(t_{ij})},1 \right) \leq \eps \cdot |\cT_{ij}|
$$
Thus by union bound, $\Pr[C_j=0,\forall j]\geq 1-\sum_j\Pr[C_j=1]\geq 1-\eps\cdot T$.
Thus under this allocation we have that
\begin{align*}
    w_{ij}(t_{ij}) = \pi_{ij}(t_{ij}) \geq (1-\epsilon\cdot M)  \min\left(\frac{\epsilon}{f_{ij}(t_{ij})}, 1 \right)f_{ij}(t_{ij}) =
    (1-\epsilon\cdot M)  \min\left(\epsilon, f_{ij}(t_{ij})\right)
\end{align*}
Thus for any target 
$$\{\hat{\pi}_{ij}(t_{ij})\leq (1-\epsilon \cdot M)\min(\epsilon,f_{ij}(t_{ij})\}_{j\in[m],t_{ij}\in T_{ij}}$$ 
there exists a set of variables $\{p_{j,t_{ij}} \in [0,1]\}_{j\in[m],t_{ij}\in T_{ij}}$ such that when the agent has type $t_{i}$, then with probability $p_{j,t_{ij}}$ we set $C(j)=0$ and then run the mechanism as we described it.
This allocation distribution ensures that:
\begin{align*}
    w_{ij}(t_{ij}) = \pi_{ij}(t_{ij}) =\hat{\pi}_{ij}(t_{ij})
\end{align*}
Thus we proved that for every $(\pi,w)$ such that:
\begin{align*}
    w_{ij}(t_{ij}) = \pi_{ij}(t_{ij}) \leq
    (1-\epsilon\cdot M)  \min\left(\epsilon, f_{ij}(t_{ij})\right)
\end{align*}
$(\pi,w_i)\in  W_i$.
Observe that for each $(\pi',w_i')$ such that:
\begin{align*}
    w_{ij}'(t_{ij})\leq w_{ij}(t_{ij}) \\
    \pi_{ij}'(t_{ij})\geq \pi_{ij}(t_{ij})
\end{align*}

then $(\pi',w')\in W_i$.
This concludes that $W_{i,\eps}^r \subseteq W_i$ for all $\epsilon < \frac{1}{T}$.}
\end{proof}

\notshow{

\begin{theorem}\label{thm:multiplicative approx for constraint additive}
Let $T = \sum_{j\in[m]}|\cT_{ij}|$ and consider parameters $\eps <\frac{1}{T}$ and \\
$k\geq\max \left(\log\left(\frac{1}{\eps}\right),\Omega\left(d^4\left( b + \log\left( \frac{1}{\eps}\right)\right)\right)\right)$.
Let $b$ be an upper bound on the bit complexity of each element in $\{f_{ij}(t_{ij})\}_{j\in[m],t_{ij}\in\cT_{ij}}$.
Then with probability at least $1-2T\exp(-2Tk)$,
having access to $\left\lceil\frac{8kT}{(1-\eps T)^2}\right\rceil$ samples from $\cD_i$,
we can construct a convex polytope $\widehat{W}_i$ such that
\begin{enumerate}
\item $\frac{1-\eps T}{6}W_i\subseteq \widehat{W}_i\subseteq W_i$

\item There exists a separation oracle $SO$ for $\widehat{W}_i$,
whose running time on input with bit complexity $b''$, is $\poly\left(b,b'',k,d,\frac{1}{\eps}\right)$ and performs $\poly\left(b,b'', k,d,\frac{1}{\eps}\right)$ queries to demand oracle $\dem_i(\cdot,\cdot)$ with inputs of bit complexity $\poly\left(b, b'',d,\frac{1}{\eps}\right)$.
\end{enumerate} 
\end{theorem}

\begin{proof}
We notice that:
\begin{enumerate}
\item By \Cref{obs:constraint additive length}, we have that the bit complexity of each element in $\{l_j(W_i)\}$ is $b'$.
\item By Observation~\ref{obs:mix constraint additive}, we have that $W_{i}$ is mixture of $\{W_{t_i}\}_{t_i\in \cT_i}$ over distribution $D_i$.
\item By Lemma~\ref{lem:monotone-closed polytope W}, for each $t_i\in\cT_i$, $W_{t_{i}}$ is lowered-closed and given access to demand oracle $\dem_i(t_i,\cdot)$ we can find an element in $\text{argmax}\{\bm{a}\cdot w_i\mid w_i\in W_{t_i}\}$.
Note that each output of the demand oracle has bit complexity at most $T$. 
\item By Lemma~\ref{lem:box polytope constraint additive}, $(1-\eps T)W_i^{box(\eps)}\subseteq W_i$ 
\end{enumerate}
The proof then follows directly by applying \Cref{thm:special case of multiplicative approx} to $W_i$ for every $i$.
\end{proof}

}


\begin{prevproof}{Theorem}{thm:multiplicative approx for constraint additive-main body}

We simply verify that $W_i$ satisfies all assumptions in \Cref{thm:special case of multiplicative approx}. Recall that $T=\sum_{i,j}|\cT_{ij}|$ and {$b$ is an upper bound of the bit complexity of all $f_{ij}(t_{ij})$'s and all $t_{ij}$'s.} We have the following: 

\begin{enumerate}
\item By \Cref{obs:mix constraint additive}, we have that $W_{i}$ is a mixture of $\{W_{t_i}\}_{t_i\in \cT_i}$ over distribution $D_i$.

\item By Lemma~\ref{lem:monotone-closed polytope W}, for each $t_i\in\cT_i$, $W_{t_{i}}$ is a convex and down-monotone polytope. Given access to the demand oracle $\dem_i(t_i,\cdot)$, we can find an element in $\argmax\{\bm{a}\cdot w_i: w_i\in W_{t_i}\}$ in time {$poly(b',b,T)$ and a single query to $\dem_i(t_i,\cdot)$}, where $b'$ is the bit complexity of the input $\bm{a}$. Note that each output of the demand oracle has bit complexity at most $T$. 

\item By \Cref{obs:constraint additive length}, {$l_{t_{ij}}(W_i)=f_{ij}(t_{ij})$,for all $j$ and $t_{ij}$.  Thus, each $l_{t_{ij}}(W_i)$ has bit complexity at most $b$.} 


%

%
\item By Lemma~\ref{lem:box polytope constraint additive}, $(1-\eps T)W_i^{box(\eps)}\subseteq W_i$ for any $\eps<\frac{1}{T}$. Choosing $\eps=\frac{1}{2T}$ obtains $\frac{1}{2}W_i^{box(\frac{1}{2T})}\subseteq W_i$.

\end{enumerate}

For any $\delta\in (0,1)$, we apply \Cref{thm:special case of multiplicative approx} with parameter $k=poly(n,m,T,b,\log(1/\delta))$, $c=\frac{1}{2}$ and $\eps=\frac{1}{2T}$. The probability that the algorithm successfully constructs a polytope that satisfies both properties of \Cref{thm:special case of multiplicative approx} is at least $1-\delta$. We have $\frac{1}{12}\cdot W_i\subseteq \widehat{W}_i\subseteq W_i$ by the first property of \Cref{thm:special case of multiplicative approx} with $c=\frac{1}{2}$. Since the  vertex-complexity of $W_{t_i}$ for each $t_i$ is no more than $T$, and the vertex-complexity for $W_i^{box(\frac{1}{2T})}$ is no more than $\poly(b,T)$, the vertex-complexity for $\widehat{W}_i$ is no more than $poly(n,m,T,b,\log(1/\delta))$. The running time of the algorithm and the separation oracle $SO$ for $\widehat{W}_i$ follows from the second property of \Cref{thm:special case of multiplicative approx}. 

\end{prevproof}

At last, we give the proof of \Cref{thm:main} by combining \Cref{thm:bounding-lp-simple-mech} and \Cref{thm:multiplicative approx for constraint additive-main body}.

\begin{prevproof}{Theorem}{thm:main}
Fix any $\delta\in (0,1)$. Recall that in the LP of \Cref{fig:bigLP}, we use an estimation of $\prev$, which is $\estprev$, according to \Cref{thm:chms10}. Denote $\cE_1$ the event that an RPP mechanism is successfully computed and $\cE_2$ the event that the algorithm in \Cref{thm:multiplicative approx for constraint additive-main body} successfully constructs a convex polytope $\widehat{W}_i$ that satisfies both properties in the statement of \Cref{thm:multiplicative approx for constraint additive-main body} {for each buyer $i\in [n]$}. Note that $\cE_1$ happens with probability at least $1-\frac{2}{nm}$ and {we take enough samples to make sure that} $\cE_2$ happens with probability at least $1-\delta$, by the union bound, the probability that both $\cE_1$ and $\cE_2$ happen is at least  $1-\delta-\frac{2}{nm}$. From now on, we will condition on both events $\cE_1$ and $\cE_2$.

Now in the LP of \Cref{fig:bigLP}, we replace $W_i$ by $\widehat{W}_i$ for every $i$ (we will call it the modified LP). By property 1 and 2 of \Cref{thm:multiplicative approx for constraint additive-main body}, we can solve the modified LP using the separation oracle for $\widehat{W}_i$, in time $\poly(b,n,m,T,\log(1/\delta))$ (recall that $T=\sum_{i,j}|\cT_{ij}|$) according to \Cref{thm:ellipsoid}. Let $x^*=(w^*,\lambda^*,\hat\lambda^*,\textbf{d}^*)$ be an optimal solution of the modified LP. Then $w_i^*\in \widehat{W}_i$ for every $i$ according to Constraint {\Wconstraint} in the modified LP. By property 1 of \Cref{thm:multiplicative approx for constraint additive-main body}, we have $\widehat{W}_i\subseteq W_i$. Thus $w_i^*\in W_i,\forall i$ and hence $x^*$ is also a feasible solution of the original LP in \Cref{fig:bigLP}.

Recall that $\optlp$ is the optimum objective of the original LP. Denote $\optlp'$ the optimum objective of the modified LP. Thus in order to prove that $(w^*,\lambda^*,\hat\lambda^*,\textbf{d}^*)$ is an approximately-optimal solution in the original LP, it suffices to show that $\optlp'\geq c\cdot \optlp$. Take any feasible solution $(w,\lambda,\hat\lambda,\textbf{d})$ of the original LP. We have that $w_i\in W_i$ for every $i$. Now consider another set of variables $(w',\lambda',\hat\lambda,\textbf{d})$ such that $w'_{ij}(t_{ij})=c\cdot w_{ij}(t_{ij})$ and $\lambda'_{ij}(t_{ij},\beta_{ij},\delta_{ij})=c\cdot \lambda_{ij}(t_{ij},\beta_{ij},\delta_{ij})$, for all $i,j,t_{ij},\beta_{ij},\delta_{ij}$, where $c=1/12$. 
We verify that $(w',\lambda',\hat\lambda,\textbf{d})$ is a feasible solution for the modified LP. For Constraint {\Wconstraint}, since $c\cdot W_i\subseteq \widehat{W}_i$ (property 1 of \Cref{thm:multiplicative approx for constraint additive-main body}), we have that $w'_i\in \widehat{W}_i$, for all $i$. For Constraint {\LambdaMarginalConstraint}, it holds since we multiply both $\lambda$ and $w$ by $c$. Constraint {\PiConstraint}, {\CompareMarginalConstraint} and {\MarginalToGlobalConstraint} hold, as for each of them their LHS is smaller while the RHS remains unchanged. Every other constraint holds since both of their LHS and RHS remain the same. Thus $(w',\lambda',\hat\lambda,\textbf{d})$ is a feasible solution for the modified LP.


Now notice that the objective of the solution $(w',\lambda',\hat\lambda,\textbf{d})$ is exactly $c$ times the objective of the solution $(w,\lambda,\hat\lambda,\textbf{d})$. By choosing $(w,\lambda,\hat\lambda,\textbf{d})$ to be the optimal solution of the original LP, we have that $\optlp'\geq c\cdot \optlp$. It implies that the objective of $(w^*,\lambda^*,\hat\lambda^*,\textbf{d}^*)$ is at least $c\cdot \optlp$. Thus if we compute the simple mechanisms using the decision variables $(w^*,\lambda^*,\hat\lambda^*,\textbf{d}^*)$. By \Cref{thm:multiplicative approx for constraint additive-main body}, we have    $$c_1\cdot \rev(\Mpp)+c_2\cdot\rev(\Mtpt)\geq\optlp'\geq c\cdot \optlp\geq c\cdot \opt$$
We finish our proof by noticing that the simple mechanisms can be computed in time $\poly(n,m,T)$ given the solution $(w^*,\lambda^*,\hat\lambda^*,\textbf{d}^*)$.
\end{prevproof}

\section{Accessing Single-Bidder Marginal Reduced Form Polytopes for XOS Valuations}\label{sec:mrf for xos}

Our goal in this section is to prove \Cref{thm:main XOS-main body} for XOS valuations. 




\begin{theorem}\label{thm:main XOS}
(Restatement of \Cref{thm:main XOS-main body} for XOS valuations) Let $T=\sum_{i,j}|\cT_{ij}|$ and $b$ be the bit complexity of the problem instance (\Cref{def:bit complexity}). 
For XOS buyers, for any $\delta>0$, there exists an algorithm that computes a rationed posted price mechanism or a two-part tariff mechanism, such that the revenue of the mechanism is at least $c\cdot \opt$ for some absolute constant $c>0$ with probability $1-\delta-\frac{2}{nm}$.
Our algorithm assumes query access to a value oracle and an adjustable demand oracle (see \Cref{sec:prelim}) of buyers' valuations, and has running time $\poly(n,m,T,b,\log (1/\delta))$.
\end{theorem}








We remind the readers the definition of the single-bidder marginal reduced form polytope $W_i$ for XOS valuations:

\begin{definition}[Restatement of \Cref{def:W_i}]\label{def:W_i-restate}
For every $i\in [n]$, the single-bidder marginal reduced form polytope $W_i\subseteq [0,1]^{2\cdot \sum_j|\cT_{ij}|}$ is defined as follows. Let $\pi_i=(\pi_{ij}(t_{ij}))_{j,t_{ij}\in \cT_{ij}}$ and $w_i=(w_{ij}(t_{ij}))_{j,t_{ij}\in \cT_{ij}}$. Then $(\pi_i,w_i)\in W_i$ if and only if there exist a number $\sigma_S^{(k)}(t_i)\in [0,1]$ for every $t_i\in \cT_i, S\subseteq [m],k\in [K]$, such that
\begin{enumerate}
    \item $\sum_{S,k}\sigma_S^{(k)}(t_i)\leq 1 $, $\forall t_i\in \cT_i$.
    \item $
    {\pi_{ij}(t_{ij})=}f_{ij}(t_{ij})\cdot\sum_{t_{i,-j}}f_{i,-j}(t_{i,-j})\cdot \sum_{S:j\in S}\sum_{k\in [K]}\sigma_S^{(k)}(t_{ij},t_{i,-j})$, for all $i,j,t_{ij}\in \cT_{ij}$.
    \item ${w_{ij}(t_{ij})\leq} f_{ij}(t_{ij})\cdot\sum_{t_{i,-j}}f_{i,-j}(t_{i,-j})\cdot \sum_{S:j\in S}\sum_{k\in [K]}\sigma_S^{(k)}(t_{ij},t_{i,-j})\cdot \frac{\alpha_{ij}^{(k)}(t_{ij})}{V_{ij}(t_{ij})}$, for all $i,j,t_{ij}\in \cT_{ij}$. {If $V_{ij}(t_{ij})=0$, we slightly abuse notation and treat $\frac{\alpha_{ij}^{(k)}(t_{ij})}{V_{ij}(t_{ij})}$ as $0$ for all $k$.}
\end{enumerate}
\end{definition}

The following observation follows directly from \Cref{def:W_i}.

\begin{observation}\label{obs:W_i-upper-bound}
For any $(\pi_i,w_i)\in W_i$ and any $j\in [m], t_{ij}\in \cT_{ij}$, $w_{ij}(t_{ij})\leq \pi_{ij}(t_{ij})\leq f_{ij}(t_{ij})$.
\end{observation}
\begin{proof}
The first inequality follows directly from the fact that $\alpha_{ij}^{(k)}(t_{ij})\leq V_{ij}(t_{ij})=\max_{k'}\alpha_{ij}^{(k')}(t_{ij})$. The second inequality follows from $\sum_{S:j\in S}\sum_{k\in [K]}\sigma_S^{(k)}(t_{ij},t_{i,-j})\leq \sum_{S}\sum_{k}\sigma_S^{(k)}(t_{ij},t_{i,-j})\leq 1$.
\end{proof}

In \Cref{sec:program-XOS}, we provide an LP (\Cref{fig:XOSLP}) that helps us to compute the simple mechanisms efficiently. In \Cref{thm:bounding-lp-simple-mech-XOS}, we have proved that given any optimal solution to the LP in \Cref{fig:XOSLP}, we can compute a simple mechanism in polynomial time, whose revenue is a constant factor of the optimal revenue. However, constraint {\Wconstraint} is implicit and thus it's unclear if we can solve the LP in polynomial time. Similar to the idea in \Cref{sec:multi-approx-polytope}, we fix this issue by constructing another polytope $\widehat{W}_i$. Unfortunately, for XOS valuations, $W_i$ is not a down-monotone polytope anymore. To see this, we simply notice that by \Cref{obs:W_i-upper-bound}, $w_{ij}(t_{ij})\leq \pi_{ij}(t_{ij})$ for every coordinate $(j,t_{ij})$. Thus given any $(\pi_i,w_i)\in W_i$ where $w_{ij}(t_{ij})>0$ for some $j,t_{ij}$, the vector $(\textbf{0},w_i)$ is clearly not in $W_i$. Thus the argument in \Cref{sec:multi-approx-polytope} does not apply here.

\subsection{Basic Properties of the Single-Bidder Marginal Reduced Form for XOS Valuations}

In this section we present some basic definitions and properties of the single-bidder Marginal Reduced Form polytope $W_i$ (\Cref{def:W_i}). We fixed any buyer $i$ throughout this section unless otherwise specified. 

\begin{definition}\label{def:welfare frequent}
For any $\eps>0$,
we denote as $W_i^{tr(\eps)}\subseteq [0,1]^{2\sum_{j\in[m]}|\cT_{ij}|}$ the \emph{$\eps$-truncated polytope} of $W_i$.
An element $(\widehat{\pi}_i,\widehat{w}_i)\in W_i^{tr(\eps)}$
if there exists $(\pi_i,w_i)\in W_i$ such that for all $j\in[m]$ and $t_{ij}\in \cT_{ij}$:
\begin{align*}
    \widehat{w}_{ij}(t_{ij}) &=
w_{ij}(t_{ij})\cdot\ind[ f_{ij}(t_{ij})\geq \eps]\\ \widehat{\pi}_{ij}(t_{ij}) &=
\pi_{ij}(t_{ij})\cdot\ind[ f_{ij}(t_{ij})\geq \eps]
\end{align*}

\end{definition}

Similar to Section~\ref{sec:multi-approx-polytope}, we show that $W^{tr(\eps)}_i$ is a mixture (\Cref{def:mixture}) of a set of polytopes $\{W^{tr(\eps)}_{t_i}\}_{t_i\in \cT_i}$ defined in \Cref{def:welfare frequent-t_i} over $\cD$.

\begin{definition}\label{def:welfare frequent-t_i}
For any $i, t_i \in \cT_i$ and $\eps>0$, we define the polytopes $W_{t_i},W^{tr(\eps)}_{t_i} \subseteq [0,1]^{2\sum_{j\in[m]}|\cT_{ij}|}$ as follows: An element{ $\left(x=\left\{x(t'_{ij})\right\}_{t'_{ij}\in \cT_{ij}},y=\left\{y(t'_{ij})\right\}_{t'_{ij}\in \cT_{ij}}\right)\in W_{t_i}$} if there exists a collection of non-negative numbers $\{\sigma_S^{(k)}\}_{S\subseteq [m], k\in [K]}$, such that $\sum_{S \subseteq [m]}\sum_{k \in [K]} \sigma_S^{(k)} \leq 1$,  and for any $j,t_{ij}'\in \cT_{ij}$,
{
\begin{align*}
&x(t_{ij}')=\sum_{S:j\in S} \sum_{k \in [K]} \sigma_S^{(k)}\cdot \ind[t_{ij}'=t_{ij}]\\
&y(t_{ij}')\leq \sum_{S:j\in S} \sum_{k \in [K]} \sigma_S^{(k)}\cdot \frac{\alpha_{ij}^{(k)}(t_{ij})}{V_{ij}(t_{ij})}\cdot\ind[t_{ij}'=t_{ij}].
\end{align*}
}


{Moreover, an element $\left(\hat{x}=\left\{\hat{x}(t'_{ij})\right\}_{t'_{ij}\in \cT_{ij}}, \hat{y}=\left\{\hat{y}(t'_{ij})\right\}_{t'_{ij}\in \cT_{ij}}\right)\in W^{tr(\eps)}_{t_i}$ if there exists $(x,y)\in W_{t_i}$ such that for any $j,t_{ij}'$,
\begin{align*}
\hat{x}(t_{ij}')&=
x(t_{ij}')\cdot\ind[ f_{ij}(t_{ij})\geq \eps]\\
\hat{y}(t_{ij}')&= y(t_{ij}')\cdot\ind[ f_{ij}(t_{ij})\geq \eps]
\end{align*}}
\end{definition}

The following observation directly follows from \Cref{def:W_i} and \Cref{def:welfare frequent-t_i}.

\begin{observation}\label{obs:mix}
$W_{i}$ is a mixture of $\{W_{t_i}\}_{t_i\in \cT_i}$ over distribution $D_i$. For any $\eps>0$, $W_{i}^{tr(\eps)}$ is a mixture of $\{W^{tr(\eps)}_{t_i}\}_{t_i\in \cT_i}$ over distribution $D_i$.
\end{observation}

The following observation is useful in later proofs.

\notshow{
\begin{observation}\label{obs:bound xos mrf}
Let $\widetilde{W}_i$ be a mixture of $\{W_{t_i}\}_{t_i\in \cT_i}$ over distribution $\widetilde{D}_i$.
Then for any $a'\geq a>0$,  $a\widetilde{W}_i\subseteq a' \widetilde{W}_i$
\end{observation}
}
\begin{observation}\label{obs:bound xos mrf}
For any $\eps>0$ and $a'\geq a>0$, $a\cdot W_i\subseteq a'\cdot W_i$ and $a\cdot W^{tr(\eps)}_i\subseteq a'\cdot W^{tr(\eps)}_i$.
\end{observation}

\begin{proof}
Let $c=\frac{a}{a'}\leq 1$. For the first statement, it suffices to prove that $(c\pi_i,cw_i)\in W_i$, for all $(\pi_i,w_i)\in W_i$. Let $\{\sigma_S^{(k)}(t_i)\}_{t_i,S,k}$ be the collection of numbers that satisfy all properties of \Cref{def:W_i}. Then since $c\leq 1$, by considering the collection of numbers $\{c\cdot \sigma_S^{(k)}(t_i)\}_{t_i,S,k}$, we immediately have that $(c\pi_i,cw_i)\in W_i$. For the second statement, let $(\pi_i',w_i')$ be the vector achieved by zeroing out all coordinates $(j,t_{ij})$ where $f_{ij}(t_{ij})<\eps$ for the vector $(\pi_i,w_i)$. By the definition of $W^{tr(\eps)}_i$, we immediately have $(\pi_i',w_i')\in W^{tr(\eps)}_i$ and $(c\pi_i',cw_i')\in W^{tr(\eps)}_i$.
\end{proof}

{We next present several desirable properties of the polytopes we consider here.}
\notshow{
\begin{proof}
For any $(\pi_i,w_i)\in aW^{tr(\eps)}_i$,
there exists for each $t_i\in \cT_i$, $\left(\pi^{(t_i)}_i,w^{(t_i)}_i\right)\in W^{tr(\eps)}_{t_i}$ such that
$$
(\pi_i,w_i) = \sum_{t_i\in \cT_i} a \cdot f_i(t_i)\left(\pi^{(t_i)}_i,w^{(t_i)}_i\right)
$$
If we show that for each $t_i\in \cT_i$, $\frac{a}{a'} \cdot \left(\pi^{(t_i)}_i,w^{(t_i)}_i\right)\in W^{tr(\eps)}_{t_i}$,
then we have that $(\pi_i,w_i)=\sum_{t_i\in \cT_i} a' \cdot
f_i(t_i)\frac{a}{a'}\left(\pi^{(t_i)}_i,w^{(t_i)}_i\right)$,
which implies that $(\pi_i,w_i)\in a'W^{tr(\eps)}_i$.
Fix some $t_i\in \cT_i$ and   consider the numbers $\sigma_S^{(k)}$ for $S\subseteq [m]$, $k\in [K]$ associated with $\left(\pi^{(t_i)}_i,w^{(t_i)}_i\right)$ according to Definition~\ref{def:welfare frequent-t_i}.

Observe that $\frac{a}{a'}\leq 1$ and for each $S\subseteq[m]$ and $k\in [K]$,
we consider the number $\widehat{\sigma}_S^{(k)}=\frac{a}{a'}\sigma_S^k$.
Using Definition~\ref{def:welfare frequent-t_i} with numbers $\{\widehat{\sigma}_S^{(k)}\}_{S,k}$ is enough to show that $\frac{a}{a'} \cdot \left(\pi^{(t_i)}_i,w^{(t_i)}_i\right)\in W^{tr(\eps)}_{t_i}$ which concludes the proof.
\end{proof}
}


\begin{lemma}\label{lem:W_t_i modify}
For any $t_i\in \cT_i$, any subset of items $B\subseteq[m]$, 
and any {$(x,y)\in W_{t_i}$},
consider any 
{$(\hat{x},\hat{y})\in [0,1]^{2\sum_{j\in[m]}|\cT_{ij}|}$ }such that for each $j\in[m]$ and $t_{ij}'\in \cT_{ij}$,

{
\begin{align*}
\hat{x}(t_{ij}')= & x(t_{ij}') \ind[j\in B]\quad\text{ and } \quad \hat{y}(t_{ij}')\leq  y(t_{ij}') \ind[j\in B].
\end{align*}}
Then $(\widehat{x},\widehat{y})\in W_{t_i}$. {Moreover,  if $(x,y)\in W_{t_i}^{tr(\eps)}$, then $(\hat{x},\hat{y})$ is also in $W_{t_i}^{tr(\eps)}$. Finally, $W_{t_i}^{tr(\eps)}\subseteq W_{t_i}$.}
\end{lemma}

\begin{proof}

It suffices to prove the case where $\hat{y}(t_{ij}')=y(t_{ij}') \ind[j\in B], \forall j, t_{ij}'$, since by \Cref{def:welfare frequent-t_i}, we can decrease any $\hat{y}(t_{ij}')$ while maintaining the vector $(\hat{x},\hat{y})$ to be in $W_{t_i}$. 


Since $(x,y)\in W_{t_i}$, let $\{\sigma_S^{(k)}\}_{S\subseteq [m], k\in [K]}$ be the collection of numbers from \Cref{def:welfare frequent-t_i}. Each $\sigma_S^{(k)}$ can be viewed as the probability of the buyer receiving bundle $S$, and enabling the $k$-th additive function. Consider another collection of numbers $\{\widehat{\sigma}_S^{(k)}\}_{S\subseteq [m], k\in [K]}$ by simply discarding items in $[m]\backslash B$. Formally, $\widehat{\sigma}_S^{(k)}=\sum_{T\subseteq [m]\backslash B}\sigma_{S\cup T}^{(k)}, \forall S\subseteq B,k\in [K]$, and $\widehat{\sigma}_S^{(k)}=0$ otherwise. Notice that for every $j\in B$ and $k\in [K]$, $\sum_{S:j\in S} \sigma_S^{(k)}=\sum_{S:j\in S}\widehat{\sigma}_S^{(k)}$. It is not hard to verify that $(\hat{x},\hat{y})$ and $\{\widehat{\sigma}_S^{(k)}\}_{S,k}$ satisfy all inequalities in \Cref{def:welfare frequent-t_i}. Thus $(\hat{x},\hat{y})\in W_{t_i}$.

{Now $W_{t_i}^{tr(\eps)}\subseteq W_{t_i}$ follows from choosing $B$ to be $\{j: f_{ij}(t_{ij})\geq \eps \}$. For any $(x,y)\in W_{t_i}^{tr(\eps)}$ and any choice set $B\subseteq [m]$, $(\hat{x},\hat{y})\in W_{t_i}$. Since $\hat{x}(t_{ij}')=
\hat{x}(t_{ij}')\cdot\ind[ f_{ij}(t_{ij})\geq \eps]$
 and $\hat{y}(t_{ij}')= \hat{y}(t_{ij}')\cdot\ind[ f_{ij}(t_{ij})\geq \eps]$. By \Cref{def:welfare frequent-t_i}, $(\hat{x},\hat{y})$ also lies in $W_{t_i}^{tr(\eps)}$.}
\notshow{
Let $\widetilde{\sigma}_S^{(k)}$ be the numbers associated with an arbitrary element $(\widetilde{\pi}_i,\widetilde{w}_i)\in W_{t_i}$ according to Definition~\ref{def:W_i-t_i}.
Notice that $\widetilde{\sigma}_S^{(k)}$ can be viewed as a joint distribution $\widetilde{\cJ}$ over $2^{[m]} \times [K]$,
where for $S\subseteq [m]$ and $k\in[K]$,
we sample element $(S,k)$ with probability $\sigma_S^{(k)}$.
Under this lens, $\widetilde{\pi}_{ij}(t_{ij}')$ is the probability that we allocate the item to the agent $$\widetilde{\pi}_{ij}(t_{ij}')=\ind[t_{ij}'=t_{ij}]\E_{(S,k)\sim \widetilde{\cJ}}[\ind[j\in S]]=\ind[t_{ij}'=t_{ij}]\Pr_{(S,k)\sim \widetilde{\cJ}}[j\in S]$$
and $\widetilde{w_{ij}}(t_{ij}')$ is a lower bound on the expected contribution to the welfare of the agent when we are able to choose which additive function to activate
$$\widetilde{w_{ij}}(t_{ij}')\leq \ind[t_{ij}'=t_{ij}]\E_{(S,k)\sim \widetilde{\cJ}}\left[\alpha_{ij}^{(k)}(t_{ij})\cdot \ind[j\in S]\right]$$

Let $\cJ$ be the associated joint distribution over $2^{[m]}\times [K]$ of the element $(\pi_i,w_i)$.
We are going to show how to construct distribution $\widehat{\cJ}$ that certifies that $(\widehat{\pi}_i,\widehat{w}_i)\in W_i$.
We consider the distribution $\widehat{\cJ}$ that takes a sample from $2^{[m]}\times [K]$ in the following way:
\begin{enumerate}
    \item First, take a sample $(S,k)\sim \cJ$ and consider the set $\widehat{S}= S \backslash B$.
    \item Return $(\widehat{S},k)$.
\end{enumerate}

Now we are going to show that $\widehat{\cJ}$ certifies, that $(\widehat{\pi}_i,\widehat{w}_i)\in W_i$.
Observe that for each $j\in B$,
$$
\Pr_{(S,k)\sim \widehat{\cJ}}[j\in S]=0
$$
which implies that for each $j\in B$ and $t_{ij}'\in \cT_{ij}$
$$\ind[t_{ij}'=t_{ij}]\Pr_{(S,k)\sim \widehat{\cJ}}[j\in S]= \ind[t_{ij}'=t_{ij}]\E_{(S,k)\sim \widehat{\cJ}}\left[\alpha_{ij}^{(k)}(t_{ij})\cdot \ind[j\in S]\right]=0$$

For each $j\notin B$,
notice that when we take a sample $(\widehat{S},\widehat{k})\sim \widehat{\cJ}$,
the probability that $j\in \widehat{S}$ is equal to the probability that when we take a sample $(S,k)\sim \cJ$, $j\in S$
$$
\Pr_{(S,k)\sim \widehat{\cJ}}[j\in S]= \Pr_{(S,k)\sim \cJ}[j\in S] $$
which implies that for each $j\notin B$ and $t_{ij}'\in \cT_{ij}$
$$\ind[t_{ij}'=t_{ij}]\Pr_{(S,k)\sim \widehat{\cJ}}[j\in S]=\ind[t_{ij}'=t_{ij}] \widehat{\pi}_{ij}(t_{ij})$$

Moreover, by the construction of $\widehat{\cJ}$ for any $k^*\in[K]$ and $j\notin B$,
notice that
$$
\Pr_{(\widehat{S},\widehat{k})\sim \widehat{\cJ}} \left[\ind[j\in \widehat{S} \land k^*=\widehat{k}]\right] = \Pr_{({S},{k})\sim {\cJ}} \left[\ind[j\in S \land k^*=k]\right]
$$
which implies that for each $j\notin B$ and $t_{ij}'\in \cT_{ij}$
\begin{align*}
\ind[t_{ij}'=t_{ij}]\E_{(S,k)\sim \widehat{\cJ}}\left[\alpha_{ij}^{(k)}(t_{ij})\cdot \ind[j\in S]\right]= &
\ind[t_{ij}'=t_{ij}]\E_{(S,k)\sim {\cJ}}\left[\alpha_{ij}^{(k)}(t_{ij})\cdot \ind[j\in S]\right] \\
\geq & \ind[t_{ij}'=t_{ij}]{w_{ij}}(t_{ij}') \\
\geq & \ind[t_{ij}'=t_{ij}]\widehat{w_{ij}}(t_{ij}')    
\end{align*}
which certifies that $(\widehat{\pi}_i,\widehat{w}_i)\in W_{t_i}$.
}
\end{proof}

\notshow{
\begin{proof}
Let $\widetilde{\sigma}_S^{(k)}$ be the numbers associated with an arbitrary element $(\widetilde{\pi}_i,\widetilde{w}_i)\in W_{t_i}$ according to Definition~\ref{def:W_i-t_i}.
Notice that $\widetilde{\sigma}_S^{(k)}$ can be viewed as a joint distribution $\widetilde{\cJ}$ over $2^{[m]} \times [K]$,
where for $S\subseteq [m]$ and $k\in[K]$,
we sample element $(S,k)$ with probability $\sigma_S^{(k)}$.
Under this lens, $\widetilde{\pi}_{ij}(t_{ij}')$ is the probability that we allocate the item to the agent $$\widetilde{\pi}_{ij}(t_{ij}')=\ind[t_{ij}'=t_{ij}]\E_{(S,k)\sim \widetilde{\cJ}}[\ind[j\in S]]=\ind[t_{ij}'=t_{ij}]\Pr_{(S,k)\sim \widetilde{\cJ}}[j\in S]$$
and $\widetilde{w_{ij}}(t_{ij}')$ is a lower bound on the expected contribution to the welfare of the agent when we are able to choose which additive function to activate
$$\widetilde{w_{ij}}(t_{ij}')\leq \ind[t_{ij}'=t_{ij}]\E_{(S,k)\sim \widetilde{\cJ}}\left[\alpha_{ij}^{(k)}(t_{ij})\cdot \ind[j\in S]\right]$$.

Let $\cJ$ be the associated joint distribution over $2^{[m]}\times [K]$ of the element $(\pi_i,w_i)$.
We are going to show how to construct distribution $\widehat{\cJ}$ that certifies that $(\widehat{\pi}_i,\widehat{w}_i)\in W_i$.
We consider the distribution $\widehat{\cJ}$ that takes a sample from $2^{[m]}\times [K]$ in the following way:
\begin{enumerate}
    \item First, take a sample $(S,k)\sim \cJ$ and consider the set $\widehat{S}= S \backslash B$.
    \item For each $j\notin (S \cup B)$, take a sample $p_j$ from a Bernoulli distribution with probability of success $\frac{\widehat{\pi}_{ij}(t_{ij})-\Pr_{(S,k)\sim {\cJ}}[j\in S]}{1-\Pr_{(S,k)\sim {\cJ}}[j\in S]}$. If $p_j=1$, then $\widehat{S}=\widehat{S}\cup \{j\}$. Notice that since $\Pr_{(S,k)\sim {\cJ}}[j\in S]=\pi_{ij}(t_{ij})$, $0 \leq \frac{\widehat{\pi}_{ij}(t_{ij})-\pi_{ij}(t_{ij})}{1-\pi_{ij}(t_{ij})}\leq 1$.
    \item Return $(\widehat{S},k)$.
\end{enumerate}

Now we are going to show that $\widehat{\cJ}$ certifies, that $(\widehat{\pi}_i,\widehat{w}_i)\in W_i$.
Observe that for each $j\in B$,
$$
\Pr_{(S,k)\sim \widehat{\cJ}}[j\in S]=0
$$
which implies that for each $j\in B$ and $t_{ij}'\in \cT_{ij}$
$$\ind[t_{ij}'=t_{ij}]\Pr_{(S,k)\sim \widehat{\cJ}}[j\in S]= \ind[t_{ij}'=t_{ij}]\E_{(S,k)\sim \widehat{\cJ}}\left[\alpha_{ij}^{(k)}(t_{ij})\cdot \ind[j\in S]\right]=0$$

For each $j\notin B$,
notice that when we take a sample $(\widehat{S},\widehat{k})\sim \widehat{\cJ}$,
the probability that $j\in \widehat{S}$ is equal to
$$
\Pr_{(S,k)\sim \widehat{\cJ}}[j\in S]= \Pr_{(S,k)\sim \cJ}[j\in S] + (1-\Pr_{(S,k)\sim \widehat{\cJ}}[j\in S])\frac{\widehat{\pi}_{ij}(t_{ij})-\Pr_{(S,k)\sim \cJ}[j\in S]}{1-\Pr_{(S,k)\sim {\cJ}}[j\in S]}
= \widehat{\pi}_{ij}(t_{ij})$$
The first term is the probability that when we take a sample $(S,k)\sim \cJ$,
then $j \in S$ and the second term is the probability that $j\notin S$ times $\frac{\widehat{\pi}_{ij}(t_{ij})-\Pr_{(S,k)\sim \cJ}[j\in S]}{1-\Pr_{(S,k)\sim {\cJ}}[j\in S]}$.
Thus if
$j\notin B$,
and for any $t_{ij}'\in \cT_{ij}$
$$\ind[t_{ij}'=t_{ij}]\Pr_{(S,k)\sim \widehat{\cJ}}[j\in S]=\ind[t_{ij}'=t_{ij}] \widehat{\pi}_{ij}(t_{ij})$$

Moreover, by the construction of $\widehat{\cJ}$ for any $k^*\in[K]$ and $j\notin B$,
notice that
$$
\Pr_{(\widehat{S},\widehat{k})\sim \widehat{\cJ}} \left[\ind[j\in \widehat{S} \land k^*=\widehat{k}]\right] \geq \Pr_{({S},{k})\sim {\cJ}} \left[\ind[j\in S \land k^*=k]\right]
$$
which implies that for each $j\notin B$ and $t_{ij}'\in \cT_{ij}$
\begin{align*}
\ind[t_{ij}'=t_{ij}]\E_{(S,k)\sim \widehat{\cJ}}\left[\alpha_{ij}^{(k)}(t_{ij})\cdot \ind[j\in S]\right]\geq &
\ind[t_{ij}'=t_{ij}]\E_{(S,k)\sim {\cJ}}\left[\alpha_{ij}^{(k)}(t_{ij})\cdot \ind[j\in S]\right] \\
\geq & \ind[t_{ij}'=t_{ij}]{w_{ij}}(t_{ij}') \\
\geq & \ind[t_{ij}'=t_{ij}]\widehat{w_{ij}}(t_{ij}')    
\end{align*}

\end{proof}
}

\begin{lemma}\label{lem:contain subtypes}
Given any $\eps>0$ and any distribution $\widetilde{D}_i$ over $\cT_i$. Let $\widetilde{W}_i^{tr(\eps)}$ be a mixture of 
$\left\{W^{tr(\eps)}_{t_i}\right\}_{t_i\in \cT_i}$ over $\widetilde{D}_i$, {that is, $\widetilde{W}_i^{tr(\eps)}:=\sum_{t_i\in\cT_i} \Pr_{s\sim \widetilde{D}_i}[s=t_i] \cdot W^{tr(\eps)}_{t_i}$.} For each $j\in [m]$, let $S_j \subseteq \cT_{ij}$ be any set. For any $(\pi_i,w_i)\in \widetilde{W}_i^{tr(\eps)}$ and any $(\widehat{\pi}_i,\widehat{w}_i)\in [0,1]^{2\sum_{j\in[m]}|\cT_{ij}|}$ such that for each $j\in[m]$ and $t_{ij}'\in \cT_{ij}$,
\begin{align*}
\widehat{\pi}_{ij}(t_{ij}')=  \pi_{ij}(t_{ij}') \ind[t_{ij}'\in S_j]\quad \text{and} \quad \widehat{w}_{ij}(t_{ij}')\leq  w_{ij}(t_{ij}') \ind[t_{ij}'\in S_j].
\end{align*}
then $(\widehat{\pi}_i,\widehat{w}_i)\in \widetilde{W}_i^{tr(\eps)}$.
\end{lemma}

\begin{proof}

It suffices to prove the case where $\widehat{w}_{ij}(t_{ij}')= w_{ij}(t_{ij}') \ind[t_{ij}'\in S_j], \forall j, t_{ij}'$. This is because $\widetilde{W}_i^{tr(\eps)}$ is a mixture of $\{W^{tr(\eps)}_{t_i}\}_{t_i\in \cT_i}$. By \Cref{def:welfare frequent-t_i}, for every $t_i$ and any vector $(x,y)\in W^{tr(\eps)}_{t_i}$, we can decrease any $y(t_{ij}')$ while maintaining the vector $(x,y)$ to be in $W^{tr(\eps)}_{t_i}$. Thus, for any $(\widehat{\pi}_i,\widehat{w}_i)\in \widetilde{W}_i^{tr(\eps)}$, it remains in $\widetilde{W}_i^{tr(\eps)}$ after decreasing any $\widehat{w}_{ij}(t_{ij}')$.  

For each $t_i\in \cT_i$, let $\left(\pi^{(t_i)}_i,w^{(t_i)}_i\right)\in W_{t_i}^{tr(\eps)}$ such that $\left(\pi_i,w_i\right) = \sum_{t_i\in \cT_i}\Pr_{s\sim \widetilde{D}_i}\left[s=t_i\right]\cdot\left(\pi^{(t_i)}_i,w^{(t_i)}_i\right)$. Consider vector $\left(\widehat{\pi}^{(t_i)}_i,\widehat{w}^{(t_i)}_i\right)$ such that for every $j\in [m]$ and $t_{ij}'\in\cT_{ij}$
$$
\widehat{\pi}^{(t_i)}_{ij}(t_{ij}')=  \pi_{ij}^{(t_i)}(t_{ij}') \ind[t_{ij}\in S_j]\quad \text{and} \quad \widehat{w}_{ij}^{(t_i)}(t_{ij}')=  w_{ij}^{(t_i)}(t_{ij}') \ind[t_{ij}\in S_j].
$$


{For each $t_i \in \cT_i$, define set $B(t_i):=\{j: t_{ij} \in S_j\}$,
by applying Lemma~\ref{lem:W_t_i modify} to  $\left(\pi^{(t_i)}_i,w^{(t_i)}_i\right)$ and set $B(t_i)$},
we have that $\left(\widehat{\pi}^{(t_i)}_i,\widehat{w}^{(t_i)}_i \right)\in W_{t_i}^{tr(\eps)}$.

We notice that by \Cref{def:welfare frequent-t_i}, $w_{ij}^{(t_i)}(t_{ij}')=\pi_{ij}^{(t_i)}(t_{ij}')=0$ if $t'_{ij}\not=t_{ij}$. Thus, for every $j\in [m]$ and $t'_{ij}\in \cT_{ij}$, 
$$
\widehat{\pi}^{(t_i)}_{ij}(t_{ij}')=  \pi_{ij}^{(t_i)}(t_{ij}') \ind[t_{ij}'\in S_j] \quad \text{and} \quad \widehat{w}_{ij}^{(t_i)}(t_{ij}')=  w_{ij}^{(t_i)}(t_{ij}') \ind[t_{ij}'\in S_j].
$$
The proof concludes by noticing that $\left(\widehat{\pi}_i,\widehat{w}_i\right) = \sum_{t_i\in \cT_i}\Pr_{s\sim \widetilde{D}_i}\left[s=t_i\right]\cdot\left(\widehat{\pi}^{(t_i)}_i,\widehat{w}^{(t_i)}_i\right)\in \widetilde{W}_i^{tr(\eps)}$.
\end{proof}

The following corollary follows from Observation~\ref{obs:mix} and Lemma~\ref{lem:W_t_i modify}.
\begin{corollary}\label{obs:contained trauncated xos}
For any $\eps>0$, $W_i^{tr(\eps)}\subseteq W_i$.
\end{corollary}

\begin{proof}
This follows easily from the claim that $W_{t_i}^{tr(\eps)}\subseteq W_{t_i}$ (\Cref{lem:W_t_i modify}).
\end{proof}

\notshow{
\begin{proof}
Let $(\pi_i,w_i)\in W_i^{tr(\eps)}$,
by Observation~\ref{obs:mix} for each $t_i\in \cT_i$, there exists $\left(\pi^{(t_i)}_i,w^{(t_i)}_i\right)\in W_{t_i}^{tr(\eps)}$ such that $(\pi_i,w_i)= \sum_{t_i\in \cT_i}f_i(t_i)\cdot(\pi^{(t_i)}_i,w^{(t_i)}_i)$.
In order to prove the claim,
it suffices to show that $\left( \pi^{(t_i)}_i,w^{(t_i)}_i\right)\in W_{t_i}$. Then $(\pi_i,w_i)\in W_i$ due to Observation~\ref{obs:mix}.

For each $t_i\in \cT_i$,
by Definition~\ref{def:welfare frequent-t_i},
there exists a $\left(\widehat{\pi}^{(t_i)}_i,\widehat{w}^{(t_i)}_i \right)\in W_{t_i}$ such that for each $j\in[m], t_{ij}\in \cT_{ij}$,
\begin{align*}
{\pi}_{ij}^{(t_i)}(t_{ij}) = \widehat{\pi}_{ij}^{(t_i)}(t_{ij})\ind[f_{ij}(t_{ij})\geq \eps],\quad
{w}_{ij}^{(t_i)}(t_{ij}) = \widehat{w}_{ij}^{(t_i)}(t_{ij})\ind[f_{ij}(t_{ij})\geq \eps]
\end{align*}

By applying Lemma~\ref{lem:W_t_i modify} to each $\left(\widehat{\pi}^{(t_i)}_i,\widehat{w}^{(t_i)}_i \right)$ with set $B=\{j\in [m]: f_{ij}(t_{ij})\geq\eps\}$,
we conclude that $\left({\pi}^{(t_i)}_i,{w}^{(t_i)}_i \right)\in W_{t_i}$. Thus $(\pi_i,w_i)\in W_i$. 
\end{proof}
}

Similar to the constrained-additive case, we define the $\eps$-box polytope of $W_i$ for XOS valuations in Definition~\ref{def:welfare rare}. 

\begin{definition}\label{def:welfare rare}
For $\eps>0$,
we denote as $W_i^{box(\eps)}\subseteq [0,1]^{2\sum_{j\in[m]}|\cT_{ij}|}$ the $\eps$-box polytope of $W_i$: $(\pi_i,w_i)\in W_i^{box(\eps)}$ if and only if for every $j\in[m]$ and  $t_{ij}\in \cT_{ij}$ it holds that
\begin{align*}
    0 \leq w_{ij}(t_{ij}) \leq \pi_{ij}(t_{ij}) \leq \min(\eps,f_{ij}(t_{ij}))
\end{align*}
\end{definition}

\notshow{
I don't think they are no longer needed

\begin{observation}\label{obs:W_i(t_i) null}
Let $\textbf{0}$ be the vector with 0 in all entries. For any $t_i$ and $\pi_i\in [0,1]^{M_i}$such that $\pi_{ij}(t_{ij}')\leq \ind[t_{ij}'=t_{ij}]$, $(\pi_i,\textbf{0})\in W_{t_i}$. 
\end{observation}

\begin{proof}
We are going to choose numbers $\{\sigma_S^k\}_{S\subseteq[m],k\in[K]}$ and use Definition~\ref{def:welfare frequent-t_i} to prove the statement.

We are going to choose numbers $\{\sigma_S^k(t_i)\}_{t_i\in \cT_i,S\subseteq[m],k\in[K]}$ and use Definition~\ref{def:welfare frequent-t_i} to prove the statement.
Fix a $k^*\in [K]$ (the choice of $k^*$ is not going to matter in our proof).
For each $S\subseteq [m]$ and $k\in[K]$
$$
\sigma_S^{k} = \prod_{j\in S}\pi_{ij}(t_{ij}) \prod_{j\notin S}(1-\pi_{ij}(t_{ij}))
\ind[k=k^*]$$
Now we are going to prove the first property in Definition~\ref{def:welfare frequent-t_i}.
Notice that
$$
\sum_{S\subseteq[m]}\sum_{k\in[K]}\sigma_S^{k} =
\sum_{S\subseteq[m]}\sigma_S^{k^*} =\sum_{S\subseteq[m]} \prod_{j\in S}\pi_{ij}(t_{ij})\prod_{j\notin S}(1-\pi_{ij}(t_{ij}))
$$
Consider a random variables $C_j$ that is sampled from a Bernoulli Distribution, with probability of success $0\leq \pi_{ij}(t_{ij})\leq 1$.
Then $\prod_{j\in S}\pi_{ij}(t_{ij})\prod_{j\notin S}(1-\pi_{ij}(t_{ij}))$ 
is equal to the probability that for each $j\in S$, $C_j=1$ and for $j\notin S$, $C_j=0$,
which implies that $\sum_{S\subseteq[m]} \prod_{j\in S}\pi_{ij}(t_{ij})\prod_{j\notin S}(1-\pi_{ij}(t_{ij}))= 1$.

Now we prove that $\sigma$ satisfied the second property of Definition ~\ref{def:welfare frequent-t_i}.
For each $j\in[m]$,
observe that 
\begin{align*}
\sum_{ S: k \in S}\sum_{k\in[K]}\sigma_S^{k} =&
\sum_{S:j\in S}\sigma_S^{k^*} \\ =& \pi_{ij}(t_{ij})\sum_{S: j \in S} \prod_{j'\in S \land j'\neq j}\pi_{ij'}(t_{ij'})
\prod_{j'\notin S}(1-\pi_{ij'}(t_{ij'})) \\
=& \pi_{ij}(t_{ij})
\end{align*}

where the last equality follows by considering random variable $C_j$ as described above.
Then for $S:j \in S$, $\prod_{j'\in S\land j'\neq j}\pi_{ij'}(t_{ij'})\prod_{j'\notin S}(1-\pi_{ij'}(t_{ij'}))$ 
is equal to the probability that for each $j\in S\backslash\{j\}$, $C_j=1$ and for $j\notin S$, $C_j=0$,
which implies that $\sum_{S: j \in S} \prod_{j'\in S \land j'\neq j}\pi_{ij'}(t_{ij'})
\prod_{j'\notin S}(1-\pi_{ij'}(t_{ij'}))= 1$.

The third property of Definition~\ref{def:welfare frequent-t_i} follow trivially since the LHS is $0$.
\end{proof}

A corollary of Observation~\ref{obs:W_i(t_i) null} is the following.
\begin{corollary}\label{obs:W_i(t_i)-1}
Let $\textbf{0}$ be the vector with 0 in all entries. For any $t_i$ and $\pi_i\in [0,1]^{M_i}$such that $\pi_{ij}(t_{ij}')\leq \ind[t_{ij}'=t_{ij}\wedge f_{ij}(t_{ij})\geq \eps]$, $(\pi_i,\textbf{0})\in W_{t_i}^{tr(\eps)}$. 
\end{corollary}

\begin{proof}
We Choosing $\sigma_{S}^k=0$ for all $S,k$ finishes the proof.
\end{proof}

\begin{definition}\label{def:welfare rare}
For $\eps>0$,
we denote as $W_i^{box(\eps)}\subseteq [0,1]^{2M_i}$ the $\eps$-box polytope of $W_i$,
such that $(\widehat{\pi}_i,\widehat{w}_i)\in W_i^{box(\eps)}$ if and only if for every $j\in[m]$ and  $t_{ij}$ it holds that
\begin{align*}
    0 \leq w_{ij}(t_{ij}) \leq \pi_{ij}(t_{ij}) \leq \cdot\min(\eps,f_{ij}(t_{ij}))
\end{align*}
\end{definition}


Clearly the first polytope $W_i^{tr(\eps)}\subseteq W_i$. 
We prove in the following lemma that the second polytope $\left(1-\eps \cdot M_i \right) W_i^{box(\eps)} \subseteq W_i, \forall \eps<\nicefrac{1}{M_i}$. 

}
\notshow{
\begin{observation}\label{obs:box within a box}
For any $\eps>0$ and $a'\geq a >0$,
$a W_i^{box(\eps)}\subseteq a' W_i^{box(\eps)}$
\end{observation}
}

The following lemma is similar to \Cref{lem:box polytope constraint additive}. 

\begin{lemma}\label{lem:welfare rare}
Let $T =\sum_{i\in[n]}\sum_{j\in[m]}|\cT_{ij}|$, then $\left(1-\eps \cdot T \right) W_i^{box(\eps)} \subseteq  W_i$, for all $\eps < \nicefrac{1}{T}$.
\end{lemma}
\begin{proof}

For any $(\pi_i,w_i)\in W_i^{box(\eps)}$, we will prove that $(\pi_i,\pi_i)\in W_i$. Then $(\pi_i,w_i)\in W_i$ since $w_{ij}(t_{ij})\leq \pi_{ij}(t_{ij})$ for any $j,t_{ij}$.

To prove $(\pi_i,\pi_i)\in W_i$, we consider the following set of numbers $\{\sigma_{S}^{(k)}(t_i)\}_{t_i,S,k}$ (see \Cref{def:W_i}): For each $j\in [m], t_{ij}\in \cT_{ij}$, let $c_j(t_{ij})=\min\left(\frac{\eps}{f_{ij}(t_{ij})},1 \right)$ and
$$p_j(t_{ij})=\frac{\pi_{ij}(t_{ij})}{f_{ij}(t_{ij})\cdot \sum_{t_{i,-j}}f_{i,-j}(t_{i,-j})\cdot c_j(t_{ij})\cdot \prod_{j'\not=j}(1-c_{j'}(t_{ij'}))}.$$ 

{Note that for every $j'\in[m]$, there exists a  $t_{ij'}\in\cT_{ij'}$ such that $f_{ij'}(t_{ij'})\geq 1/|\cT_{ij'}|$. Due to our choice of $\eps$, the corresponding $c_{j'}(t_{ij'})<1$. Hence, $\sum_{t_{i,-j}}f_{i,-j}(t_{i,-j})\cdot \prod_{j'\not=j}(1-c_{j'}(t_{ij'}))>0$, and $p_j(t_{ij})$  is well-defined.}

For every $t_i$, define~\footnote{When there are two indices $k,k'\in [K]$ such that $k,k^*\in \argmax_{k'}\alpha_{ij}^{(k')}(t_{ij})$, we break ties in lexicographic order.}
\begin{align*}
\sigma_{S}^{(k)}(t_i)=
\begin{cases}
p_j(t_{ij})\cdot c_j(t_{ij})\cdot \prod_{j'\not=j}(1-c_{j'}(t_{ij'})), &\text{if } S=\{j\} \text{ and } k=\argmax_{k'}\alpha_{ij}^{(k')}(t_{ij})\\
0, &\text{o.w.}
\end{cases}
\end{align*}

For every $j$, let $C_j(t_{ij})$ be the independent Bernoulli random variable that activates with probability $c_j(t_{ij})$. Then for every $j$, 
$$
\Pr_{C_j(t_{ij}), t_{ij}\sim D_{ij}}\left[C_j(t_{ij})=1\right] \leq \sum_{t_{ij}\in \cT_{ij}} f_{ij}(t_{ij})\cdot\min\left(\frac{\eps}{f_{ij}(t_{ij})},1 \right) \leq \eps \cdot |\cT_{ij}|
$$
By the union bound,
$$\sum_{t_{i,-j}}f_{i,-j}(t_{i,-j})\cdot \prod_{j'\not=j}(1-c_{j'}(t_{ij'}))=\Pr_{\substack{t_{i,-j}\sim\cD_{i-j}\\ \forall j' \neq j, C_{j'}(t_{ij'})}}[C_{j'}(t_{ij'}) =0,\forall j'\not=j]\geq 1-\sum_j\Pr[C_j(t_{ij})=1]\geq 1-\eps\cdot T$$

Now we prove that $(\pi_i,\pi_i)\in W_i$ by verifying all three conditions in \Cref{def:W_i}. For the first condition, since $(\pi_i,w_i)\in (1-\eps T)W_i^{box(\eps)}$, $0 \leq w_{ij}(t_{ij}) \leq \pi_{ij}(t_{ij}) \leq \left(1-\eps \cdot T \right)\cdot f_{ij}(t_{ij})\cdot c_j(t_{ij})$. Thus $p_j(t_{ij})\leq 1$ for every $j,t_{ij}$. We have that $\sum_{S,k}\sigma_S^{(k)}(t_i)=\sum_jp_j(t_{ij})\cdot c_j(t_{ij})\cdot \prod_{j'\not=j}(1-c_{j'}(t_{ij'}))\leq \sum_j c_j(t_{ij})\cdot \prod_{j'\not=j}(1-c_{j'}(t_{ij'}))$. We notice that $\sum_j c_j(t_{ij})\cdot \prod_{j'\not=j}(1-c_{j'}(t_{ij'}))$ is exactly the probability that there exists a unique $C_j(t_{ij})=1$. Thus $\sum_{S,k}\sigma_S^{(k)}(t_i)\leq 1$ for all $t_i\in \cT_i$.

The second condition: 
\begin{align*}
&f_{ij}(t_{ij})\cdot\sum_{t_{i,-j}}f_{i,-j}(t_{i,-j})\cdot \sum_{S:j\in S}\sum_{k\in [K]}\sigma_S^{(k)}(t_{ij},t_{i,-j})\\
=&f_{ij}(t_{ij})\sum_{t_{i,-j}}f_{i,-j}(t_{i,-j})\cdot p_j(t_{ij})\cdot c_j(t_{ij})\cdot \prod_{j'\not=j}(1-c_{j'}(t_{ij'}))=\pi_{ij}(t_{ij})
\end{align*}

The third condition:
\begin{align*}
&f_{ij}(t_{ij})\cdot\sum_{t_{i,-j}}f_{i,-j}(t_{i,-j})\cdot \sum_{S:j\in S}\sum_{k\in [K]}\sigma_S^{(k)}(t_{ij},t_{i,-j})\cdot \frac{\alpha_{ij}^{(k)}(t_{ij})}{V_{ij}(t_{ij})}\\
=&f_{ij}(t_{ij})\sum_{t_{i,-j}}f_{i,-j}(t_{i,-j})\cdot p_j(t_{ij})\cdot c_j(t_{ij})\cdot \prod_{j'\not=j}(1-c_{j'}(t_{ij'}))\cdot \frac{\max_k\alpha_{ij}^{(k)}(t_{ij})}{V_{ij}(t_{ij})}=\pi_{ij}(t_{ij})
\end{align*}

By \Cref{def:W_i}, $(\pi_i,\pi_i)\in W_i$. Thus, $(\pi_i,w_i)\in W_i$.
\end{proof}

\vspace{.2in}

Similar to \Cref{lem:boxable polytope contained}, we prove in the following lemma that
$W_i$ can be approximated by the polytope $\frac{1}{2}W_i^{tr(\eps)} + \frac{1-\eps T}{2}W_i^{box(\eps)}$ within a multiplicative factor.

\begin{lemma}\label{lem:xos mrf contain}
Let $T=\sum_{i\in[n]}\sum_{j\in[m]}|\cT_{ij}|$ and any $0\leq \eps<\frac{1}{T}$, then $$\frac{1-\eps T}{2}W_i \subseteq \frac{1}{2}W_i^{tr(\eps)} + \frac{1-\eps T}{2}W_i^{box(\eps)} \subseteq  W_i.$$
\end{lemma}
\begin{proof}
Let $W'=\frac{1}{2}W_{i}^{tr(\eps)} + \frac{1-\eps T}{2}W_{i}^{box(\epsilon)}$. Then $W'\subseteq  W_i$ directly follows from Corollary~\ref{obs:contained trauncated xos} and Lemma~\ref{lem:welfare rare}. We are going to show that for each $(\pi_i,w_i) \in \frac{1-\eps \cdot T}{2}\cdot W_i$, $(\pi_i,w_i) \in W'$. We consider the following vector $(\pi_{i}^{tr}, w_{i}^{tr})$ such that for every $j\in [m], t_{ij}\in \cT_{ij}$,
\begin{align*}
\pi_{ij}^{tr}(t_{ij})=
\pi_{ij}(t_{ij})\ind[f_{ij}(t_{ij})\geq \eps], \qquad
w_{ij}^{tr}(t_{ij})=
w_{ij}(t_{ij})\ind[f_{ij}(t_{ij})\geq \eps]
\end{align*}
Then by \Cref{def:welfare frequent} and Observation~\ref{obs:bound xos mrf}, $(\pi_{i}^{tr}, w_{i}^{tr})\in \frac{1-\eps\cdot T}{2}W_{i}^{tr(\eps)}\subseteq \frac{1}{2}W_{i}^{tr(\eps)}$. Consider the vector $(\pi_{i}^{box}, w_{i}^{box})$ such that
\begin{align*}
\pi_{ij}^{box}(t_{ij})=
\pi_{ij}(t_{ij})\ind[f_{ij}(t_{ij})< \eps] , \qquad
w_{ij}^{box}(t_{ij})=
w_{ij}(t_{ij})\ind[f_{ij}(t_{ij}) < \eps]
\end{align*}
For every $j\in [m], t_{ij}\in \cT_{ij}$, since $(\pi_i,w_i)\in \frac{1-\eps \cdot T}{2}\cdot W_i$, by \Cref{obs:W_i-upper-bound} we have $w_{ij}(t_{ij})\leq \pi_{ij}(t_{ij})\leq \frac{1-\eps \cdot T}{2}\cdot f_{ij}(t_{ij})$. Thus $w_{ij}^{box}(t_{ij})\leq \pi_{ij}^{box}(t_{ij})\leq \frac{1-\eps \cdot T}{2}\cdot f_{ij}(t_{ij})\cdot \ind[f_{ij}(t_{ij})<\eps]\leq \frac{1-\eps \cdot T}{2}\cdot \min(f_{ij}(t_{ij}),\eps)$. Thus $(\pi_{i}^{box}, w_{i}^{box})\in \frac{1-\eps T}{2}W_{i}^{box(\eps)}$. Now observe that $(\pi_i,w_i) = (\pi^{tr}_i + \pi^{box}_i,w^{tr}_i + w^{box}_i)\in \frac{1}{2}W_{i}^{tr(\eps)} + \frac{1-\eps T}{2}W_{i}^{box(\eps)}$,
which concludes the proof.
\end{proof}

\subsection{Approximating the Single-Bidder Marginal Reduced Form Polytope}

In this section we provide the important step that allows us to prove \Cref{thm:main XOS}. In the constrained-additive case (\Cref{thm:multiplicative approx for constraint additive-main body}), we construct (with high probability) another polytope $\widehat{W}_i$ such that (i) $c\cdot W_i\subseteq \widehat{W}_i\subseteq W_i$ for some absolute constant $c>0$, and (ii) we have an efficient separation oracle for $\widehat{W}_i$. The proof heavily relies on the fact that $W_i$ is down-monotone and thus cannot be easily extended  to the single-bidder marginal reduced form polytope $W_i$ in the XOS case. Here we prove a similar statement with a weaker version of property (i): For every vector $x$ in $W_i$, there exists another vector $x'$ in $\widehat{W}_i$ such that for every coordinate $j$, $x_j/x_j'\in [a,b]$ for some absolute constant $0<a<b$, and vice versa. It's not hard to see that the original property (i) directly implies this weaker version. A formal statement is shown in \Cref{cor:rfs for xos mult}.

\begin{theorem}\label{cor:rfs for xos mult}
Let $T=\sum_{i\in[n]}\sum_{j\in[m]}|\cT_{ij}|$ and $b$ be the bit complexity of the problem instance (\Cref{def:bit complexity}). For any $i\in [n]$ and $\delta>0$, there is an algorithm that constructs a convex polytope $\widehat{W}_i\in [0,1]^{2\cdot\sum_{j\in [m]}|\cT_{ij}|}$ using $\poly(n,m,T,\log(1/\delta))$ samples from $D_i$, such that with probability at least $1-\delta$ all of the following are satisfied:

\begin{enumerate}
    \item For each $(\widetilde{\pi}_i,\widetilde{w}_i) \in \widehat{W}_i$,
there exists a $(\pi_i,w_i) \in W_i$ such that for all $j\in[m]$ and $t_{ij} \in \cT_{ij}$:

$$\frac{{\pi}_{ij}(t_{ij})}{\widetilde{\pi}_{ij}(t_{ij})}\in \left[\frac{1}{4},\frac{3}{2}\right],\quad \frac{{w}_{ij}(t_{ij})}{\widetilde{w}_{ij}(t_{ij})}\in \left[\frac{1}{4},\frac{5}{4}\right]$$

\item For each $(\pi_i,w_i) \in W_{i}$,
there exists a $(\widetilde{\pi}_i,\widetilde{w}_i) \in \widehat{W}_i$ such that for all $j\in[m]$ and $t_{ij} \in \cT_{ij}$:

$$\frac{\widetilde{\pi}_{ij}(t_{ij})}{\pi_{ij}(t_{ij})}\in \left[\frac{1}{16},\frac{3}{8}\right],\quad \frac{\widetilde{w}_{ij}(t_{ij})}{w_{ij}(t_{ij})}\in \left[\frac{1}{16},\frac{5}{16}\right]$$

\item The vertex-complexity (\Cref{def:vertex-complexity}) of $\widehat{W}_i$ is $\poly(n,m,T,b,\log(1/\delta))$. 

\item There exists a separation oracle $SO$ for $\widehat{W}_i$, given access to the value oracle and an adjustable demand oracle (\Cref{def:adjustable demand oracle}) of buyer $i$'s valuation. The running time of $SO$ on any input with bit complexity $b'$ is $\poly(n,m,T,b,b',\log(1/\delta))$ and makes $\poly(n,m,T,b,b',\log(1/\delta))$ queries to both oracles.
\end{enumerate}

The algorithm constructs the polytope and the separation orcle $SO$ in time 
$\poly(n,m,T,b,\log(1/\delta))$.
\end{theorem}

\subsubsection{Efficiently Optimizing over the Polytope}

Both \Cref{obs:bit complexity mrf xos} and \Cref{lem:optimize} are useful to show that there exists an efficient separation oracle for our constructed polytope.

In \Cref{lem:optimize}, we show that given access to the adjustable demand oracle and the value oracle, we can optimize any linear objective over the polytope $W_{t_i}^{tr(\eps)}$ (\Cref{def:welfare frequent-t_i}).

\begin{lemma}\label{lem:optimize}
Let $T=\sum_{i\in[n]}\sum_{j\in[m]}|\cT_{ij}|$,
and $b$ be the bit complexity of the instance.
For any $t_i\in \cT_i$ and any $\eps>0$, given access to the adjustable demand oracle $\adem_i(\cdot,\cdot,\cdot)$ and the value oracle for buyer $i$'s valuation $v_i(\cdot,\cdot)$, there exists an algorithm that takes arbitrary $\bm{x},\bm{y}\in \mathbb{R}^{\sum_j|\cT_{ij}|}$ as input and finds $$(\pi_i^*,w_i^*)\in \text{argmax}_{(\pi_i,w_i)\in W_{t_i}^{tr(\eps)}}\bm{x}\cdot \pi_i+ \bm{y}\cdot w_i.$$ The algorithm runs in time $poly(n,m,T,b,b',1/\eps)$ and makes $O(m)$ queries to both oracles, where $b'$ is the bit complexity of $(\bm{x},\bm{y})$. Moreover the bit complexity of $(\pi_i^*,w_i^*)$ is at most $O(b T)$.
\end{lemma}

\begin{proof}
We are going to solve the problem:
\begin{align*}
    \max&\quad \bm{x}\cdot \pi_i + \bm{y}\cdot w_i \\
    s.t.&\quad (\pi_i,w_i) \in W_{t_i}^{tr(\eps)}
\end{align*}
By \Cref{def:welfare frequent-t_i}, $\pi_{ij}(t_{ij}')=w_{ij}(t_{ij}')=0$ if $t_{ij}'\not=t_{ij}$ or $f_{ij}(t_{ij})<\eps$. Let $Q=\{j\in[m]: f_{ij}(t_{ij})\geq \eps\}$. Then according to \Cref{def:welfare frequent-t_i}, the problem is equivalent to
\begin{align*}
  (O_1)\quad   \max &\sum_{j\in Q}\left( x_j(t_{ij})\cdot \pi_{ij}(t_{ij}) + y_j(t_{ij})\cdot w_{ij}(t_{ij})  \right)\\
    s.t.& \sum_{S\subseteq [m]}\sum_{k\in[K]}\sigma_S^{(k)} \leq 1 \\
    & \pi_{ij}(t_{ij}) = \sum_{S:j\in S}\sum_{k\in[K]}\sigma_S^{(k)} \quad &\forall j\in[m] \\
    &0\leq w_{ij}(t_{ij}) \leq \sum_{S:j\in S}\sum_{k\in[K]}\sigma_S^{(k)} \frac{\alpha^{(k)}_{ij}(t_{ij})}{V_{ij}(t_{ij})} &\forall j\in[m] \\
    &\sigma_S^{(k)}\geq 0  &\forall S\subseteq [m], k\in[K]
\end{align*}

We notice that at the maximum, $w_{ij}(t_{ij})$ must be equal to $\sum_{S:j\in S}\sum_{k\in[K]}\sigma_S^{(k)} \frac{\alpha^{(k)}_{ij}(t_{ij})}{V_{ij}(t_{ij})}$ if $y_j(t_{ij})>0$, and $0$ otherwise. $(O_1)$ is equivalent to (denote $y_j(t_{ij})^+=\max\{y_j(t_{ij}),0\}$):
\begin{align*}
    (O_2)\quad\max &\sum_{j\in Q}\left( x_j(t_{ij})\cdot \pi_{ij}(t_{ij}) + y_j(t_{ij})^+\cdot w_{ij}(t_{ij})  \right)\\
    s.t.& \sum_{S\subseteq [m]}\sum_{k\in[K]}\sigma_S^{(k)} \leq 1 \\
    & \pi_{ij}(t_{ij}) = \sum_{S:j\in S}\sum_{k\in[K]}\sigma_S^{(k)} \quad &\forall j\in[m] \\
    &w_{ij}(t_{ij}) = \sum_{S:j\in S}\sum_{k\in[K]}\sigma_S^{(k)} \frac{\alpha^{(k)}_{ij}(t_{ij})}{V_{ij}(t_{ij})} &\forall j\in[m]\\
    &\sigma_S^{(k)}\geq 0  &\forall S\subseteq [m], k\in[K]
\end{align*}

or equivalently
\begin{align*}
   (O_3)\quad \max &\sum_{S\subseteq [m]} \sum_{k\in[K]}  \sigma_S^{(k)}\left(\sum_{j\in S\cap Q}\left( y_j(t_{ij})^+\frac{\alpha_{ij}^{(k)}(t_{ij})}{V_{ij}(t_{ij})} + x_j(t_{ij})\right) \right)\\
    s.t.& \sum_{S\subseteq [m]}\sum_{k\in[K]}\sigma_S^{(k)} \leq 1 \\
    &\sigma_S^{(k)}\geq 0  \qquad \qquad \qquad \qquad \qquad S\subseteq [m], k\in[K]
\end{align*}


Clearly, to solve $(O_3)$, it suffices to find $S^*\subseteq Q$ and $k^*\in [K]$ that lies in $$\argmax_{S\subseteq Q, k\in[K]}\sum_{j\in S}\left(y_j(t_{ij})^+\cdot\frac{\alpha_{ij}^{(k)}(t_{ij})}{V_{ij}(t_{ij})} + x_j(t_{ij})\right).$$
We notice that since $y_j(t_{ij})^+\cdot\frac{\alpha_{ij}^{(k)}(t_{ij})}{V_{ij}(t_{ij})}\geq 0$ for every $j,k$, thus $j\in S^*$ for every $j$ such that $x_j(t_{ij})\geq 0$. Consider the vector $\textbf{b}\in \mathbb{R}^m_+$ and $\textbf{p}\in \mathbb{R}^m_+$ such that for each $j\in[m]$,
\begin{align*}
b_j=\frac{y_j(t_{ij})^+}{V_{ij}(t_{ij})},\quad
p_j =
\begin{cases}
-x_j(t_{ij})\ind[x_j(t_{ij})<0], & j\in Q\\
\infty, & j\not\in Q
\end{cases}
\end{align*}

Then $(S^*,k^*)\in \argmax\sum_{j\in S}\left( b_j\alpha_{ij}^{(k)}(t_{ij}) - p_j \right)$
, which can be achieved by a single query to the adjustable demand oracle with input $(t_i,\textbf{b},\textbf{p})$. Now the corresponding vector $(\pi_i^*,w_i^*)$ can be computed according to $(O_2)$, with $\sigma_S^{(k)}=\ind[S=S^*\cup S_{+} \land k=k^*]$, where $S_{+}:=\{j\in[m]: x_{j}(t_{ij})>0\}$. ~\footnote{Both $\alpha^{(k^*)}_{ij}(t_{ij})$ and $V_{ij}(t_{ij})=\max_{k}\alpha^{(k)}_{ij}(t_{ij})$ can be computed with $O(1)$ queries to the adjustable demand oracle.}

\notshow{

Let $(S^*,k^*)$ be the optimal solution of $(O_3)$.
Note that the solution $\widehat{\sigma}_{S}^{k}=\ind[S=S^*\land k=k^*]$ is an optimal solution of $(O_2)$.
Consider the set $S_\pi=\{j\in [m]: c_\pi(t_{ij})>0\}$ and the following program
\begin{align*}
(O_4)  \quad  {\arg\max}_{S\subseteq [m], k\in [K]} &\sum_{j\in S}\left( \frac{c_{w}(t_{ij})^+}{V_j(t_j)}\alpha_j^{(k)}(t_j) + c_{\pi}(t_{ij})\ind[c_{\pi}(t_{ij}) \leq 0]\right)
\end{align*}

Let $(S^*,k^*)$ be the optimal solution of $(O_4)$.
Observe that $(S^*\cup S_\pi,k^*)$ is an optimal solution of $(O_3)$.
Consider the vector $\textbf{b}\in \mathbb{R}^m_+$ and $\textbf{p}\in \mathbb{R}^m_+$ such that for each $j\in[m]$.
\begin{align*}
\textbf{b}_j =& \frac{c_w(t_{ij})^+}{V_{ij}(t_{ij})}\\
\textbf{p}_j =& -c_{\pi}(t_{ij})\ind[c_{\pi}(t_{ij})\leq 0]
\end{align*}

Observe we can rewrite Program $(O_4)$

\begin{align*}
(O_5)  \quad  {\arg\max}_{S\subseteq [m], k\in [K]} &\sum_{j\in S}\left( \textbf{b}_j\alpha_j^{(k)}(t_j) - \textbf{p}_j \right)
\end{align*}

Also observe that we can find $(S^*,k^*)$ by a querying a scalable demand oracle with input $\dem(t_i,\textbf{b},\textbf{p})$.
Given our reasoning above,
having values $(S^*,k^*)$ that optimize Program $(O_5)$,
we can find $(\pi^*_i,w_i^*)\in W_i$ that maximizes Program $(O_1)$,
where
\begin{align*}
\pi_{ij}^*(t_{ij}') =& \ind[t_{ij}'=t_{ij}\land j\in S_\pi \cup S^*] \\
w_{ij}^*(t_{ij}') =& \frac{\alpha_{ij}^{(k^*)}(t_{ij}')}{V_{ij}(t_{ij}')}\ind[t_{ij}'=t_{ij}\land j\in S^*]
\end{align*}
where we can calculate the values of $\alpha^{(k^*)}_{ij}(t_{ij})$ and $V_{ij}(t_{ij})=\max_{k}\alpha^{(k)}_{ij}(t_{ij})$ by one more call to the adjastuble demand oracle and two accesses to value oracles.
The bit complexity of each coordinate of $(\pi_i^*,w_i^*)$ is at most $O(b)$,
which implies that the bit complexity of $(\pi_i^*,w_i^*)$ is at most $O(b T)$.

We can similarly calculate 
\begin{align*}
    \arg\max_{(\pi,w)} c_\pi^T\pi^{tr} + c_w^T w^{tr} \\
    s.t.: (\pi^{tr},w^{tr}) \in W^{tr(\eps)}_{t_i}
\end{align*}

}

Finally, the bit complexity of each coordinate of $(\pi_i^*,w_i^*)$ is at most $O(b)$,
which implies that the bit complexity of $(\pi_i^*,w_i^*)$ is at most $O(b T)$.
\end{proof}

\begin{corollary}\label{obs:bit complexity mrf xos}
Let $T=\sum_{i\in[n]}\sum_{j\in[m]}|\cT_{ij}|$ and $b$ be the bit complexity of the instance.
For any $t_i\in \cT_i$, 
$W_{t_i}^{tr(\eps)}$ has vertex-complexity $O(bT)$. 

\end{corollary}

\subsubsection{Proof of \Cref{cor:rfs for xos mult}}

\notshow{
\begin{theorem}\label{thm:multi-approx-mixture}
Let $T=\sum_{i\in[n]}\sum_{j\in[m]}|\cT_{ij}|$ and $b$ be the bit complexity of the problem instance (\Cref{def:bit complexity}).
Consider parameters $0 < \eps \leq 1/T$, $\eta\in (\eps T,1)$ \argyrisnote{Argyris: I think we only need that $\eta \in [0,1]$ now    } and
integer $k\geq \Omega\left( T^4\left( Tb+\log\left(\frac{1}{\eta \eps}\right)\right) \right)$.
If we have access to $n=\left\lceil\frac{8kT}{(\eta\eps)^2}\right\rceil$ independent samples from $\cD$,
then with probability at least $1-2T\exp(-2T k)$,
we can construct a polytope $\widehat{W}_i$ such that

\begin{enumerate}
    \item For each $(\widetilde{\pi}_i,\widetilde{w}_i) \in \widehat{W}_i$,
there exists a $(\pi_i,w_i) \in \frac{1}{2}{W}_{i}^{tr(\eps)} + \frac{1-\eps T}{2}W^{box(\eps)}_{i}$ such that for all $j\in[m]$ and $t_{ij} \in \cT_{ij}$:

$$\frac{{\pi}_{ij}(t_{ij})}{\widetilde{\pi}_{ij}(t_{ij})}\in \left[\frac{1}{2}-\frac{\max\left(\eps\cdot  T,\eta\right)}{2}, \frac{3}{2}+\frac{\eta}{2}\right],\quad \frac{{w}_{ij}(t_{ij})}{\widetilde{w}_{ij}(t_{ij})}\in \left[\frac{1}{2}-\frac{\max\left(\eps\cdot  T,\eta\right)}{2} ,1+\frac{\eta}{2}\right]$$

\item For each $(\pi_i,w_i) \in \frac{1}{2}{W}_{i}^{tr(\eps)} + \frac{1-\eps T}{2}W^{box(\eps)}_{i}$,
there exists a $(\widetilde{\pi}_i,\widetilde{w}_i) \in \widehat{W}_i$ such that for all $j\in[m]$ and $t_{ij} \in \cT_{ij}$:

$$\frac{\widetilde{\pi}_{ij}(t_{ij})}{\pi_{ij}(t_{ij})}\in \left[\frac{1}{2}-\frac{\max\left(\eps\cdot  T,\eta\right)}{2}, \frac{3}{2}+\frac{\eta}{2}\right],
\quad 
\frac{\widetilde{w}_{ij}(t_{ij})}{w_{ij}(t_{ij})}\in \left[\frac{1}{2}-\frac{\max\left(\eps\cdot  T,\eta\right)}{2} ,1+\frac{\eta}{2}\right]$$

\item Given access to a scalable demand oracle $\dem_i(\cdot,\cdot,\cdot)$ and a value oracle $\val_i(\cdot,\cdot)$,
we can construct a separation oracle $SO$ for $\widehat{W}_i$. The running time of $SO$ on any input with bit complexity $b'$ is $\poly\left(b,b',k,T,\frac{1}{\eps},\frac{1}{\eta}\right)$ and perform $\poly\left(b,b',k,T,\frac{1}{\eps},\frac{1}{\eta}\right)$ queries to $\dem_i,\val_i$ with input of bit complexity at most $\poly\left(b,b',T\right)$.
\end{enumerate}
\end{theorem}
}

Choose parameter $\eps=\frac{1}{2T}$ and any $k=\Omega\left( T^4\left( Tb+\log\left(1/\eps\right)+n\log(1/\delta)\right) \right)$ 
Hence, in this proof, $1-\eps T=1/2$. Let $\widehat{D}_i$ be the empirical distribution induced by $N=\left\lceil\frac{32kT}{\eps^2}\right\rceil$ independent samples from $D_i$.
Let $\widehat{W}_{i}^{tr(\eps)}$ be the mixture of $\{W_{t_i}^{tr(\eps)}\}_{t_i\in T_i}$ over the empirical distribution $\widehat{D}_i$. By~\Cref{obs:bit complexity mrf xos}, the bit complexity of each corner of polytopes in $\{W_{t_i}^{tr(\eps)}\}_{t_i\in\cT_i}$ is at most $O(T b)$.
By Theorem~\ref{thm:approx mixture}, with probability at least $1-2Te^{-2Tk}\geq 1-\delta$, 
both of the following properties hold:
\begin{enumerate}[label=(\roman*)]
    \item For each $(\widehat{\pi}_i,\widehat{w}_i)\in \widehat{W}_i^{tr(\eps)}$,
    there exists a $(\pi_i,w_i)\in {W}_i^{tr(\eps)}$ such that $||(\widehat{\pi}_i,\widehat{w}_i)-(\pi_i,w_i)||_\infty \leq \frac{\eps}{2}$
    \item For each $(\pi_i,w_i)\in {W}_i^{tr(\eps)}$,
    there exists a $(\widehat{\pi}_i,\widehat{w}_i)\in \widehat{W}_i^{tr(\eps)}$ such that $||(\widehat{\pi}_i,\widehat{w}_i)-(\pi_i,w_i)||_\infty \leq \frac{\eps}{2}$
\end{enumerate}

For the rest of the proof, we condition on the event where both properties above hold. Let $W_i'=\frac{1}{2}W_{i}^{tr(\eps)} + \frac{1}{4}W^{box(\eps)}_{i}$. Consider the polytope $\widehat{W}_i = \frac{1}{2}\widehat{W}_{i}^{tr(\eps)} + \frac{1}{4}W^{box(\eps)}_{i}$.
We are going to show that for each $(\pi_i,w_i) \in W_i'$, there exists a $(\widetilde{\pi}_i,\widetilde{w}_i) \in \widehat{W}_i$ such that for all $j,t_{ij}$, 

$$\frac{\widetilde{\pi}_{ij}(t_{ij})}{\pi_{ij}(t_{ij})}\in \left[\frac{1}{4}, \frac{3}{2}\right]\quad \text{and} \quad
\frac{\widetilde{w}_{ij}(t_{ij})}{w_{ij}(t_{ij})}\in \left[\frac{1}{4},\frac{5}{4}\right].$$
It is not hard to see that this implies Property 2 in the statement of \Cref{cor:rfs for xos mult}, as $W_i'=\frac{1}{2}W_{i}^{tr(\eps)} + \frac{1}{4}W^{box(\eps)}_{i}\supseteq \frac{1}{4}W_i$(due to \Cref{lem:xos mrf contain} and our choice of $\eps$). 


For any $(\pi_i,w_i) \in W_i'$,
we rewrite $(\pi_i,w_i)$ as  $\frac{1}{2}(\pi_i^{tr},w_i^{tr}) + \frac{1}{2}(\pi_i^{box},w_i^{box})$ where  $(\pi_i^{tr},w_i^{tr}) \in W^{tr(\eps)}_{i}$ and $(\pi_i^{box},w_i^{box}) \in \frac{1}{2} W_i^{box(\eps)}$.
For every $j\in [m]$, we partition $\cT_{ij}$ into four disjoint sets: 
\begin{align*}
S_{j}^{(1)}=&\{t_{ij}\in \cT_{ij} : f_{ij}(t_{ij}) \geq \eps \land \eps < w_{ij}(t_{ij})\leq \pi_{ij}(t_{ij})  \} \\
S_j^{(2)}=&\{t_{ij}\in \cT_{ij} : f_{ij}(t_{ij}) \geq \eps \land w_{ij}(t_{ij})\leq \eps < \pi_{ij}(t_{ij})  \} \\
S_j^{(3)}=&\{t_{ij}\in \cT_{ij} : f_{ij}(t_{ij}) \geq \eps \land w_{ij}(t_{ij})\leq \pi_{ij}(t_{ij}) \leq \eps \}\\
S_j^{(4)}=&\{t_{ij}\in \cT_{ij} : f_{ij}(t_{ij}) <\eps\}
\end{align*}

By property (ii), there exists a $(\widehat{\pi}^{tr}_{i},\widehat{w}^{tr}_i)\in \widehat{W}^{tr(\eps)}_{i}$ such that for all $j\in[m]$ and $t_{ij} \in \cT_{ij}$:
\begin{align*}
\pi^{tr}_{ij}(t_{ij}) - \eps/2\leq \widehat{\pi}^{tr}_{ij}(t_{ij}) \leq \pi^{tr}_{ij}(t_{ij}) + \eps/2 \\
w^{tr}_{ij}(t_{ij}) - \eps/2\leq \widehat{w}^{tr}_{ij}(t_{ij}) \leq w^{tr}_{ij}(t_{ij}) + \eps/2
\end{align*}

Now consider the following vector $(\widetilde{\pi}^{tr}_i,\widetilde{w}^{tr}_i) \in [0,1]^{2\sum_{j\in[m]}|\cT_{ij}|}$:
\begin{align*}
\widetilde{\pi}^{tr}_{ij}(t_{ij}) =&
\begin{cases}
\widehat{\pi}_{ij}^{tr}(t_{ij}) \quad & \text{ if $t_{ij}\in S_j^{(1)}\cup S_j^{(2)}$} \\
0 \quad & \text{ o.w.}
\end{cases}\\
\widetilde{w}^{tr}_{ij}(t_{ij}) =&
\begin{cases}
\widehat{w}_{ij}^{tr}(t_{ij}) \quad & \text{ if $t_{ij}\in S_j^{(1)}$} \\
0 \quad & \text{ o.w.}
\end{cases}\\
\end{align*}

{In other words, $$(\widetilde{\pi}^{tr}_{ij}(t_{ij}),\widetilde{w}^{tr}_{ij}(t_{ij}))=(\widehat{\pi}^{tr}_{ij}(t_{ij})\cdot \ind[{\pi}_{ij}(t_{ij})> \eps],\widehat{w}^{tr}_{ij}(t_{ij})\cdot\ind[{w}_{ij}(t_{ij})> \eps]),$$}
and 
$(\widetilde{\pi}^{tr}_i,\widetilde{w}^{tr}_i)\in \widehat{W}_i^{tr(\eps)}$.
\notshow{
\begin{lemma}
$(\widetilde{\pi}^{tr}_i,\widetilde{w}^{tr}_i) \in \widehat{W}_{i}^{tr(\eps)}$.
\end{lemma}
\begin{proof}
Since $(\widehat{\pi}^{tr}_{i},\widehat{w}^{tr}_i)\in \widehat{W}_{i}^{tr(\eps)}$, for every $t_i$, there exists $(\widehat{\pi}^{(t_i)},\widehat{w}^{(t_i)})\in {W}_{t_i}^{tr(\eps)}$, such that $(\widehat{\pi}^{tr}_{i},\widehat{w}^{tr}_i)=\sum_{t_i\in \cT_i}\Pr_{t_i'\sim \widehat{D}_i}[t_i'=t_i]\cdot (\widehat{\pi}^{(t_i)},\widehat{w}^{(t_i)})$. Define $(\widetilde{\pi}^{(t_i)},\widetilde{w}^{(t_i)})$ as follows:

\begin{align*}
\widetilde{\pi}^{(t_i)}_j(t_{ij}') =&
\begin{cases}
\widehat{\pi}^{(t_i)}_j(t_{ij}) \quad & \text{ if $t_{ij}=t_{ij}'\wedge t_{ij}\in S_j^{(1)}$} \\
\min\left(1 ,a\cdot(\pi_{ij}(t_{ij}) - w_{ij}(t_{ij}))/\widehat{f}_{ij}(t_{ij})\right)  \quad & \text{ if $t_{ij}=t_{ij}'\wedge t_{ij}\in S_j^{(2)}$} \\
0 \quad & \text{ o.w.}
\end{cases}\\
\widetilde{w}^{(t_i)}_{j}(t_{ij}') =&
\begin{cases}
\widehat{w}_{j}^{(t_i)}(t_{ij}) \quad & \text{ if $t_{ij}=t_{ij}'\wedge t_{ij}\in S_j^{(1)}$} \\
0 \quad & \text{ o.w.}
\end{cases}\\
\end{align*}

By \Cref{cor:W_t_i truncate} we have $(\widetilde{\pi}^{(t_i)},\widetilde{w}^{(t_i)})\in {W}_{t_i}^{tr(\eps)}$. Since $\widehat{W}_{i}^{tr(\eps)}$ is a mixture of ${W}_{t_i}^{tr(\eps)}$ over distribution $\widehat{D}_i$, we have that $(\widetilde{\pi}^{tr}_{i},\widetilde{w}^{tr}_i)=\sum_{t_i\in \cT_i}\Pr_{t_i'\sim \widehat{D}_i}[t_i'=t_i]\cdot (\widetilde{\pi}^{(t_i)},\widetilde{w}^{(t_i)})$ is in $\widehat{W}_{i}^{tr(\eps)}$. 
\end{proof}
}

Consider another vector $(\widetilde{\pi}^{box}_i,\widetilde{w}^{box}_i) \in [0,1]^{2\sum_{j\in[m]}|\cT_{ij}|}$: 
\begin{align*}
(\widetilde{\pi}^{box}_{ij}(t_{ij}),\widetilde{w}^{box}_{ij}(t_{ij})) =&
\begin{cases}
({\pi}_{ij}^{box}(t_{ij}),w^{box}_{ij}(t_{ij})) \quad & \text{ if $t_{ij}\in S_{j}^{(1)}\cup S_{j}^{(4)}$} \\ 
(w_{ij}(t_{ij})/2, w_{ij}(t_{ij})/2)  \quad & \text{ if $t_{ij}\in S_{j}^{(2)}$} \\
({\pi}_{ij}(t_{ij})/2, w_{ij}(t_{ij})/2) \quad & \text{ if $t_{ij}\in S_{j}^{(3)}$} \\
\end{cases}
\end{align*}

For every $j\in[m]$ and $t_{ij}\in \cT_{ij}$, clearly $\widetilde{\pi}_{ij}^{box}(t_{ij})\geq \widetilde{w}_{ij}^{box}(t_{ij})$ since $\pi_{ij}^{box}(t_{ij})\geq w_{ij}^{box}(t_{ij})$ and $\pi_{ij}(t_{ij})\geq w_{ij}(t_{ij})$. Moreover, if $t_{ij}\in S_j^{(1)}\cup S_j^{(4)}$, $\widetilde{\pi}_{ij}^{box}(t_{ij})=\pi_{ij}^{box}(t_{ij})\leq \frac{1}{2}\cdot \min\{\eps,f_{ij}(t_{ij})\}$, since $(\pi_i^{box},w_i^{box}) \in \frac{1}{2}W_i^{box(\eps)}$.
If $t_{ij}\in S_j^{(2)}\cup S_j^{(3)}$, by the definitions of $S_j^{(2)}$ and $S_j^{(3)}$, we have that $\widetilde\pi_{ij}^{box}(t_{ij})\leq \eps/2=\frac{1}{2}\cdot \min\{\eps,f_{ij}(t_{ij})\}$.
Thus $(\widetilde{\pi}^{box}_i,\widetilde{w}^{box}_i)\in \frac{1}{2} {W}_{i}^{box(\eps)}$.
Now define $(\widetilde{\pi}_i,\widetilde{w}_i)=\frac{1}{2}(\widetilde{\pi}^{tr}_i,\widetilde{w}^{tr}_i) + \frac{1}{2}(\widetilde{\pi}^{box}_i,\widetilde{w}^{box}_i)$.
Then $(\widetilde{\pi}_i,\widetilde{w}_i)\in \widehat{W}_i$.
It remains to prove that $(\widetilde{\pi}_i,\widetilde{w}_i)$ satisfies: for all $j, t_{ij}$

$$\frac{\widetilde{\pi}_{ij}(t_{ij})}{\pi_{ij}(t_{ij})}\in \left[\frac{1}{4}, \frac{3}{2}\right]\quad \text{and} \quad
\frac{\widetilde{w}_{ij}(t_{ij})}{w_{ij}(t_{ij})}\in \left[\frac{1}{4},\frac{5}{4}\right].$$


We verify all cases based on which set $t_{ij}$ is in. 
\paragraph{Case 1: $t_{ij}\in S_j^{(1)}$.} 
Recall that for $t_{ij}\in S_j^{(1)}$, $\eps<w_{ij}(t_{ij}) \leq \pi_{ij}(t_{ij})$. We have that $$\widetilde\pi_{ij}(t_{ij})-\pi_{ij}(t_{ij})=\frac{1}{2}(\hat\pi_{ij}^{tr}(t_{ij})-\pi_{ij}^{tr}(t_{ij})).$$ According to property (ii), $\frac{1}{2}(\hat\pi_{ij}^{tr}(t_{ij})-\pi_{ij}^{tr}(t_{ij}))\in [-\frac{\eps}{4},\frac{\eps}{4}]$, so $\frac{1}{2}(\hat\pi_{ij}^{tr}(t_{ij})-\pi_{ij}^{tr}(t_{ij}))$ also lies in $[-\frac{\pi_{ij}(t_{ij})}{4},\frac{\pi_{ij}(t_{ij})}{4}]$. 
Similarly, 
$$\widetilde {w}_{ij}(t_{ij})-w_{ij}(t_{ij})=\frac{1}{2}(\hat{w}_{ij}^{tr}(t_{ij})-w_{ij}^{tr}(t_{ij}))\in [-\frac{\eps }{4},\frac{\eps }{4}] \subseteq [-\frac{w_{ij}(t_{ij})}{4}, \frac{w_{ij}(t_{ij})}{4}].$$
Hence, $$\frac{\widetilde{\pi}_{ij}(t_{ij})}{\pi_{ij}(t_{ij})}\in \left[\frac{3}{4}, \frac{5}{4}\right]\quad \text{and} \quad
\frac{\widetilde{w}_{ij}(t_{ij})}{w_{ij}(t_{ij})}\in \left[\frac{3}{4},\frac{5}{4}\right].$$

\paragraph{Case 2: $t_{ij}\in S_j^{(2)}$.}
Recall that
$
\widetilde{w}_{ij}(t_{ij}) = \frac{1}{4}\cdot {w}_{ij}(t_{ij})
$ and $\widetilde{\pi}_{ij}(t_{ij}) = \frac{1}{4}\cdot {w}_{ij}(t_{ij}) + \frac{1}{2}\cdot\widehat{\pi}^{tr}_{ij}(t_{ij})$. We have that
$$
\widetilde{\pi}_{ij}(t_{ij}) \leq \frac{1}{4}\cdot {w}_{ij}(t_{ij}) + \frac{{\pi}^{tr}_{ij}(t_{ij})}{2}  + \frac{\eps}{4} \leq \frac{3}{2}\cdot{\pi}_{ij}(t_{ij}).
$$
The first inequality follows from property (ii), and the second inequality follows from $\pi_{ij}(t_{ij}) =\frac{\pi^{tr}_{ij}(t_{ij})}{2} + \frac{\pi^{box}_{ij}(t_{ij})}{2}\geq \frac{\pi^{tr}_{ij}(t_{ij})}{2}$ and $w_{ij}(t_{ij})\leq \eps < \pi_{ij}(t_{ij})$ when $t_{ij}\in S_j^{(2)}$. We also have that
$$
\widetilde{\pi}_{ij}(t_{ij})\geq \frac{\widetilde{\pi}^{tr}_{ij}(t_{ij})}{2} \geq \frac{{\pi}^{tr}_{ij}(t_{ij})}{2}  - \frac{\eps}{4} = \pi_{ij}(t_{ij}) - \frac{\pi_{ij}^{box}(t_{ij})}{2} - \frac{ \eps}{4}\geq \pi_{ij}(t_{ij})- \frac{\eps}{2} \geq \frac{1}{2}{\pi}_{ij}(t_{ij}) 
$$
The first inequality follows from the non-negativity of $\widetilde{\pi}^{box}_{ij}(t_{ij})$; the second inequality follows from property (ii); the third inequality is due to the fact that $\pi^{box}_{ij}(t_{ij})\leq \eps/2$; the last inequality is because $\eps<\pi_{ij}(t_{ij})$.

Hence, $$\frac{\widetilde{\pi}_{ij}(t_{ij})}{\pi_{ij}(t_{ij})}\in \left[\frac{1}{2}, \frac{3}{2}\right]\quad \text{and} \quad
\frac{\widetilde{w}_{ij}(t_{ij})}{w_{ij}(t_{ij})}= \frac{1}{4}.$$


\paragraph{Case 3: $t_{ij}\in S_j^{(3)}\cup S_j^{(4)}$.}
Recall that when $t_{ij}\in S_j^{(3)}$, $\widetilde{\pi}_{ij}(t_{ij}) = \frac{1}{4}\cdot {\pi}_{ij}(t_{ij})$ and $\widetilde{w}_{ij}(t_{ij}) = \frac{1}{4}\cdot {w}_{ij}(t_{ij})$. 
When $t_{ij}\in S_j^{(4)}$, as $f_{ij}(t_{ij})<\eps$, $w_{ij}^{tr}(t_{ij})=\pi_{ij}^{tr}(t_{ij})=0$. Thus $\widetilde{\pi}_{ij}(t_{ij})=\frac{1}{2}\pi_{ij}^{box}(t_{ij})=\pi_{ij}(t_{ij})$, and $w_{ij}(t_{ij})=\frac{1}{2}w_{ij}^{box}(t_{ij})=w_{ij}(t_{ij})$. 

To sum up, we have argued that for each $(\pi_i,w_i) \in W_i'$, there exists a $(\widetilde{\pi}_i,\widetilde{w}_i) \in \widehat{W}_i$ such that for all $j,t_{ij}$, 

$$\frac{\widetilde{\pi}_{ij}(t_{ij})}{\pi_{ij}(t_{ij})}\in \left[\frac{1}{4}, \frac{3}{2}\right]\quad \text{and} \quad
\frac{\widetilde{w}_{ij}(t_{ij})}{w_{ij}(t_{ij})}\in \left[\frac{1}{4},\frac{5}{4}\right].$$

With a similar analysis,
~{\footnote{{We only need to switch the role of $(\widetilde{\pi}_i,\widetilde{w}_i)$ and $(\pi_i,w_i)$. We provide a brief sketch here. First, rewrite $(\widetilde{\pi}_i,\widetilde{w}_i)$ as $\frac{1}{2}(\widehat{\pi}^{tr}_i,\widehat{w}^{tr}_i)+\frac{1}{2}(\widetilde{\pi}^{box}_i,\widetilde{w}^{box}_i)$, where $(\widehat{\pi}^{tr}_i,\widehat{w}^{tr}_i)\in \widehat{W}^{tr(\eps)}_i$ and $(\widetilde{\pi}^{box}_i,\widetilde{w}^{box}_i)\in \frac{1}{2}W_i^{box(\eps)}$. Also, redefine $S_j^{(i)}$ in the same fashion but according to $(\widetilde{\pi}_i,\widetilde{w}_i)$. Take $({\pi}^{tr}_i,{w}^{tr}_i)\in W_i^{tr(\eps)}$ to be the point guaranteed to exist by property (i), and define $(\bar{\pi}^{tr}_{ij}(t_{ij}),\bar{w}^{tr}_{ij}(t_{ij})):= ({\pi}^{tr}_{ij}(t_{ij})\cdot \ind[\widetilde{\pi}_{ij}(t_{ij})>\eps],w^{tr}_{ij}(t_{ij})\cdot \ind[\widetilde{w}_{ij}(t_{ij})>\eps])$. Also, define $(\pi_{ij}^{box}(t_{ij}),w_{ij}^{box}(t_{ij}))$ according to which set $t_{ij}$ belongs to in a fashion similar to the proof above. Now define $(\pi_i,w_i)=\frac{1}{2}(\bar{\pi}^{tr}_{i},\bar{w}^{tr}_{i})+\frac{1}{2}({\pi}^{box}_{i},{w}^{box}_{i})$. Using a similar case analysis, we can prove the claim that $\frac{{\pi}_{ij}(t_{ij})}{\widetilde{\pi}_{ij}(t_{ij})}\in \left[\frac{1}{4},\frac{3}{2}\right]$ and  $\frac{{w}_{ij}(t_{ij})}{\widetilde{w}_{ij}(t_{ij})}\in \left[\frac{1}{4},\frac{5}{4}\right]$ for all $j\in[m]$ and $t_{ij}\in \cT_{ij}$.}}}
we can also show that for each $(\widetilde{\pi}_i,\widetilde{w}_i) \in \widehat{W}_i$, there exists a $({\pi}_i,{w}_i) \in W_i'$ such that for all $j,t_{ij}$, 

$$\frac{{\pi}_{ij}(t_{ij})}{\widetilde{\pi}_{ij}(t_{ij})}\in \left[\frac{1}{4},\frac{3}{2}\right]\quad\text{and}\quad \frac{{w}_{ij}(t_{ij})}{\widetilde{w}_{ij}(t_{ij})}\in \left[\frac{1}{4},\frac{5}{4}\right].$$

Thus Property 1 in the statement of \Cref{cor:rfs for xos mult} follows from the fact that $W_i'=\frac{1}{2}W_{i}^{tr(\eps)} + \frac{1}{4}W^{box(\eps)}_{i}\subseteq W_i$ (\Cref{lem:xos mrf contain}). 

Let $\{t_i^{(1)},...,t_i^{(N)}\}$
be the $N$ samples from $D_i$. Then 
$$
\widehat{W}_i = \sum_{\ell\in[N]}\frac{1}{2N}\cdot W_{t_i^{(\ell)}}^{tr(\eps)} 
+ \frac{1}{4}W^{box(\eps)}_{i} 
$$
is the Minkowski addition of $N+1$ polytopes. For Property 3 of the statement,
since the vertex-complexity of $W_{t_i}^{tr(\eps)}$ is $O(bT)$ for each $t_i$ (\Cref{obs:bit complexity mrf xos}), and the vertex-complexity of $W_i^{box(\eps)}$ is no more than $\poly(b,T)$, the vertex-complexity of $\widehat{W}_i$ is no more than $poly(n,m,T,b,\log(1/\delta))$. 

At last, we show the existence of an efficient separation oracle $SO$ for $\widehat{W}_i$, by efficiently optimizing any linear objective over $\widehat{W}_i$. Since $\widehat{W}_i$ is the Minkowski addition of polytopes $\{W_{t_i^{(\ell)}}^{tr(\eps)}\}_{\ell\in [N]}$ and $W^{box(\eps)}_{i}$, in order to maximize over $\widehat{W}_i$, it's sufficient to maximize over each polytope. By \Cref{lem:optimize}, we can efficiently optimize any linear objective over $W_{t_i}^{tr(\eps)}$ for every $t_i$, given the adjustable demand oracle and value oracle. Thus it is sufficient to solve
$\max\{\bm{x}\cdot \pi_i+\bm{y}\cdot w_i: (\pi_i,w_i)\in W^{box(\eps)}_{i}\}$ for any vector $\bm{x},\bm{y}$. Since in $W^{box(\eps)}_{i}$, the constraint for each coordinate $(j,t_{ij})$ is separate: $0\leq w_{ij}(t_{ij})\leq \pi_{ij}(t_{ij})\leq \min(\eps,f_{ij}(t_{ij}))$. Thus the optimum can be achieved by solving the following LP for every coordinate:
\begin{align*}
    \max&\quad x_j(t_{ij})\cdot \pi_{ij}(t_{ij}) + y_j(t_{ij})\cdot w_{ij}(t_{ij}) \\
    s.t.&\quad 0\leq w_{ij}(t_{ij})\leq \pi_{ij}(t_{ij})\leq \min(\eps,f_{ij}(t_{ij}))
\end{align*}
Note that the bit complexity of the output of our optimization algorithm is $poly(n,m,T,b,\log(1/\delta))$.
Thus by \Cref{thm:equivalence of opt and sep}, there exists a separation oracle $SO$ of $\widehat{W}_i$, that satisfies Property 4 in the statement of \Cref{cor:rfs for xos mult}.


\notshow{

\begin{proof}
By applying Theorem~\ref{thm:multi-approx-mixture} with $\eps=\frac{1}{2T}$, $\eta=\frac{1}{2}$, integer $k=\Omega\left(\frac{1}{2T}\log\left(\frac{n2T}{\delta} \right),T^4\left(Tb +\log\left( \frac{1}{\eta\eps}\right) \right)\right)$  and
having access to $\left\lceil\frac{8kT}{(\eta \eps )^2} \right\rceil=\poly\left(b,T, \log\left(\frac{1}{\delta}\right) \right)$ samples,
we can construct a polytope $\widehat{W}_i$ such that with probability at least $1-2T\exp(-2Tk)\geq 1-\frac{\delta}{n}$    
\begin{enumerate}
    \item For each $(\widetilde{\pi}_i,\widetilde{w}_i) \in \widehat{W}_i$,
there exists a $(\pi_i,w_i) \in \frac{1}{2}{W}_{i}^{tr(\eps)} + \frac{1}{4}W^{box(\eps)}_{i}$ such that for all $j\in[m]$ and $t_{ij} \in \cT_{ij}$:

$$\frac{{\pi}_{ij}(t_{ij})}{\widetilde{\pi}_{ij}(t_{ij})}\in \left[\frac{1}{4},\frac{7}{4}\right],\quad \frac{{w}_{ij}(t_{ij})}{\widetilde{w}_{ij}(t_{ij})}\in \left[\frac{1}{4},\frac{5}{4}\right]$$

\item For each $(\pi_i,w_i) \in \frac{1}{2}{W}_{i}^{tr(\eps)} + \frac{1}{4}W^{box(\eps)}_{i}$,
there exists a $(\widetilde{\pi}_i,\widetilde{w}_i) \in \widehat{W}_i$ such that for all $j\in[m]$ and $t_{ij} \in \cT_{ij}$:

$$\frac{\widetilde{\pi}_{ij}(t_{ij})}{\pi_{ij}(t_{ij})}\in \left[\frac{1}{4},\frac{7}{4}\right],\quad \frac{\widetilde{w}_{ij}(t_{ij})}{w_{ij}(t_{ij})}\in \left[\frac{1}{4},\frac{5}{4}\right]$$

\item  Given access to a scalable demand oracle $\dem_i(\cdot,\cdot,\cdot)$ and a value oracle $\val_i(\cdot,\cdot)$,
we can construct a separation oracle $SO_i$ for $\widehat{W}_i$. The running time of $SO_i$ on any input with bit complexity $b'$ is $\poly\left(b,b',T,\log\left(\frac{1}{\delta}\right)\right)$ and perform $\poly\left(b,b',T,\log\left(\frac{1}{\delta}\right)\right)$ queries to $\dem_i,\val_i$ with input of bit complexity $\poly\left(b,b',T,\log\left(\frac{1}{\delta}\right)\right)$. 
\end{enumerate}
By Lemma~\ref{lem:xos mrf contain},
we know that $\frac{1}{4}W_i\subseteq\frac{1}{2}{W}_{i}^{tr(\eps)} + \frac{1}{4}W^{box(\eps)}_{i}\subseteq W_i$.

Thus for each $(\widetilde{\pi}_i,\widetilde{w}_i)\in \widehat{W}_i$ there exists a $(\pi_i,w_i) \in \frac{1}{2}{W}_{i}^{tr(\eps)} + \frac{1}{4}W^{box(\eps)}_{i} \subseteq W_i$ such that
$$\frac{{\pi}_{ij}(t_{ij})}{\widetilde{\pi}_{ij}(t_{ij})}\in \left[\frac{1}{4},\frac{7}{4}\right],\quad \frac{{w}_{ij}(t_{ij})}{\widetilde{w}_{ij}(t_{ij})}\in \left[\frac{1}{4},\frac{5}{4}\right]$$

And for 
$(\pi_i,w_i)\in 
\frac{1}{2}W_i^{tr(\eps)} + \frac{1}{4}W_i^{box} \subseteq \frac{1}{4}W_i$
there exists a $(\widetilde{\pi}_i,\widetilde{w}_i) \in \widehat{W}_i$ such that for all $j\in[m]$ and $t_{ij} \in \cT_{ij}$:

$$\frac{\widetilde{\pi}_{ij}(t_{ij})}{\pi_{ij}(t_{ij})}\in \left[\frac{1}{4},\frac{7}{4}\right],\quad \frac{\widetilde{w}_{ij}(t_{ij})}{w_{ij}(t_{ij})}\in \left[\frac{1}{4},\frac{5}{4}\right]$$

\end{proof}

}



\notshow{

\begin{lemma}\label{lem:W_t_i modify}
For any $t_i$ and $(\pi_i,w_i)\in W_{t_i}$ and set of items $B\subseteq [m]$.
Consider $\widehat{\pi}_i,\widehat{w}_i\in [0,1]^{\sum_{j\in[m]}|\cT_{ij}|}$ such that for $j\in S$:
\begin{align*}
\widehat{w}_{ij}(t_{ij}) = w_{ij}(t_{ij}) \\
\widehat{\pi}_{ij}(t_{ij}) = \pi_{ij}(t_{ij})    
\end{align*}
and for $j\notin S$:
\begin{align*}
\widehat{w}_{ij}(t_{ij}) = 0 \\
0 \leq \widehat{\pi}_{ij}(t_{ij}) \leq 1   
\end{align*}
For $j\in [m]$ and $t_{ij}'\neq t_{ij}$
\begin{align*}
\widehat{w}_{ij}(t_{ij}) = 0 \\
\widehat{\pi}_{ij}(t_{ij}) = 0
\end{align*}
then $(\pi,w)\in W_{t_i}$.
\end{lemma}

\begin{proof}
Note that for each $j\in[m]$ and $t_{ij}'\neq t_{ij}$,
then $\widehat{w}_{ij}(t_{ij}')=\widehat{\pi}_{ij}(t_{ij}')={w}_{ij}(t_{ij}')={\pi}_{ij}(t_{ij}')=0$.
Let $\{\sigma_{S}^k\}_{S,k}$ be the set of numbers associated with $(\pi,w)\in W_{t_i}$ according to Definition~\ref{def:welfare frequent-t_i}.
We consider the set of numbers $\{\widehat{\sigma}_{S}^k\}_{S,k}$,
and we are going to prove that the element in $W_{t_i}$ associated with $\{\widehat{\sigma}_{S}^k\}_{S,k}$, is $(\widehat{w}_i,\widehat{\pi}_i)$.
For each $k\in [K]$ and $S\subseteq [m]$,
where $S=S_b\cup S_g$, where $S_B\subseteq B$ and $S_G\subseteq [m] \backslash B$:
$$
\widehat{\sigma}_{S}^k = 
\left(\prod_{j\in S_B} \widehat{\pi}_{ij}(t_{ij})\prod_{j \in B\backslash S_B} (1-\widehat{\pi}_{ij}(t_{ij})) \right) \left(\sum_{\substack{S_G \subseteq T \subseteq S_G \cup B}} {\sigma}_{T}^k\right)
$$
First observe that each set for $T\subseteq [m]$,
there exists a unique set $U\subseteq [m]\backslash B$, such that $U\subseteq T\subseteq U\cup B$ (in that case $U=T \backslash B$).
Thus

\begin{align*}
\sum_{S_B\subseteq B}\sum_{k\in [K]}\sum_{S_G\subseteq[m]\backslash B} \widehat{\sigma}_{S_B\cup S_G}^k = &\sum_{S_B\subseteq B}
\left(\prod_{j\in S_B} \widehat{\pi}_{ij}(t_{ij})\prod_{j \in B\backslash S_B} (1-\widehat{\pi}_{ij}(t_{ij})) \right) \sum_{k\in [K]}\sum_{S_G\subseteq[m]\backslash B}\left(\sum_{\substack{S_G \subseteq T \subseteq S_G \cup B}} {\sigma}_{T}^k\right) \\
=&\sum_{S_B\subseteq B}\left(\prod_{j\in S_B} \widehat{\pi}_{ij}(t_{ij})\prod_{j \in B\backslash S_B} (1-\widehat{\pi}_{ij}(t_{ij})) \right) \sum_{k\in [K]}\sum_{T\subseteq[m]}{\sigma}_{T}^k \\
=&\sum_{S_B\subseteq B}\prod_{j\in S_B} \widehat{\pi}_{ij}(t_{ij})\prod_{j \in B\backslash S_B} (1-\widehat{\pi}_{ij}(t_{ij})) 
\end{align*}

Consider $|B|$ random variables,
such that for each $j\in B$, random variables $C_j$ is sampled from a Bernoulli distribution with probability of success $\widehat{\pi}_{ij}(t_{ij})$.
Observe that $\prod_{j\in S_B} \widehat{\pi}_{ij}(t_{ij})\prod_{j \in B\backslash S_B} (1-\widehat{\pi}_{ij}(t_{ij}))$,
is the probability that once we take a sample from the $|B|$ random variables,
$C_j=1$ only if $j\in S_B$ and $C_j=0$ if $j\notin S_B$.
Thus
\begin{align*}
\sum_{S_B\subseteq B}\sum_{k\in [K]}\sum_{S_G\subseteq[m]\backslash B} \widehat{\sigma}_{S_B\cup S_G}^k 
=&\sum_{S_B\subseteq B}\prod_{j\in S_B} \widehat{\pi}_{ij}(t_{ij})\prod_{j \in B\backslash S_B} (1-\widehat{\pi}_{ij}(t_{ij})) \\
=& 1
\end{align*}

which proves that $\{\widehat{\sigma}_{S}^k\}_{S,k}$ first property of Definition~\ref{def:welfare frequent-t_i}.
Now we are going to prove the second property.
For $j\in B$:

\begin{align*}
\sum_{S_B\subseteq B:j\in S_B}\sum_{k\in [K]}\sum_{S_G\subseteq[m]\backslash B} \widehat{\sigma}_{S_B\cup S_G}^k 
=& \widehat{\pi}_{ij}(t_{ij}) \sum_{\substack{S_B\subseteq B\\j\in B}}\prod_{\substack{j'\in S_B \\ \land j'\neq j}} \widehat{\pi}_{ij'}(t_{ij'})\prod_{j' \in B\backslash S_B} (1-\widehat{\pi}_{ij'}(t_{ij'})) \\
=& \widehat{\pi}_{ij}(t_{ij})
\end{align*}

The last equation follows by a similar way of thinking as above,
by interpreting $\prod_{\substack{j'\in S_B \\\land j'\neq j}} \widehat{\pi}_{ij'}(t_{ij'})\prod_{j' \in B\backslash S_B} (1-\widehat{\pi}_{ij'}(t_{ij'}))$,
as the event where we have $|B|$ random variables $C_j$,
where $C_j$ is sampled from a Bernoulli distribution with probability of success $\widehat{\pi}_{ij}(t_{ij})$.
Notice that for $j\in B$, setting $\widehat{w}_{ij}(t_{ij})$ trivially satisfy the third requirement.

For $j\notin B$,
observe that

\begin{align*}
\sum_{\substack{S_G\subseteq[m]\backslash B\\j \in S_G}}\sum_{S_B\subseteq B}\sum_{k\in [K]} \widehat{\sigma}_{S_B\cup S_G}^k = 
&\sum_{S_B\subseteq B}
\left(\prod_{j\in S_B} \widehat{\pi}_{ij}(t_{ij})\prod_{j \in B\backslash S_B} (1-\widehat{\pi}_{ij}(t_{ij})) \right) \sum_{k\in [K]}\sum_{\substack{S_G\subseteq[m]\backslash B\\j \in S_G}}\left(\sum_{\substack{S_G \subseteq T \subseteq S_G \cup B}} {\sigma}_{T}^k\right) \\
=&\sum_{k\in [K]}\sum_{\substack{S_G\subseteq[m]\backslash B\\j \in S_G}}\left(\sum_{\substack{S_G \subseteq T \subseteq S_G \cup B}} {\sigma}_{T}^k\right)\\
=&\sum_{k\in [K]}\sum_{\substack{S\subseteq[m]\\j \in S}} {\sigma}_{S}^k=\pi_{ij}(t_{ij})
\end{align*}

We can similarly prove that 
\begin{align*}
\sum_{\substack{S_G\subseteq[m]\backslash B\\j \in S_G}}\sum_{S_B\subseteq B}\sum_{k\in [K]} \widehat{\sigma}_{S_B\cup S_G}^k \frac{\alpha_j^{(k)}}{V_{ij}(t_{ij})}= 
w_{ij}(t_{ij})
\end{align*}

which concludes the proof.
\end{proof}

A corollary of Lemma~\ref{lem:W_t_i modify} is the following.

\begin{corollary}\label{cor:W_t_i truncate}
For any $t_i$ and $(\pi^{tr}_i,w^{tr}_i)\in W_{t_i}^{tr(\eps)}$ and set of items $B\subseteq [m]$.
Consider $\widehat{\pi}^{tr}_i,\widehat{w}^{tr}_i\in [0,1]^{\sum_{j\in[m]}|\cT_{ij}|}$ such that for $j\in S$:
\begin{align*}
\widehat{w}^{tr}_{ij}(t_{ij}) = w_{ij}^{tr}(t_{ij}) \\
\widehat{\pi}^{tr}_{ij}(t_{ij}) = \pi^{tr}_{ij}(t_{ij})    
\end{align*}
and for $j\notin S$:
\begin{align*}
\widehat{w}^{tr}_{ij}(t_{ij}) = 0 \\
0 \leq \widehat{\pi}^{tr}_{ij}(t_{ij}) \leq 1   
\end{align*}
For $j\in [m]$ and $t_{ij}'\neq t_{ij}$
\begin{align*}
\widehat{w}_{ij}^{tr}(t_{ij}) = 0 \\
\widehat{\pi}_{ij}^{tr}(t_{ij}) = 0
\end{align*}
then $(\pi^{tr},w^{tr})\in W_{t_i}^{tr(\eps)}$.
\end{corollary}

A corollary of Lemma~\ref{lem:W_t_i modify} is the following.
\begin{observation}\label{obs:W_i(t_i)-2}
Fix any $t_i$. For any $j$, $S_j\subseteq \cT_{ij}$, $(\pi,w)\in W_{t_i}$, define $(\pi',w')\in [0,1]^{2\sum_{j\in[m]}|\cT_{ij}|}$ as: 
$$\pi_{j}'(t_{ij}')=\pi_j(t_{ij}')\cdot\ind[t_{ij}'\in S_j],\quad w_{j}'(t_{ij}')=w_j(t_{ij}')\cdot\ind[t_{ij}'\in S_j].$$ 
Then $(\pi',w')\in W_{t_i}$. 
\end{observation}

We conclude with the main theorem of this sections.
We show how to construct a polytope $\widehat{W}_i$,
such that each element in $\widehat{W}_i$ is within a multiplicative factor of an element in $\frac{1}{2}{W}_{i}^{tr(\eps)} + \frac{1-\eps T}{2}W^{box(\eps)}_{i}$ and vice versa and we have access to an efficient seperation oracle for $\widehat{W}_i$.
\argyrisnote{Unfortunately we cannot show that our polytope is contained in $W_i$ and contains a constant factor of $W_i$ as we did in Section~\ref{sec:multi-approx-polytope},
but we have the necessary properties to be able to ``round" our solution to one that can give us meaningful guarantees.}

\begin{theorem}\label{thm:multi-approx-mixture}
Let $T=\sum_{i\in[n]}\sum_{j\in[m]}|\cT_{ij}|$ and $b$ be the bit complexity of the problem instance (\Cref{def:bit complexity}).
Consider parameters $0 < \eps < 1/T$, $\delta\in (\eps T,1)$ and
integer $k\geq \Omega\left( T^4\left( Tb+\log\left(\frac{1}{\delta \eps}\right)\right) \right)$.
If we have access to $n=\left\lceil\frac{8kT}{(\delta\eps)^2}\right\rceil$ independent samples from $\cD$,
then with probability at least $1-2T\exp(-2T k)$,
we can construct a polytope $\widehat{W}_i$ such that there exists a constant $c\geq 1$

\begin{enumerate}
    \item For each $(\widetilde{\pi}_i,\widetilde{w}_i) \in \widehat{W}_i$,
there exists a $(\pi_i,w_i) \in \frac{1}{2}{W}_{i}^{tr(\eps)} + \frac{1-\eps T}{2}W^{box(\eps)}_{i}$ such that for all $j\in[m]$ and $t_{ij} \in \cT_{ij}$:

$$\frac{{\pi}_{ij}(t_{ij})}{\widetilde{\pi}_{ij}(t_{ij})}\in \left[\frac{1}{2}-c\cdot(\delta+\eps \cdot T),1+c\cdot\delta\right],\quad \frac{{w}_{ij}(t_{ij})}{\widetilde{w}_{ij}(t_{ij})}\in \left[\frac{1}{2}-c\cdot(\delta+\eps \cdot T),1+c\cdot\delta\right]$$

\item For each $(\pi_i,w_i) \in \frac{1}{2}{W}_{i}^{tr(\eps)} + \frac{1-\eps T}{2}W^{box(\eps)}_{i}$,
there exists a $(\widetilde{\pi}_i,\widetilde{w}_i) \in \widehat{W}_i$ such that for all $j\in[m]$ and $t_{ij} \in \cT_{ij}$:

$$\frac{\widetilde{\pi}_{ij}(t_{ij})}{\pi_{ij}(t_{ij})}\in \left[\frac{1}{2}-c\cdot(\delta+\eps \cdot T),1+c\cdot\delta\right],\quad \frac{\widetilde{w}_{ij}(t_{ij})}{w_{ij}(t_{ij})}\in \left[\frac{1}{2}-c\cdot(\delta+\eps \cdot T),1+c\cdot\delta\right]$$

\item Given access to a scalable demand oracle and an value oracle for the $i$-th buyer, we can construct an efficient separation oracle for $\widehat{W}_i$.

\item \argyrisnote{Write instead:} Given access to a scalable demand oracle $\dem_i(\cdot,\cdot,\cdot)$ and a value oracle $\val_i(\cdot,\cdot)$,
we can construct a separation oracle $SO$ for $\widehat{W}_i$. The running time of $SO$ on any input with bit complexity $b'$ is $\poly(d,b,b',\frac{1}{\eps},\log(1/\delta),\mathcal{A}(\poly(d,b,b',\frac{1}{\eps},\log(1/\delta))))$.
\end{enumerate}
\end{theorem}

\mingfeinote{
1.$\alpha_{ij}^k(t_{ij})$ not $a_j$.

2. The notation is inconsistent for XOS and c-a. Here you use $M_i=\sum_j|\cT_{ij}|$. Also there is a `T' in the statement.

3. Figure out the exact constant for $\delta$. I suggest you pick a specific constant for $\delta$ at the end of the proof.  

4. Argue for $W_i$ instead of $\frac{1}{2}{W}_{i}^{tr(\eps)} + \frac{1-\eps M_i}{2}W^{box(\eps)}_{i}$. Put Lemma 22 inside in the proof.

}

\begin{proof}
Let $\widehat{D}_i$ be the empirical distribution induced by $n$ independent samples from $D_i$.
Let $\widehat{W}_{i}^{tr(\eps)}$ be the mixture of $\{W_{t_i}^{tr(\eps)}\}_{t_i\in T_i}$ over the empirical distribution $\widehat{D}_i$.
By Observation~\ref{obs:bit complexity mrf xos},
the bit complexity of corners of polytopes in $\{W_{t_i}^{tr(\eps)}\}_{t_i\in\cT_i}$ is at most $O(T b)$.
By Theorem~\ref{thm:approx mixture} we have that with probability at least $1-2T\exp(-2Tk)$ the two following properties hold:
\begin{enumerate}
    \item Property (1): for each $(\widehat{\pi},\widehat{w})\in \widehat{W}_i^{tr(\eps)}$,
    there exists a $({\pi},{w})\in {W}_i^{tr(\eps)}$ such that $||(\widehat{pi},\widehat{w})-({pi},{w})||_\infty \leq \eps \delta$
    \item Property (2): for each $({\pi},{w})\in {W}_i^{tr(\eps)}$,
    there exists a $(\widehat{\pi},\widehat{w})\in \widehat{W}_i^{tr(\eps)}$ such that $||(\widehat{pi},\widehat{w})-({pi},{w})||_\infty \leq \eps \delta$
\end{enumerate}
For the rest of the proof, we condition that both of the events we described above happens.
Consider the polytope $\widehat{W}_i = \frac{1}{2}\widehat{W}_{i}^{tr(\eps)} + \frac{1-\eps T}{2}W^{box(\eps)}_{i}$.
We are going to show that for each $(\pi_i,w_i) \in \frac{1}{2}{W}_{i}^{tr(\eps)} + \frac{1-\eps T}{2}W^{box(\eps)}_{i}$,
there exists a $(\widetilde{\pi}_i,\widetilde{w}_i) \in \widehat{W}_i$ and constant $c\geq $, such that

$$\frac{\widetilde{\pi}_{ij}(t_{ij})}{\pi_{ij}(t_{ij})}\in \left[\frac{1}{2}-c\cdot(\delta+\eps \cdot T),1+c\cdot\delta\right],\quad \frac{\widetilde{w}_{ij}(t_{ij})}{w_{ij}(t_{ij})}\in \left[\frac{1}{2}-c\cdot(\delta+\eps \cdot T),1+c\cdot\delta\right]$$
The other direction is similar.

For every $j,t_{ij}$, let $\widehat{f}_{ij}(t_{ij})=\Pr_{t_i'\sim \widehat{D}_i}[t_{ij}'=t_{ij}]$ be the marginal probability of $t_{ij}$ when $t_i\sim \widehat{D}_i$. We first prove the following lemma using Property (1) and Property (2). 

\begin{lemma}\label{lem:f_hat and f} 
For every $t_{ij}$ such that $f_{ij}(t_{ij})\geq \eps$, we have $\frac{\widehat{f}_{ij}(t_{ij})}{f_{ij}(t_{ij})}\in [1-\delta,1+\delta]$.
\end{lemma}
\begin{proof}

For any $t_i$, consider the following vector $(\pi_{i}^{(t_i)},w_i^{(t_i)})$ such that $\pi_{ij}^{(t_i)}(t_{ij}')=\ind[t_{ij}'=t_{ij}\wedge f_{ij}(t_{ij})\geq \eps], w_{ij}^{(t_i)}(t_{ij}')=0$. By \Cref{def:welfare frequent-t_i}, we have $(\pi_{i}^{(t_i)},w_i^{(t_i)})\in W_{t_i}^{tr(\eps)}$ by considering the set of numbers $\{\sigma_S^{(k)}\}_{S,k}$ such that 
\argyrisnote{$\sigma_S^{(k)}=\ind[S=[j]\wedge k=\argmax_{k'}\alpha_{ij}^{(k')}(t_{ij})]$}. Now let $(\pi_{i},w_i)=\sum_{t_{i}\in \cT_i}f_i(t_i)\cdot (\pi_{i}^{(t_i)},w_i^{(t_i)})$. 
Since $W_{i}^{tr(\eps)}$ is a mixture of $W_{t_i}^{tr(\eps)}$ over $t_i\sim D_i$,
then $(\pi_{i},w_i)\in W_{i}^{tr(\eps)}$. Moreover, $w_{ij}(t_{ij})=0$ and $$\pi_{ij}(t_{ij})=\sum_{t_i'\in \cT_i}f_i(t_i')\cdot \ind[t_{ij}'=t_{ij}\wedge f_{ij}(t_{ij}')\geq \eps]=f_{ij}(t_{ij})\cdot \ind[f_{ij}(t_{ij})\geq \eps]$$  

Similarly, let $(\widehat\pi_{i},\widehat w_i)=\sum_{t_{i}\in \cT_i}\Pr_{t_i'\sim \widehat{D}_i}[t_i'=t_i]\cdot (\pi_{i}^{(t_i)},w_i^{(t_i)})$.
Then $(\widehat\pi_{i},\widehat w_i)\in \widehat{W}_{i}^{tr(\eps)}$. Moreover, $\widehat w_{ij}(t_{ij})=0$ and $$\widehat\pi_{ij}(t_{ij})=\sum_{t_i'\in \cT_i}\Pr_{x\sim \widehat{D}_i}[x=t_i']\cdot \ind[t_{ij}'=t_{ij}\wedge f_{ij}(t_{ij}')\geq \eps]=\widehat f_{ij}(t_{ij})\cdot \ind[f_{ij}(t_{ij})\geq \eps]$$   

Now by Property (2), there exists $(\pi_{i}',w_i')\in  \widehat{W}_{i}^{tr(\eps)}$, such that $||\pi_i-\pi_i'||_{\infty}\leq \delta \eps$. Thus for $t_{ij}$ such that $f_{ij}(t_{ij})\geq \eps$, we have 
$$(1-\delta)f_{ij}(t_{ij})\leq f_{ij}(t_{ij})-\delta \eps =\pi_{ij}(t_{ij})-\delta \eps \leq  \pi_{ij}'(t_{ij}) \leq \widehat{f}_{ij}(t_{ij})$$

the last inequality is because $\widehat{W}_{i}^{tr(\eps)}$ is a mixture of $W_{t_i}^{tr(\eps)}$ over $t_i\sim \widehat D_i$,
which implies  
$$\pi_{ij}'(t_{ij})\leq \sum_{t_{i'}\in \cT_i}\Pr_{x\sim \widehat{D}_i}[x=t_i']\cdot \ind[t_{ij}=t_{ij}']=\widehat{f}_{ij}(t_{ij}).$$

On the other hand, again by Property (1), there exists $(\pi_{i}'',w_i'')\in W_{i}^{tr(\eps)}$, such that $||\widehat\pi_i-\pi_i''||_{\infty}<\delta \eps$. Thus for $t_{ij}$ such that $f_{ij}(t_{ij})\geq \eps$, we have 
$$\widehat f_{ij}(t_{ij})=\widehat\pi_{ij}(t_{ij})\leq \pi_{ij}''(t_{ij})+\delta \eps\leq (1+\delta)\cdot f_{ij}(t_{ij})$$
the last inequality uses the fact that $\pi_{ij}''(t_{ij})\leq f_{ij}(t_{ij})$ and $f_{ij}(t_{ij})\geq \eps$.

\end{proof}


Back to the proof of \Cref{thm:multi-approx-mixture}. Since $(\pi_i,w_i) \in \frac{1}{2}{W}_{i}^{tr(\eps)} + \frac{1-\eps T}{2}W^{box(\eps)}_{i}$,
we can decompose $(\pi_i,w_i) = \frac{1}{2}(\pi_i^{tr},w_i^{tr}) + \frac{1}{2}(\pi_i^{box},w_i^{box})$ such that $(\pi_i^{tr},w_i^{tr}) \in W^{tr(\eps)}_{i}$ and $(\pi_i^{box},w_i^{box}) \in (1-\eps T)W_i^{box(\eps)}$.
For every $j\in [m]$, we partition $\cT_{ij}$ into four disjoint sets: 
\begin{align*}
S_{j}^{(1)}=&\{t_{ij}\in \cT_{ij} : f_{ij}(t_{ij}) \geq \eps \land \eps < w_{ij}(t_{ij})\leq \pi_{ij}(t_{ij})  \} \\
S_j^{(2)}=&\{t_{ij}\in \cT_{ij} : f_{ij}(t_{ij}) \geq \eps \land w_{ij}(t_{ij})\leq \eps < \pi_{ij}(t_{ij})  \} \\
S_j^{(3)}=&\{t_{ij}\in \cT_{ij} : f_{ij}(t_{ij}) \geq \eps \land w_{ij}(t_{ij})\leq \pi_{ij}(t_{ij}) \leq \eps \}\\
S_j^{(4)}=&\{t_{ij}\in \cT_{ij} : f_{ij}(t_{ij}) <\eps\}
\end{align*}

By Property (2), there exists a $(\widehat{\pi}^{tr}_{i},\widehat{w}^{tr}_i)\in \widehat{W}^{tr(\eps)}_{i}$ such that for all $j\in[m]$ and $t_{ij} \in \cT_{ij}$:
\begin{align*}
\pi^{tr}_{ij}(t_{ij}) - \delta\eps\leq \widehat{\pi}^{tr}_{ij}(t_{ij}) \leq \pi^{tr}_{ij}(t_{ij}) + \delta\eps \\
w^{tr}_{ij}(t_{ij}) - \delta\eps\leq \widehat{w}^{tr}_{ij}(t_{ij}) \leq w^{tr}_{ij}(t_{ij}) + \delta\eps
\end{align*}

Let $a=1-\eps\cdot T>0$. Now consider the following vector $(\widetilde{\pi}^{tr}_i,\widetilde{w}^{tr}_i) \in [0,1]^{2M_i}$:
\begin{align*}
\widetilde{\pi}^{tr}_{ij}(t_{ij}) =&
\begin{cases}
\widehat{\pi}_{ij}^{tr}(t_{ij}) \quad & \text{ if $t_{ij}\in S_j^{(1)}$} \\
\min\left( \widehat{f}_{ij}(t_{ij}) ,a\cdot(\pi_{ij}(t_{ij}) - w_{ij}(t_{ij}))\right)  \quad & \text{ if $t_{ij}\in S_j^{(2)}$} \\
0 \quad & \text{ o.w.}
\end{cases}\\
\widetilde{w}^{tr}_{ij}(t_{ij}) =&
\begin{cases}
\widehat{w}_{ij}^{tr}(t_{ij}) \quad & \text{ if $t_{ij}\in S_j^{(1)}$} \\
0 \quad & \text{ o.w.}
\end{cases}\\
\end{align*}

\begin{lemma}
$(\widetilde{\pi}^{tr}_i,\widetilde{w}^{tr}_i) \in \widehat{W}_{i}^{tr(\eps)}$.
\end{lemma}
\begin{proof}
Since $(\widehat{\pi}^{tr}_{i},\widehat{w}^{tr}_i)\in \widehat{W}_{i}^{tr(\eps)}$, for every $t_i$, there exists $(\widehat{\pi}^{(t_i)},\widehat{w}^{(t_i)})\in {W}_{t_i}^{tr(\eps)}$, such that $(\widehat{\pi}^{tr}_{i},\widehat{w}^{tr}_i)=\sum_{t_i\in \cT_i}\Pr_{t_i'\sim \widehat{D}_i}[t_i'=t_i]\cdot (\widehat{\pi}^{(t_i)},\widehat{w}^{(t_i)})$. Define $(\widetilde{\pi}^{(t_i)},\widetilde{w}^{(t_i)})$ as follows:

\begin{align*}
\widetilde{\pi}^{(t_i)}_j(t_{ij}') =&
\begin{cases}
\widehat{\pi}^{(t_i)}_j(t_{ij}) \quad & \text{ if $t_{ij}=t_{ij}'\wedge t_{ij}\in S_j^{(1)}$} \\
\min\left(1 ,a\cdot(\pi_{ij}(t_{ij}) - w_{ij}(t_{ij}))/\widehat{f}_{ij}(t_{ij})\right)  \quad & \text{ if $t_{ij}=t_{ij}'\wedge t_{ij}\in S_j^{(2)}$} \\
0 \quad & \text{ o.w.}
\end{cases}\\
\widetilde{w}^{(t_i)}_{j}(t_{ij}') =&
\begin{cases}
\widehat{w}_{j}^{(t_i)}(t_{ij}) \quad & \text{ if $t_{ij}=t_{ij}'\wedge t_{ij}\in S_j^{(1)}$} \\
0 \quad & \text{ o.w.}
\end{cases}\\
\end{align*}

By \Cref{cor:W_t_i truncate} we have $(\widetilde{\pi}^{(t_i)},\widetilde{w}^{(t_i)})\in {W}_{t_i}^{tr(\eps)}$. Since $\widehat{W}_{i}^{tr(\eps)}$ is a mixture of ${W}_{t_i}^{tr(\eps)}$ over distribution $\widehat{D}_i$, we have that $(\widetilde{\pi}^{tr}_{i},\widetilde{w}^{tr}_i)=\sum_{t_i\in \cT_i}\Pr_{t_i'\sim \widehat{D}_i}[t_i'=t_i]\cdot (\widetilde{\pi}^{(t_i)},\widetilde{w}^{(t_i)})$ is in $\widehat{W}_{i}^{tr(\eps)}$. 
\end{proof}

Now we consider another vector $(\widetilde{\pi}^{box}_i,\widetilde{w}^{box}_i)$: 
\begin{align*}
(\widetilde{\pi}^{box}_{ij}(t_{ij}),\widetilde{w}^{box}_{ij}(t_{ij})) =&
\begin{cases}
({\pi}_{ij}^{box}(t_{ij}),w^{box}_{ij}(t_{ij})) \quad & \text{ if $t_{ij}\in S_{j}^{(1)}\cup S_{j}^{(4)}$} \\ 
(a\cdot w_{ij}(t_{ij}),a\cdot w_{ij}(t_{ij}))  \quad & \text{ if $t_{ij}\in S_{j}^{(2)}$} \\
(a\cdot{\pi}_{ij}(t_{ij}),a\cdot w_{ij}(t_{ij})) \quad & \text{ if $t_{ij}\in S_{j}^{(3)}$} \\
\end{cases}
\end{align*}

For every $j,t_{ij}$, clearly $\widetilde{\pi}_{ij}^{box}(t_{ij})\geq \widetilde{w}_{ij}^{box}(t_{ij})$ since $\pi_{ij}^{box}(t_{ij})\geq w_{ij}^{box}(t_{ij})$ and $\pi_{ij}(t_{ij})\geq w_{ij}(t_{ij})$. Moreover, if $t_{ij}\in S_j^{(1)}\cup S_j^{(4)}$, $\widetilde{\pi}_{ij}^{box}(t_{ij})=\pi_{ij}^{box}(t_{ij})\leq a\cdot \min\{\eps,f_{ij}(t_{ij})\}$.
If $t_{ij}\in S_j^{(2)}\cup S_j^{(3)}$, by definition of $S_j^{(2)}$ and $S_j^{(3)}$, we have that $\widetilde\pi_{ij}^{box}(t_{ij})\leq a\cdot\eps=a\cdot \min\{\eps,f_{ij}(t_{ij})\}$.
Thus $(\widetilde{\pi}^{box}_i,\widetilde{w}^{box}_i)\in a {W}_{i}^{box(\eps)}$.
Now define $(\widetilde{\pi},\widetilde{w})=\frac{1}{2}(\widetilde{\pi}^{tr}_i,\widetilde{w}^f_i) + \frac{1}{2}(\widetilde{\pi}^{box}_i,\widetilde{w}^r_i)$.
Then $(\widetilde{\pi}_i,\widetilde{w}_i)\in \widehat{W}_i$.
We remain to prove that $(\widetilde{\pi}_i,\widetilde{w}_i)$ satisfies:

$$\frac{\widetilde{\pi}_{ij}(t_{ij})}{\pi_{ij}(t_{ij})}\in [\frac{1-\delta}{2},1+\frac{\delta}{2}],\quad \frac{\widetilde{w}_{ij}(t_{ij})}{w_{ij}(t_{ij})}\in [\frac{1-\delta}{2},1+\frac{\delta}{2}]$$

We will verify all cases depending on which set $t_{ij}$ is in. First observe that for $t_{ij}\in S_j^{(1)}$, $\eps<w_{ij}(t_{ij}) \leq \pi_{ij}(t_{ij})$. We have $$-\frac{\delta}{2}\cdot \pi_{ij}(t_{ij})\leq-\frac{\eps \delta}{2}\leq\widetilde\pi_{ij}(t_{ij})-\pi_{ij}(t_{ij})=\frac{1}{2}(\hat\pi_{ij}^{tr}(t_{ij})-\pi_{ij}^{tr}(t_{ij}))\leq\frac{\eps\delta}{2}\leq \frac{\delta}{2}\cdot \pi_{ij}(t_{ij})$$
Similarly, 
$$-\frac{\delta}{2}\cdot w_{ij}(t_{ij})\leq-\frac{\eps \delta}{2}\leq\widetilde {w}_{ij}(t_{ij})-w_{ij}(t_{ij})=\frac{1}{2}(\hat{w}_{ij}^f(t_{ij})-w_{ij}^f(t_{ij}))\leq\frac{\eps \delta}{2}\leq \frac{\delta}{2}\cdot w_{ij}(t_{ij})$$

For $t_{ij}\in S_j^{(2)}$,
we have
$
\widetilde{w}_{ij}(t_{ij}) = \frac{a}{2}\cdot {w}_{ij}(t_{ij})
$.
For $\widetilde{\pi}_{ij}(t_{ij})$, when $\widehat{f}_{ij}(t_{ij}) \geq a\cdot(\pi_{ij}(t_{ij}) - w_{ij}(t_{ij}))$,
$
\widetilde{\pi}_{ij}(t_{ij}) = \frac{a}{2}\cdot {\pi}_{ij}(t_{ij})
$.

Now we deal with the case where $\widehat{f}_{ij}(t_{ij}) < a\cdot(\pi_{ij}(t_{ij}) - w_{ij}(t_{ij}))$.
By Lemma~\ref{lem:xos mrf contain}, since $(\pi_i,w_i) \in \frac{1}{2}{W}_{i}^{tr(\eps)} + \frac{1-\eps T}{2}W^{box(\eps)}_{i}\subseteq W_i$, we have $\pi_{ij}(t_{ij})-w_{ij}(t_{ij})\leq \pi_{ij}(t_{ij}) \leq f_{ij}(t_{ij})$. Also by \Cref{lem:f_hat and f}, $\widehat{f}_{ij}(t_{ij}) \geq (1-\delta)f_{ij}(t_{ij})$. Thus if $\widehat{f}_{ij}(t_{ij}) < a\cdot(\pi_{ij}(t_{ij}) - w_{ij}(t_{ij}))$,
then
$$
\left(1-\delta\right)\left(\pi_{ij}(t_{ij}) - w_{ij}(t_{ij}) \right)\leq \left(1-\delta \right)f_{ij}(t_{ij})\leq \widehat{f}_{ij}(t_{ij})
$$


Thus
\begin{align*}
    \widetilde{\pi}_{ij}(t_{ij}) = \frac{1}{2} \left(\widehat{f}_{ij}(t_{ij}) + a\cdot w_{ij}(t_{ij})\right) \geq \frac{1}{2} \left(1-\delta\right)(\pi_{ij}(t_{ij}) - w_{ij}(t_{ij})) + \frac{a}{2}w_{ij}(t_{ij}) \geq \frac{1-\delta}{2} \pi_{ij}(t_{ij}),
\end{align*}
where the last inequality uses the fact that $a\geq 1-\delta$. On the other hand,
\begin{align*}
    \widetilde{\pi}_{ij}(t_{ij})<\frac{a}{2} (\pi_{ij}(t_{ij}) - w_{ij}(t_{ij}) + w_{ij}(t_{ij})) \leq \frac{a}{2} \pi_{ij}(t_{ij})
\end{align*}

For $t_{ij}\in S_j^{(3)}$, we have $\widetilde{\pi}_{ij}(t_{ij}) = \frac{a}{2}\cdot {\pi}_{ij}(t_{ij})$ and $\widetilde{w}_{ij}(t_{ij}) = \frac{a}{2}\cdot {w}_{ij}(t_{ij})$. 
For $t_{ij}\in S_j^{(4)}$, since $f_{ij}(t_{ij})<\eps$, $w_{ij}^f(t_{ij})=\pi_{ij}^f(t_{ij})=0$. Thus $\widetilde{\pi}_{ij}(t_{ij})=\frac{1}{2}\pi_{ij}^r(t_{ij})=\pi_{ij}(t_{ij})$, and $w_{ij}(t_{ij})=\frac{1}{2}w_{ij}^r(t_{ij})=w_{ij}(t_{ij})$. 

Now we are left to prove that we can have access to an efficient separation oracle for $\widehat{W}_i$.
Similar to the proof of Theorem~\ref{thm:special case of multiplicative approx},
$\widehat{W}_i$ is the Minkowski sum of polytope $W_i^{tr(\eps)}$ and $n$ polytopes in $\{W_{t_i}^{tr(\eps)}\}_{t_i\in \cT_i}$.
With Lemma~\ref{lem:optimize} and using arguments similar to the proof of Theorem~\ref{thm:special case of multiplicative approx},
we can argue that we can construct an efficient seperation oracle for $\widehat{W}_i$.
\end{proof}
}


\subsection{Putting Everything Together}


In this section, we put all pieces together and provide a complete proof of \Cref{thm:main XOS}.
Denote $(P)$ the LP in \Cref{fig:XOSLP} and $\optlp$ the optimal objective of $(P)$. We consider another LP denoted as $(P')$. In $(P')$, in addition to all variables in $(P)$, we introduce new variables $\widehat{\pi}_i=\{\widehat{\pi}_{ij}(t_{ij})\}_{j\in[m],t_{ij}\in\cT_{ij}}$ and $\widehat{w}_i=\{\widehat{w}_{ij}(t_{ij})\}_{j\in[m],t_{ij}\in\cT_{ij}}$ for every $i\in [n]$. Both $(P)$ and $(P')$ have the same objective function. The only difference between $(P)$ and $(P')$ is that in $(P')$, we replace Constraint {\Wconstraint} with the following constraints:
$$\text{Constraint {\WconstraintNew}}: \quad (\widehat{\pi}_i,\widehat{w}_i) \in \widehat{W}_i, \pi_i \geq \frac{3}{2}\widehat{\pi}_i\geq \bm{0}, w_i \leq \frac{1}{4}\widehat{w}_i,\quad\forall i\in [n].$$ Here $\widehat{W}_i$ is the proxy polytope from \Cref{cor:rfs for xos mult} for each $i\in [n]$. Both inequalities hold coordinate-wisely. Denote $\optlp'$ the optimal objective of $(P')$.     
By Property 4 of \Cref{cor:rfs for xos mult}, there exists an efficient separation oracle for each $\widehat{W}_i$. Thus we can solve $(P')$ in polynomial time using the Ellipsoid algorithm (\Cref{thm:ellipsoid}). The following lemma shows the relationship between $(P)$ and $(P')$. 

\notshow{
\begin{figure}[H]
\colorbox{MyGray}{
\begin{minipage}{.98\textwidth}
$$\quad\textbf{max  } \sum_{i\in[n]} \sum_{j\in[m]} \sum_{t_{ij}\in \cT_{ij}} 
 f_{ij}(t_{ij})\cdot V_{ij}(t_{ij})\cdot \sum_{\substack{\beta_{ij}\in \cV_{ij}\\
 \delta_{ij} \in \Delta}} \lambda_{ij}(t_{ij},\beta_{ij}, \delta_{ij})\cdot \ind[V_{ij}(t_{ij})\leq \beta_{ij} + \delta_{ij}]$$
\vspace{-.3in}
  \begin{align*}
 & s.t.\\
 &\quad\textbf{Allocation Feasibility Constraints:}\\
 &\quad\Wconstraint \quad (\widehat{\pi}_i,\widehat{w}_i) \in \widehat{W}_i ,\quad \pi_i \geq \frac{3}{2}\widehat{\pi}_i,w_i \leq \frac{1}{4}\widehat{w}_i & \forall i \\
&\quad \PiConstraint \quad \sum_i \sum_{t_{ij}\in \cT_{ij}} \pi_{ij}(t_{ij})\leq 1 
& \forall j \\  
 &\quad\textbf{Natural Feasibility Constraints:}\\
    &\quad\LambdaMarginalConstraint\quad f_{ij}(t_{ij})\cdot\sum_{\beta_{ij}\in \cV_{ij}}\sum_{\delta_{ij}\in \Delta} \lambda_{ij}(t_{ij},\beta_{ij},\delta_{ij}) = w_{ij}(t_{ij}) & \forall i,j,t_{ij}\in \cT_{ij}\\
    &\quad\CompareMarginalConstraint\quad\lambda_{ij}(t_{ij},\beta_{ij},\delta_{ij})\leq \hat\lambda_{ij}(\beta_{ij}, \delta_{ij}) & \forall i,j, t_{ij},\beta_{ij}\in \cV_{ij},\delta_{ij}\\
    &\quad\HatLambdaDistributionConstraint\quad \sum_{\substack{\beta_{ij}\in \cV_{ij}\\ \delta_{ij} \in \Delta}} \hat\lambda_{ij}(\beta_{ij},\delta_{ij}) = 1 & \forall i,j\\
    &\quad\textbf{Problem Specific Constraints:}\\
    &\quad\ReduceDemandConstaint \quad \sum_{i\in[n]}\sum_{\beta_{ij} \in \cV_{ij}}\sum_{\delta_{ij}\in \Delta} \hat\lambda_{ij}(\beta_{ij},\delta_{ij}) \cdot \Pr_{t_{ij}\sim D_{ij}}[V_{ij}(t_{ij})\geq \beta_{ij}] \leq \frac{1}{2}& \forall j\\
    & \quad\MarginalToGlobalConstraint\quad \frac{1}{2}\sum_{t_{ij} \in \cT_{ij}} f_{ij}(t_{ij}) \left(\lambda_{ij}(t_{ij},\beta_{ij},\delta_{ij})+ \lambda_{ij}(t_{ij},\beta_{ij}^+,\delta_{ij})\right) \leq \\
    &\qquad\qquad\qquad\hat\lambda_{ij}(\beta_{ij},\delta_{ij}) \cdot \Pr_{t_{ij}}[t_{ij}\geq \beta_{ij}]+\hat\lambda_{ij}(\beta_{ij}^+,\delta_{ij}) \cdot \Pr_{t_{ij}}[t_{ij}\geq\beta_{ij}^+] & \forall i,j,\beta_{ij}\in \cV_{ij}^0,\delta_{ij}\in \Delta\\
    &\quad\BoundMeanDeltaConstraint\quad \sum_{\substack{\beta_{ij}\in \cV_{ij}\\ \delta_{ij} \in \Delta}} \delta_{ij}\cdot \hat\lambda_{ij}(\beta_{ij},\delta_{ij}) \leq d_i  & \forall i,j\\
    &\quad\BoundSumDeltaConstraint\quad\sum_{i\in[n]} d_i \leq 
    {111\cdot \estprev}\\
     & \lambda_{ij}(t_{ij},\beta_{ij},\delta_{ij})\geq 0,\hat\lambda_{ij}(\beta_{ij},\delta_{ij})\geq 0,\pi_{ij}(t_{ij})\geq 0, w_{ij}(t_{ij})\geq 0, d_i\geq 0 & \forall i,j,t_{ij},\beta_{ij}\in \cV_{ij},\delta_{ij}
\end{align*}
\end{minipage}}
\caption{Efficiently Computable LP for XOS Valuations}~\label{fig:computable XOS}.

\end{figure}

}

\notshow{
\begin{theorem}\label{thm:XOS mrf}
Let $T=\sum_{i\in[n]}\sum_{j\in[m]}|\cT_{ij}|$ and $b$ be the bit complexity of the problem instance.
Consider the LP in Figure~\ref{fig:computable XOS},
where for each $i\in[n]$, polytope $\widehat{W}_i$ is defined in \Cref{cor:rfs for xos mult}. 
We condition on the event that each proxy polytope $\widehat{W}_i$ satisfies the properties in \Cref{cor:rfs for xos mult}. Then
\begin{enumerate}
\item For any $(\pi,w,\hat\pi,\hat w,\lambda,\hat\lambda, \bd)$  feasible solution of the LP in \Cref{fig:computable XOS},
there exists $\widetilde{\pi}\in [0,1]^{\sum_{i\in[n]}\sum_{j\in[m]}|\cT_{ij}|}$ such that $(\widetilde{\pi},w,\lambda,\hat\lambda, \bd)$ is a feasible solution of the LP in Figure~\ref{fig:XOSLP}.
\item Let $\opt$ be the value of the optimal solution of the LP in Figure~\ref{fig:XOSLP} and $\widehat{\opt}$ be the value of the optimal solution of the LP in Figure~\ref{fig:computable XOS}. Then $\widehat{\opt}\geq \frac{\opt}{64}$
\end{enumerate}
\end{theorem}
}

\begin{lemma}\label{lem:XOS-P to P'}
Suppose for every $i\in [n]$, $\widehat{W}_i$ satisfies the properties in \Cref{cor:rfs for xos mult}. Then
\begin{itemize}
\item For any feasible solution $(\pi,w,\hat\pi,\hat w,\lambda,\hat\lambda, \bd)$ to $(P')$,
there exists $\widetilde{\pi}\in [0,1]^{\sum_{i,j}|\cT_{ij}|}$ such that $(\widetilde{\pi},w,\lambda,\hat\lambda, \bd)$ is a feasible solution to $(P)$.
\item $\optlp\leq 64\cdot \optlp'$.
\end{itemize}
\end{lemma}

\begin{proof}







We prove the first part of the statement. Let $(\pi,w,\hat\pi,\hat w,\lambda,\hat\lambda, \bd)$ be any feasible solution to $(P')$. Then for every $i\in[n]$, $(\widehat{\pi}_i,\widehat{w}_i)\in \widehat{W}_i$. By Property 1 of \Cref{cor:rfs for xos mult}, there exists a $(\widetilde{\pi}_i,\widetilde{w}_i)\in W_i$ such that for each $j\in[m]$ and $t_{ij}\in \cT_{ij}$,
$$\widetilde{\pi}_{ij}(t_{ij}) \leq \frac{3}{2} \widehat{\pi}_{ij}(t_{ij})\leq{\pi}_{ij}(t_{ij}) \quad \widetilde{w}_{ij}(t_{ij}) \geq \frac{1}{4} \widehat{w}_{ij}(t_{ij}) \geq{w}_{ij}(t_{ij}) $$

Let $\widetilde{\pi}=\{\widetilde{\pi}_i\}_{i\in [n]}$. We are going to show that $(\widetilde\pi,w,\lambda,\hat\lambda, \bd)$ is a feasible solution to $(P)$ by verifying all constraints. 
For Constraint~$\Wconstraint$, since $(\widetilde{\pi}_i,\widetilde{w}_i)\in W_i$ and $w_{ij}(t_{ij})\leq \widetilde{w}_{ij}(t_{ij}),\forall j,t_{ij}$, thus by \Cref{def:W_i} $(\widetilde{\pi}_i,w_i)\in W_i$.



For Constraint {\PiConstraint}, since $(\pi,w,\hat\pi,\hat w,\lambda,\hat\lambda, \bd)$ is a feasible solution to $(P')$, for every $j\in [m]$,  $$\sum_i\sum_{t_{ij}\in \cT_{ij}}\widetilde{\pi}_{ij}(t_{ij})\leq \sum_i\sum_{t_{ij}\in \cT_{ij}}\pi_{ij}(t_{ij})\leq 1.$$ Furthermore, Constraints~$\LambdaMarginalConstraint-\BoundSumDeltaConstraint$ are clearly satisfied since $(\pi,w,\hat\pi,\hat w,\lambda,\hat\lambda, \bd)$ is a feasible solution to $(P')$.
Thus $(\widetilde{\pi},w,\lambda,\hat\lambda, \bd)$ is a feasible solution to $(P)$.

Now we prove the second part of the statement.
Let $(\pi^*,w^*,\lambda^*,\hat\lambda^*, \bd^*)$ be any optimal feasible solution to $(P)$. 
For every buyer $i$, $(\pi_i^*,w_i^*)\in W_i$. By Property 2 of \Cref{cor:rfs for xos mult}, there exists a $(\widehat{\pi}_i,\widehat{w}_i)\in \widehat{W}_i$ such that for every $j,t_{ij}$, 
$$\widehat{\pi}_{ij}(t_{ij})\leq \frac{3}{8}\pi_{ij}^*(t_{ij}),\quad \widehat{w}_{ij}(t_{ij})\geq \frac{1}{16}w_{ij}^*(t_{ij})$$
We are going to show that $\left(\pi=\frac{3}{2}\widehat{\pi},w=\frac{1}{64}w^*,\hat\pi,\hat w,\lambda=\frac{1}{64}\lambda^*,\hat\lambda=\hat\lambda^*,\bd= \bd^*\right)$ is a feasible solution to $(P')$, which implies that $\optlp'\geq \optlp/64$. Firstly, for each $i\in[n]$, $(\widehat{\pi}_i,\widehat{w}_i)\in \widehat{W}_i$ and 
$$
\pi_i = \frac{3}{2}\widehat{\pi}_i, \qquad w_i = \frac{1}{64}w_i^* \leq \frac{1}{4}\widehat{w}_i 
$$
Thus Constraint {\WconstraintNew} is satisfied. For Constraint~$\PiConstraint$, since $(\pi^*,w^*,\lambda^*,\hat\lambda^*, \bd^*)$ is a feasible solution to $(P)$, we have that for every $j$, 
$$
\sum_{i\in[n]} \sum_{t_{ij}\in \cT_{ij}}\pi_{ij}(t_{ij})=\sum_{i\in[n]} \sum_{t_{ij}\in \cT_{ij}}\frac{3}{2}\hat\pi_{ij}(t_{ij}) \leq \sum_{i\in[n]} \sum_{t_{ij}\in \cT_{ij}}\frac{9}{16}\pi^*_{ij}(t_{ij}) \leq \frac{9}{16}<1
$$

One can easily verify that when $(\pi^*,w^*,\lambda^*,\hat\lambda^*, \bd^*)$ is a feasible solution to $(P)$, 
$\left(\frac{1}{64}\pi^*,\frac{1}{64}w^*,\frac{1}{64}\lambda^*,\hat\lambda^*, \bd^*\right)$ is also a feasible solution to $(P)$, 
which implies that $\left(\pi,w,\hat\pi,\hat w,\lambda,\hat\lambda,\bd\right)$ satisfies Constraints~$\LambdaMarginalConstraint-\BoundSumDeltaConstraint$. Thus $\left(\pi,w,\hat\pi,\hat w,\lambda,\hat\lambda,\bd \right)$ is a feasible solution to $(P')$. The objective of the solution is a $\frac{1}{64}$-fraction of the objective of $(\pi^*,w^*,\lambda^*,\hat\lambda^*, \bd^*)$ since $\lambda=\frac{1}{64}\lambda^*$, which concludes the proof.  

\end{proof}

Now we are ready to give the proof of \Cref{thm:main XOS}.

\begin{prevproof}{Theorem}{thm:main XOS}

{We consider a fixed $\delta \in(0,1)$.
For each $i\in[n]$,
let $\widehat{W}_i$ be the proxy polytope for the single-bidder marginal reduced form polytope $W_i$ that is constructed in \Cref{cor:rfs for xos mult} with parameter $\delta'=\frac{\delta}{n}$.
Let $\cE_1$ be the event that the RPP mechanism  computed in \Cref{thm:chms10} has revenue $\estprev=\Omega(\prev)$, and let $\cE_2$  be the event that for each $i\in[n]$,
the proxy polytope $\widehat{W}_i$, satisfies the properties in \Cref{cor:rfs for xos mult}.
By the union bound combined with \Cref{thm:chms10} and \Cref{cor:rfs for xos mult}, with probability at least $1-\delta-\frac{2}{nm}$, both events happen.
}


We condition on the event that both $\cE_1$ and $\cE_2$ happens. By \Cref{thm:ellipsoid} and Property 3 and 4 of \Cref{cor:rfs for xos mult}, there exists an algorithm that solves the LP $(P')$ in time $\poly(n,m,T,b,\log (1/\delta))$, given access to the adjustable demand oracle and value oracle for all buyers' valuations. Let $(\pi^*,w^*,\hat\pi^*,\hat w^*,\lambda^*,\hat\lambda^*, \bd^*)$ be an optimal solution to $(P')$. 
By \Cref{thm:chms10}, \Cref{lem:XOS-P to P'} and \Cref{lem:bound rev by opt-XOS}, $\opt \leq 28 \prev + 4 \opt_{LP}\leq \frac{nm}{nm-1}189\cdot \estprev + 256\cdot\opt_{LP}'$.

\Cref{lem:XOS-P to P'} also guarantees the existence of $\widetilde{\pi}\in [0,1]^{\sum_{i,j}|\cT_{ij}|}$ such that $(\widetilde{\pi},w^*,\lambda^*,\hat\lambda^*, \bd^*)$ is a feasible solution to $(P)$. Although we do not know the value $\widetilde{\pi}$, it turns out sufficient to know $\lambda^*$ to compute a simple and approximately optimal mechanism. We compute the prices $\{Q^*_j\}_{j\in[m]}$ using $\lambda^*$ according to \Cref{def:Q_j-XOS}. In particular,
$$Q^*_j = \frac{1}{2}\cdot\sum_{i\in[n]}\sum_{t_{ij}\in \cT_{ij}}  f_{ij}(t_{ij})\cdot V_{ij}(t_{ij})\cdot
    \sum_{\substack{\beta_{ij}\in \cV_{ij}\\\delta_{ij} \in \Delta}}\lambda^*_{ij}(t_{ij},\beta_{ij},\delta_{ij})\cdot\ind[V_{ij}(t_{ij})\leq \beta_{ij}+\delta_{ij}],
$$

and $$2\sum_{j\in[m]} Q^*_j=\opt_{LP}'.$$ According to \Cref{thm:bounding-lp-simple-mech-XOS}, we can construct a two-part tariff mechanism $\Mtpt$ with prices $\{Q^*_j\}_{j\in[m]}$ and a rationed posted price mechanism $\Mpp$ (computed in \Cref{thm:chms10}) in time $\poly(n,m,T)$, such that $$\opt_{LP}'=2\sum_{j\in[m]} Q^*_j= O(\rev(\Mtpt))+O(\estprev).$$

To sum up, we can compute in time $\poly(n,m,T,b,\log (1/\delta))$ a two-part tariff mechanism $\Mtpt$ and a rationed posted price mechanism $\Mpp$, such that $\opt\leq c_1\cdot\rev(\Mpp) + c_2\cdot \rev(\Mtpt)$ for some absolute constants $c_1,c_2>0$ with probability at least $1-\delta-\frac{2}{nm}$.
\end{prevproof}

\section{Missing Details from \Cref{sec:sample access}}\label{appx:sample access}

\begin{definition}\label{def:kolmogorov}
The Kolmogorov distance between two distributions $\cD$ and $\widehat{\cD}$ supported on $\mathbb{R}$ is defined as
$$
d_K(\cD,\widehat{\cD}) = \sup_{z\in \mathbb{R}}\left| \Pr_{t\sim \cD}[t\leq z] -\Pr_{\widehat{t}\sim \widehat{\cD}}[\widehat{t}\leq z]\right|
$$
\end{definition}

We need the following robustness result from Cai and Daskalakis~\cite{CaiD17}.
\begin{theorem}[Theorem~3 in \cite{CaiD17}]\label{thm:robust mechanism}
Suppose all bidders' valuations are constrained additive. Let $\cM$ be a Sequential Posted Price with Entry Fee Mechanism (as defined in Mechanism~\ref{def:general SPEM}~\footnote{Indeed the result holds for any buyers' order, we choose the lexicographical order to keep the notation light.}) whose prices are $\{p_{ij}\}_{i\in[n], j\in[m]}$ and its entry fee function for the $i$-th bidder is $\xi_i:2^{[m]}\rightarrow\mathbb{R}_+$.
If for each $i\in[n]$ and $j\in[m]$ we have $d_K(\cD_{ij},\widehat{\cD}_{ij})\leq \eps$, and $\cD_{i,j}$ and $\widehat{\cD}_{i,j}$ are both supported on  $[0,1]$ (that is when each bidder's value for a single item is at most $1$), then
$$
|\rev(\cM,\cD) - \rev(\cM,\widehat{\cD})| \leq 4 n m^2 \eps,
$$ where $\rev(\cM,\cD)$ and $\rev(\cM,\widehat{\cD})$ are the revenues of $\cM$ under $\cD=\bigtimes_{i,j} \cD_{ij}$ and  $\widehat{\cD}=\bigtimes_{i,j} \widehat{\cD}_{ij}$  respectively.
\end{theorem}

\begin{algorithm}[H]
\floatname{algorithm}{Mechanism}
\begin{algorithmic}[1]
\setcounter{ALG@line}{-1}
\State Before the mechanism starts, the seller decides on a collection of $\{p_{ij}\}_{i\in[n],j\in[m]}$ and a collection of  entry fee functions $\{\xi_i(\cdot)\}_{i\in [n]}$, where $\xi_i:2^{[m]}\rightarrow\mathbb{R}_+$ is buyer $i$'s entry fee.  
\State Bidders arrive sequentially in the lexicographical order.
\State When buyer $i$ arrive, the seller shows her the set of available items $S\subseteq [m]$, as well as their prices $\{p_{ij}\}_{j\in S}$ and asks buyer $i$ to pay an \emph{entry fee} of $\xi_i(S)$. Note that $S$ is the set of items that are not purchased by the first $i-1$ buyers.
\If{Bidder $i$ pays the entry fee ${\delta_i}(S)$}
        \State $i$ receives her favorite bundle $S_i^{*}$ and pays $\sum_{j\in S_i^{*}}p_{ij}$.
        \State $S\gets S\backslash S_i^{*}$.
    \Else
        \State $i$ gets nothing and pays $0$.
    \EndIf
\end{algorithmic}
\caption{{\sf \quad Sequential Posted Price with Entry Fee Mechanism (SPEM)}}\label{def:general SPEM}
\end{algorithm}

\begin{lemma}\label{lem:revenue close kolmogorov}
Suppose all bidders' valuations are constrained additive. If for each $i\in[n]$ and $j\in[m]$ we have $d_K(\cD_{ij},\widehat{\cD}_{ij})\leq \eps$, and $\cD_{i,j}$ and $\widehat{\cD}_{i,j}$ are both supported on  $[0,1]$ (that is when each bidder's value for a single item is at most $1$), then 
$$ c_1\cdot \opt(\widehat{\cD})- O(nm^2\eps)\leq \opt(\cD)\leq \frac{\opt(\widehat{\cD})}{c_2}+ O(nm^2\eps)$$ for some absolute constant $c_1, c_2>0$, where 
$\opt({\cD})$ (or $\opt(\widehat{\cD})$) is the optimal revenue for distribution $\cD=\bigtimes_{i,j} \cD_{ij}$ (or $\widehat{\cD}=\bigtimes_{i,j} \widehat{\cD}_{ij}$).
\end{lemma}

\begin{proof}

Let $M_1$, $M_2$ be the optimal RPP and TPT for $\cD$, and we denote their revenue as $\rev(M_1)$, $\rev(M_2)$ respectively. Let $M_3$, $M_4$ be the optimal RPP and TPT for $\widehat{\cD}$, and we denote their revenue as $\rev(M_3)$, $\rev(M_4)$ respectively. Since any TPT is a SPEM, by Theorem~\ref{thm:robust mechanism}, we know that $|\rev(M_2)-\rev(M_4)|\leq 4nm^2\varepsilon$, as $$\rev(M_4)\geq \rev(M_2,\widehat{\cD})\geq \rev(M_2,\cD)-4nm^2\varepsilon=\rev(M_2)-4nm^2\varepsilon,$$ and $$\rev(M_2)\geq \rev(M_4,{\cD})\geq \rev(M_4,\widehat{\cD})-4nm^2\varepsilon=\rev(M_4)-4nm^2\varepsilon.$$
Similarly, since any RPP is also a SPEM if we treat the buyers' valuation as Unit-Demand, we have $|\rev(M_1)-\rev(M_3)|\leq 4nm^2\varepsilon$. By \Cref{lem:caiz17-constrained-additive}, $\max\{\rev(M_1),\rev(M_2)\}\geq \Omega(\opt(\cD))$. Hence, $\opt(\widehat{\cD})\geq \max\{\rev(M_3),\rev(M_4)\}\geq \max\{\rev(M_1),\rev(M_2)\}-4nm^2\varepsilon\geq  \Omega(\opt(\cD))-4nm^2\varepsilon$. Similarly, $\opt({\cD})\geq \Omega(\opt(\widehat{\cD}))-4nm^2\varepsilon$.

\end{proof}

\begin{prevproof}{Theorem}{thm:sample access}
For each $D_{ij}$, we first take $O\left(\frac{\log(nm/\delta)}{\varepsilon^2}\right)$ samples. Let $\widehat{D}_{ij}$ be the uniform distribution over the samples. By the union bound and the DKW inequality~\cite{DvoretzkyKW56}, $d_K(D_{ij},\widehat{D}_{ij})\leq \eps$ for all $i\in[n]$ and $j\in[m]$ with probability at least $1-\delta$. Now we run the algorithm from \Cref{thm:main XOS-main body} on $\bigtimes_{i,j}\widehat{D}_{ij}$ and let $\cM$ be the computed mechanism. Note that $\rev(\cM,\widehat{D})$, the revenue of $\cM$ under $\widehat{D}=\bigtimes_{i,j} \widehat{D}_{ij}$, is $\Omega(\opt(\widehat{D}))$ -- the optimal revenue for distribution $\widehat{D}$. By Theorem~\ref{thm:robust mechanism}, $\rev(\cM,{D})\geq \rev(\cM,\widehat{D})-O(nm^2\eps)$. By Lemma~\ref{lem:revenue close kolmogorov}, $\opt(\widehat{D})\geq\Omega(\opt({D}))-O(nm^2\eps)$. Chaining the ineuqalities together, we have that $\rev(\cM,{D})\geq \Omega(\opt({D}))-O(nm^2\eps)$.
\end{prevproof}
\section{Counterexample for Adjustable Demand Oracle}\label{sec:example_adjustable_oracle}


For XOS valuations, our algorithm for constructing the simple mechanism requires access to a special adjustable demand oracle $\adem_i(\cdot,\cdot,\cdot)$. Readers may wonder if this enhanced oracle (rather than a demand oracle) is necessary to prove our result. In this section we show that (even an approximation of) $\adem_i$ can not be implemented using polynomial number of queries from the value oracle, demand oracle and a classic XOS oracle. All the oracles are defined as follows. Throughout this section, we only consider a single buyer and thus drop the subscript $i$. 
Recall that the XOS valuation $v(\cdot)$ satisfies that $v(S)=\max_{k\in [K]}\left\{\sum_{j\in S}\alpha_j^{(k)}\right\}$ for every set $S$, where $\{\alpha_j^{(k)}\}_{j\in [m]}$ is the $k$-th additive function.
\begin{itemize}
    \item Demand Oracle ($\dem$): takes a price vector $p\in \mathbb{R}^m$ as input, and outputs\\ $S^*\in \argmax_{S\subseteq [m]} \left(v(S) - \sum_{j\in S}p_j\right)$.
    \item XOS Oracle (\xos): takes a set $S\subseteq [m]$ as input, and outputs the $k^*$-th additive function $\{\alpha_j^{(k^*)}\}_{j\in [m]}$, where $k^*\in \argmax_{k\in [K]}\left\{\sum_{j\in S}\alpha_j^{(k)}\right\}$.
\item Value Oracle: takes a set $S\subseteq [m]$ as input, and outputs $v(S)$. We notice that a value oracle can be easily simulated with an XOS oracle. Thus we focus on XOS oracles for the rest of this section.
    \item Adjustable Demand Oracle (\adem): takes a coefficient vector $b\in \mathbb{R}^m$ and a price vector $p\in \mathbb{R}^m$ as inputs, and outputs $(S^*,\{\alpha_j^{(k^*)}\}_{j\in [m]})$ where $(S^*,k^*) \in \argmax_{S\subseteq [m],k\in[K]} \left\{\sum_{j\in S}b_j \alpha_{j}^{(k)}- \sum_{j\in S}p_j\right\}$. 
\end{itemize}

An (approximate) implementation of $\adem$ is an algorithm that takes inputs $b,p\in \mathbb{R}^m$, and outputs a set $S\subseteq [m]$ and $k\in [K]$. The algorithm has access to the demand oracle and XOS oracle of $v$. We denote $\alg(v,b,p)$ the output of the algorithm. For any $\alpha>1$, $\alg$ is an $\alpha$-approximation to $\adem$ if for every XOS valuation $v$ and every $b,p\in \mathbb{R}^m$, the algorithm outputs $(S',k')$ that satisfies:
$$\max_{S\subseteq [m],k\in[K]} \left\{\sum_{j\in S}b_j \alpha_{j}^{(k)}- \sum_{j\in S}p_j\right\}\leq \alpha\cdot \left(\sum_{j\in S}b_j \alpha_{j}^{(k')}- \sum_{j\in S'}p_j\right)$$

In Theorem~\ref{lem:example_adjustable_oracle} we show that we cannot approximate the output of an Adjustable Demand Oracle within any finite factor,
if we are permitted to query polynomial many times the XOS, Value and Demand Oracle.

\begin{theorem}\label{lem:example_adjustable_oracle}
Given any $\alpha>1$, there \textbf{does not exist} an implementation of $\adem$ (denoted as $\alg$) that satisfies both of the following properties:
\begin{enumerate}
    \item For any XOS valuation $v$ over $m$ items, $\alg$ makes $\poly(m)$ queries to the value oracle, the demand oracle and XOS oracle of $v$, and runs in time $\poly(m, b)$. Here $b$ is the bit complexity of the problem instance (See Definition~\ref{def:bit complexity}).
    \item $\alg$ is an $\alpha$-approximation to $\adem$.
\end{enumerate}
\end{theorem}

\begin{proof}
Recall that a value oracle can be easily simulated with an XOS oracle. Thus we only argue for demand queries and XOS queries. For the sake of contradiction, assume there exists such an algorithm $\alg$. Let $\ell>e^{2\alpha}$ be an arbitrary even integer. Let $m=2\ell$.
Denote by $L = {\ell \choose \ell/2}$.
We decompose the items into two sets $S_1 = \left\{1,\ldots,\ell \right\}$ and $S_2=\left\{\ell+1,\ldots,m \right\}$.

We consider an XOS valuation $\widehat{v}$ generated by $2\ell$ additive functions denoted by $\{\widehat{\alpha}_j^{(k)}\}_{j\in[m],k\in[2\ell]}$ parameterized by variables $\eps',\eps$ such that $0 < (\ell+1)\eps' < \eps < \frac{1}{2}$. For $k=1,\ldots,\ell$ and $j\in [m]$:
\notshow{
$$
\widehat{\alpha}_{j}^{(k)}
=
\begin{cases}
j + \eps+ \left(1 - \frac{1}{\ell} \right) \eps' \quad &\text{if $j=k$} \\
0 &\text{if $j\in S_1$ and $j\not=k$} \\
0 \quad &\text{if $j\in S_2$} \\
\end{cases}
$$
}

{$$
\widehat{\alpha}_{j}^{(k)}
=
\begin{cases}
j + \eps+ \left(1 - \frac{1}{\ell} \right) \eps' \quad &\text{if $j=k$} \\
0 &\text{otherwise} \\
\end{cases}
$$}

For $k=\ell + 1,\ldots,2\ell$ and $j\in [m]$, define
$$
\widehat{\alpha}_{j}^{(k)}
=
\begin{cases}
j \quad &\text{if $j=k-\ell$} \\
0 \quad &\text{if $j\in S_1$ and $j\not=k-\ell$}\\
\frac{2}{\ell}\cdot\eps \quad &\text{if $j\in S_2$} \\
\end{cases}
$$

Next, we introduce a family of XOS valuations $\{v_r\}_{r\in [L]}$ over $m$ items. For every $r$, the valuation $v_r$ is generated by $K=2\ell + 1$ additive functions, denoted as $\{\alpha_{r, j}^{(k)}\}_{j\in [m]}, \forall k\in [K]$. We define all the additive functions as follows (and hence $v_r$ is defined as $v_r(S)=\max_{k\in [K]}\sum_{j\in S}\alpha_{r,j}^{(k)}, \forall S\subseteq [m]$):

\paragraph{For $k\in[2\ell]$:}  For every $r\in [L]$ and $j\in[m]$, define
$$
\alpha_{r, j}^{(k)}
=
\widehat{\alpha}_{j}^{(k)}
$$
\paragraph{For $k=2\ell + 1$:} Take any bijective mapping $\mathcal{C}$ between $[L]$ and subsets of $S_2$ of size $\frac{\ell}{2}$. 
For every $r\in [L]$, define
$$
\alpha_{r, j}^{(2\ell+1)}
=
\begin{cases}
1\quad &\text{if $j\in S_1$ and $j\not=\ell$}\\
1+\eps \quad &\text{if $j = \ell$}\\
\frac{2}{\ell}\cdot\eps' \quad &\text{if $j\in \mathcal{C}(r)$} \\
0 \quad &\text{if $j\in S_2\backslash \mathcal{C}(r)$} 
\end{cases}
$$

In the following lemmas, we prove that given access to both the XOS and Demand oracle, the algorithm can not distinguish between valuations $v_r$ and $\widehat{v}$, unless the algorithm knows $\mathcal{C}(r)$ (or $r$).


\begin{lemma}\label{lem:v_r val}
For any $r\in [L]$ and any nonempty set $S\subseteq [m]$,
$
\sum_{j\in S}\alpha_{r,j}^{(2\ell+1)} > \sum_{j\in S}\alpha_{r,j}^{(k)} 
$ for all $k\in [2\ell]$ if and only if $S=S_1\cup \mathcal{C}(r)$.
Hence for any $r\in [\ell]$ and any nonempty set $S\subseteq [m]$,
$$\xos(v_r,S)=
\begin{cases}
\{\alpha_{r,j}^{(2\ell+1)}\}_{j\in[m]} \quad & \text{if $S=S_1\cup \mathcal{C}(r)$ 
}\\
\xos(\widehat{v},S) \quad & \text{otherwise}
\end{cases}$$
\end{lemma}

\begin{proof}
Let $j^* = \max\{j: j\in S_1\cap S\}$ (define it to be 0 if $S\cap S_1=\emptyset$). We consider the following cases:

\begin{itemize}
\item Case 1: $S\cap S_1=\emptyset$. Then $S\subseteq S_2$. For any $k^*=\ell + 1,\ldots,2\ell$,
$$
\sum_{j\in S} \alpha_{r,j}^{(k^*)} = \sum_{j\in S}  \frac{2}{\ell}\eps
> \sum_{j\in S} \frac{2}{\ell}\eps'\geq \sum_{j\in S\cap \mathcal{C}(r)} \frac{2}{\ell}\eps'= \sum_{j\in S} \alpha_{r,j}^{(2\ell+1)}
$$

\item Case 2: $S\cap S_2=\emptyset$. Then $S\subseteq S_1$ and $S\cap S_1 \neq \emptyset$, which implies that $j^* > 0$. We have
$$
\sum_{j\in S} \alpha_{r,j}^{(j^*)} = j^* + \eps+ \left(1 - \frac{1}{\ell} \right) \eps'
> j^* + \eps
\geq \sum_{j=1}^{j^*} \alpha_{r,j}^{(2\ell+1)}
\geq 
\sum_{j\in S} \alpha_{r,j}^{(2\ell+1)}
$$
Here the second last inequality follows from the fact that $S\subseteq S_1$ and $\alpha_{r,j}^{(2\ell+1)}=1$ for every $j\in S_1/\{\ell\}$ and $\alpha_{r,\ell}^{(2\ell+1)}=1+\eps$. The last inequality follows from the definition of $j^*$. 

\item Case 3: $S_1\not\subseteq S$, $S\cap S_1\neq \emptyset$ and $S\cap S_2 \neq \emptyset$. We have that

$$
\sum_{j\in S} \alpha_{r,j}^{(\ell+j^*)} = j^* + \sum_{j\in S_2\cap S}\frac{2}{\ell}\eps 
>
j^* +  \sum_{j\in S_2\cap S}\frac{2}{\ell}\eps'
\geq 
\sum_{j\in S} \alpha_{r,j}^{(2\ell+1)}
$$

The last inequality follows from the fact that $\sum_{j\in S\cap S_1}\alpha_{r,j}^{(2\ell+1)}\leq j^*$: If $j^*\not=\ell$, then $\alpha_{r,j}^{(2\ell+1)}=1,\forall j\in S\cap S_1$; If $j^*=\ell$, since $S_1\not\subseteq S$, $\sum_{j\in S\cap S_1}\alpha_{r,j}^{(2\ell+1)}=|S\cap S_1|+\eps\leq (\ell-1)+\eps<\ell$. 

\item Case 4: $S_1\subseteq S$ and $S\cap S_2 \neq \mathcal{C}(r)$. We have

$$
\sum_{j\in S} \alpha_{r,j}^{(\ell)} = \alpha_{r,\ell}^{(\ell)}=
\ell+ \eps+ \left(1 - \frac{1}{\ell} \right) \eps'
>
\left(\ell + \eps\right) +  \sum_{j\in \mathcal{C}(r)\cap S}\frac{2}{\ell}\eps'
= 
\sum_{j\in S} \alpha_{r,j}^{(2\ell+1)}
$$
Here the inequality is because: $\mathcal{C}(r)\cap S \neq \mathcal{C}(r)$, then $|\mathcal{C}(r)\cap S| \leq \frac{\ell}{2} - 1$,
which implies that $\sum_{j\in \mathcal{C}(r)\cap S}\frac{2}{\ell}\eps' \leq \left(\frac{\ell}{2}-1\right)\frac{2}{\ell}\eps' < \left(1-\frac{1}{\ell}\right)\eps'$.

\item Case 5: $S_1\subseteq S$ and $S\cap S_2 = \mathcal{C}(r)$. Recall that $|\mathcal{C}(r)|=\frac{\ell}{2} $. We notice that $\sum_{j\in S}  \alpha_{r,j}^{(2\ell+1)} = \ell + \eps + \eps'
$. On the other hand,
$$
\max_{k\in [\ell]} \sum_{j\in S}  \alpha_{r,j}^{(k)} 
= \sum_{j\in S} \alpha_{r,j}^{(\ell)}= \ell + \eps+ \left(1 - \frac{1}{\ell} \right) \eps' < \sum_{j\in S}  \alpha_{r,j}^{(2\ell+1)} 
$$

and

$$
\max_{\ell< k\leq 2\ell} \sum_{j\in S}  \alpha_{r,j}^{(k)} 
= \sum_{j\in S} \alpha_{r,j}^{(2\ell)}
= \ell + \sum_{j\in \mathcal{C}(r)} \alpha_{r,j}^{(2\ell)}= \ell + \eps < \sum_{j\in S}  \alpha_{r,j}^{(2\ell+1)} 
$$
\end{itemize}

We have proved the first part of the statement. The second part of the statement then follows by noticing that the first $2\ell$ additive functions of $v_r$ and $\widehat{v}$ are exactly the same.  
\end{proof}



\begin{lemma}\label{lem:output-Demand-oracle}
For any $r\in [L]$ and any set of prices $\bm{p} \in \mathbb{R}^m_{\geq 0}$ such that $\{j\in S_2: p_j \leq \frac{2}{\ell}\eps'\}\neq \mathcal{C}(r)$, we have $\dem(v_r,\bm{p})= \dem(\widehat{v},\bm{p})$.
\end{lemma}

\begin{proof}
Recall that

$$ \dem(v_r,\bm{p}) \in \argmax_{S\subseteq [m]} \left\{v_r(S) - \sum_{j\in S} p_j\right\} = \argmax_{S\subseteq [m]}\left\{\max_{k \in [K]}  \sum_{j\in S} (\alpha_{r,j}^{(k)} - p_j)\right\}$$

We notice that valuations $v_r$ and $\widehat{v}$ differ only in the $(2\ell+1)$-th additive valuation (with coefficients $\{\alpha_{r,j}^{(2\ell+1)}\}_{j\in[m]}$). Thus by \Cref{lem:v_r val}, $\dem(v_r,\bm{p})= \dem(\widehat{v},\bm{p})$ unless $S^*=S_1\cup \mathcal{C}(r)$ is the favorite bundle for valuation $v_r$ at price $\bm{p}$, i.e. 
$$\sum_{j\in S_1\cup \mathcal{C}(r)}(\alpha_{r,j}^{(2\ell+1)} - p_j) = \max_{S\subseteq [m], k \in [K]}  \sum_{j\in S} (\alpha_{r,j}^{(k)} - p_j)$$

Let $S_0= \{j\in S_2 : p_j \leq \frac{2}{\ell}\eps'\}$.
Firstly, if $\mathcal{C}(r)\not\subseteq S_0$, then there exists $j\in S_2$ such that $p_j>\frac{2}{\ell}\eps'$. Since $\alpha_{r,j}^{(2\ell+1)}=\frac{2}{\ell}\eps',\forall j\in \mathcal{C}(r)$, we have  
$$\sum_{j\in S_1\cup \mathcal{C}(r)}(\alpha_{r,j}^{(2\ell+1)} - p_j)
<
\sum_{j\in S_1\cup({\mathcal{C}(r)\cap} S_0)}(\alpha_{r,j}^{(2\ell+1)} -  p_j)
\leq \max_{S\subseteq [m], k \in [K]}  \sum_{j\in S} (\alpha_{r,j}^{(k)} - p_j),$$

which implies that $\dem(v_r,\bm{p})=\dem(\widehat{v},\bm{p})$. It remains to consider the case where $\mathcal{C}(r)\subseteq S_0$ and $\mathcal{C}(r)\not=S_0$. We have
$$\sum_{j\in S_1\cup \mathcal{C}(r)}(\alpha_{r,j}^{(2\ell+1)} - p_j)
\leq \ell +\eps - \sum_{j\in S_1}p_j 
+ \sum_{j \in \mathcal{C}(r)} \alpha_{r,j}^{(2\ell+1)} 
= \ell +\eps - \sum_{j\in S_1}p_j 
+ \eps'
$$
Here the first inequality follows from $\sum_{j\in S_1}\alpha_{r,j}^{(2\ell+1)}=\ell+\eps$ and $p_j\geq 0,\forall j\in \mathcal{C}(r)$. And the equality follows from $\sum_{j \in \mathcal{C}(r)} \alpha_{r,j}^{(2\ell+1)}= |\mathcal{C}(r)|\cdot \frac{2}{\ell}\eps'=\ell/2\cdot \frac{2}{\ell}\eps'=\eps'$. 
On the other hand, 

\begin{align*}
\sum_{j\in S_1\cup S_0}(\alpha_{r,j}^{(2\ell)} -  p_j)
&= \ell  - \sum_{j\in S_1}p_j 
+ \sum_{j \in S_0} \left(\alpha_{r,j}^{(2\ell)} - p_j\right) \\
&\geq \ell  - \sum_{j\in S_1}p_j 
+ \sum_{j \in S_0} \left(\frac{2}{\ell}\eps - \frac{2}{\ell}\eps'\right) \\
&\geq  \ell  - \sum_{j\in S_1}p_j 
+  \left( \frac{\ell}{2}+1\right)\left(\frac{2}{\ell}\eps - \frac{2}{\ell}\eps'\right)\\
&=  \ell  - \sum_{j\in S_1}p_j 
+ \eps -\eps' + \frac{2}{\ell}\left( \eps - \eps'\right) \\
&> \ell  - \sum_{j\in S_1}p_j 
+ \eps -\eps' + \frac{2}{\ell}\left( (\ell+1)\eps' - \eps'\right) \\
&=  \ell + \eps + \eps' - \sum_{j\in S_1}p_j
\end{align*}
Here the first inequality follows from $p_j \leq \frac{2}{\ell}\eps'$, $\alpha_{r,j}^{(2\ell)}=\frac{2}{\ell}\eps$ for $j\in {S_0}$. The second inequality follows from $|S_0| \geq \frac{\ell}{2}+1$, since $|\mathcal{C}(r)|=\frac{\ell}{2}$, $\mathcal{C}(r)\subseteq S_0$, and $\mathcal{C}(r)\not=S_0$. The third inequality follows from our choice of $\eps$ and $\eps'$ such that $\eps > (\ell+1)\eps'$. Thus,
$$
\sum_{j\in S_1\cup \mathcal{C}(r)}(\alpha_{r,j}^{(2\ell+1)} - p_j) < \sum_{j\in S_1\cup S_0}(\alpha_{r,j}^{(2\ell)} - p_j) \leq \max_{S\subseteq [m], k \in [K]}  \sum_{j\in S} (\alpha_{r,j}^{(k)} - p_j)
$$
which implies that $\dem(v_r,\bm{p})=\dem(\widehat{v},\bm{p})$.
\end{proof}

To complete the proof of \Cref{lem:example_adjustable_oracle}, set $\eps = 0.1$ and $\eps' = \frac{0.1}{2(\ell+1)}$. We notice that the bit complexity of our input is $b=O(\poly(\ell))$ for any valuation $v_r$.
Now consider the coefficient vector $\bm{c}=(1, 1/2,\ldots, 1/\ell, 0, \ldots, 0)$ with price vector $\bm{p}=\textbf{0}$. For any valuation $v_r$ in the family, clearly $\adem$ will select the whole set $[m]$ since all coefficients are non-negative. 

For any $k\in [2\ell]$, $\sum_{j\in [m]}c_j\alpha_{r,j}^{(k)}\leq 1 + \eps + \left(1-\frac{1}{\ell}\right)\eps'<2$. 
On the other hand, when $k=2\ell+1$, since $\ell>e^{2\alpha}$ we have
$$\sum_{j\in [m]}c_j\alpha_{r,j}^{(2\ell+1)}\geq\sum_{j\in [\ell]}\frac{1}{j}\cdot 1>\log(\ell)>2\alpha>\alpha\cdot \max_{k\not=2\ell+1}\left\{\sum_{j\in [m]}c_j\alpha_{r,j}^{(k)}\right\}=\alpha\cdot \max_{k}\left\{\sum_{j\in [m]}c_j\widehat{\alpha}_{j}^{(k)}\right\}$$




Thus informally speaking, to obtain an $\alpha$-approximation to $\adem$ for every valuation $v_r$, the algorithm needs to distinguish $v_r$ from $\widehat{v}$ by identifying the $(2\ell+1)$-th additive function. 
By Lemmas~\ref{lem:v_r val} and \ref{lem:output-Demand-oracle}, the algorithm must be able to identify $r$ or $\mathcal{C}(r)$. However, there are $L = {\ell \choose \ell/2}$ different valuations in the family $\{v_r\}_{r\in [L]}$. Thus there must exist one $v_r$ such that the algorithm can not distinguish $v_r$ from $\widehat{v}$ within $poly(\ell)=poly(m,b)$ queries to both oracles. 

Formally, consider the execution of $\alg$ on valuation $\widehat{v}$. Let $Q$ be the total number of queries during the execution and define set $S^{(1)},\ldots,S^{(Q)}\subseteq [m]$ as follows: For every $q\in [Q]$, if the $q$-th query is an XOS query, define $S^{(q)}$ as the input set to the XOS oracle; If it's a demand query, let $S^{(q)}=\{j\in S_2: p_j^{(q)} \leq \frac{2}{\ell}\eps'\}$, where $\bm{p}^{(q)}=\{p_j^{(q)}\}_{j\in [m]}$ is the input price vector to the demand oracle.  
Recall that $m=2\ell$ and $b=\poly(\ell)$, we have $Q=\poly(m, b)<L = {\ell \choose \ell/2}$ for sufficiently large $\ell$. Thus there exists some $r^*\in [L]$ such that $\mathcal{C}(r^*)\not=S_2\cap S^{(q)}$ 
for any $q\in [Q]$.
By \Cref{lem:v_r val}, we have that $\xos(\widehat{v},S^{(q)})=\xos(v_r,S^{(q)})$ and by \Cref{lem:output-Demand-oracle} we have that $\dem(\widehat{v},\bm{p}^{(q)}) = \dem(v_r,\bm{p}^{(q)})$ for any $q\in [Q]$. 
This implies that the execution (and thus output) of $\alg$ on input valuation $\widehat{v}$ is exactly the same as its execution on $v_{r^*}$.  

We notice that from the above calculation, $\max_{k}\left\{\sum_{j\in [m]}c_j\alpha_{r^*,j}^{(k)}\right\}\geq\sum_{j\in [\ell]}\frac{1}{j}>\log(\ell)>2\alpha$, while
$\max_{k}\left\{\sum_{j\in [m]}c_j\widehat{\alpha}_{j}^{(k)}\right\}=\max_{k\not=2\ell+1}\left\{\sum_{j\in [m]}c_j\alpha_{r^*,j}^{(k)}\right\}<2$. Thus on input $(v_{r^*},\bm{c})$, $\alg$ achieves less than a $\frac{1}{\alpha}$-fraction of the optimal objective from $\adem$, contradicting with the fact that $\alg$ is an $\alpha$-approximation to $\adem$. 
\end{proof}


\newpage
\bibliographystyle{alpha}
\bibliography{Yang.bib}

\end{document}